\documentclass[11pt]{amsart}

\usepackage[T1]{fontenc}
\usepackage[utf8]{inputenc}
\usepackage{lmodern}
\usepackage{microtype}
\usepackage{amsmath,amssymb,amsfonts,amsthm,mathtools}
\usepackage{bm}
\usepackage{mathrsfs}
\usepackage{enumitem}
\usepackage{graphicx}
\usepackage{subcaption}
\usepackage{xcolor}
\usepackage{cite}
\usepackage{hyperref}
\usepackage{aliascnt}
\usepackage[nameinlink,noabbrev]{cleveref}
\setlist[itemize]{leftmargin=2.0em,itemsep=0.15em,topsep=0.25em,parsep=0pt}
\setlist[enumerate]{leftmargin=2.0em,itemsep=0.15em,topsep=0.25em,parsep=0pt}
\numberwithin{equation}{section}

\DeclareMathOperator{\TV}{TV}

\DeclareMathOperator{\dist}{dist}
\DeclareMathOperator{\sgn}{sgn}

\theoremstyle{plain}

\newtheorem{thm}{Theorem}[section]

\newaliascnt{prop}{thm}
\newtheorem{prop}[prop]{Proposition}
\aliascntresetthe{prop}

\newaliascnt{lem}{thm}
\newtheorem{lem}[lem]{Lemma}
\aliascntresetthe{lem}

\newaliascnt{cor}{thm}
\newtheorem{cor}[cor]{Corollary}
\aliascntresetthe{cor}

\theoremstyle{definition}

\newaliascnt{defn}{thm}
\newtheorem{defn}[defn]{Definition}
\aliascntresetthe{defn}

\newaliascnt{assump}{thm}

\aliascntresetthe{assump}

\newaliascnt{example}{thm}
\newtheorem{example}[example]{Example}
\aliascntresetthe{example}

\theoremstyle{remark}

\newaliascnt{rem}{thm}
\newtheorem{rem}[rem]{Remark}
\aliascntresetthe{rem}

\crefname{thm}{Theorem}{Theorems}
\Crefname{thm}{Theorem}{Theorems}

\crefname{prop}{Proposition}{Propositions}
\Crefname{prop}{Proposition}{Propositions}

\crefname{lem}{Lemma}{Lemmas}
\Crefname{lem}{Lemma}{Lemmas}

\crefname{cor}{Corollary}{Corollaries}
\Crefname{cor}{Corollary}{Corollaries}

\crefname{defn}{Definition}{Definitions}
\Crefname{defn}{Definition}{Definitions}

\crefname{assump}{Assumption}{Assumptions}
\Crefname{assump}{Assumption}{Assumptions}

\crefname{example}{Example}{Examples}
\Crefname{example}{Example}{Examples}

\crefname{rem}{Remark}{Remarks}
\Crefname{rem}{Remark}{Remarks}
\newcommand{\C}{\mathbb{C}}
\newcommand{\R}{\mathbb{R}}

\newcommand{\NF}{\mathrm{NF}}
\newcommand{\FR}{\mathrm{Fr}}

\newcommand{\aFR}{a_{\FR}}

\newcommand{\Mset}{\mathcal{M}}

\newcommand{\FNF}{\mathcal F_{\NF}}
\newcommand{\FFR}{\mathcal F_{\FR}}

\newcommand{\Qset}{\mathcal Q}
\newcommand{\Rset}{\mathcal R}
\newcommand{\Thetaset}{\Theta}

\newcommand{\mx}[1]{\mathbf{#1}}
\newcommand{\RS}{\mathrm{RS}}
\newcommand{\LCS}{\mathrm{LCS}}

\newcommand{\Jlc}{J^{\rm lc}}
\newcommand{\Cmax}{C_{\max}}

\newcommand{\bs}[1]{\boldsymbol{#1}}
\title[Near-field super-resolution]{A Mathematical Theory of Near-Field Super-Resolution}

\author{Sajad Daei}
\address{KTH Royal Institute of Technology, Stockholm, Sweden}
\email{sajado@kth.se}

\author{G\'abor Fodor}
\address{KTH Royal Institute of Technology, Stockholm, Sweden,
and Ericsson Research, Stockholm, Sweden}
\email{gaborf@kth.se}
\author{Mikael Skoglund}
\address{KTH Royal Institute of Technology, Stockholm, Sweden}
\email{skoglund@kth.se}
\begin{document}

\begin{abstract}
Finite-aperture near-field sensing leads to a super-resolution geometry
fundamentally different from the translation-invariant Fourier setting.  In the Fresnel regime, wavefront curvature introduces a range-dependent
quadratic aperture phase, so distinguishability is governed by incomplete
quadratic exponential sums of the form
\(
        \sum_{n=0}^{N_r-1}
        a_n e^{i(\omega_1 n+\omega_2 n^2)}
\),
rather than by angular separation alone.  We develop a
deterministic recovery theory for sparse measures with ranges on a finite grid
and continuous angles, and introduce a support-uniform quadratic-phase aperture
criterion replacing classical minimum separation.  Under this criterion,
total-variation minimization exactly recovers every sparse measure in the
admissible support class, uniformly over all nonzero complex amplitudes.  The
proof develops nonasymptotic support-uniform bounds for finite quadratic sums
and combines them with a gauged Hermite dual certificate controlling
interpolation, local curvature, and off-support leakage.  We further construct
a finite-harmonic Bessel-Vandermonde lift with explicit truncation error.  In
the far-field limit, the quadratic phase disappears and the theory reduces to
Fourier-type angular super-resolution.
\end{abstract}

\maketitle

% \tableofcontents

% ============================================================
% ============================================================
% ============================================================
\section{Introduction}
\label{sec:introduction}
% ============================================================

\subsection{What can a finite aperture distinguish?}

A finite aperture does not observe a scene directly.  It observes only the finite trace that the
scene casts on its sensors.
In that trace, geometry can lose its identity.  Distinct configurations of
point sources may become indistinguishable, not because they coincide in space,
but because they coincide at the aperture.  Thus the inverse problem begins
before reconstruction: it begins with distinguishability.  The fundamental
question is not what the scene is, but whether the aperture can tell one scene
from another.

This question is common to many forms of coherent sensing and inverse problems:
diffraction-limited optical imaging, Fourier optics, and Fresnel diffraction
\cite{Goodman2005,BornWolf1999}, coherent wavefield imaging and inverse source
problems \cite{Devaney2012}, inverse acoustic and electromagnetic scattering
\cite{ColtonKress2013}, array processing and direction-of-arrival estimation
\cite{VanTrees2002,JohnsonDudgeon1993,Schmidt1986,RoyKailath1989}, radar and
sonar localization \cite{VanTrees2002,JohnsonDudgeon1993}, spectral estimation
and line-spectrum recovery
\cite{TangBhaskarShahRecht2013,BhaskarTangRecht2013}, and off-the-grid sparse
recovery \cite{Donoho1992,CandesFernandezGranda2014,
deCastroGamboa2012,DuvalPeyre2015}.  In sparse super-resolution, the scene is
a finite collection of point sources, and the aim is to recover this collection
from measurements whose nominal resolution is limited by the aperture.

In the far field, this question has a classical answer.  Sources are
sufficiently distant that their wavefronts are approximately planar across the
aperture.  The aperture response is Fourier-like, and the interaction between
two atoms is governed by a translation-invariant low-pass kernel.  Exact
recovery by total-variation minimization is certified by a dual polynomial that
interpolates the source signs and remains strictly bounded by one away from
the support.  In the Fourier setting, a standard sufficient hypothesis for exact
recovery is minimum separation in the angular spatial-frequency variable
\cite{Donoho1992,CandesFernandezGranda2014}.
For this introductory comparison, let \(N_r\) denote the number of aperture
samples, let \(d>0\) denote the inter-sensor spacing, let \(\lambda>0\) denote
the wavelength, and set
\(
        k_\lambda:=\frac{2\pi}{\lambda}.
\)
The comparison with Fourier super-resolution is structural rather than
literal.  In the standard finite-band Fourier super-resolution model, one
observes finitely many Fourier coefficients of an unknown measure,
\[
        y_k
        =
        \int e^{-ikx}\,d\mu(x),
        ~
        |k|\le k_{\max},
\]
where \(k_{\max}\in\mathbb N\) is the largest observed Fourier index.
Thus the observation domain is a finite frequency band.  In the present array
problem, the measurements are instead spatial aperture samples,
\[
        y_n=\sum_i\int a(r_i,\theta)[n]\,d\mu_i(\theta),
        ~ n=0,\ldots,N_r-1 .
\]
In the far-field approximation the steering vector has the trigonometric form
\(
        a_{\rm FF}(\theta)[n]
        =
        e^{ink_\lambda d\cos\theta}.
\)
Thus the sensor index \(n\) plays the algebraic role of a Fourier index, while
the physical limitation is aperture length rather than bandwidth.
\begin{figure}[t]
    \centering
    \IfFileExists{far-near.pdf}
    {\includegraphics[scale=0.5]{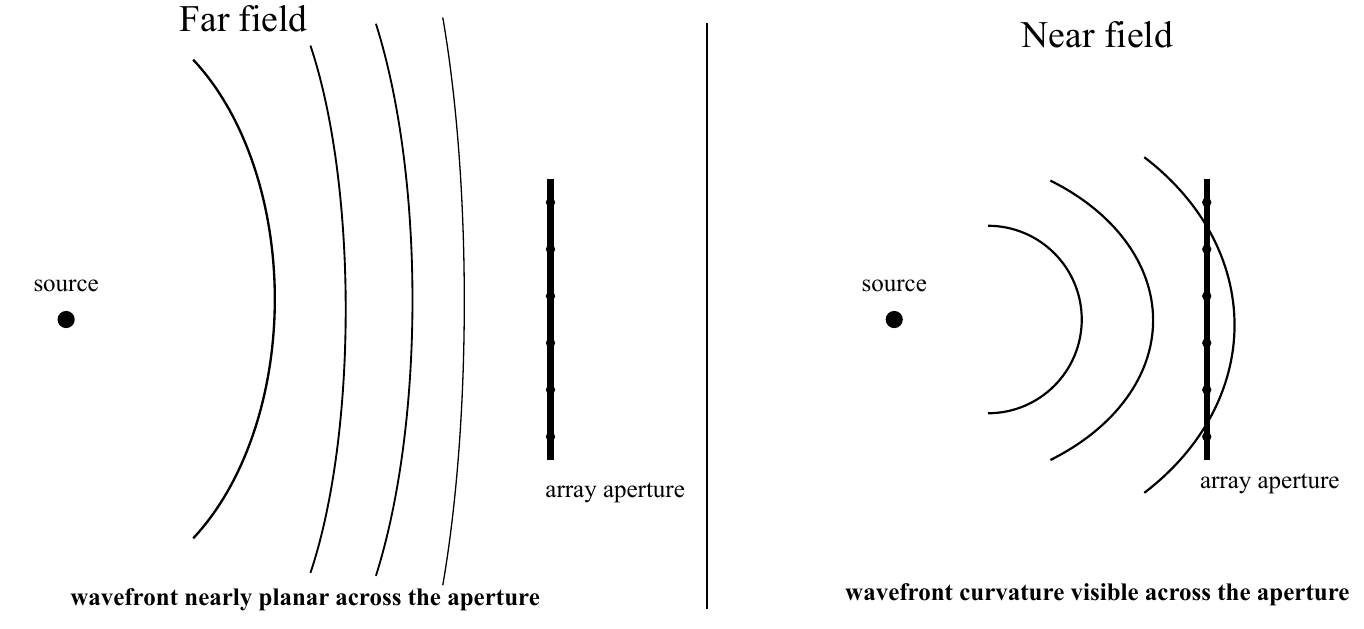}}
    {\fbox{\parbox{0.75\linewidth}{\centering Placeholder for far-near aperture figure.}}}
    \caption{Far-field and near-field aperture geometry.  In the far field,
    wavefronts are approximately planar and the response depends only on angle.
    In the Fresnel near field, wavefront curvature is visible across the
    aperture and couples angle with range.}
    \label{fig:near_far_image}
\end{figure}
Near-field sensing changes the question.  When wavefront curvature is visible
across the aperture, the response is no longer determined by angle alone.
Range and angle are coupled through curvature
\cite{Goodman2005,BornWolf1999,VanTrees2002,JohnsonDudgeon1993}; see
\Cref{fig:near_far_image}.  The kernel is
not translation invariant in the range-angle domain, and the one-dimensional Fourier separation principle
no longer describes what the aperture can distinguish.
The central problem of this paper is therefore the following:
\begin{quote}
What replaces far-field minimum separation when the sensing kernel is curved
and non-translation-invariant?
\end{quote}
Our answer is \emph{quadratic-phase aperture certification} (QPAC), a
support-uniform sufficient condition based on the finite-aperture quadratic
phase. Rather than imposing a universal scalar distance threshold, QPAC
controls the value, tangent, and higher angular derivative interactions
required by the Fresnel dual certificate through deterministic
quadratic-phase cancellation.

\subsection{Range blindness and Fresnel curvature}

The distinction already appears in the elementary narrowband spatial model.
Consider a one-dimensional aperture with samples \(n=0,\dots,N_r-1\), spacing
\(d\), wavelength \(\lambda\), and wavenumber \(k_\lambda:=\tfrac{2\pi}{\lambda}\).
In the ideal far-field model, the response of a source at range \(r\) and angle
\(\theta\) is
\cite{VanTrees2002,JohnsonDudgeon1993}
\(
        a_{\rm FF}(r,\theta)[n]
        =
        e^{i k_\lambda d n\cos\theta}.
\)
% Equivalently, with
% \(\kappa(\theta):=k_\lambda d\cos\theta\), the far-field atom is
% \(a_{\rm FF}(\theta)[n]=e^{in\kappa(\theta)}\).
Thus the far-field array response is Fourier-type in the angular spatial
frequency \(\kappa(\theta):=k_\lambda d\cos\theta\).  However, these samples
are not finite-band Fourier measurements of an unknown measure.  They are spatial measurements collected
across a finite aperture.  The finite index set \(n=0,\ldots,N_r-1\) is the
set of sensor samples, and the aperture length, not a frequency cutoff, is the
physical resolution-limiting resource.
The range variable does not appear: for any two ranges \(r_1,r_2\),
\(a_{\rm FF}(r_1,\theta)=a_{\rm FF}(r_2,\theta)\).  Thus, at a single
narrowband carrier, the ideal far-field spatial aperture is blind to range
along a fixed ray.
This is not a statement about wideband time delay or propagation-time
estimation; it is a statement about the narrowband spatial response.

The Fresnel approximation retains the first curvature correction in the
spherical-wave expansion \cite{Goodman2005,BornWolf1999}:
\[
        a_{\rm Fr}(r,\theta)[n]
        :=
        \exp\!\left(
        i k_\lambda d n\cos\theta
        -
        i\tfrac{k_\lambda d^2 n^2}{2r}\sin^2\theta
        \right).
\]
The second term is quadratic in the aperture index and depends on range.
Thus range enters the spatial data through wavefront curvature.  In this sense, the ideal far-field model forgets range, while the Fresnel
near-field model remembers it through curvature.

For an ordered pair consisting of an evaluation point
\(q=(r,\theta)\) and a source point \(p=(r',\theta')\), the relative Fresnel
interaction is
\[
        \overline{a_{\rm Fr}(q)[n]}\,a_{\rm Fr}(p)[n]
        =
        \exp\!\left(
        i\omega_1(q,p)n+i\omega_2(q,p)n^2
        \right),
\]
where
\(
        \omega_1(q,p)
        =
        k_\lambda d(\cos\theta'-\cos\theta),
\)
and
\(
        \omega_2(q,p)
        =
        -\tfrac{k_\lambda d^2}{2r'}\sin^2\theta'
        +
        \tfrac{k_\lambda d^2}{2r}\sin^2\theta .
\)
The coefficient \(\omega_1\) is the angular, or linear-phase, mismatch, while
\(\omega_2\) is the curvature mismatch.
In the far-field limit \(\omega_2\to0\), the phase becomes linear and range
disappears from the narrowband spatial response.  In the Fresnel near field,
the aperture therefore sees not only bearing but also curvature, encoded by the
ordered phase coordinates \((\omega_1,\omega_2)\). For \(p=(r,\theta)\) and \(p'=(r',\theta')\), we quantify the visual
distinction in \Cref{fig:matched_pairs} by the normalized aperture coherence
\[
       \mathcal C_{\rm model}(p,p')
:=
\frac{
\left|
\left\langle
\mathbf a_{\rm model}(p),
\mathbf a_{\rm model}(p')
\right\rangle_{\C^{N_r}}
\right|
}{
\|\mathbf a_{\rm model}(p)\|_2
\|\mathbf a_{\rm model}(p')\|_2
},
        ~
        {\rm model}\in\{{\rm FF},{\rm Fr}\},
\]
where the aperture-response vector is defined by
\[
        \mathbf a_{\rm model}(p)
        :=
        \bigl(
        a_{\rm model}(p)[0],\ldots,
        a_{\rm model}(p)[N_r-1]
        \bigr)^T
        \in\mathbb C^{N_r}.
\]
Thus \(\mathcal C_{\rm model}=1\) means that the two finite-aperture responses
are identical up to a global phase, while smaller values indicate greater
distinguishability at the aperture.  The degeneracy and its Fresnel resolution
are illustrated in \Cref{fig:matched_pairs}.

\begin{figure}[t]
    \centering
    \IfFileExists{cpam_distingushibility_fig.png}
    {\includegraphics[scale=0.3]{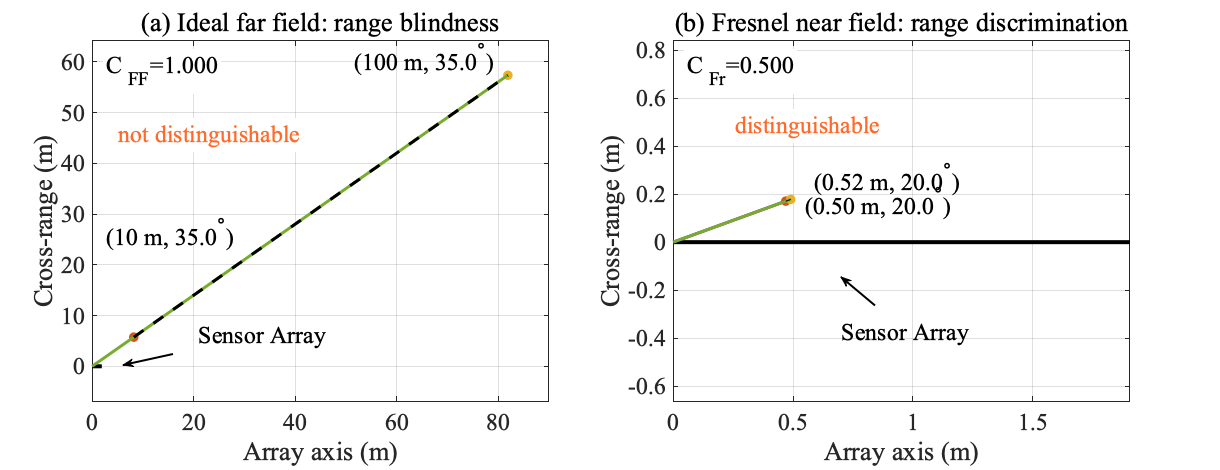}}
   
    \caption{Range blindness in the ideal far field and range discrimination in the
Fresnel near field.  The displayed quantities \(\mathcal C_{\rm FF}\) and
\(\mathcal C_{\rm Fr}\) are normalized aperture coherences between the two
plotted steering vectors.  In the ideal far-field model, same-bearing sources
at different ranges have identical spatial steering vectors, so
\(\mathcal C_{\rm FF}=1\).  In the Fresnel model, the quadratic aperture phase
depends on range; in the example shown, the coherence drops to
\(\mathcal C_{\rm Fr}=0.5\), indicating that curvature has made the two
same-bearing responses distinguishable at the finite aperture.}
    \label{fig:matched_pairs}
\end{figure}

\subsection{Why angular separation is not the right invariant}

In the far-field model, the interaction between two atoms
depends only on the difference in angular spatial frequency,
\[
\omega_1
=
k_\lambda d(\cos\theta'-\cos\theta).
\]
In the Fresnel model, the interaction also depends on the
curvature mismatch \(\omega_2\). For fixed normalized
nonnegative weights \(b_n\), it has the form
\[
\sum_{n=0}^{N_r-1}
b_n\overline{a_{\rm Fr}(q)[n]}a_{\rm Fr}(p)[n]
=
\sum_{n=0}^{N_r-1}
b_n e^{i(\omega_1(q,p)n+\omega_2(q,p)n^2)}.
\]
Thus the ungauged kernel is a difference kernel in the two
phase coordinates, although it is not translation invariant
in physical range-angle coordinates.

Angular separation alone does not control this quadratic sum.
Cancellation depends on the joint behavior of its linear and
quadratic phases across the finite aperture. Moreover, the
angular tangent and higher derivative channels depend on the
physical parameterization of the individual points.

We therefore bound the value, tangent, and derivative
interactions directly. The resulting estimates provide the
interpolation, curvature, and leakage bounds needed for a
Fresnel dual certificate.

\subsection{Quadratic-phase aperture certification}

The recovery theorem is driven by a finite oscillatory question: when does
\[
        n\mapsto \omega_1 n+\omega_2 n^2,
        ~ n=0,\ldots,N_r-1,
\]
fail to add coherently?  For the Fresnel kernel, the normalized angular
tangent channels, and the higher angular derivative channels, the relevant
objects are finite quadratic sums
\[
        \sum_{n=0}^{N_r-1}
        a_n e^{i\omega_1 n+i\omega_2 n^2}.
\]
The coefficient sequence \(a=(a_n)\) contains the fixed certificate taper
\(\rho\), through its normalized version, together with the channel-dependent
tangent or derivative multiplier.  Thus the coefficient sequence changes from
one certificate channel to another, but the taper and the phase geometry are
fixed.
QPAC combines three complementary bounds.  The derivative branch
controls genuinely nonstationary phases.  The lag-correlation branch controls
constructive alignment after squaring the sum and grouping residue lag terms.
The residue-linear branch controls near-rational curvature by reducing the
remaining oscillation on residue classes to a linear one.
QPAC is not formulated as a universal scalar minimum-separation condition.
Instead, it directly certifies that the finite Fresnel phase exhibits
sufficient cancellation in the kernel, normalized tangent, and higher angular
derivative channels entering the TV dual certificate.  Its constants form the
three recovery budgets controlling support interpolation, near-support
curvature, and far-region leakage.
\subsection{Dual certificates: interpolation, curvature, leakage}

The recovery proof follows the dual-certificate method, but the certificate is
adapted to the Fresnel geometry
\cite{CandesFernandezGranda2014,deCastroGamboa2012,
DuvalPeyre2015,AzaisDeCastroGamboa2015}.  For a support
\(S=\{p_\ell\}_{\ell=1}^L\), we construct a gauged Hermite certificate
\(P_{S,\widetilde {\mathbf v}}\) satisfying
\(
        P_{S,\widetilde {\mathbf v}}(p_\ell)=\widetilde v_\ell,
        ~
        \partial_\theta P_{S,\widetilde {\mathbf v}}(p_\ell)=0.
\)
The range coordinate is discrete in the exact Fresnel theorem, so the only
continuous local direction is angular.  The gauge makes the value atom
orthogonal to the angular tangent atom and gives a Hermite interpolation system
with identity diagonal blocks.

Strict off-support control is obtained by decomposing the parameter domain into
near and far regions.  Near a support point, stationarity and negative angular
curvature force the certificate modulus below one.  Away from the support,
QPAC estimates control the leakage generated by all support atoms.  The
certificate must do three things: interpolate on the support, curve downward
nearby, and remain small far away.

These tasks are summarized in \Cref{subsec:explicit_qpac_number_sec2} by a
QPAC recovery number combining the support-interpolation, near-support, and
far-region budgets.  The condition that this recovery number is strictly less
than one controls Hermite interpolation, Taylor near-support decay, and
far-region leakage.

The theorem is deterministic and amplitude independent.  The QPAC constants
are defined uniformly over the prescribed support class; they do not depend
on the source magnitudes or phases, or on the particular support locations
within that class.
\subsection{Main recovery result}

The precise recovery theorem is stated in
\Cref{thm:main_exact_recovery_sec2}, after the QPAC budgets have been
defined.  Its content may be summarized as follows.  Fix a finite admissible
range grid, a closed angular interval \(\Theta\subset(0,\pi)\), and a prescribed nonempty
class of \(L\)-point supports.  If the
support-uniform QPAC estimates satisfy
\[
        \eta_{\rm SS}<1,
        ~
        \eta_{\rm near}(\varpi_{\rm loc})<1,
        ~
        \eta_{\rm far}(\varpi_{\rm loc})<1
\]
for some \(\varpi_{\rm loc}>0\), then every sparse measure in the prescribed
support class is the unique minimizer of the noiseless semi-discrete Fresnel
TV problem.  The conclusion holds uniformly over all nonzero complex
amplitudes.  This is a sufficient recovery criterion; no converse is
claimed.
\subsection{Computation and the finite-harmonic lift}

The exact theorem in \Cref{thm:main_exact_recovery_sec2} concerns the
semi-discrete Fresnel TV problem and does not use the lift. For computation,
we truncate the Bessel-Vandermonde expansion of each Fresnel atom. The
resulting finite-dimensional model has an explicit uniform error bound and
admits a semidefinite formulation whose dual polynomial can be used for
localization. The lift is therefore a numerical approximation of the
Fresnel TV problem, not an additional exact-recovery theorem.

\subsection{Comparison with related work and sparse recovery}
\label{subsec:related-work}
\label{subsec:sparse-recovery-connection}

This paper belongs to the sparse-recovery field in which convex optimization
is used to recover structured sparse objects from incomplete or indirect
measurements \cite{CandesRombergTao2006}.  It is closest in spirit to sparse
super-resolution by total-variation minimization, but the sensing geometry is
different from the standard translation-invariant Fourier setting.
In the classical Fourier model, one observes finitely many low-frequency
Fourier coefficients of an unknown measure, and exact recovery by
total-variation minimization follows under a minimum-separation condition
\cite{Donoho1992,CandesFernandezGranda2014}.  Related work developed broader
measure-space formulations, noisy and stable recovery, support localization,
and multidimensional extensions
\cite{deCastroGamboa2012,BrediesPikkarainen2013,
CandesFernandezGranda2013Noisy,DuvalPeyre2015,
AzaisDeCastroGamboa2015,ValiulahiDaeiHaddadiParvaresh2019}.  In that theory,
the dual certificate is a low-frequency trigonometric polynomial, and the
principal obstruction is the interaction of nearby points through a
translation-invariant low-pass kernel. Among general theories for off-the-grid recovery with
non-translation-invariant sensing operators, the Fisher-geometric framework
of Poon, Keriven, and Peyr\'e \cite{PoonKerivenPeyre2023} provides an
important point of comparison.  Their theory introduces an intrinsic metric
induced by the sensing operator and establishes stable Beurling LASSO (BLASSO) recovery under
geometric separation together with regularity and nondegeneracy assumptions
on the associated kernel.  It therefore supplies a general framework that
is not restricted to translation-invariant Fourier measurements.

The contribution developed here is complementary and operator specific.
We consider a fixed deterministic finite Fresnel aperture and derive explicit
finite-\(N_r\) estimates for the value, normalized-tangent, and higher
angular derivative interactions entering a Hermite dual certificate.  These
interactions reduce to structured sums
\(
        \sum_{n=0}^{N_r-1}
        a_n e^{i(\omega_1 n+\omega_2 n^2)},
\)
and QPAC supplies support-uniform bounds through derivative,
lag-correlation, and residue-linear mechanisms.  The resulting constants
enter directly into interpolation, local-curvature, and far-leakage
inequalities for the deterministic Fresnel TV problem.

Thus the distinction is not that non-translation-invariant sparse recovery
lacks a general geometric theory. Rather, the present work provides an
explicit finite-aperture arithmetic verification for the Fresnel operator.
In regimes where the general geometric hypotheses can also be verified, the
two viewpoints are compatible: the intrinsic metric describes the sensing
geometry, while QPAC gives a deterministic finite-\(N_r\) route for
controlling the concrete oscillatory channels appearing in the Fresnel dual
certificate.

At the level of the physical measurement model, the data considered here
are not Fourier-band samples of an unknown measure, but spatial aperture
samples of a narrowband Fresnel field.  Even though the aperture index plays an algebraic role
analogous to a Fourier index in the far-field approximation, the Fresnel phase
contains a range-dependent quadratic aperture term.  Consequently, the
interaction between two atoms is governed by the ordered finite trajectory
\(
        n\mapsto \omega_1(q,p)n+\omega_2(q,p)n^2,
        ~ n=0,\ldots,N_r-1,
\)
rather than by a scalar angular difference.  Thus, instead of imposing a universal scalar minimum-separation condition,
we derive a support-uniform sufficient recovery criterion directly from
finite quadratic-phase aperture cancellation.

From the viewpoint of sparse recovery, this quadratic-phase certification
provides a natural sufficient alternative to scalar coherence or
minimum-separation conditions. After
discretizing both range and angle, the Fresnel inverse problem can be viewed as
sparse recovery in a finite dictionary of aperture steering vectors.  This
viewpoint is useful but incomplete.  The relevant dictionaries are highly
coherent, and in the far-field limit same-bearing sources at different ranges
generate identical spatial steering vectors.  Therefore ordinary coherence or
restricted-isometry conditions do not capture the continuum Fresnel geometry.
QPAC plays the role of a Fresnel-specific sufficient incoherence certificate:
it controls the kernel, angular tangent, and higher angular derivative
channels entering the TV dual certificate through deterministic
quadratic-phase cancellation estimates, without requiring a scalar
dictionary-incoherence condition.

This also distinguishes the present theorem from polar-domain sparse dictionary
methods.  Polar dictionaries discretize both angle and distance in order to
recover sparsity in a near-field representation
\cite{CuiDai2022FarNear}.  Such discretizations are effective computational
models, but they introduce grid resolution, basis mismatch, and coherence
questions\cite{Chi2020,chi2011sensitivity}.  Gridless and lifted methods reduce this mismatch by replacing the finite polar
grid with atomic-norm or low-dimensional harmonic representations
\cite{DaeiZamaniChatterjeeSkoglundFodor2025ISAC,XiYang2025Gridless,
DaeiFodorSkoglund2026Convexity,DaeiFodorSkoglund2026LivingOffGrid}.
Those works primarily address computational representations, convex
formulations, and estimation procedures for near-field communication and
sensing models.  The principal contribution of the present paper is different:
we prove a support-uniform deterministic exact-recovery theorem for the
semi-discrete Fresnel TV problem by constructing and controlling a gauged
Hermite dual certificate through finite quadratic-phase cancellation.  The
finite-harmonic lift used later in this paper is retained only as a
computational bridge to this exact certificate theory; it is not itself the
basis of the QPAC recovery theorem.

The physical relevance of the Fresnel regime is not merely a matter of
terminology.  Classical Rayleigh-type boundaries give geometric transition
rules, but recent work shows that application-dependent discrepancy criteria
can place the effective near-field/far-field transition substantially farther
away from the aperture
\cite{DaeiFodorSkoglund2025NearFar}.  Thus the Fresnel region is not a thin
intermediate layer between two simpler models.  For large apertures and high
carrier frequencies, wavefront curvature can remain visible over ranges that
are large on the scale of the array.  This is precisely the regime in which the
quadratic aperture phase becomes an intrinsic part of the inverse problem
rather than a perturbation of the far-field Fourier model.

Off-the-grid convex recovery has also appeared beyond ordinary source
localization, for example in blind deconvolution, blind demixing, and message
recovery problems where continuous delay parameters must be recovered jointly
with additional unknown signals or user messages
\cite{DaeiRazavikiaSkoglundFodorFischione2025}.  These works share the
convex-duality and atomic-norm viewpoint.  The Fresnel problem considered here
has a different analytic obstruction: the off-support dual leakage is
controlled by finite quadratic exponential sums generated by aperture
curvature.

Finally, the physical setting connects with diffraction-limited imaging and
Fresnel optics \cite{Goodman2005,BornWolf1999}, as well as coherent wavefield
imaging, inverse source problems, and scattering theory
\cite{Devaney2012,ColtonKress2013}.  The contribution here is not a new
wave-propagation approximation, nor primarily an algorithmic channel-estimation
method.  It is a deterministic super-resolution theory explaining when a finite
aperture can distinguish range-angle point sources through quadratic Fresnel
curvature.

\subsection{Contributions and organization}

The main contributions are as follows.

\begin{enumerate}
\item We identify finite quadratic aperture phases
\(
        \omega_1 n+\omega_2 n^2, ~ n=0,\ldots,N_r-1,
\)
as the resolution geometry of semi-discrete Fresnel super-resolution.

\item We prove deterministic bounds for finite-aperture incomplete quadratic
exponential sums with channel-dependent coefficient sequences, with explicit
support-uniform envelopes for the Fresnel kernel, normalized angular tangent channels,
and higher angular derivative channels needed in a TV dual certificate.

\item We introduce QPAC, a support-uniform aperture-cancellation criterion that
combines derivative separation, lag-correlation cancellation, and
residue-linear cancellation.

\item We prove uniform exact recovery for the semi-discrete Fresnel TV problem
under a sufficient QPAC recovery condition.  The QPAC recovery number combines
support interpolation, Taylor near-support decay, and far leakage.  The result
is uniform over all nonzero source amplitudes and all supports in the
prescribed support class.  No converse is asserted: configurations not
certified by QPAC may still be exactly recoverable.

\item We provide a finite-harmonic Bessel-Vandermonde lift as a computational
surrogate and prove a perturbative bridge from the ideal Fresnel TV certificate
to the lifted dual polynomial.
\end{enumerate}
The lifted SDP is used as a computational relaxation; the exact recovery
theorem itself is the semi-discrete Fresnel TV theorem certified by QPAC.

After comparing the result with related work, sparse-recovery viewpoints, and
quadratic exponential-sum estimates in
\Cref{subsec:related-work,subsec:quadratic-sum-literature},
\Cref{sec:model-main-results} states the semi-discrete Fresnel model, the QPAC
recovery number, and the main exact-recovery theorem.
\Cref{sec:qsum-bounds} develops the finite-aperture quadratic-sum estimates
that enter QPAC.
The proof of the exact recovery theorem is given in \Cref{app:exact-proof}.
\Cref{sec:lifted_model} develops the finite-harmonic lifted SDP relaxation and
the perturbative bridge from the QPAC certificate to the lifted dual
polynomial.
\Cref{sec:numerics} illustrates the resulting dual-polynomial localization
mechanism, and \Cref{sec:discussion} discusses limitations and extensions.
In short, the paper establishes finite quadratic-phase cancellation as a
deterministic sufficient mechanism for certifying semi-discrete Fresnel
super-resolution, without requiring a universal scalar angular-separation
condition.

\paragraph{Notation and conventions.}
We write \(\mathbb T=\mathbb R/(2\pi\mathbb Z)\),
\([N]=\{0,\ldots,N-1\}\), and use \(\|\cdot\|_2\) for the Euclidean norm.
Ordinary italic symbols are used for scalars, measures, and scalar-valued
functions. Bold lowercase symbols denote finite-dimensional vectors, bold
uppercase symbols denote finite-dimensional matrices, and Greek
finite-dimensional vectors and matrices are written in boldsymbol form.  Calligraphic symbols are
reserved for sets, spaces, structured families, and linear operators, unless
explicitly stated otherwise.  Thus, for example, \(c_\ell\) is a scalar,
\(\mathbf c=(c_\ell)_\ell\) is a vector,
\(\boldsymbol\zeta\in\mathbb C^{N_r}\) is a dual vector, and
\(\mathbf G_{\rm SS}\in\mathbb R_+^{2\times2}\) is a matrix.
Components of vectors and matrices are written without boldface, e.g.,
\(c_\ell\), \(\zeta_n\), and \([\mathbf A]_{j\ell}\).
For the closed interval
\(\Theta=[\theta_{\min},\theta_{\max}]\subset(0,\pi)\), \(\mathcal M(\Theta)\) denotes the space of finite
complex Radon measures on \(\Theta\), with total variation norm
\(\|\cdot\|_{\TV}\).
For \(c\neq0\),
\(\operatorname{sgn}(c)=c/|c|\).
For \(x\in\mathbb R\), set
\(
\dist_{2\pi}(x)
:=
\min_{k\in\mathbb Z}|x-2\pi k|, 
[x]_+
:=
\max\{x,0\}.
\)
For a bounded real set \(E\),
\(\operatorname{hull}(E)\) denotes the smallest closed interval containing
\(E\).
The symbol \(\mathbf I_2\) denotes the \(2\times2\) identity matrix.
The notation \(\mathbf A^\dagger\) denotes the Moore-Penrose pseudoinverse
of a matrix \(\mathbf A\), and
\(\mathbf u\odot\mathbf v\) denotes entrywise (Hadamard) multiplication of
vectors of equal dimension.
We write \(\mathbb B_2\) for the Euclidean unit ball in the ambient
finite-dimensional space.
The spectral radius of a square matrix \(\mathbf G\) is denoted by
\(\rho_{\rm sp}(\mathbf G)\).
Complex inner products are conjugate-linear in the first argument:
\(
\langle \mathbf u,\mathbf v\rangle_{\mathbb C^N}
=
\sum_{n=0}^{N-1}\overline{u_n}v_n .
\)
For weighted aperture vectors we use
\(
\langle \mathbf u,\mathbf v\rangle_{\rho}
=
\sum_{n=0}^{N_r-1}\rho_n\overline{u_n}v_n, 
\langle \mathbf u,\mathbf v\rangle_{b}
=
\sum_{n=0}^{N_r-1}b_n\overline{u_n}v_n .
\)
Throughout the paper, \(\rho_n\) denotes the fixed certificate taper and
\(b_n=\rho_n/\sum_m\rho_m\) its normalized version.
The symbols \(a_n\) and \(a_n^X\) denote channel-dependent scalar coefficient
sequences in the QPAC estimates; they are not additional tapers.
The taper may vanish at some aperture indices, so these pairings are not
assumed positive definite on all of \(\mathbb C^{N_r}\).
When a genuine Hilbert-space interpretation is needed, we restrict the
weighted vectors to the active aperture
\(
        \mathcal A_\rho
        :=
        \{n:\rho_n>0\}
\)
and identify the resulting weighted Hilbert space with
\(
        \mathbb C^{|\mathcal A_\rho|}.
\)
Linear independence of Fresnel atoms is always understood in the full space
\(\mathbb C^{N_r}\).
For matrices,
\(
\langle \mathbf X,\mathbf Y\rangle_F
=
\operatorname{trace}(\mathbf X^H\mathbf Y).
\)
For finite-dimensional complex matrices,
\(\mathbf A^H\) denotes the conjugate transpose.
For operators between Hilbert spaces, \(T^*\) denotes the adjoint, defined by
\(
\langle Tx,y\rangle
=
\langle x,T^*y\rangle .
\)
A matrix \(\mathbf R\) is Hermitian if
\(
\mathbf R=\mathbf R^H.
\)
% ============================================================
\section{Models and main results}
\label{sec:model-main-results}
% ============================================================
We first formulate the semi-discrete Fresnel inverse problem and state the
deterministic recovery theorem.  The exact result is proved for the Fresnel
operator.  Later, we develop a computable finite-harmonic lifting for the same
semi-discrete Fresnel inverse problem.

% ============================================================
\subsection{The semi-discrete Fresnel model}
\label{subsec:semidiscrete_fresnel_model}
% ============================================================

Fix integers
\(
        N_r\ge1,~
        N_d\ge1.
\)
Here \(N_r\) is the number of aperture samples and \(N_d\) is the number of
admissible range bins.  Let \(d>0\) denote the inter-sensor spacing, let
\(\lambda>0\) denote the wavelength, and define
\(
        k_\lambda:=\frac{2\pi}{\lambda}.
\)
Fix
\(
0<r_{\min}\le r_{\max}<\infty,
\)
and let
\[
\Rset:=\{r_1,\ldots,r_{N_d}\}
\subset[r_{\min},r_{\max}]
\]
be a finite grid of pairwise distinct ranges. Let
\[
0<\theta_{\min}\le\theta_{\max}<\pi,
\qquad
\Thetaset:=[\theta_{\min},\theta_{\max}],
\]
and define
\[
\Qset:=\{1,\ldots,N_d\}\times\Thetaset .
\]
The feasible domain of the TV problem is \(\Qset\). Whenever physical
coordinates are needed, we identify \(p=(i,\theta)\in\Qset\) with
\((r_i,\theta)\), and we write
\[
        \delta_p:=\delta_{r_i}\otimes\delta_\theta .
\]
For aperture samples \(n=0,\dots,N_r-1\), define the
Fresnel atom
\begin{equation}\label{eq:fresnel_atom_sec2}
        \aFR(r_i,\theta)[n]
        :=
        \exp\!\left(
        i k_\lambda d n\cos\theta
        -
        i\tfrac{k_\lambda d^2 n^2}{2r_i}\sin^2\theta
        \right).
\end{equation}

A semi-discrete scene is a vector of complex measures
\begin{equation}\label{eq:semi_discrete_scene_sec2}
        \nu
        =
        \sum_{i=1}^{N_d}\delta_{r_i}\otimes\mu_i,
        ~
        \mu_i\in\Mset(\Thetaset).
\end{equation}
The Fresnel measurement operator is
\begin{equation}\label{eq:fresnel_forward_model_sec2}
        (\FFR\nu)[n]
        =
        \sum_{i=1}^{N_d}
        \int_{\Thetaset}
        \aFR(r_i,\theta)[n]\,d\mu_i(\theta),
        ~ n=0,\dots,N_r-1.
\end{equation}
The total variation norm is
\(\|\nu\|_{\TV,\Qset}:=\sum_{i=1}^{N_d}\|\mu_i\|_{\TV}\).  Given noiseless
data \(\mx y=\FFR\nu_\star\), we consider
\begin{equation}\label{eq:exact_tv_program_sec2}
        \min_\nu \|\nu\|_{\TV,\Qset}
        ~
        \text{subject to}
        ~
        \FFR\nu=\mx y .
\end{equation}

For a sparse scene
\[
        \nu_\star
        =
        \sum_{\ell=1}^{L}
        c_\ell\,\delta_{r_{i_\ell}}\otimes\delta_{\theta_\ell},
        ~ c_\ell\ne0,
\]
write \(p_\ell=(i_\ell,\theta_\ell)\) and
\(S=\{p_\ell\}_{\ell=1}^{L}\subset\Qset\).

% ============================================================
\subsection{Quadratic phase interactions}
\label{subsec:quadratic_phase_interactions_sec2}
% ============================================================

For an evaluation point \(q=(i,\theta)\in\Qset\) and a source point
\(p_\ell=(i_\ell,\theta_\ell)\in\Qset\), we use the ordered convention
\begin{equation}\label{eq:relative_fresnel_phase_sec2}
        \overline{\aFR(r_i,\theta)[n]}\,
        \aFR(r_{i_\ell},\theta_\ell)[n]
        =
        e^{i\omega_1(q,p_\ell)n+i\omega_2(q,p_\ell)n^2}.
\end{equation}
Here
\begin{equation}\label{eq:relative_phase_coefficients_sec2}
        \omega_1(q,p_\ell)
        =
        k_\lambda d(\cos\theta_\ell-\cos\theta),
\end{equation}
and
\begin{equation}\label{eq:relative_quadratic_coefficients_sec2}
        \omega_2(q,p_\ell)
        =
        -\tfrac{k_\lambda d^2}{2r_{i_\ell}}\sin^2\theta_\ell
        +
        \tfrac{k_\lambda d^2}{2r_i}\sin^2\theta .
\end{equation}
Thus every pairwise interaction is governed by the ordered finite-aperture
phase \(\omega_1(q,p_\ell)n+\omega_2(q,p_\ell)n^2\).  The coefficient
\(\omega_1\) is the angular phase mismatch, while
\(\omega_2\) is the curvature phase mismatch.  The far-field model is obtained in the limit \(\omega_2=0\).

% ============================================================
\subsection{Support classes and near-far decomposition}
\label{subsec:QPAC_support_classes_sec2}
% ============================================================

For an ordered support-support pair \((p_j,p_\ell)\), with
\(p_j=(i_j,\theta_j)\) as evaluation point and
\(p_\ell=(i_\ell,\theta_\ell)\) as source point, define
\[
        \overline{\aFR(r_{i_j},\theta_j)[n]}\,
        \aFR(r_{i_\ell},\theta_\ell)[n]
        =
        e^{i\omega_{1,j\ell}n+i\omega_{2,j\ell}n^2},
\]
where
\(
        \omega_{1,j\ell}
        =
        k_\lambda d(\cos\theta_\ell-\cos\theta_j),
        ~
        \omega_{2,j\ell}
        =
        -\tfrac{k_\lambda d^2}{2r_{i_\ell}}\sin^2\theta_\ell
        +
        \tfrac{k_\lambda d^2}{2r_{i_j}}\sin^2\theta_j .
\)
Set
\[
        d_{N_r}^+(p_j,p_\ell)
        :=
        \min_{0\le n\le N_r-1}
        \dist_{2\pi}
        \bigl(\omega_{1,j\ell}+(2n+1)\omega_{2,j\ell}\bigr).
\]
The endpoint \(n=N_r-1\) is included because the Abel-summation estimate uses
this extended derivative range.

For a support \(S=\{p_\ell=(i_\ell,\theta_\ell)\}_{\ell=1}^L\), define the
feasible same-range angular near set
\[
\begin{aligned}
\mathcal N_\theta(S;\varpi_{\rm loc})
&:=
\bigcup_{\ell=1}^{L}
\{(i_\ell,\theta)\in\Qset:
|\theta-\theta_\ell|\le \varpi_{\rm loc}\},\\
\mathcal N^\circ(S;\varpi_{\rm loc})
&:=
\mathcal N_\theta(S;\varpi_{\rm loc})\setminus S,\\
\mathcal F(S;\varpi_{\rm loc})
&:=
\Qset\setminus \mathcal N_\theta(S;\varpi_{\rm loc}) .
\end{aligned}
\]
Because \(\Theta\) is an interval, every angular segment joining a support
point to a point of
\(\mathcal N_\theta(S;\varpi_{\rm loc})\)
remains in \(\Theta\). Therefore the same near set is used both for strict
dual boundedness and for the Taylor derivative envelopes.
Fix a nonnegative taper \(\rho\) with
\(W_0:=\sum_n\rho_n>0\), and set
\[
b_n:=\tfrac{\rho_n}{W_0},\qquad
\bar n:=\sum_n b_n n,\qquad
\overline{n^2}:=\sum_n b_n n^2.
\]
Define
\[
\chi_i(\theta)
:=
k_\lambda d\,\bar n\cos\theta
+
\frac{k_\lambda d^2}{4r_i}
\overline{n^2}\cos(2\theta),
\]
and the normalized gauged atom
\[
\boldsymbol\psi_{i,\theta}[n]
:=
W_0^{-1/2}e^{-i\chi_i(\theta)}
a_{\rm Fr}(r_i,\theta)[n].
\]
The gauge gives
\[
\|\boldsymbol\psi_{i,\theta}\|_\rho=1,
\qquad
\langle\boldsymbol\psi_{i,\theta},
\partial_\theta\boldsymbol\psi_{i,\theta}\rangle_\rho=0.
\]
Set
\(
\sigma(i,\theta)
:=
\|\partial_\theta\boldsymbol\psi_{i,\theta}\|_\rho .
\)
Whenever normalized tangent channels are used, assume
\(
\inf_{(i,\theta)\in\Qset}\sigma(i,\theta)>0,
\)
and define
\(
\mathbf h_{i,\theta}
:=
\frac{\partial_\theta\boldsymbol\psi_{i,\theta}}
{\sigma(i,\theta)}.
\)
The condition \(q_{\min}>0\) imposed in the main theorem guarantees this
nondegeneracy. These formulas also define the
atoms on \(\mathcal Q\), where the Taylor
estimates are taken.

For \(q,p\in\mathcal Q\), define
\[
\begin{aligned}
K(q,p)&:=\langle\boldsymbol\psi_q,\boldsymbol\psi_p\rangle_\rho,
&
H(q,p)&:=\langle\boldsymbol\psi_q,\mathbf h_p\rangle_\rho,\\
dK(q,p)&:=\langle \mathbf h_q,\boldsymbol\psi_p\rangle_\rho,
&
dH(q,p)&:=\langle \mathbf h_q,\mathbf h_p\rangle_\rho .
\end{aligned}
\]
Throughout the paper, \(K\) denotes the gauged value kernel and \(H\) denotes
the support-side normalized tangent kernel.  The symbols \(dK\) and \(dH\)
denote normalized evaluation-side tangent channels, not physical angular
derivatives.  For \(m=2,3\), we use the shorthand
\(
d^mK(q,p)
:=
\partial_{\theta_q}^mK(q,p),
~
d^mH(q,p)
:=
\partial_{\theta_q}^mH(q,p).
\)
Thus \(K,H,dK,dH\) are the normalized Hermite interpolation channels,
whereas \(d^2K,d^3K,d^2H,d^3H\) denote the physical evaluation-side angular
derivative channels used in the curvature and Taylor estimates.
Writing
\(
        \sigma(q)
        :=
        \|\partial_{\theta_q}\boldsymbol\psi_q\|_\rho,
\)
the physical first evaluation-side derivatives satisfy
\(
        \partial_{\theta_q}K(q,p)
        =
        \sigma(q)dK(q,p),
        ~
        \partial_{\theta_q}H(q,p)
        =
        \sigma(q)dH(q,p).
\)
Thus \(dK,dH\) enter the normalized Hermite interpolation matrix.
The stationarity-based local estimate uses the physical second- and
third-order channels \(d^2K,d^3K,d^2H,d^3H\).

\begin{defn}[Support-uniform QPAC envelope system]
\label{def:qpac-envelope-system}
Assume \(N_r\ge10\). A support-uniform QPAC envelope system is a tuple
\(
       \mathscr E
=
(\rho,Q_{\max},\mathscr C_{\rm coeff}) .
\)
Here \(\rho=(\rho_n)_{n=0}^{N_r-1}\) is the nonnegative aperture taper used in
the gauged certificate.  It satisfies
\(
        W_0:=\sum_{n=0}^{N_r-1}\rho_n>0
\)
and the fourth-order flat-end condition
\[
        \rho_0=\rho_1=\rho_2=\rho_3=0,
        ~
        \rho_{N_r-4}=\rho_{N_r-3}=\rho_{N_r-2}=\rho_{N_r-1}=0 .
\]
The integer \(Q_{\max}\) satisfies
\(
        2\le Q_{\max}\le N_r
\)
and is the maximal residue denominator used by the lag-correlation and
residue-linear branches. The three QPAC cancellation branches are the derivative, lag-correlation,
and residue-linear branches. The \(\ell_1\) and Cauchy-Schwarz bounds are
universal fallback bounds, not cancellation branches.

The coefficient-enclosure scheme \(\mathscr C_{\rm coeff}\) associates with
each compact physical cell \(C\), each channel
\[
        X\in
        \{K,H,dK,dH,d^2K,d^3K,d^2H,d^3H\},
\]
and each branch \(b\in\{\mathrm{Der},\LCS,\mathrm{ResLin}\}\) a value
\(
        \mathcal U_{X,C}^{(b,\mathscr E)}
        \in[0,+\infty].
\)
Whenever
\(
        \mathcal U_{X,C}^{(b,\mathscr E)}<+\infty,
\)
the hypotheses of branch \(b\) are required to hold uniformly on \(C\), and
the resulting envelope is required to satisfy
\(
        |X(q,p)|
        \le
        \mathcal U_{X,C}^{(b,\mathscr E)}
\)
for every ordered pair \((q,p)\) whose physical coordinates satisfy
\(
        (r_{i_p},\theta_p,r_{i_q},\theta_q)\in C .
\)
If these uniform branch conditions or the corresponding channel bound have
not been established on the whole cell, we set
\(
        \mathcal U_{X,C}^{(b,\mathscr E)}
        :=+\infty .
\)
For
\(
        q=(i_q,\theta_q), 
        p=(i_p,\theta_p),
\)
define
\[
        \mathscr C_{X,b}^{\mathscr E}(q,p)
        :=
        \left\{
        C:
        (r_{i_p},\theta_p,r_{i_q},\theta_q)\in C,\;
        \mathcal U_{X,C}^{(b,\mathscr E)}<+\infty
        \right\}.
\]
Here the physical-cell coordinates are ordered as
(source range, source angle, evaluation range, evaluation angle), consistently
with the cellwise construction in the appendix.  Define
\(
       \mathcal U_X^{(b,\mathscr E)}(q,p)
:=
\begin{cases}
\displaystyle
\inf_{C\in\mathscr C_{X,b}^{\mathscr E}(q,p)}
\mathcal U_{X,C}^{(b,\mathscr E)},
&
\mathscr C_{X,b}^{\mathscr E}(q,p)\ne\varnothing,\\[2ex]
+\infty,
&
\mathscr C_{X,b}^{\mathscr E}(q,p)=\varnothing.
\end{cases}
\)
Thus overlapping cells cause no ambiguity, and branch validity is explicitly
channel dependent.
\end{defn}
\begin{defn}[Support-support QPAC class]
\label{def:ss_qpac_class_sec2}
Let
\[
\mathbf u_{\rm SS}
:=(u_K^{\rm SS},u_H^{\rm SS},
   u_{dK}^{\rm SS},u_{dH}^{\rm SS})
\]
be nonnegative constants. For an ordered support-support pair
\((p_j,p_\ell)\), \(j\ne\ell\), and
\(X\in\{K,H,dK,dH\}\), define
\[
B_X^{\rm SS}(p_j,p_\ell)
:=
\min_{b\in\{\mathrm{Der},\mathrm{LCS},\mathrm{ResLin}\}}
\mathcal U_X^{(b,\mathscr E)}(p_j,p_\ell),
\]
where an uncertified branch has value \(+\infty\).

We write
\(\mathfrak S_L^{\rm SS}(\mathbf u_{\rm SS};\mathscr E)\)
for the class of all \(L\)-point supports \(S\subset\Qset\)
such that, for every ordered pair \(p_j,p_\ell\in S\),
\(j\ne\ell\),
\[
B_K^{\rm SS}(p_j,p_\ell)\le u_K^{\rm SS},
\qquad
B_H^{\rm SS}(p_j,p_\ell)\le u_H^{\rm SS},
\]
and
\[
B_{dK}^{\rm SS}(p_j,p_\ell)\le u_{dK}^{\rm SS},
\qquad
B_{dH}^{\rm SS}(p_j,p_\ell)\le u_{dH}^{\rm SS}.
\]
The certifying branch may depend on the ordered pair and
the channel.
\end{defn}

\subsection{The QPAC recovery number}
\label{subsec:explicit_qpac_number_sec2}

\paragraph{Roadmap:}
The hierarchy of estimates is as follows.  The scalar QPAC estimates bound one
finite quadratic sum.  The cellwise envelope construction makes these bounds
uniform over physical range-angle parameter cells.  The support-support
envelopes enter the gauged Hermite interpolation system.  The near-support and
far-region envelopes then control the off-support dual certificate.  Thus QPAC
turns finite quadratic cancellation into the three recovery budgets used below.
% ============================================================

QPAC produces three recovery budgets: the support-interpolation budget
\(\eta_{\rm SS}\), the Taylor near-support budget
\(\eta_{\rm near}(\varpi_{\rm loc})\), and the far-region leakage budget
\(\eta_{\rm far}(\varpi_{\rm loc})\).  These three budgets provide all
estimates needed for the strict TV dual certificate.  All constants
below are deterministic functions of the aperture geometry, the range grid,
the angular domain, the certificate taper, and the branch envelopes developed
in \Cref{sec:qsum-bounds}.  The recovery number is defined after the three
budgets have been specified.

\paragraph{Support interpolation:}
Let
\(\mathbf u_{\rm SS}
=(u_K^{\rm SS},u_H^{\rm SS},u_{dK}^{\rm SS},u_{dH}^{\rm SS})\)
be support-support envelopes, and fix a nonempty prescribed support class
\(
        \varnothing\ne\mathfrak S_L
        \subseteq
        \mathfrak S_L^{\rm SS}
        (\mathbf u_{\rm SS};\mathscr E).
\)
Thus, for every \(S\in\mathfrak S_L\) and every ordered pair of distinct
support points in \(S\),
\(
        |K(p_j,p_\ell)|\le u_K^{\rm SS},
        ~
        |H(p_j,p_\ell)|\le u_H^{\rm SS},
\)
and
\(
        |dK(p_j,p_\ell)|\le u_{dK}^{\rm SS},
        ~
        |dH(p_j,p_\ell)|\le u_{dH}^{\rm SS}.
\)
Define
\begin{equation}\label{eq:GSS_sec2}
        \mathbf G_{\rm SS}
        :=
        \begin{bmatrix}
        u_K^{\rm SS}
        &
        u_H^{\rm SS}
        \\[0.5ex]
        u_{dK}^{\rm SS}
        &
        u_{dH}^{\rm SS}
        \end{bmatrix}.
\end{equation}
Set
\begin{equation}\label{eq:eta_SS_sec2}
        \eta_{\rm SS}
        :=
        (L-1)\rho_{\rm sp}(\mathbf G_{\rm SS}).
\end{equation}
If \(\eta_{\rm SS}<1\), define
\begin{equation}\label{eq:Gamma_sec2}
\boldsymbol\Gamma
:=
\left(\mathbf I_2-(L-1)\mathbf G_{\rm SS}\right)^{-1}
\begin{bmatrix}1\\0\end{bmatrix}
=
\begin{bmatrix}
\Gamma_K\\
\Gamma_H
\end{bmatrix}.
\end{equation}

\begin{equation}\label{eq:Xi_sec2}
\boldsymbol\Xi
:=
\boldsymbol\Gamma-
\begin{bmatrix}1\\0\end{bmatrix}
=
\begin{bmatrix}
\Xi_K\\
\Xi_H
\end{bmatrix}.
\end{equation}
Since \(\mathbf G_{\rm SS}\) is entrywise nonnegative and
\(\rho_{\rm sp}((L-1)\mathbf G_{\rm SS})<1\), the Neumann-series representation gives
\[
        \boldsymbol\Gamma
        =
        \sum_{m=0}^{\infty}
        \bigl((L-1)\mathbf G_{\rm SS}\bigr)^m
        \begin{bmatrix}1\\0\end{bmatrix}.
\]
Consequently,
\(
\boldsymbol\Gamma\ge
\begin{bmatrix}1\\0\end{bmatrix},
~
\boldsymbol\Xi\ge0
\)
componentwise.
The vector \(\boldsymbol\Gamma\) bounds the Hermite certificate
coefficients, and \(\boldsymbol\Xi\) bounds the diagonal correction
from the ideal self-kernel.

\paragraph{Near-support curvature:}
The certificate uses the taper \(\rho\) from
\Cref{def:qpac-envelope-system}.  Set \(W_0:=\sum_n\rho_n\) and
\(b_n:=\rho_n/W_0\).  Define
\(
        \bar n:=\sum_{n=0}^{N_r-1}b_n\,n,
        ~
        \overline{n^2}:=\sum_{n=0}^{N_r-1}b_n\,n^2,
\)
\(
        x_n:=n-\bar n,
        ~
        y_n:=n^2-\overline{n^2}.
\)
Let
\(
        s_{\min}:=
        \min_{\theta\in[\theta_{\min},\theta_{\max}]}\sin\theta>0
\)
and define the conservative tangent-parameter hull
\[
        I_\tau
        :=
        \operatorname{hull}
        \left\{
        \tfrac{d\cos\theta}{r}:
        r\in[r_{\min},r_{\max}],
        ~
        \theta\in[\theta_{\min},\theta_{\max}]
        \right\}.
\]
For \(\tau\in I_\tau\), define
\(
        q(\tau)
        :=
        \sum_{n=0}^{N_r-1}b_n(x_n+\tau y_n)^2 .
\)
Assume the taper is tangent-nondegenerate on \(I_\tau\), namely
\(
        q_{\min}
        :=
        \min_{\tau\in I_\tau}q(\tau)
        >0.
\)
Set
\(
        q_{\max}:=\max_{\tau\in I_\tau}q(\tau).
\)
With \(\kappa_0:=k_\lambda d\), define
\(
        \underline\sigma^2
        :=
        \kappa_0^2s_{\min}^2q_{\min},
        ~
        \overline\sigma_\theta
        :=
        \kappa_0\sqrt{q_{\max}} .
\)
Then
\(
        \|\partial_\theta \boldsymbol\psi_p\|_\rho^2\ge\underline\sigma^2,~
        \|\partial_\theta\boldsymbol \psi_p\|_\rho\le\overline\sigma_\theta
\)
uniformly over \(\Qset\).

Let
\[
\begin{aligned}
        U_{K,2}^{\rm self}
        &:=
        \sup_{p\in\Qset}
        |d^2K(p,p)|,
        &
        U_{H,2}^{\rm self}
        &:=
        \sup_{p\in\Qset}
        |d^2H(p,p)|.
\end{aligned}
\]
Let \(U_{K,2}^{\rm SS}\) and \(U_{H,2}^{\rm SS}\) be certified
support-uniform off-diagonal support-support envelopes for \(d^2K\) and
\(d^2H\), respectively:
\[
        |d^2K(p_j,p_\ell)|\le U_{K,2}^{\rm SS},
        ~
        |d^2H(p_j,p_\ell)|\le U_{H,2}^{\rm SS},
        ~ j\ne \ell,
\]
uniformly for all \(S\in\mathfrak S_L\).  Set
\[
\begin{aligned}
        E_{\rm curv}
        &:=
        2\Xi_KU_{K,2}^{\rm self}
        +
        2\Xi_HU_{H,2}^{\rm self}+
        2(L-1)
        \bigl(
        \Gamma_KU_{K,2}^{\rm SS}
        +
        \Gamma_HU_{H,2}^{\rm SS}
        \bigr),
        \\
        m_{\rm near}
        &:=
        2\underline\sigma^2-E_{\rm curv}.
\end{aligned}
\]
For \(X\in\{K,H\}\) and \(a=2,3\), let
\(B_{X,a}^{\rm loc}(q,p)\) be the pointwise envelope obtained by taking
the infimum of the certified physical-cell bounds for
\(|\partial_{\theta_q}^{a}X(q,p)|\) over the cells containing \((q,p)\).
The cells and their bounds are supplied by
\(\mathscr C_{\rm coeff}\), using the valid derivative-channel bounds
and derivative-size caps in
\Cref{thm:app_cellwise_support_uniform_qpac}.
Wherever this local coverage has been certified,
\begin{equation}
\label{eq:local_channel_envelopes_sec2}
 |\partial_{\theta_q}^{a}X(q,p)|
 \le B_{X,a}^{\rm loc}(q,p),
 \qquad X\in\{K,H\},\quad a=2,3.
\end{equation}
These are physical evaluation-angle derivatives.

For the fixed localization radius \(\varpi_{\rm loc}\) and \(a=2,3\),
choose a certified constant \(D_a\in[0,+\infty]\) such that
\begin{equation}
\label{eq:local_Da_sec2}
\sup_{\substack{
S\in\mathfrak S_L\\
q\in\mathcal N_\theta(S;\varpi_{\rm loc})}}
\sum_{\ell=1}^{L}
\left[
\Gamma_K
\left|\partial_{\theta_q}^{\,a}K(q,p_\ell)\right|
+
\Gamma_H
\left|\partial_{\theta_q}^{\,a}H(q,p_\ell)\right|
\right]
\le D_a.
\end{equation}
If no finite certified bound is available, set \(D_a=+\infty\).

The pairwise cell envelopes provide the admissible choice
\[
D_a^{\rm cell}(\varpi_{\rm loc})
:=
\sup_{\substack{
S\in\mathfrak S_L\\
q\in{\mathcal N}_{\theta}
(S;\varpi_{\rm loc})}}
\sum_{\ell=1}^{L}
\left[
\Gamma_K B_{K,a}^{\rm loc}(q,p_\ell)
+
\Gamma_H B_{H,a}^{\rm loc}(q,p_\ell)
\right].
\]

All source contributions are summed at one common evaluation point and one
common support before the supremum is taken. Because \(\Theta\) is an
interval, the complete near set contains every angular segment used in the
Taylor argument.

Define the near-support number by
\begin{equation}
\label{eq:eta_near_sec2}
\eta_{\rm near}(\varpi_{\rm loc})
:=
\begin{cases}
\begin{aligned}
&\displaystyle
\frac{2\varpi_{\rm loc}D_3(\varpi_{\rm loc})}{3m_{\rm near}}\\
&\displaystyle\quad+
\frac{\varpi_{\rm loc}^{2}D_2(\varpi_{\rm loc})^{2}}
     {2m_{\rm near}},
\end{aligned}
&
\begin{array}{l}
m_{\rm near}>0,\\[-0.2ex]
D_2(\varpi_{\rm loc})+D_3(\varpi_{\rm loc})<+\infty,
\end{array}\\[2ex]
+\infty,&\text{otherwise}.
\end{cases}
\end{equation}
The condition \(\eta_{\rm near}(\varpi_{\rm loc})<1\) includes positive
curvature and finite local derivative bounds.
By \Cref{thm:near_support_strict_decay_app}, it implies strict decay
of the gauged Hermite certificate on the punctured feasible near set.

\paragraph{Far-region leakage:}
For \(X\in\{K,H\}\), define
\[
B_X^\cup(q,p_\ell)
:=
\min_{b\in\{\mathrm{Der},\LCS,\mathrm{ResLin}\}}
\mathcal U_X^{(b,\mathscr E)}(q,p_\ell),
\]
where an uncertified branch has value \(+\infty\). Choose a certified far
budget \(\eta_{\rm far}(\varpi_{\rm loc})\in[0,+\infty]\) such that
\begin{equation}
\label{eq:eta_far_rowsum_sec2}
\sup_{\substack{
S\in\mathfrak S_L\\
q\in\mathcal F(S;\varpi_{\rm loc})}}
\sum_{\ell=1}^{L}
\left[
\Gamma_K|K(q,p_\ell)|
+
\Gamma_H|H(q,p_\ell)|
\right]
\le
\eta_{\rm far}(\varpi_{\rm loc}).
\end{equation}
The supremum over an empty far set is defined to be zero. The branch
envelopes provide the admissible choice
\[
\eta_{\rm far}^{\cup}(\varpi_{\rm loc})
:=
\sup_{\substack{
S\in\mathfrak S_L\\
q\in\mathcal F(S;\varpi_{\rm loc})}}
\sum_{\ell=1}^{L}
\left[
\Gamma_K B_K^\cup(q,p_\ell)
+
\Gamma_H B_H^\cup(q,p_\ell)
\right].
\]
If several valid bounds are available, their minimum may be taken.

In what follows, \(D_2\), \(D_3\), and \(\eta_{\rm far}\) denote these
fixed admissible certified choices.
Finally, define the QPAC recovery number by
\begin{equation}\label{eq:QPAC_recovery_number_sec2}
        \mathfrak C_{\rm QPAC}^{\mathscr E}
        (\mathfrak S_L,\varpi_{\rm loc})
        :=
        \begin{cases}
        \displaystyle
        \max\left\{
        \eta_{\rm SS},
        \eta_{\rm near}(\varpi_{\rm loc}),
        \eta_{\rm far}(\varpi_{\rm loc})
        \right\},
        & \eta_{\rm SS}<1,\\[1.5ex]
        +\infty,
        & \eta_{\rm SS}\ge1 .
        \end{cases}
\end{equation}
\paragraph{QPAC-admissibility:}
Let
\(
        \varnothing\ne\mathfrak S_L
        \subseteq
        \mathfrak S_L^{\rm SS}
        (\mathbf u_{\rm SS};\mathscr E)
\)
be a prescribed support class in \(\mathcal Q\), and let
\(\varpi_{\rm loc}>0\).  We say that
\((\mathfrak S_L,\varpi_{\rm loc},\mathscr E)\) is QPAC-admissible if
\[
        \eta_{\rm SS}<1,
        ~
        \eta_{\rm near}(\varpi_{\rm loc})<1,
        ~
        \eta_{\rm far}(\varpi_{\rm loc})<1.
\]
The local derivative envelopes entering
\(\eta_{\rm near}(\varpi_{\rm loc})\) are taken over the complete near set
\(\mathcal N_\theta(S;\varpi_{\rm loc})\). Because \(\Theta\) is an interval,
this set contains every angular segment used in the Taylor argument.
The quantities \(D_2(\varpi_{\rm loc})\), \(D_3(\varpi_{\rm loc})\), and
\(\eta_{\rm far}(\varpi_{\rm loc})\) must be finite certified aggregate
bounds over the complete near and far classes, respectively.  These bounds may be obtained from the branch-union
cell envelopes or from a separately proved common-point continuum
enclosure. Moreover,
\Cref{cor:qpac_near_implies_noncollision}
shows that
\(
        \eta_{\rm near}(\varpi_{\rm loc})<1
\)
automatically excludes any second same-range support point from the angular
localization neighborhood of radius \(\varpi_{\rm loc}\).  Hence the local
noncollision required by the Taylor argument is already a consequence of the
QPAC near-support condition and is not imposed as an additional separation
or radius-compatibility hypothesis.
 
Note that failure of QPAC does not imply failure of recovery; it only means
that the present deterministic envelopes do not certify the class.

% ============================================================
\subsection{Main exact recovery theorem}
\label{subsec:main_exact_theorem_sec2}

The three QPAC budgets give a single sufficient condition for uniform exact
recovery; no necessity claim is made.

\begin{thm}
\label{thm:main_exact_recovery_sec2}
Assume \(N_r\ge10\), and fix integers
\(
        N_d\ge1, L\ge1.
\)
Fix an aperture spacing \(d>0\), a wavelength
\(\lambda>0\), and set \(k_\lambda:=\tfrac{2\pi}{\lambda}\).  Let
\(
        \mathcal R=\{r_1,\ldots,r_{N_d}\}\subset [r_{\min},r_{\max}],
        ~
        0<r_{\min}\le r_{\max}<\infty,
\)
and let
\(
0<\theta_{\min}\le\theta_{\max}<\pi,
\qquad
\Theta:=[\theta_{\min},\theta_{\max}].
\)
Define
\(
\mathcal Q:=\{1,\ldots,N_d\}\times\Theta .
\)
Fix a nonempty prescribed class \(\mathfrak S_L\) of
\(L\)-point supports in \(\mathcal Q\). Suppose there exist
a support-uniform QPAC envelope system \(\mathscr E\) as in
\Cref{def:qpac-envelope-system}, nonnegative
support-support budgets
\(
\mathbf u_{\rm SS}
=(u_K^{\rm SS},u_H^{\rm SS},
  u_{dK}^{\rm SS},u_{dH}^{\rm SS}),
\)
and a radius \(\varpi_{\rm loc}>0\) such that
\[
\mathfrak S_L
\subseteq
\mathfrak S_L^{\rm SS}(\mathbf u_{\rm SS};\mathscr E),
\]
and the corresponding QPAC recovery number defined in
\Cref{eq:QPAC_recovery_number_sec2} satisfies
\[
\mathfrak C_{\rm QPAC}^{\mathscr E}
(\mathfrak S_L,\varpi_{\rm loc})<1.
\]
Then every measure
\(
        \nu_\star
        =
        \sum_{\ell=1}^{L}
        c_\ell\,\delta_{r_{i_\ell}}\otimes\delta_{\theta_\ell},
        ~ c_\ell\neq0,
\)
with distinct support points
\(
        S=\{(i_\ell,\theta_\ell)\}_{\ell=1}^{L}\in\mathfrak S_L
\)
is the unique minimizer of
\[
        \min_{\nu}\|\nu\|_{\TV,\mathcal Q}
        ~
        \text{\rm subject to}
        ~
        \FFR\nu=\FFR\nu_\star .
\]
The conclusion is uniform over all nonzero amplitude vectors
\((c_\ell)_{\ell=1}^{L}\).
\end{thm}

The QPAC condition in the theorem is branch-agnostic: different ordered
pairs and different channels may be controlled by different valid branches.
We now record two concrete support subclasses.  The first obtains the
support-support envelopes from derivative separation, while the second
obtains them from residue-lag cancellation.  Both are specializations of the
same support-uniform recovery theorem.

\begin{defn}[Derivative verification class]
\label{def:derivative_verification_class_sec2}
For parameters \(0<d_{\rm SS}\le\pi\),
\({s_{2,{\rm SS}}^{\rm max}}\in[0,1]\), and an ambient domain
\(\mathcal Q\), let
\(
        \mathfrak S_L^{\rm der}
        (d_{\rm SS},{s_{2,{\rm SS}}^{\rm max}};\mathcal Q)
\)
be the collection of all \(L\)-point supports \(S\subset\mathcal Q\) whose
distinct support-support pairs satisfy
\(
        d_{N_r}^+(p_j,p_\ell)\ge d_{\rm SS},
        ~
        |\sin\omega_{2,j\ell}|
        \le {s_{2,{\rm SS}}^{\rm max}},
        ~ j\ne\ell .
\)
The first inequality enforces nonstationary behavior of the support-support
quadratic phase increments.  The second inequality is not a separation
condition; it controls the constants in the fourth-order derivative branch by
bounding the curvature factor \(|\sin\omega_2|\). This class provides an explicit derivative-based route for constructing the
support-support envelopes required by the general QPAC condition.
\end{defn}
The derivative verification class becomes a concrete QPAC recovery class
once the derivative branch provides uniform bounds for the four
support-support Hermite channels.  Applying the main theorem to this
subclass gives the following consequence.

\begin{cor}
\label{cor:derivative_verification_route_sec2}
Let
\(
\varnothing\ne\mathfrak S_L
\subseteq
\mathfrak S_L^{\rm der}
(d_{\rm SS},s_{2,\rm SS}^{\max};\mathcal Q)
\)
be a prescribed support subclass. Suppose that the derivative branch
provides support-support envelopes \(\mathbf u_{\rm SS}\), uniformly on
\(\mathfrak S_L\), such that
\(
\mathfrak S_L
\subseteq
\mathfrak S_L^{\rm SS}(\mathbf u_{\rm SS};\mathscr E).
\)
If there exists \(\varpi_{\rm loc}>0\) such that
\(
\mathfrak C_{\rm QPAC}^{\mathscr E}
(\mathfrak S_L,\varpi_{\rm loc})<1,
\)
then the exact-recovery conclusion of
\Cref{thm:main_exact_recovery_sec2} holds uniformly for every support
in \(\mathfrak S_L\).
\end{cor}

\begin{defn}[Lag-correlation support class]
\label{def:lcs_support_class}
Let
\(
        \mathbf u_{\rm SS}
        =
        (u_K^{\rm SS},u_H^{\rm SS},u_{dK}^{\rm SS},u_{dH}^{\rm SS})
\)
be support-support interaction budgets. For an ordered support-support pair \((p_j,p_\ell)\) and
\(X\in\{K,H,dK,dH\}\), define
\(
        B_X^{\LCS}(p_j,p_\ell)
        :=
        \mathcal U_X^{(\LCS,\mathscr E)}(p_j,p_\ell).
\)
By the convention in
\Cref{def:qpac-envelope-system}, this quantity equals \(+\infty\) whenever
the LCS branch is not certified for channel \(X\) on any certified cell
containing the ordered pair.  We
define \(\mathfrak S_L^{\LCS}(\mathbf u_{\rm SS};\mathscr E)\) as the class of
all \(L\)-point supports \(S\subset\Qset\) such that, for
every ordered off-diagonal pair \(p_j,p_\ell\in S\), \(j\ne\ell\),
\(
        B_X^{\LCS}(p_j,p_\ell)
        \le
        u_X^{\rm SS},
        ~
        X\in\{K,H,dK,dH\}.
\)
The class is controlled by full residue lag-phase nonalignment over support
cells, not by the derivative separation quantity \(d_{N_r}^+\).
\end{defn}
\begin{rem}[LCS is not a derivative class]
The derivative verification class is controlled by
\(
       d_{N_r}^+(\omega_1,\omega_2)
=
\min_{0\le n\le N_r-1}
\dist_{2\pi}
\bigl(\omega_1+(2n+1)\omega_2\bigr),
\)
together with a curvature-size parameter.  The LCS class is controlled instead by the full residue lag-phase intervals
\(\Jlc_{s,h,m}(\Omega)\) and by the associated cosine majorants
\(\Cmax(\Jlc_{s,h,m})\).
Thus LCS certifies support interactions through nonalignment of complete
residue lag phases, rather than through pointwise lower bounds on the discrete
phase derivative.
\end{rem}

The same specialization applies when support-support interactions are
controlled by residue-lag cancellation rather than by derivative
separation.  Since the LCS support class is contained in the general QPAC
support class with the same interaction budgets, the main recovery theorem
gives the following consequence.

\begin{cor}
\label{cor:lcs_verification_route_sec2}
Let
\(
        \mathbf u_{\rm SS}
        =
        (u_K^{\rm SS},u_H^{\rm SS},
        u_{dK}^{\rm SS},u_{dH}^{\rm SS}),
\)
and let
\(
        \mathfrak S_L^{\LCS}
        (\mathbf u_{\rm SS};\mathscr E)
\)
be the lag-correlation support class from
\Cref{def:lcs_support_class}.  Let
\(
        \varnothing\ne\mathfrak S_L
        \subseteq
        \mathfrak S_L^{\LCS}
        (\mathbf u_{\rm SS};\mathscr E)
\)
be any prescribed support subclass.  By the definitions of
\(B_X^{\LCS}\) and \(B_X^{\rm SS}\),
\(
        \mathfrak S_L
        \subseteq
        \mathfrak S_L^{\LCS}
        (\mathbf u_{\rm SS};\mathscr E)
        \subseteq
        \mathfrak S_L^{\rm SS}
        (\mathbf u_{\rm SS};\mathscr E).
\)
If there exists \(\varpi_{\rm loc}>0\) such that
\(
        \mathfrak C_{\rm QPAC}^{\mathscr E}
        \left(
        \mathfrak S_L,
        \varpi_{\rm loc}
        \right)
        <1,
\)
then the exact-recovery conclusion of
\Cref{thm:main_exact_recovery_sec2} holds uniformly for every support in
\(\mathfrak S_L\).
\end{cor}
\begin{rem}[Modeling scope]
The exact theorem is a theorem for the semi-discrete Fresnel model: ranges lie
on a finite admissible grid and angles are continuous.  This choice isolates
the Fresnel curvature mechanism while leaving only one continuous local
direction, so a one-dimensional angular Hermite certificate suffices.  A fully
continuous range-angle theorem would require a two-parameter certificate with
range, angle, and mixed derivative control.  The finite-harmonic lift below is
therefore used only as a computational surrogate for the Fresnel TV certificate,
not as a replacement for the exact QPAC theorem.
\end{rem}

The proof decomposes the dual-certificate construction into three QPAC
estimates: \(\eta_{\rm SS}<1\) gives the gauged Hermite certificate,
\(\eta_{\rm near}<1\) gives strict Taylor decay on the punctured feasible near
set, and \(\eta_{\rm far}<1\) gives strict boundedness on the far region.
Together these three estimates produce the strict TV dual certificate, as
shown in \Cref{app:exact-proof}.
\begin{table}[t]
\centering
\small
\begin{tabular}{p{0.18\linewidth}|p{0.34\linewidth}|p{0.34\linewidth}}
\hline
QPAC quantity & Certificate role & Proof mechanism\\
\hline
\(\eta_{\rm SS}\) &
Hermite interpolation on the support &
Block Neumann domination of the gauged interpolation system\\
\(\eta_{\rm near}(\varpi_{\rm loc})\) &
Strict decay on the punctured feasible near set &
Negative gauged curvature plus Taylor remainder control\\
% Radius compatibility &
% Same-range local noncollision &
% \(K\)-coherence bound plus
% \(0<\varpi_{\rm loc}<
% \sqrt{(1-u_K^{\rm SS})/(2\overline\sigma_\theta^2)}\)\\
\(\eta_{\rm far}(\varpi_{\rm loc})\) &
Strict off-support boundedness &
Support-uniform best-branch QPAC leakage bounds\\
\hline
\end{tabular}
\caption{The three quadratic-phase aperture certification (QPAC) recovery
budgets used in the exact-recovery proof.}
\label{tab:qpac_roadmap}
\end{table}
\begin{rem}[Amplitude independence and coefficient recovery]
The QPAC condition is independent of the nonzero amplitudes and their phases.
It certifies the geometry of the support class.  After exact support recovery,
the amplitudes are obtained by finite-dimensional least squares. Indeed, for \(S=\{(i_\ell,\theta_\ell)\}_{\ell=1}^{L}\), define
\(
        \mathbf A_S
        :=
        \bigl[
        \mathbf \aFR(r_{i_1},\theta_1)
         \cdots
        \mathbf \aFR(r_{i_L},\theta_L)
        \bigr].
\)
Under the hypotheses of \Cref{thm:main_exact_recovery_sec2}, the active atoms
are linearly independent by the Hermite invertibility argument in
\Cref{app:exact-Hermite-system}. Hence, once \(S\) has been identified,
\(\mathbf c
=
\mathbf A_S^\dagger\mathbf y
=
\left(
\mathbf A_S^H\mathbf A_S
\right)^{-1}
\mathbf A_S^H\mathbf y.\)
\end{rem}

The numerical examples in \Cref{sec:numerical_derivative_qpac_example,sec:numerical_range_lcs_qpac_example}
use distinct QPAC branches. The derivative example covers a
two-range support class whose angles vary independently over
prescribed windows. The lag-correlation example covers a fixed
two-range support at a common bearing. Both evaluate the
support-interpolation, near-support, and far-region budgets.ear-support, and far-region budgets appearing
in the main theorem.

% % ============================================================
\subsection{Finite-harmonic lifted Fresnel localization}
\label{subsec:main_lifted_statement}

The exact theorem concerns the semi-discrete Fresnel operator \(\FFR\).  The
finite-harmonic lifted model developed in \Cref{sec:lifted_model} is a
computable surrogate, not a second exact-recovery theorem.  It replaces each
Fresnel atom by a finite Bessel-Vandermonde harmonic expansion and leads to a
finite-dimensional SDP with a lifted trigonometric dual polynomial.  The connection with QPAC is perturbative: for any fixed QPAC dual vector, the
finite-harmonic polynomial generated by that same vector approximates the ideal
Fresnel dual certificate uniformly, with an explicit truncation margin that
tends to zero as the harmonic orders increase.  The two results play complementary roles: the perturbation estimate controls the finite-harmonic representation of a fixed QPAC certificate, while \Cref{prop:lifted_dual_saturation} characterizes the dual polynomial produced by the lifted SDP.  The precise
definitions of the lifted atom, the finite Fresnel operator, the lifted dual
polynomial, and the truncation bound are given in \Cref{sec:lifted_model}.

\subsection{Far-field limit}
\label{subsec:far_field_degeneration_sec2}

The Fresnel theorem contains the far-field geometry as a limiting case.  If the
range grid is sent to infinity while the aperture and angular domain are fixed,
then
\(
        \tfrac{k_\lambda d^2}{2r_i}\sin^2\theta\to0,
        ~
        \omega_1n+\omega_2n^2\to\omega_1n .
\)
Range then disappears from the narrowband spatial response: same-bearing
sources at different ranges generate identical far-field steering vectors and
can only be recovered through their combined complex amplitude.  The limiting
atom is
\(
        a_{\rm FF}(\theta)[n]
        =
        e^{in\kappa(\theta)},
        ~
        \kappa(\theta):=k_\lambda d\cos\theta .
\)
This is the precise sense in which the far-field limit is Fourier-like.  The
array observes spatial aperture samples \(n=0,\ldots,N_r-1\), not Fourier
coefficients over a physical frequency band.  Thus QPAC degenerates to a
linear-phase finite-aperture certificate geometry in the angular spatial
frequency \(\kappa(\theta)\), while the Fresnel theorem extends that geometry
to the regime where the quadratic curvature coordinate is visible.
% ============================================================
% ============================================================
\section{QPAC estimates and kernel localization}
\label{sec:qsum-bounds}
% ============================================================

This section develops the oscillatory estimates behind QPAC.  Every Fresnel
certificate channel reduces to a finite quadratic sum
\begin{equation}
\label{eq:qsum_basic_object}
        T_N(a;\omega_1,\omega_2)
        :=
        \sum_{n=0}^{N-1}a_n e^{i(\omega_1 n+\omega_2 n^2)} .
\end{equation}
Here \(N\) denotes an abstract summation length.  In the Fresnel certificate
applications, \(N=N_r\), where \(N_r\) is the number of aperture samples.
The coefficient sequence \(a\) contains the normalized taper
\(b_n=\rho_n/W_0\) and the channel-specific tangent or derivative multiplier; it depends on the channel,
while the phase coordinates encode the ordered Fresnel interaction. QPAC bounds these sums by
combining derivative separation, lag-correlation cancellation, and
residue-linear cancellation.  These scalar estimates are then converted into
support-uniform channel envelopes over physical support/evaluation cells; the
cellwise construction is given in the appendix. Although the estimates below are connected in spirit to classical bounds for
quadratic exponential sums, the requirements here are different.  The number of
aperture samples is finite, the coefficient sequences are channel-dependent,
and the bounds must hold uniformly over physical support/evaluation cells.  We therefore give
estimates with explicit constants suitable for the Hermite dual-certificate
inequalities.

The normalized gauged Hermite channels are
\(K,H,dK,dH\), corresponding to \((\bs \psi,\bs \psi)\), \((\bs\psi,\mx h)\), \((\mx h,\bs \psi)\),
and \((\mx h,\mx h)\).  The higher angular derivative channels
\(d^2K,d^3K,d^2H,d^3H\) control the local curvature and Taylor remainder.
They are not normalized kernel inner products, so their fallback caps are
derivative-size Cauchy-Schwarz bounds.  The superscript \({\rm SS}\) is reserved for
support-support quantities such as \(\mathbf G_{\rm SS}\), \(d_{\rm SS}\), and
\(\eta_{\rm SS}\), not for channel labels.

\subsection{Quadratic exponential sums and finite-aperture cancellation}
\label{subsec:quadratic-sum-literature}

A central analytic component of this paper is a nonasymptotic theory of
incomplete quadratic exponential sums with channel-dependent coefficient
sequences, adapted to finite apertures and dual certificates.  The scalar estimates are stated for a generic summation
length \(N\):
\[
        T_N(a;\omega_1,\omega_2)
        =
        \sum_{n=0}^{N-1}
        a_n e^{i(\omega_1 n+\omega_2 n^2)} .
\]
In the Fresnel certificate these estimates are applied with \(N=N_r\), where
\(N_r\) is the number of aperture samples.  The physical aperture length is
\(
        D_{\rm ap}=(N_r-1)d .
\)
The phase coordinates \((\omega_1,\omega_2)\) are not abstract frequencies:
they are the ordered linear and curvature mismatches between an evaluation
point and a source point.  Likewise, the coefficient sequences \(a=(a_n)\) are not arbitrary tapers.
They are generated by the fixed certificate taper \(\rho\), the Fresnel kernel,
the gauged angular tangent channels, and the higher angular derivative channels
entering the Hermite dual certificate.  These finite quadratic sums describe how wavefront curvature enters the
recovery theorem.

The classical theory of exponential sums provides the natural background.
Bounds for sums of the form
\(
        \sum_n a_n e^{i\phi(n)}
\)
are fundamental in analytic number theory and harmonic analysis, including the
van der Corput method, Weyl differencing, completion, exponent-pair estimates,
and the Hardy-Littlewood circle method
\cite{GrahamKolesnik1991,Vaughan1997,Montgomery1994,IwaniecKowalski2004}.
For quadratic phases, the literature includes complete Gauss sums, incomplete
Gauss sums, theta sums, and finite or smoothly truncated quadratic sums
\cite{FiedlerJurkatKorner1977,AkarsuMarklof2013,MarklofWelsh2023}.  These
works give powerful cancellation estimates and asymptotic information for
arithmetic phase families.

The estimates required here have a different mathematical target.  The number
of aperture samples \(N_r\) is fixed by the sensing device, and the constants
must be explicit enough to enter a strict deterministic recovery inequality,
\(
        \mathfrak C_{\rm QPAC}^{\mathscr E}<1 .
\)
Moreover, the same quadratic phase must be controlled simultaneously for
several structured coefficient families: the kernel channel, the normalized
tangent channels, and the second and third angular derivative channels used in
the local curvature argument.  The estimates are therefore not large-\(N_r\)
asymptotic bounds for a single arithmetic sum.  They are quantitative,
channel-dependent bounds whose constants become interpolation, curvature, and
leakage budgets in a TV dual certificate.

A second essential feature is support uniformity.  The theorem does not require
a bound at one fixed pair of points.  It requires certified envelopes over
physical range-angle cells:
\[
        \sup_{(q,p)\in\Omega}
        \left|
        \sum_{n=0}^{N_r-1}
        a_n(q,p)
        e^{i(\omega_1(q,p)n+\omega_2(q,p)n^2)}
        \right|.
\]
This distinction is structural.  The certificate must work uniformly for every
support in a prescribed admissible class and for every nonzero amplitude vector
on that support.  Consequently, the analysis must track not only cancellation
at a point in phase space, but cancellation over a continuum of Fresnel
interactions.

QPAC is the resulting finite-aperture quadratic-sum bounds.  It combines
three complementary mechanisms.  The derivative branch is a discrete
nonstationary-phase estimate controlled by lower bounds for
\(\omega_1+(2n+1)\omega_2\) modulo \(2\pi\).  The lag-correlation branch is
related in spirit to Weyl differencing, but it is formulated for
support-uniform finite-aperture certification: after squaring the sum and
splitting into residue classes, it retains the complete residue lag-pair phase,
including both the linear and quadratic contributions, and controls
constructive alignment through interval cosine majorants.  The residue-linear
branch treats near-rational curvature by decomposing the aperture into residue
classes and applying a linear Abel estimate to the remaining oscillation.

Thus the contribution is not an asymptotic exponent-pair theorem, nor a direct
application of a black-box Gauss-sum estimate. It is a
nonasymptotic calculus for incomplete quadratic sums with channel-dependent
coefficient sequences and explicit support-uniform envelopes. Classical bounds control oscillation; the bounds
developed here control the certificate budgets needed for recovery.  These
budgets enter the three analytic tasks in the Fresnel TV certificate:
invertibility of the Hermite interpolation system, strict angular decay near
the support, and off-support leakage control on the far region.  In this sense,
QPAC turns finite quadratic cancellation into a deterministic recovery
principle for near-field super-resolution.

% ============================================================
\subsection{Scalar QPAC branches}
\label{subsec:weighted_qsum_branches_main}
% ============================================================

For \(x\in\mathbb R\), define the distance to the \(2\pi\)-lattice by
\(\dist_{2\pi}(x):=\min_{k\in\mathbb Z}|x-2\pi k|\).
For a phase pair \((\omega_1,\omega_2)\), define the extended derivative
separation
\begin{equation}
\label{eq:dN_plus_main}
        d_N^+(\omega_1,\omega_2)
        :=
        \min_{0\le n\le N-1}
        \dist_{2\pi}
        \bigl(\omega_1+(2n+1)\omega_2\bigr).
\end{equation}
The endpoint \(n=N-1\) is included because the fourth-order Abel summation
argument uses one auxiliary endpoint.  This is the separation quantity used
throughout the QPAC estimates.

We use \(\RS\) for an auxiliary pointwise residue-class split.  The notation
\(\LCS\) is reserved for the support-uniform lag-correlation construction
that expands squared residue sums, retains complete lag-pair phases, and
enters the QPAC support classes.

For \(a\in\mathbb C^N\), define the best-branch scalar QPAC bound
\begin{equation}
\label{eq:Bbest_qpac_main}
\begin{aligned}
     B_{\rm best}^{(4,Q_{\max})}
        (a;\omega_1,\omega_2;N)
        :=
        \min\Bigl\{&
        \|a\|_1,\,
        B_{\rm der}^{(4)}(a;\omega_1,\omega_2;N),\\
        &
        B_{\RS}^{(Q_{\max})}(a;\omega_1,\omega_2;N),\\
        &
        B_{\rm res,lin}^{(Q_{\max})}(a;\omega_1,\omega_2;N)
        \Bigr\}.
\end{aligned}
\end{equation}
The minimum contains the universal \(\ell^1\) cap, two analytic pointwise
cancellation bounds, and one auxiliary residue-split bound.  The derivative
branch applies when the discrete phase increments remain uniformly away from
\(2\pi\mathbb Z\).  The quantity
\(B_{\RS}^{(Q_{\max})}\) is an exact residue decomposition followed by the
triangle inequality across residue classes; it should not be confused with
the support-uniform lag-correlation branch \(\LCS\).  The latter is
constructed separately at the cell level by expanding squared residue sums,
retaining the complete lag-pair phase, and bounding the real lag terms by
interval cosine majorants.  The residue-linear branch treats near-rational
quadratic phases by reducing the residual oscillation on each residue class
to a linear one. Explicit formulas are
collected in \Cref{app:qsum}; the support-uniform cellwise conversion is given
in \Cref{thm:app_cellwise_support_uniform_qpac}.

\begin{defn}[Pointwise residue split for the lag-correlation branch]
\label{def:pointwise_lc}
Fix \(Q\in\{2,\ldots,N\}\).  Write \(n=s+Qm\), set
\(a_m^{(Q,s)}:=a_{s+Qm}\), \(M_s:=\#\{m:s+Qm\le N-1\}\), and define
\(
        \alpha_s:=\omega_1s+\omega_2s^2,~
        \beta_s:=Q\omega_1+2Qs\omega_2,~
        \gamma:=Q^2\omega_2 .
\)
For each residue class,
\(
        T_s^{(Q)}
        :=
        \sum_{m=0}^{M_s-1}
        a_m^{(Q,s)}e^{i(\beta_s m+\gamma m^2)} .
\)
Then
\[
        T_N(a;\omega_1,\omega_2)
        =
        \sum_{s=0}^{Q-1}e^{i\alpha_s}T_s^{(Q)},
        ~
        |T_N(a;\omega_1,\omega_2)|
        \le
        \sum_{s=0}^{Q-1}|T_s^{(Q)}|.
\]
The pointwise residue-split precursor used below is
\[
        B_{\RS}^{(Q_{\max})}(a;\omega_1,\omega_2;N)
        :=
        \min_{2\le Q\le Q_{\max}}
        \sum_{s=0}^{Q-1}|T_s^{(Q)}|.
\]
We use \(Q\ge2\) for nontrivial comparison; allowing \(Q=1\) gives the exact
quantity \(|T_N(a;\omega_1,\omega_2)|\).
\end{defn}

The pointwise bounds introduced above are now collected into a single
estimate.  Each candidate in
\(B_{\rm best}^{(4,Q_{\max})}\) bounds the same weighted quadratic sum, so
taking their minimum preserves the upper-bound property.

\begin{thm}
\label{thm:qpac_sum_bounds_main}
Let \(N\ge9\), and let \(Q_{\max}\in\{2,\ldots,N\}\).  Let \(a\in\mathbb C^N\) satisfy the
fourth-order flat-end condition
\(
        a_0=a_1=a_2=a_3=0,
        ~
        a_{N-4}=a_{N-3}=a_{N-2}=a_{N-1}=0.
\)
Then, for every \((\omega_1,\omega_2)\in\mathbb R^2\),
\begin{equation}
\label{eq:scalar_qpac_bound_main}
        |T_N(a;\omega_1,\omega_2)|
        \le
        B_{\rm best}^{(4,Q_{\max})}
        (a;\omega_1,\omega_2;N).
\end{equation}
\end{thm}

\begin{rem}[Normalization is separate from the scalar QPAC estimate]
The scalar estimate in \eqref{eq:scalar_qpac_bound_main} applies to the
inserted coefficient sequence \(a\).  It does not, by itself, know whether
that coefficient sequence came from a normalized kernel.  Thus the unit
Cauchy-Schwarz bound
\(
        |K|,\ |H|,\ |dK|,\ |dH|\le 1
\)
follows separately from the weighted Cauchy-Schwarz inequality for the
normalized channels.  It is not part
of the scalar QPAC theorem.
For the higher derivative channels
\(
        d^2K,\ d^3K,\ d^2H,\ d^3H,
\)
the unit cap is not valid.  These channels contain angular derivative
multipliers, so their natural size may exceed one.  Their fallback caps are
therefore derivative-size Cauchy bounds.
\end{rem}
\subsection{From scalar sums to certificate envelopes}
\label{subsec:scalar_to_certificate_envelopes}

The scalar QPAC estimate is applied to the Fresnel certificate through the
gauged channel reductions.  The normalized Hermite channels \(K,H,dK,dH\) control support interpolation
and first tangent interactions.  The higher evaluation-side angular derivative
channels
\(\partial_\theta^2K,\partial_\theta^3K,\partial_\theta^2H,\partial_\theta^3H\)
control the near-support curvature and Taylor remainder.
For a fixed ordered pair, each channel is a quadratic sum with the fixed taper
\(\rho\) absorbed into a channel-dependent coefficient sequence.  For the recovery theorem, these
pairwise estimates are converted into support-uniform envelopes over physical
support/evaluation cells.  This conversion uses interval enclosures of the
phase variables, tangent-arc hulls for normalized tangent factors, and
dictionary envelopes for the higher derivative profiles.  The complete
support-uniform construction is given in
\Cref{thm:app_cellwise_support_uniform_qpac}.  The resulting envelopes are the
constants entering
\(\eta_{\rm SS},~ \eta_{\rm near}(\varpi_{\rm loc}),~ \eta_{\rm far}(\varpi_{\rm loc}).\)
% ============================================================
\subsection{Physical domain, taper, and tangent normalization}
\label{subsec:physical_domain_taper_main}

We summarize the normalization used by the gauged Hermite channels.  Let
\(
        \alpha:=\tfrac{d}{r},
        ~
        \tau:=\tfrac{d\cos\theta}{r}.
\)
The conservative tangent-parameter hull is
\[
        I_\tau
        =
        \operatorname{hull}
        \left\{
        \tfrac{d\cos\theta}{r}:
        r\in[r_{\min},r_{\max}],
        ~
        \theta\in[\theta_{\min},\theta_{\max}]
        \right\},
\]
where \(\operatorname{hull}(E)\) denotes the smallest closed interval containing
the real set \(E\). Because \(\cos\theta\) may be negative on \((0,\pi)\), interval enclosures for
\(\tau\) are always computed by the full endpoint hull over the relevant
range-angle cell.

Let \(b_n=\rho_n/W_0\), where \(W_0=\sum_n\rho_n>0\), and define
\(
        \bar n=\sum_n b_n\,n,
        ~
        \overline{n^2}=\sum_n b_n\,n^2,
\)
\(
        x_n=n-\bar n,
        ~
        y_n=n^2-\overline{n^2}.
\)
The tangent Gram polynomial is
\[
        q(\tau):=\sum_{n=0}^{N_r-1}b_n(x_n+\tau y_n)^2 .
\]
Assume the taper is tangent-nondegenerate on \(I_\tau\), namely
\(
        q_{\min}
        :=
        \min_{\tau\in I_\tau}q(\tau)
        >0.
\)
Set
\(
        q_{\max}:=\max_{\tau\in I_\tau}q(\tau).
\)
The normalized tangent profile is
\(
       \mathbf h_\tau
:=
\bigl(
h_\tau[0],\ldots,h_\tau[N_r-1]
\bigr)^T,
~
h_\tau[n]
:=
\frac{x_n+\tau y_n}{\sqrt{q(\tau)}},
        ~
        \sum_n b_nh_\tau[n]^2=1.
\)
Writing \(\kappa_0=k_\lambda d\), the physical angular tangent energy is
\(
        \|\partial_\theta\psi_{r,\theta}\|_\rho^2
        =
        (\kappa_0\sin\theta)^2q(\tau).
\)
With
\(
        s_{\min}:=
        \min_{\theta\in[\theta_{\min},\theta_{\max}]}\sin\theta,
\)
the uniform constants used in the main theorem are
\(
      \underline\sigma^2
:=
\kappa_0^2s_{\min}^2q_{\min},
~
\overline\sigma_\theta
:=
\kappa_0\sqrt{q_{\max}}.
\)
The normalized profile \(h_\tau\) is the tangent factor used in the Hermite
channels \(H,dK,dH\).  Full interval details are used in the support-uniform
cell construction in \Cref{thm:app_cellwise_support_uniform_qpac}.

% ============================================================
\subsection{Gauged Fresnel atoms and ordered phase increments}
\label{subsec:gauged_pairwise_coefficients_main}
% ============================================================

Write the ungauged Fresnel atom at \(a=(r_a,\theta_a)\) as
\[
        s_a[n]=e^{i\mu_a n-i\eta_a n^2},
        ~
        \mu_a:=k_\lambda d\cos\theta_a,
        ~
        \eta_a:=\tfrac{k_\lambda d^2}{2r_a}\sin^2\theta_a .
\]
The gauged normalized atom has the form
\(\psi_a[n]=W_0^{-1/2}e^{-i\chi_a}s_a[n]\), where the gauge phase \(\chi_a\)
is independent of \(n\).  It is chosen so
that the angular derivative is orthogonal to the atom:
\(\langle \psi_a,\partial_\theta\psi_a\rangle_\rho=0.\)
With this gauge,
\[
        \partial_\theta\psi_a[n]
        =
        i\widetilde u_a(n)\psi_a[n],
        ~
        \widetilde u_a(n)
        =
        -\kappa_0\sin\theta_a(x_n+\tau_a y_n),
        ~
        \tau_a=\tfrac{d\cos\theta_a}{r_a}.
\]
The normalized gauged tangent atom is
\(
        h_a
        :=
        \tfrac{\partial_\theta\psi_a}
        {\|\partial_\theta\psi_a\|_\rho}.
\)
Since \(\sin\theta_a>0\), the normalized tangent is exactly
\begin{equation}
\label{eq:h_profile_main}
h_a[n]
=
-i\,h_{\tau_a}[n]\psi_a[n]
=
-iW_0^{-1/2}
\frac{x_n+\tau_a y_n}{\sqrt{q(\tau_a)}}
e^{-i\chi_a}s_a[n].
\end{equation}

For two physical points \(u\) and \(v\), define the ordered phase increments
by
\(\overline{s_u[n]}s_v[n] = e^{i(\omega_{1,uv}n+\omega_{2,uv}n^2)}.\)
Since \(s_a[n]=e^{i\mu_a n-i\eta_a n^2}\), we obtain
\begin{equation}
\label{eq:ordered_phase_increments_main}
        \omega_{1,uv}:=\mu_v-\mu_u,
        ~
        \omega_{2,uv}:=\eta_u-\eta_v.
\end{equation}
In the support/evaluation convention used below, \(e\) denotes the evaluation
point and \(s\) denotes the support point.  Thus
\(\omega_1=\mu_s-\mu_e\) and \(\omega_2=\eta_e-\eta_s\).

The four normalized gauged pairwise kernels are
\(K_{uv}:=\langle\psi_u,\psi_v\rangle_\rho, ~ H_{uv}:=\langle\psi_u,h_v\rangle_\rho,\)
and
\(dK_{uv}:=\langle h_u,\psi_v\rangle_\rho, ~ dH_{uv}:=\langle h_u,h_v\rangle_\rho.\)
With the ordering in \eqref{eq:ordered_phase_increments_main}, these kernels
are quadratic sums with the following coefficient families:
\begin{equation}
\label{eq:basic_coeff_families_main}
        a^K_n:=b_n,
\end{equation}
\begin{equation}
\label{eq:H_coeff_family_main}
        a^H_\tau[n]
        :=
        b_nh_\tau[n]
        =
        b_n
        \tfrac{x_n+\tau y_n}{\sqrt{q(\tau)}}.
\end{equation}
Set \(a^{dK}_\tau[n]:=a^H_\tau[n]\), and
\begin{equation}
\label{eq:dH_coeff_family_main}
        a^{dH}_{\tau,\tau'}[n]
        :=
        b_nh_\tau[n]h_{\tau'}[n].
\end{equation}
Writing \(g_{uv}:=e^{i(\chi_u-\chi_v)}\), the exact reductions are
\[
\begin{aligned}
K_{uv}
&=g_{uv}T_{N_r}(a^K;\omega_{1,uv},\omega_{2,uv}),\\
H_{uv}
&=-ig_{uv}T_{N_r}(a^H_{\tau_v};
                  \omega_{1,uv},\omega_{2,uv}),\\
dK_{uv}
&=ig_{uv}T_{N_r}(a^{dK}_{\tau_u};
                 \omega_{1,uv},\omega_{2,uv}),\\
dH_{uv}
&=g_{uv}T_{N_r}(a^{dH}_{\tau_u,\tau_v};
                \omega_{1,uv},\omega_{2,uv}).
\end{aligned}
\]
Because every coefficient sequence is \(b_n=\rho_n/W_0\) multiplied by a
polynomial or a rationally normalized affine expression in \(n\), all four
families inherit the fourth-order flat-end condition from the taper.

% ============================================================
\subsection{Higher angular derivative coefficient families}
\label{subsec:higher_derivative_coefficients_main}
% ============================================================

The higher angular derivative channels are needed to control local angular
variation of the gauged kernel.  Let
\(
        \alpha=\tfrac{d}{r},
        ~
        c=\cos\theta,
        ~
        s_\theta=\sin\theta.
\)
Set
\(
        I_\alpha
        :=
        \left[\tfrac{d}{r_{\max}},\tfrac{d}{r_{\min}}\right].
\)
For \((\alpha,\theta)\in I_\alpha\times[\theta_{\min},\theta_{\max}]\), define
\[
        A_n(\alpha,c):=x_n+\alpha c\,y_n,
        ~
        B_n(\alpha,c):=c\,x_n+\alpha(2c^2-1)y_n,
        ~
        u_{\alpha,\theta}[n]:=-\kappa_0 s_\theta A_n(\alpha,c).
\]
Its first and second angular derivatives are
\[
        u_{\theta,\alpha,\theta}[n]:=-\kappa_0 B_n(\alpha,c),
        ~
        u_{\theta\theta,\alpha,\theta}[n]
        :=
        \kappa_0 s_\theta(x_n+4\alpha c\,y_n).
\]
Since
\(\partial_\theta\psi = i\,u_{\alpha,\theta}\psi,\)
we have
\[
        \partial_\theta^2\psi
        =
        \left(
        i\,u_{\theta,\alpha,\theta}
        -
        u_{\alpha,\theta}^2
        \right)\psi,
\]
and
\[
        \partial_\theta^3\psi
        =
        \left(
        i\,u_{\theta\theta,\alpha,\theta}
        -
        3u_{\alpha,\theta}u_{\theta,\alpha,\theta}
        -
        i\,u_{\alpha,\theta}^3
        \right)\psi.
\]
Because the weighted inner product is conjugate-linear in the first argument,
the coefficient multipliers for second and third evaluation-side derivatives
are
\begin{equation}
\label{eq:Q2_main}
        \widetilde \Pi_2(\alpha,\theta;n)
        :=
        -i\,u_{\theta,\alpha,\theta}[n]
        -
        u_{\alpha,\theta}[n]^2,
\end{equation}
and
\begin{equation}
\label{eq:Q3_main}
        \widetilde \Pi_3(\alpha,\theta;n)
        :=
        -i\,u_{\theta\theta,\alpha,\theta}[n]
        -
        3u_{\alpha,\theta}[n]u_{\theta,\alpha,\theta}[n]
        +
        i\,u_{\alpha,\theta}[n]^3.
\end{equation}
For \(m=2,3\), define the \(K\)-type higher derivative coefficient family by
\begin{equation}
\label{eq:Km_coeff_family_main}
        a^{K,m}_{\alpha,\theta}[n]
        :=
        b_n\widetilde\Pi_m(\alpha,\theta;n).
\end{equation}
For a support-side tangent parameter \(\tau_s\), define the \(H\)-type higher
derivative coefficient family by
\begin{equation}
\label{eq:Hm_coeff_family_main}
       a^{H,m}_{\alpha,\theta;\tau_s}[n]
        :=
        b_nh_{\tau_s}[n]\widetilde\Pi_m(\alpha,\theta;n).
\end{equation}
Then the channels \(d^mK\) and \(d^mH\), for \(m=2,3\), are scalar quadratic
sums with coefficient families
\(a^{K,m}_{\alpha,\theta}, ~ a^{H,m}_{\alpha,\theta;\tau_s}.\)

% ============================================================
\subsection{Pairwise channel bounds}
\label{subsec:actual_pairwise_bounds}
% ============================================================

For a fixed ordered pair \((u,v)\), the physical parameters determine
\(
        \omega_{1,uv},
        ~
        \omega_{2,uv},
        ~
        \tau_u,
        ~
        \tau_v,
        ~
        \alpha_u,
        ~
        \theta_u.
\)
Therefore, for each channel, the corresponding coefficient sequence is fixed.
The coefficient sequences are
\(
        a^K,~
        a^H_{\tau_v},~
        a^{dK}_{\tau_u},~
        a^{dH}_{\tau_u,\tau_v},
\)
and, for the higher angular derivative channels,
\(
        a^{K,2}_{\alpha_u,\theta_u},~
        a^{K,3}_{\alpha_u,\theta_u},~
        a^{H,2}_{\alpha_u,\theta_u;\tau_v},~
        a^{H,3}_{\alpha_u,\theta_u;\tau_v}.
\)
For any one of these coefficient sequences \(a\), the exact pairwise channel
magnitude is
\(\left| T_{N_r}(a;\omega_{1,uv},\omega_{2,uv}) \right|.\)
The corresponding analytical pairwise upper bound is
\(B_{\rm best}^{(4,Q_{\max})} (a;\omega_{1,uv},\omega_{2,uv};N_r).\)
This quantity is pointwise in the physical pair \((u,v)\).  It is not a
support-uniform bound.  The distinction is important: the pairwise bound only
needs to hold at one fixed physical configuration, whereas a support-uniform
bound must hold simultaneously over a continuum of configurations.  Therefore
the pairwise bound can be much tighter than the corresponding support-uniform
cell envelope.

\subsection{Certification of the QPAC constants}
\label{subsec:qpac_recovery_constants_main}

The constants entering the QPAC recovery number were defined in
\Cref{subsec:explicit_qpac_number_sec2}.  We now explain how the scalar
quadratic-sum estimates above certify those constants.

Support-support quantities are certified uniformly over the prescribed
support class
\(
        \varnothing\ne\mathfrak S_L
        \subseteq
        \mathfrak S_L^{\rm SS}
        (\mathbf u_{\rm SS};\mathscr E).
\)
For each ordered support-support pair and each channel
\(X\in\{K,H,dK,dH\}\), the best valid QPAC branch gives a certified
upper bound for that ordered interaction.  Any certified upper bounds on the
corresponding support-uniform suprema may be used as
\(
        u_K^{\rm SS},~
        u_H^{\rm SS},~
        u_{dK}^{\rm SS},~
        u_{dH}^{\rm SS},
\)
which form the Hermite envelope matrix \(\mathbf G_{\rm SS}\).
Second-derivative bounds at the support determine
\(E_{\rm curv}\) and \(m_{\rm near}\).
On the complete near sets, bounds for the sums of
derivative-channel magnitudes give
\(D_2(\varpi_{\rm loc})\) and \(D_3(\varpi_{\rm loc})\).
On the far region, bounds for the corresponding sums of
\(K\)- and \(H\)-channel magnitudes give
\(\eta_{\rm far}(\varpi_{\rm loc})\).
In each sum, all source contributions are evaluated at one
common point before taking the supremum.
The channelwise quantities \(g_K^{\rm far}\) and
\(g_H^{\rm far}\) provide an optional coarser estimate.

\begin{rem}[Cellwise certification of suprema]
The constants above are continuum quantities over support/evaluation classes.
A rigorous numerical instantiation may obtain the required cellwise
enclosures from interval phase boxes, certified lower bounds for branch
denominators, tangent-parameter hulls for the normalized tangent families,
and validated enclosures for the higher angular derivative coefficient
families.
\end{rem}
\begin{lem}[Direct common-point continuum enclosure]
\label{lem:direct_common_point_enclosure}
Fix nonnegative constants \(\Gamma_K,\Gamma_H\), a support
\(S=\{p_\ell\}_{\ell=1}^{L}\), an evaluation range row \(i\), and a compact
angular interval \(J\). For \(a\in\{0,2\}\),
with \(\partial_\theta^0X:=X\), define
\[
F_a(\theta)
:=
\sum_{\ell=1}^{L}
\left[
\Gamma_K
\left|\partial_\theta^aK((i,\theta),p_\ell)\right|
+
\Gamma_H
\left|\partial_\theta^aH((i,\theta),p_\ell)\right|
\right].
\]
Let \(T=\{\theta_m\}\subset J\) have fill distance
\[
h_T
:=
\sup_{\theta\in J}
\min_{\theta_m\in T}|\theta-\theta_m|.
\]
If
\[
L_a\ge
\sum_{\ell=1}^{L}
\left[
\Gamma_K
\sup_{\theta\in J}
\left|\partial_\theta^{a+1}K((i,\theta),p_\ell)\right|
+
\Gamma_H
\sup_{\theta\in J}
\left|\partial_\theta^{a+1}H((i,\theta),p_\ell)\right|
\right],
\]
then
\[
\sup_{\theta\in J}F_a(\theta)
\le
\max_{\theta_m\in T}F_a(\theta_m)+h_TL_a.
\]
For a uniform grid containing both endpoints of \(J\), with maximal
consecutive-node spacing \(h\), one has \(h_T\le h/2\). Without endpoint
inclusion, the exact fill distance \(h_T\) defined above must be used.
\end{lem}

\begin{proof}
For every differentiable scalar function \(f\),
\[
\bigl||f(\theta)|-|f(\vartheta)|\bigr|
\le
|f(\theta)-f(\vartheta)|
\le
|\theta-\vartheta|
\sup_{u\in J}|f'(u)|.
\]
Apply this inequality to every kernel term in \(F_a\), sum the
resulting estimates, and choose a nearest node
\(\theta_m\in T\).
\end{proof}

% ============================================================
% ============================================================
% ============================================================
\section{Finite-harmonic lifting for Fresnel dual localization}
\label{sec:lifted_model}
% ============================================================
We now construct the finite-harmonic model used in the numerical section. It
approximates the Fresnel atoms uniformly and yields a semidefinite dual
problem. Throughout this section, the range variable is discretized on
\(\Rset=\{r_1,\ldots,r_{N_d}\}\), whereas the angle remains continuous on
\(\Thetaset\Subset(0,\pi)\).
% We use the complex Frobenius pairing
% \[
%         \langle \mathbf X,\mathbf Y\rangle_F
%         :=
%         \operatorname{trace}(\mathbf Y^H\mathbf X),
% \]
% which is linear in the first argument.

% ============================================================
\subsection{Finite-harmonic Fresnel representation}
\label{subsec:finite_harmonic_fresnel_representation}
% ============================================================

Let \(P_{\rm JA},Q_{\rm JA}\in\mathbb Z_{\ge0}\) be the two Jacobi-Anger
truncation orders used in the Bessel-Vandermonde Fresnel expansion.
Each product \(e^{ip\theta}e^{i2q\theta}\) gives an angular harmonic
\(m=p+2q\), where
\(
        |p|\le P_{\rm JA},
        ~
        |q|\le Q_{\rm JA}.
\)
Set
\(
        I:=P_{\rm JA}+2Q_{\rm JA},
        ~
        N_h:=2I+1,
\)
and define the Vandermonde vector
\[
        \mathbf v_I(\theta)
        :=
        \bigl(e^{-iI\theta},e^{-i(I-1)\theta},
        \ldots,e^{iI\theta}\bigr)^T
        \in\mathbb C^{N_h}.
\]
Repeated Bessel-product frequencies are aggregated.  Thus the finite Fresnel
atom can be written as
\begin{equation}
\label{eq:fresnel_finite_harmonic_expansion}
        a_{\rm Fr}^{P_{\rm JA},Q_{\rm JA}}(r_i,\theta)[n]
        =
        \sum_{m=-I}^{I}
        C_{n,i,m}^{\rm Fr,P_{\rm JA},Q_{\rm JA}} e^{im\theta},
        ~ n=0,\ldots,N_r-1,
\end{equation}
where \(C_{n,i,m}^{\rm Fr,P_{\rm JA},Q_{\rm JA}}=0\) if the frequency \(m\) is not generated by the
truncated Bessel-product expansion.

For each range bin \(r_i\), define the lifted atom
\[
        \mathbf A_i(\theta)
        :=
        \mathbf e_i\mathbf v_I(\theta)^H
        \in\mathbb C^{N_d\times N_h},
\]
where \(\mathbf e_i\) is the \(i\)-th canonical vector in
\(\mathbb C^{N_d}\).  The finite-harmonic Fresnel expansion is represented by
matrices
\(\boldsymbol\Phi_{n,{\rm Fr}}^{P_{\rm JA},Q_{\rm JA}} \in\mathbb C^{N_d\times N_h}, ~ n=0,\ldots,N_r-1,\)
chosen so that
\begin{equation}
\label{eq:fresnel_lifted_approximation_sec4}
    \aFR(r_i,\theta)[n]
        =
        \left\langle
        \boldsymbol\Phi_{n,{\rm Fr}}^{P_{\rm JA},Q_{\rm JA}},
        \mathbf A_i(\theta)
        \right\rangle_F
        +
        e_{n,i,{\rm Fr}}^{P_{\rm JA},Q_{\rm JA}}(\theta).
\end{equation}
The truncation error, controlled by the Bessel-tail estimate proved in the
appendix, satisfies the uniform bound
\begin{equation}
\label{eq:fresnel_lifted_uniform_error_sec4}
        \max_{0\le n\le N_r-1}
        \max_{1\le i\le N_d}
        \sup_{\theta\in\Theta}
        \left|
        e_{n,i,{\rm Fr}}^{P_{\rm JA},Q_{\rm JA}}(\theta)
        \right|
        \le
        \Delta_{P_{\rm JA},Q_{\rm JA}}^{\rm Fr},
        ~
        \Delta_{P_{\rm JA},Q_{\rm JA}}^{\rm Fr}\to0
        \quad\text{as }P_{\rm JA},Q_{\rm JA}\to\infty .
\end{equation}
Consequently,
\[
\sup_{1\le i\le N_d}\sup_{\theta\in\Theta}
\left\|
\mathcal B_{\rm Fr}^{P_{\rm JA},Q_{\rm JA}}\mathbf A_i(\theta)
-
\mathbf a_{\rm Fr}(r_i,\theta)
\right\|_2
\le
\sqrt{N_r}\,
\Delta_{P_{\rm JA},Q_{\rm JA}}^{\rm Fr}
\longrightarrow0.
\]
Thus the finite-harmonic measurement atoms converge uniformly to the Fresnel
measurement atoms.

% ============================================================
\subsection{Lifted atomic norm and finite Fresnel operator}
\label{subsec:lifted_atomic_norm_operator}
% ============================================================

A semi-discrete measure on \(\Rset\times\Thetaset\) is written as
\(\nu = \sum_{i=1}^{N_d}\delta_{r_i}\otimes\mu_i,\)
where each \(\mu_i\) is a finite complex measure on \(\Theta\).  Its lifted
matrix is
\begin{equation}
\label{eq:lifted_measure_matrix}
        \mathbf X_\nu
        :=
        \sum_{i=1}^{N_d}
        \int_{\Theta}
        \mathbf A_i(\theta)\,d\mu_i(\theta)
        \in\mathbb C^{N_d\times N_h}.
\end{equation}

The finite Fresnel lifted operator is
\begin{equation}
\label{eq:finite_fresnel_lifted_operator}
           \bigl(\mathcal B_{\rm Fr}^{P_{\rm JA},Q_{\rm JA}}\mathbf X\bigr)[n]
        :=
        \left\langle
        \boldsymbol\Phi_{n,{\rm Fr}}^{P_{\rm JA},Q_{\rm JA}},
        \mathbf X
        \right\rangle_F,
        ~ n=0,\ldots,N_r-1 .
\end{equation}
In particular,
\(\mathcal B_{\rm Fr}^{P_{\rm JA},Q_{\rm JA}}\mathbf A_i(\theta) = \mathbf a_{\rm Fr}^{P_{\rm JA},Q_{\rm JA}}(r_i,\theta),\)
where \(\mathbf a_{\rm Fr}^{P_{\rm JA},Q_{\rm JA}}\) denotes the finite-harmonic Fresnel atom.

Assume throughout this lifted construction that
\(
        \operatorname{span}
        \left\{
        \mathbf v_I(\theta):
        \theta\in\Theta
        \right\}
        =
        \mathbb C^{N_h}.
\)
This condition holds, in particular, whenever \(\Theta\) contains a
nondegenerate angular interval.  Under this assumption, the atomic gauge
defined below is a norm on
\(\mathbb C^{N_d\times N_h}\). The lifted atomic norm associated with the atoms
\(\{\mathbf A_i(\theta):1\le i\le N_d,\ \theta\in\Theta\}\)
is
\begin{equation}
\label{eq:lifted_atomic_norm}
        \|\mathbf X\|_{\mathcal A}
        :=
        \inf
        \left\{
        \sum_k |c_k|:
        \mathbf X
        =
        \sum_k c_k\mathbf A_{i_k}(\theta_k)
        \right\}.
\end{equation}

Given data \(\mathbf y\in\mathbb C^{N_r}\), the noiseless lifted atomic
program is
\begin{equation}
\label{eq:lifted_primal_noiseless}
        \min_{\mathbf X}
        \|\mathbf X\|_{\mathcal A}
        ~
        \text{subject to}
        ~
        \mathcal B_{\rm Fr}^{P_{\rm JA},Q_{\rm JA}}\mathbf X=\mathbf y .
\end{equation}
When measurement noise or finite-harmonic truncation error is explicitly
allowed, we use the robust version
\begin{equation}
\label{eq:lifted_primal_noisy}
        \min_{\mathbf X}
        \|\mathbf X\|_{\mathcal A}
        ~
        \text{subject to}
        ~
       \|\mathcal B_{\rm Fr}^{P_{\rm JA},Q_{\rm JA}}\mathbf X-\mathbf y\|_2\le\varepsilon ,
\end{equation}
where \(\varepsilon\ge0\) is a prescribed tolerance for measurement noise,
finite-harmonic truncation error, or both.
% ============================================================
\subsection{Lifted dual polynomial and SDP feasibility}
\label{subsec:lifted_dual_polynomial_sdp}
% ============================================================

For
\[
        \boldsymbol\lambda
        =
        (\lambda_0,\ldots,\lambda_{N_r-1})^T
        \in\mathbb C^{N_r},
\]
define
\begin{equation}
\label{eq:lifted_dual_matrix}
        \mathbf Z_{\boldsymbol\lambda}
        :=
        (\mathcal B_{\rm Fr}^{P_{\rm JA},Q_{\rm JA}})^*\boldsymbol\lambda
        =
        \sum_{n=0}^{N_r-1}
        \lambda_n\boldsymbol\Phi_{n,{\rm Fr}}^{P_{\rm JA},Q_{\rm JA}}.
\end{equation}
The lifted dual polynomial on the \(i\)-th range row is
\begin{equation}
\label{eq:lifted_dual_polynomial}
 p_i^{\boldsymbol\lambda}(\theta)
        :=
        \left\langle
        \mathbf A_i(\theta),
        \mathbf Z_{\boldsymbol\lambda}
        \right\rangle_F
        =
        \mathbf e_i^H
        \mathbf Z_{\boldsymbol\lambda}
        \mathbf v_I(\theta).
\end{equation}
With the conjugate-linear-in-the-first convention, this order is the one that
matches the Fresnel TV dual convention
\[
        Q_i^{\rm Fr}(\theta)
        =
        \sum_{n=0}^{N_r-1}
        \zeta_n\overline{a_{\rm Fr}(r_i,\theta)[n]} .
\]
The modulus constraint is unchanged if the conjugate polynomial is used, but
this convention keeps the phase interpolation consistent.

The dual feasibility condition is
\begin{equation}
\label{eq:lifted_dual_feasibility}
        \sup_{1\le i\le N_d}
        \sup_{\theta\in\Theta}
        |p_i^{\boldsymbol\lambda}(\theta)|
        \le 1 .
\end{equation}
For the noiseless program \eqref{eq:lifted_primal_noiseless}, the dual problem
is
\begin{equation}
\label{eq:lifted_dual_noiseless_sec4}
        \max_{\boldsymbol\lambda\in\mathbb C^{N_r}}
        \operatorname{Re}\langle \boldsymbol\lambda,\mathbf y\rangle
        ~
        \text{subject to}
        ~
        \sup_{1\le i\le N_d}\sup_{\theta\in\Theta}|p_i^{\boldsymbol\lambda}(\theta)|\le1 .
\end{equation}
For the robust program \eqref{eq:lifted_primal_noisy}, the dual objective
becomes
\begin{equation}
\label{eq:lifted_dual_noisy}
        \max_{\boldsymbol\lambda\in\mathbb C^{N_r}}
        \operatorname{Re}\langle \boldsymbol\lambda,\mathbf y\rangle
        -
        \varepsilon\|\boldsymbol\lambda\|_2
        ~
        \text{subject to}
        ~
        \sup_{1\le i\le N_d}\sup_{\theta\in\Theta}|p_i^{\boldsymbol\lambda}(\theta)|\le1 .
\end{equation}

The constraint \eqref{eq:lifted_dual_feasibility} is a bounded
trigonometric-polynomial constraint on the physical angular set \(\Theta\).
After zero-filling missing harmonics, each row polynomial has frequencies in
\(\{-I,\ldots,I\}\).  If the modulus constraint is imposed on the full circle
\(\mathbb T\), then
\(
        \sup_{\theta\in\mathbb T}
        |p_i^{\boldsymbol\lambda}(\theta)|
        \le 1
\)
has the standard Fej\'er-Riesz Toeplitz semidefinite representation.  Since
the physical angular set satisfies \(\Theta\subset(0,\pi)\), the full-circle
constraint is conservative for the physical problem.  For a single closed angular arc, the exact restricted constraint can be
represented using the corresponding localizing Toeplitz constraint.  For a
finite union of closed arcs, the corresponding collection of localizing
constraints is required.  In the
numerical experiments we use the full-circle SDP relaxation and search for
near-unit peaks only on the physical angular interval \(\Theta\).

% ============================================================
\subsection{range-angle extraction from lifted dual peaks}
\label{subsec:range_angle_extraction_lifted_dual}
% ============================================================

The lifted dual polynomial provides the localization rule. Peaks of
\(
        |p_i^{\boldsymbol\lambda}(\theta)|
\)
near one identify candidate active atoms. The row index \(i\) gives the range
estimate \(r_i\), and the angular peak location gives the angle estimate.

The dual-saturation argument has the same form for the restricted-angle and
full-circle atom sets and for the noiseless and robust fidelity constraints.
Under strong duality, every atom carrying nonzero mass in a norm-attaining
decomposition of a primal optimizer saturates the corresponding dual modulus
constraint and interpolates the phase of its coefficient. This is an
optimality statement about the support of the lifted optimizer. Equality
between that support and the ground-truth physical support requires a separate
uniqueness or stability argument; in the numerical experiment, that agreement
is observed numerically. At finite \(P_{\rm JA},Q_{\rm JA}\), exact Fresnel
data and finite-harmonic lifted data need not coincide unless the truncation
error vanishes.

\begin{prop}
\label{prop:lifted_dual_saturation}
Let \(\varepsilon\ge0\), and consider the robust lifted program
\eqref{eq:lifted_primal_noisy} and its dual
\eqref{eq:lifted_dual_noisy}. When \(\varepsilon=0\), these reduce to the
noiseless programs
\eqref{eq:lifted_primal_noiseless} and
\eqref{eq:lifted_dual_noiseless_sec4}.
Assume that the primal problem is feasible and that strong duality holds.

Let \(\widehat{\mathbf X}\) be a primal optimizer with a norm-attaining
atomic decomposition
\[
        \widehat{\mathbf X}
        =
        \sum_{\ell=1}^{\widehat L}
        \widehat c_\ell
        \mathbf A_{\widehat i_\ell}(\widehat\theta_\ell),
        \qquad
        \|\widehat{\mathbf X}\|_{\mathcal A}
        =
        \sum_{\ell=1}^{\widehat L}|\widehat c_\ell|,
\]
and let \(\widehat{\boldsymbol\lambda}\) be a dual optimizer. Define
\[
        p_i^{\widehat{\boldsymbol\lambda}}(\theta)
        :=
        \left\langle
        \mathbf A_i(\theta),
        (\mathcal B_{\rm Fr}^{P_{\rm JA},Q_{\rm JA}})^*
        \widehat{\boldsymbol\lambda}
        \right\rangle_F .
\]
Then every active atom in the norm-attaining decomposition satisfies
\[
        p_{\widehat i_\ell}^{\widehat{\boldsymbol\lambda}}
        (\widehat\theta_\ell)
        =
        \frac{\widehat c_\ell}{|\widehat c_\ell|},
        \qquad
        \ell=1,\ldots,\widehat L .
\]
In particular,
\[
        \left|
        p_{\widehat i_\ell}^{\widehat{\boldsymbol\lambda}}
        (\widehat\theta_\ell)
        \right|
        =1,
        \qquad
        \ell=1,\ldots,\widehat L .
\]
Hence every atom in every norm-attaining atomic decomposition of
\(\widehat{\mathbf X}\) belongs to the unit-modulus saturation set of the
lifted dual polynomial.

The same conclusion holds for the full-circle formulation obtained by
replacing \(\Theta\) by \(\mathbb T\) in the atomic norm and dual constraint.
\end{prop}

The phase-completed atom set
\[
\mathcal A_\circ
:=
\left\{
e^{i\phi}\mathbf A_i(\theta):
\phi\in[0,2\pi],\
i=1,\ldots,N_d,\
\theta\in\Theta
\right\}
\]
is compact. Its convex hull is the unit ball of
\(\|\cdot\|_{\mathcal A}\). Therefore finite-dimensional compactness and
Carathéodory's theorem yield a finite norm-attaining atomic decomposition.

After the dual SDP is solved, the practical estimated support is read from the
near-saturation set
\begin{equation}
\label{eq:lifted_near_saturation_set}
   \widehat{\mathcal S}_{\tau_{\rm sat}}
:=
\left\{
(r_i,\theta):
1\le i\le N_d,\ \theta\in\Theta,\ 
|p_i^{\widehat{\boldsymbol\lambda}}(\theta)|
\ge
1-\tau_{\rm sat}
\right\},
\end{equation}
where \(\tau_{\rm sat}>0\) is a numerical saturation tolerance.  Exact saturation is guaranteed for every atom in a norm-attaining
decomposition of the corresponding primal optimizer, in both the noiseless
and robust formulations under strong duality. The near-saturation tolerance is
used to accommodate finite solver accuracy and numerical peak localization.
Saturation identifies atoms of the optimizer; it does not by itself establish
equality with the ground-truth support.

Once the support has been estimated, write
\(
        \widehat S
        =
        \{(\widehat i_\ell,\widehat\theta_\ell)\}_{\ell=1}^{\widehat L},
\)
and form the lifted sensing matrix
\(
        \mathbf D_{\widehat S}
        :=
        \bigl[
        \mathcal B_{\rm Fr}^{P_{\rm JA},Q_{\rm JA}}
        \mathbf A_{\widehat i_1}(\widehat\theta_1)
        \ \cdots\
        \mathcal B_{\rm Fr}^{P_{\rm JA},Q_{\rm JA}}
        \mathbf A_{\widehat i_{\widehat L}}
        (\widehat\theta_{\widehat L})
        \bigr].
\)
The amplitudes are then estimated by
\(
        \widehat{\mathbf c}
        =
        \mathbf D_{\widehat S}^{\dagger}\mathbf y.
\)
If \(\mathbf D_{\widehat S}\) has full column rank, this becomes
\(
        \widehat{\mathbf c}
        =
        \left(
        \mathbf D_{\widehat S}^{H}\mathbf D_{\widehat S}
        \right)^{-1}
        \mathbf D_{\widehat S}^{H}\mathbf y.
\)

% ============================================================
\subsection{Connection with the Fresnel QPAC certificate}
\label{subsec:connection_qpac_lifting}
% ============================================================

Let
\[
        \mathcal S_\star
        =
        \{(i_\ell,\theta_\ell)\}_{\ell=1}^{L}
\]
be the true range-angle support on the range grid, with nonzero amplitudes
\(c_\ell\).  The QPAC theorem provides an ideal Fresnel TV dual certificate of
the form
\begin{equation}
\label{eq:fresnel_qpac_dual_polynomial}
        Q_i^{\rm Fr}(\theta)
        =
        \sum_{n=0}^{N_r-1}
        \zeta_n
        \overline{
        a_{\rm Fr}(r_i,\theta)[n]} .
\end{equation}
With the complex measure-pairing convention used throughout the paper, the
interpolation conditions are
\[
        Q_{i_\ell}^{\rm Fr}(\theta_\ell)
        =
        \tfrac{c_\ell}{|c_\ell|},
        ~
        \ell=1,\ldots,L,
\]
and the strict off-support condition is
\[
        |Q_i^{\rm Fr}(\theta)|<1
        ~
        \text{for }(i,\theta)\notin\mathcal S_\star .
\]
This strict inequality is the dual-certificate mechanism behind uniqueness in
the semi-discrete Fresnel TV problem.

For finite Jacobi-Anger truncation orders
\(P_{\rm JA},Q_{\rm JA}\), define the lifted polynomial associated
with the same QPAC vector \(\boldsymbol\zeta\) by
\begin{equation}
\label{eq:lifted_qpac_same_zeta}
        p_{i,\boldsymbol\zeta}^{P_{\rm JA},Q_{\rm JA}}(\theta)
        :=
        \left\langle
        \mathbf A_i(\theta),
        (\mathcal B_{\rm Fr}^{P_{\rm JA},Q_{\rm JA}})^*
        \boldsymbol\zeta
        \right\rangle_F .
\end{equation}
Equivalently,
\[
        p_{i,\boldsymbol\zeta}^{P_{\rm JA},Q_{\rm JA}}(\theta)
        =
        \sum_{n=0}^{N_r-1}
        \zeta_n
        \overline{
        \left(
        \mathcal B_{\rm Fr}^{P_{\rm JA},Q_{\rm JA}}\mathbf A_i(\theta)
        \right)[n]} .
\]
To compare the exact Fresnel certificate with its finite-harmonic
counterpart, we evaluate both polynomials using the same dual vector
\(\boldsymbol\zeta\).  The uniform truncation bound then controls
the difference between the two polynomials over every range bin and every
admissible angle.

\begin{prop}
\label{prop:lifted_fresnel_dual_perturbation}
Assume that the finite-harmonic approximation
\eqref{eq:fresnel_lifted_approximation_sec4} satisfies the uniform pointwise
error bound \eqref{eq:fresnel_lifted_uniform_error_sec4}.  Then, for every
\(\boldsymbol\zeta\in\C^{N_r}\),
\[
        \sup_{1\le i\le N_d}\sup_{\theta\in\Theta}
        \left|
        p_{i,\boldsymbol\zeta}^{P_{\rm JA},Q_{\rm JA}}(\theta)
        -
        Q_i^{\rm Fr}(\theta)
        \right|
        \le
        \sqrt{N_r}
        \|\boldsymbol\zeta\|_2
        \Delta_{P_{\rm JA},Q_{\rm JA}}^{\rm Fr}.
\]
\end{prop}

\begin{proof}
The pointwise approximation error gives
\[
        \left|
        \left(
        \mathcal B_{\rm Fr}^{P_{\rm JA},Q_{\rm JA}}\mathbf A_i(\theta)
        \right)[n]
        -
        a_{\rm Fr}(r_i,\theta)[n]
        \right|
        \le
        \Delta_{P_{\rm JA},Q_{\rm JA}}^{\rm Fr}
\]
for every \(0\le n\le N_r-1\), \(1\le i\le N_d\), and \(\theta\in\Theta\).  Therefore
\[
        \left\|
        \mathcal B_{\rm Fr}^{P_{\rm JA},Q_{\rm JA}}\mathbf A_i(\theta)
        -
        a_{\rm Fr}(r_i,\theta)
        \right\|_2
        \le
        \sqrt{N_r}\Delta_{P_{\rm JA},Q_{\rm JA}}^{\rm Fr}.
\]
Cauchy-Schwarz yields the stated bound.
\end{proof}

Define the computable perturbation level \( \eta_{P_{\rm JA},Q_{\rm JA}}
        :=
        \sqrt{N_r}
        \|\boldsymbol\zeta\|_2
        \Delta_{P_{\rm JA},Q_{\rm JA}}^{\rm Fr}.\)
Then
\(
        \sup_{1\le i\le N_d}\sup_{\theta\in\Theta}
        \left|
        p_{i,\boldsymbol\zeta}^{P_{\rm JA},Q_{\rm JA}}(\theta)
        -
        Q_i^{\rm Fr}(\theta)
        \right|
        \le
        \eta_{P_{\rm JA},Q_{\rm JA}}.
\)

\begin{rem}[QPAC margin and finite-harmonic peak persistence]
Fix \(\delta>0\), and let
\(
       \mathcal N_\delta(\mathcal S_\star)
:=
\bigcup_{\ell=1}^L
\{(i_\ell,\theta)\in\mathcal Q:
|\theta-\theta_\ell|\le\delta\}.
\)
Define the ideal QPAC margin outside this neighborhood by
\(
        \gamma_{\rm QPAC}(\delta)
        :=
        1-
       \sup_{(i,\theta)\in\mathcal Q\setminus
\mathcal N_\delta(\mathcal S_\star)}
        |Q_i^{\rm Fr}(\theta)|.
\)
If
\(
        \gamma_{\rm QPAC}(\delta)>0
        ~\text{and}~
        \eta_{P_{\rm JA},Q_{\rm JA}}<\tfrac12\gamma_{\rm QPAC}(\delta),
\)
then the finite-harmonic polynomial generated by the QPAC vector
\(\boldsymbol\zeta\) preserves the separation of the ideal Fresnel certificate.
At the true support,
\(
        |p_{i_{\ell},\boldsymbol\zeta}^{P_{\rm JA},Q_{\rm JA}}(\theta_{\ell})|
        \ge
        1-\eta_{P_{\rm JA},Q_{\rm JA}},
        ~
        \ell=1,\ldots,L,
\)
whereas outside \(\mathcal N_\delta(\mathcal S_\star)\),
\(
        |p_{i,\boldsymbol\zeta}^{P_{\rm JA},Q_{\rm JA}}(\theta)|
        \le
        1-\gamma_{\rm QPAC}(\delta)+\eta_{P_{\rm JA},Q_{\rm JA}}.
\)
Hence any threshold
\(
        T\in
        \left(
        1-\gamma_{\rm QPAC}(\delta)+\eta_{P_{\rm JA},Q_{\rm JA}},
        \,
        1-\eta_{P_{\rm JA},Q_{\rm JA}}
        \right)
\)
selects only points inside \(\mathcal N_\delta(\mathcal S_\star)\).  The
interval is nonempty because
\(
        \eta_{P_{\rm JA},Q_{\rm JA}}<\tfrac12\gamma_{\rm QPAC}(\delta).
\)
This statement concerns the finite-harmonic polynomial generated by the QPAC
dual vector \(\boldsymbol\zeta\).  The SDP-computed dual polynomial at finite
\(P_{\rm JA},Q_{\rm JA}\) need not equal this particular polynomial.  The QPAC perturbation
estimate explains the limiting certificate geometry as the finite-harmonic
model approaches the exact Fresnel model, while
\Cref{prop:lifted_dual_saturation} gives the contact-set optimality condition
for the corresponding computed lifted problem. Equality with the
ground-truth support is not asserted by dual saturation alone.
\end{rem}
% ============================================================
% ============================================================
\section{Numerical QPAC evaluations and computational illustrations}
\label{sec:numerics}

We give two numerical evaluations of the QPAC condition in
\Cref{thm:main_exact_recovery_sec2}. The first uses the derivative
branch for a continuously varying two-source class, with one
source on each of two range rows. The second uses lag
correlation for a fixed two-range support at a common bearing.

For each example, we report the three quantities entering
\[
\mathfrak C_{\rm QPAC}
=
\max\{\eta_{\rm SS},\eta_{\rm near},\eta_{\rm far}\}.
\]
All local and far quantities are evaluated numerically at the stated
localization radius. The continuous domains are treated using analytical cell
bounds or grids with analytical padding. The reported values are MATLAB
double-precision evaluations, rounded for display. 

Throughout these two examples,
\(
c_0=3\times10^8\,{\rm m/s},
~
\lambda=c_0/f_c,
~
d=\lambda/2.
\)
The computations are reproduced by the provided scripts
in \cite{Daei2026NearFieldRepro}. The subsequent experiments
illustrate the lifted computational model and the channel
bounds.

\subsection{Derivative-route QPAC example}
\label{sec:numerical_derivative_qpac_example}

\begin{example}
\label{ex:numerical_connected_derivative_qpac}
Let \(N_r=128\), \(f_c=10\,{\rm GHz}\), and
\(
\mathcal R=\{r_1,r_2\}
=\{10\,{\rm m},100\,{\rm m}\}.
\)
Use the fourth-order flat-end taper
\(
\rho_n=\binom{n}{4}\binom{127-n}{4},
~ n=0,\ldots,127,
\)
and set \(b_n=\rho_n/\sum_m\rho_m\). Let
\(
\Theta_{\rm ex}
=\left[\frac{\pi}{2}-0.401,\frac{\pi}{2}+0.401\right],
\)
with support windows
\(
I_-=
\left[\frac{\pi}{2}-0.401,\frac{\pi}{2}-0.399\right],
~
I_+=
\left[\frac{\pi}{2}+0.399,\frac{\pi}{2}+0.401\right].
\)
The prescribed class and evaluation domain are
\[
\mathfrak S_{2,{\rm win}}^{\rm der}
=
\left\{
\{(1,\theta_-),(2,\theta_+)\}:
\theta_-\in I_-,\ \theta_+\in I_+
\right\},
\qquad
\mathcal Q_{\rm ex}=\{1,2\}\times\Theta_{\rm ex}.
\]
Thus the two angles vary independently, while each support
contains one point on each range row. Set
\(
\varpi_{\rm loc}=0.012,~
d_{\rm SS}=1.9,~
s_{2,{\rm SS}}^{\max}=0.002.
\)
On both ordered support cells, the phase bounds give
\(
d_{128}^{+}>1.981>d_{\rm SS},
~
|\sin\omega_2|
<1.802\times10^{-3}<s_{2,{\rm SS}}^{\max}.
\)
Hence the derivative branch supplies uniform support-support
bounds for all four Hermite channels.

The near and far budgets are computed as sums at a common
evaluation point. Each support window uses 101 grid points;
the near and far evaluation grids use 601 and 4001 points
per range row, respectively. The grid maxima are enlarged
using uniform angular derivative bounds for both the source
and evaluation variables. In the far calculation, the
same-row source obeys the angular exclusion defining
\(\mathcal F(S;\varpi_{\rm loc})\); the source on the other
range row is included at every evaluation angle. Thus the
reported budgets cover the continuous class and both
evaluation rows.

\Cref{tab:numerical_connected_derivative_qpac_values}
gives conservative decimal bounds for the calculation. In
particular, the computed QPAC recovery number is below one
for the prescribed support class. The calculation is
reproduced by
\nolinkurl{run_derivative_example_qpac.m}.
\end{example}

\begin{table}[t]
\centering
\begin{tabular}{lc}
\hline
Quantity & Conservative reported bound\\
\hline
\(\eta_{\rm SS}\) & \(<2.6\times10^{-5}\)\\
\(m_{\rm near}\) & \(>4.97\times10^3\)\\
\(D_2\) & \(<3.40\times10^3\)\\
\(D_3\) & \(<2.69\times10^5\)\\
\(\eta_{\rm near}\) & \(<0.601\)\\
\(\eta_{\rm far}\) & \(<0.828\)\\
\(\mathfrak C_{\rm QPAC}\) & \(<0.828<1\)\\
\hline
\end{tabular}
\caption{Derivative-route QPAC evaluation for two independently
varying angles on distinct range rows, with
\(\varpi_{\rm loc}=0.012\). The near and far values include
analytic bounds between grid points.}
\label{tab:numerical_connected_derivative_qpac_values}
\end{table}

\subsection{Correlation-route QPAC example}
\label{sec:numerical_range_lcs_qpac_example}

\begin{example}[Two ranges at a common bearing]
\label{ex:numerical_range_lcs_qpac}
Let \(N_r=256\), \(f_c=100\,{\rm GHz}\), and
\(\mathcal R=\{r_1,r_2\}\), where
\(
r_1\approx2.81\,{\rm m},
\qquad
r_2\approx89.73\,{\rm m}.
\)
Set
\(
\theta_0=\frac{\pi}{2},~
\Theta_{\rm LCS}=[\theta_0-0.15,\theta_0+0.15],
~
\mathcal Q_{\rm LCS}=\{1,2\}\times\Theta_{\rm LCS}.
\)
We consider the fixed support
\(
S_\star=\{(1,\theta_0),(2,\theta_0)\},
~
\mathfrak S_{2,\rm range}^{\LCS}=\{S_\star\}.
\)
The source locations are fixed, while the evaluation angle
varies continuously on both range rows.
For this support,
\(
d_{N_r}^{+}\approx8.1174\times10^{-4}<0.1,
\)
so the selected derivative threshold is not met.
We instead use the lag-correlation bound for support
interpolation and choose
\(
\varpi_{\rm loc}=0.0066.
\)

\Cref{tab:numerical_range_lcs_qpac_values} reports the computed QPAC
quantities. All three reported values are below one, with
\[
\mathfrak C_{\rm QPAC}\approx0.825990<1.
\]
Thus the example illustrates the QPAC sufficient condition
for two sources at a common bearing and different ranges.
The bounds do not depend on their nonzero complex amplitudes.
The calculation is reproduced by
\nolinkurl{run_lcs_example_qpac.m}.
\end{example}

\begin{table}[t]
\centering
\begin{tabular}{lc}
\hline
Quantity & Numerical value\\
\hline
\(\eta_{\rm SS}\) & \(0.079408\)\\
\(\eta_{\rm near}\) & \(0.825990\)\\
\(\eta_{\rm far}\) & \(0.780100\)\\
\(\mathfrak C_{\rm QPAC}\) & \(0.825990\)\\
\hline
\end{tabular}
\caption{Correlation-route QPAC evaluation for the fixed
two-range support, with \(\varpi_{\rm loc}=0.0066\).}
\label{tab:numerical_range_lcs_qpac_values}
\end{table}

Having illustrated the QPAC sufficient condition, we now
turn to localization with the finite-harmonic computational
model.
\subsection{Lifted dual-polynomial localization}
We next solve the robust finite-harmonic program and localize the sources from
its dual polynomial. This experiment illustrates
\Cref{prop:lifted_fresnel_dual_perturbation,prop:lifted_dual_saturation}; it is
separate from the numerical evaluations of the QPAC recovery number above.
We use a uniform linear aperture with
\(N_r=16,~ \lambda=0.3\,{\rm m},~ d=0.031831\,{\rm m}.\)
The lifted SDP and the synthetic Fresnel observations use the untapered
physical measurement aperture.  For the separate gauged Hermite certificate
and QPAC channel-envelope diagnostics, we use the fourth-order binomial
flat-end taper
\[
        \rho_n
        =
        \binom{n}{4}\binom{N_r-1-n}{4},
        ~ n=0,\ldots,N_r-1,
\]
normalized only through the weighted inner products by
\(b_n=\tfrac{\rho_n}{\sum_m\rho_m}.\)
This taper is therefore not part of the lifted measurement operator or the SDP
data-fidelity term; it is used only in the gauged-certificate and QPAC
channel-bound diagnostics.

The aperture length is
\(D_{\rm ap}=(N_r-1)d=0.477465\,{\rm m},\)
and the dimensionless aperture size is
\(k_\lambda D_{\rm ap}=10.\)
The Jacobi-Anger truncation orders are
\(
        P_{\rm JA}=20,~ Q_{\rm JA}=8.
\)
Hence the aggregated harmonic half-bandwidth and harmonic dimension are
\(
        I=P_{\rm JA}+2Q_{\rm JA}=36,
        ~
        N_h=2I+1=73.
\)
The number of retained \((p,q)\) pairs is \(697\), and the full dual SDP
block dimension is \(592\).
The range grid contains \(N_d=8\) logarithmically spaced bins,
\(
r_i
=
1.24340\,12^{(i-1)/7}\ {\rm m},
~
i=1,\ldots,8.
\)
Thus
\(
r_{\min}=1.24340\,{\rm m},
~
r_{\max}=14.9208\,{\rm m}.
\)
Equivalently, the Fresnel curvature parameter
\(
        z_2(r_i)
        =
        \frac{k_\lambda D_{\rm ap}^2}{4r_i}
\)
varies over
\(
        0.08\le z_2(r_i)\le0.96.
\)
Thus the range grid spans a factor of \(12\) in range while the selected
Jacobi-Anger orders keep the finite-harmonic approximation error small.

The scene consists of two point sources at the active range-grid locations
\(
(r_{i_1},\theta_1)
=
(2.52899\,{\rm m},0.94248),
~
(r_{i_2},\theta_2)
=
(7.33591\,{\rm m},2.35619),
\)
with complex amplitudes
\(
c_1=-1.00147-0.34340i,
~
c_2=1.15019+0.13302i.
\)
Synthetic observations are generated from the semi-discrete Fresnel model,
while the convex program uses the finite Jacobi-Anger lifted approximation.
There is no additive measurement noise in this experiment.  The deterministic
Jacobi-Anger truncation bound is therefore used as the robustness radius of
the lifted dual problem.
Let \(\mathbf y_{\rm Fr}\) denote the data vector generated by the
semi-discrete Fresnel measurement model, and let
\(\mathbf y_{\rm JA}\) denote the data vector generated by the finite Jacobi-Anger
lifted approximation at the same source parameters.  Their relative mismatch is
\(
        \frac{\|\mathbf y_{\rm Fr}-\mathbf y_{\rm JA}\|_2}
        {\|\mathbf y_{\rm Fr}\|_2}
        =
        3.00862\times10^{-7}.
\)
The corresponding absolute mismatch is
\(
        \|\mathbf y_{\rm Fr}-\mathbf y_{\rm JA}\|_2
        =
        1.83321\times10^{-6}.
\)

By \eqref{eq:Fresnel_operator_error_bv}, the oracle-calibrated
Jacobi-Anger data-mismatch radius used in this synthetic experiment is
\(
        \varepsilon_{\rm JA}
        :=
        \sqrt{N_r}\,
        \Delta_{P_{\rm JA},Q_{\rm JA}}^{\rm Fr}
        \|\nu_\star\|_{\TV,\mathcal Q}.
\)
For the present synthetic scene,
\(
        \varepsilon_{\rm JA}
        =
        1.75056\times10^{-5},
        ~
        \frac{\varepsilon_{\rm JA}}
        {\|\mathbf y_{\rm Fr}\|_2}
        =
        2.87297\times10^{-6}.
\)
Thus the deterministic truncation radius exceeds the realized
Fresnel-to-lifted data mismatch.

\Cref{fig:dual-polynomial-fresnel-lift} shows the magnitude of the
recovered lifted dual polynomial over the range-angle grid.  Up to
permutation, the two dominant peaks occur at
\(
(\widehat r_1,\widehat\theta_1)
=
(2.52899\,{\rm m},0.942),
~
(\widehat r_2,\widehat\theta_2)
=
(7.33591\,{\rm m},2.356).
\)
After jointly pairing the recovered range-angle points with the true
support, the maximum angular error is
\(
        4.77796\times10^{-4}\ {\rm rad},
\)
and both semi-discrete range bins are recovered exactly.  A least-squares
refit on the paired recovered support gives
\(
        \widehat c_1
        =
        -1.00209-0.34131i,
        ~
        \widehat c_2
        =
        1.15059+0.13236i.
\)
These results illustrate the role of the lifted dual polynomial: its
near-unit peaks identify the active range bins and localize the continuous
angular variables.
The figure and the reported peak locations, support errors, amplitude refit,
and model-mismatch diagnostics are reproduced by the script
\nolinkurl{run_lifted_dual_figure1_repro.m} in the reproducibility
archive \cite{Daei2026NearFieldRepro}.

\begin{figure}[t]
    \centering
    \includegraphics[width=0.78\linewidth]{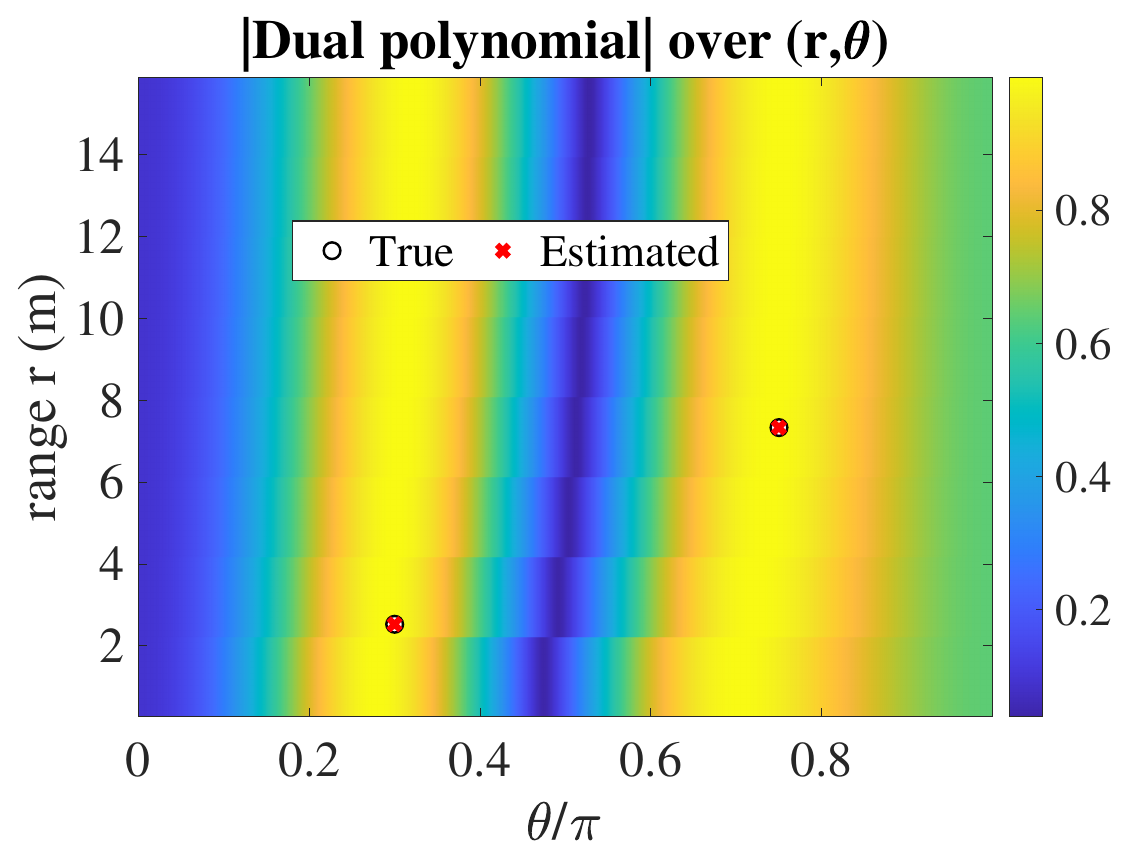}
    \caption{
  Magnitude of the lifted dual polynomial for a two-source scene generated
by the quadratic Fresnel phase model.  The dominant peaks occur near the
true range-angle support.  The maximum paired angular error is
\(4.77796\times10^{-4}\) rad, and both semi-discrete range bins are
recovered exactly.  This figure is reproduced by
\texttt{run\_lifted\_dual\_figure1\_repro.m} in the reproducibility
archive \cite{Daei2026NearFieldRepro}.
    }
    \label{fig:dual-polynomial-fresnel-lift}
\end{figure}

\subsection{QPAC channel envelopes and higher-derivative control}
\label{subsec:numerics_qpac_channel_envelopes}

The final experiment visualizes the passage from one scalar quadratic-sum
estimate to the support-uniform constants used in the recovery theorem.  For
a fixed support point and an evaluation slice, we compare three quantities:
the exact gauged channel magnitude, the pointwise QPAC bound evaluated at the
actual phase pair, and the cellwise QPAC envelope valid simultaneously over
the corresponding physical parameter cell.  The comparison shows both the
validity of the pointwise cancellation estimate and the additional margin
required to make it uniform over a support/evaluation class.
For a fixed support point
\(p_\star=(r_\star,\theta_\star)\)
and an evaluation point \(p=(r,\theta)\), we compare exact channel magnitudes
with two deterministic upper bounds.  The first is the pointwise QPAC bound
evaluated at the actual phase pair
\((\omega_1(p,p_\star),\omega_2(p,p_\star))\).
The second is the corresponding cellwise support-uniform QPAC envelope,
obtained by enforcing the bound simultaneously for every physical
support/evaluation pair in the prescribed cell. The
comparison illustrates the additional conservatism introduced by uniform
cellwise certification.

The first four channels are the normalized Hermite channels
\(K,~ H,~ dK,~ dH.\)
They determine the stability of the Hermite interpolation system.  Here
\(dK\) and \(dH\) denote normalized angular tangent channels; equivalently,
\(\partial_\theta K\) and \(\partial_\theta H\) are obtained from them by
multiplication with the local tangent norm.  These normalized channels satisfy
the weighted Cauchy-Schwarz cap
\(|K|,\ |H|,\ |dK|,\ |dH|\le 1.\)
The higher derivative channels
\(
        \partial_\theta^2K,~
        \partial_\theta^2H,~
        \partial_\theta^3K,~
        \partial_\theta^3H
\)
enter the near-support part of the proof.  The second derivatives determine
the certified negative curvature of the dual certificate at the support, while
the third derivatives control the Taylor remainder for
\(|P_{S,\widetilde {\mathbf v}}(i_\ell,\theta_\ell+t)|^2.\)
Unlike the normalized Hermite channels, these higher derivative channels are
not unit-bounded.  Differentiation introduces aperture-polynomial factors, so
their natural fallback estimates are derivative-size Cauchy bounds rather than
the unit cap.

\Cref{fig:cpam-angle-all-channels} shows the angular slices.  The first row
contains the normalized Hermite channels from the original comparison, while
the second row contains the higher derivative channels needed for the local
curvature and Taylor estimates.  \Cref{fig:cpam-range-all-channels} shows the
corresponding range slices, again keeping the original \(K,H,dK,dH\) panels and
adding the higher derivative channels.
For each channel, the cellwise curve is the object entering
\Cref{thm:app_cellwise_support_uniform_qpac}.  Maximizing these envelopes
over the support-support pairs, complete near sets, and far-region classes produces,
respectively, the constants used in
\(
        \eta_{\rm SS},~
        \eta_{\rm near}(\varpi_{\rm loc}),~
        \eta_{\rm far}(\varpi_{\rm loc}).
\)
The figures therefore display the complete analytical path from one
quadratic phase interaction to the three budgets of the main recovery
theorem: \(K,H,dK,dH\) control Hermite interpolation, the second derivatives
control local curvature, and the third derivatives control the Taylor
remainder. The angular and range slice figures are reproduced by the script
\texttt{run\_bounds\_true\_kernel\_repro.m} in the reproducibility archive
\cite{Daei2026NearFieldRepro}.

\begin{figure}[t]
    \centering

    % ------ Row 1 ----------
    \begin{subfigure}[t]{0.31\textwidth}
        \centering
        \includegraphics[width=\linewidth]{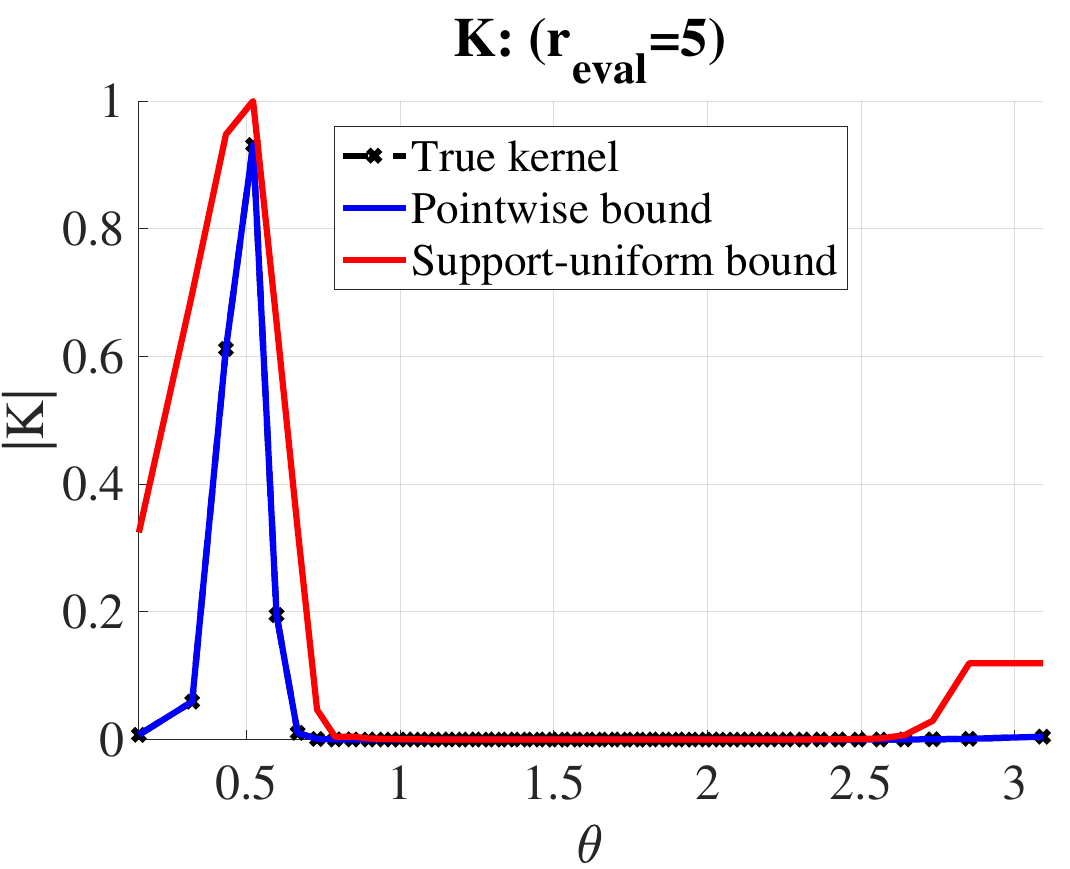}
        \caption{\(K\)}
        \label{fig:cpam-angle-K}
    \end{subfigure}
    \hfill
    \begin{subfigure}[t]{0.31\textwidth}
        \centering
        \includegraphics[width=\linewidth]{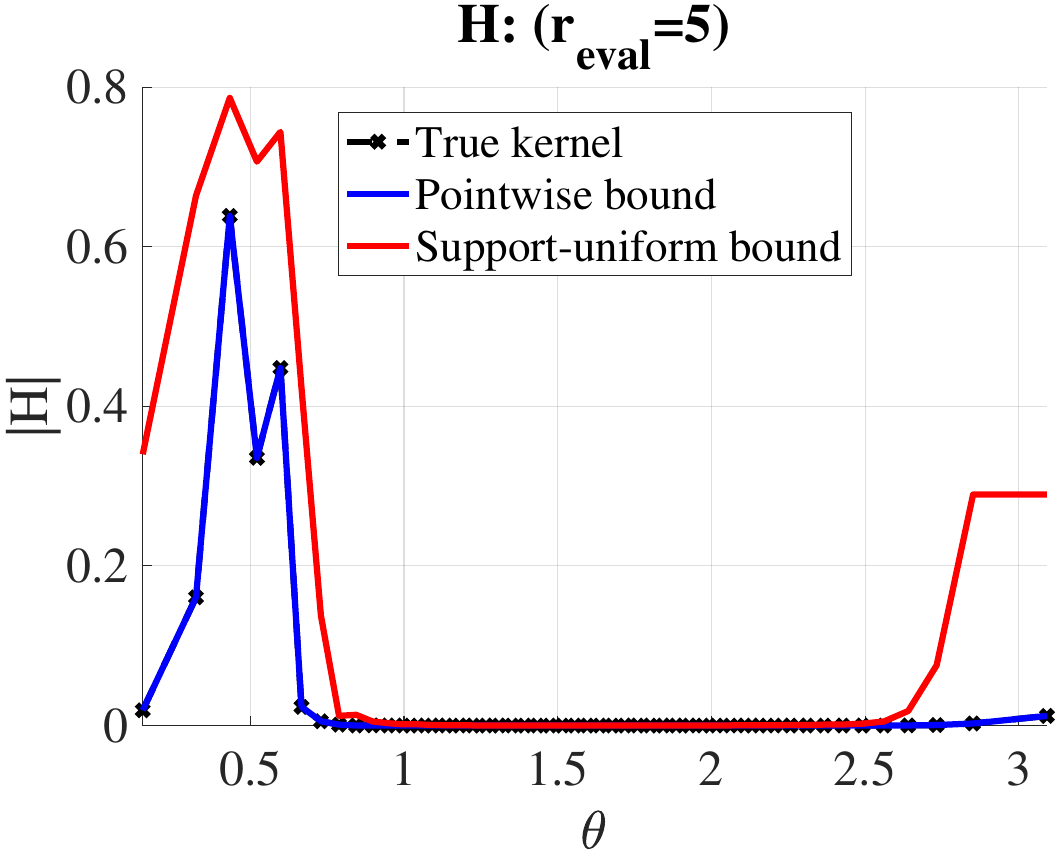}
        \caption{\(H\)}
        \label{fig:cpam-angle-H}
    \end{subfigure}
    \hfill
    \begin{subfigure}[t]{0.31\textwidth}
        \centering
        \includegraphics[width=\linewidth]{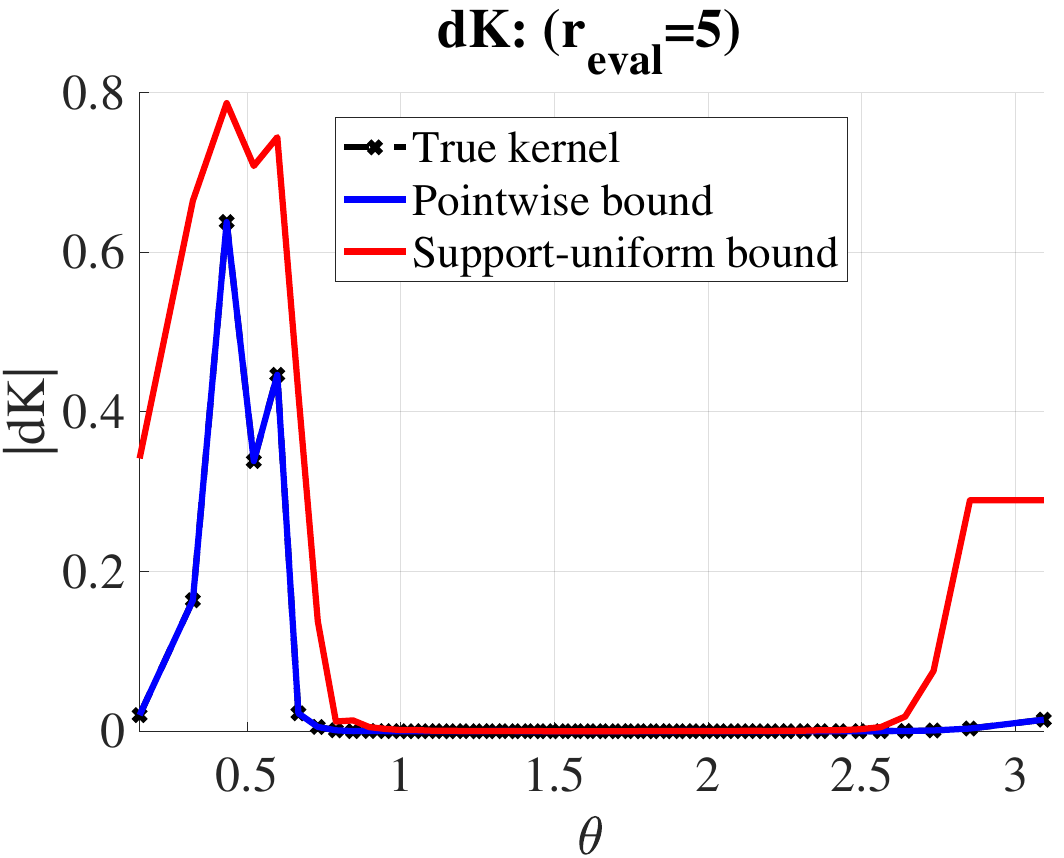}
        \caption{\(dK\)}
        \label{fig:cpam-angle-dK}
    \end{subfigure}

    \vspace{0.8em}

    % ---------- Row 2 ----------
    \begin{subfigure}[t]{0.31\textwidth}
        \centering
        \includegraphics[width=\linewidth]{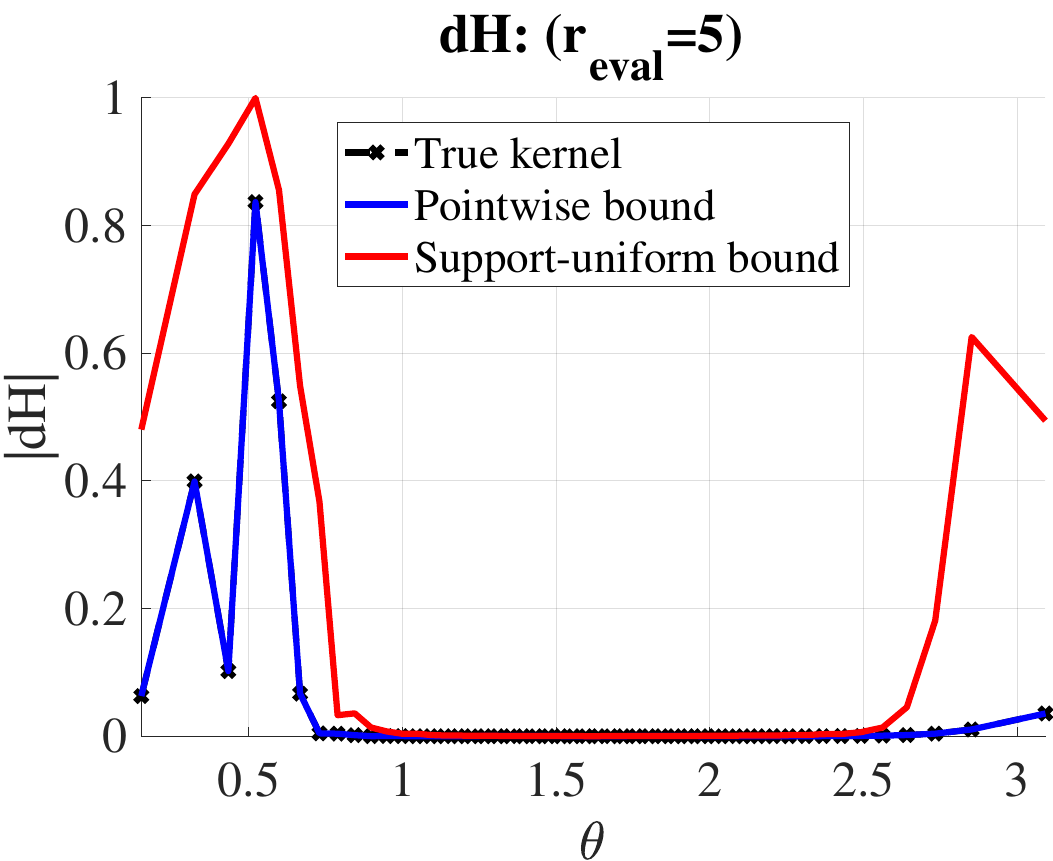}
        \caption{\(dH\)}
        \label{fig:cpam-angle-dH}
    \end{subfigure}
    \hfill
    \begin{subfigure}[t]{0.31\textwidth}
        \centering
        \includegraphics[width=\linewidth]{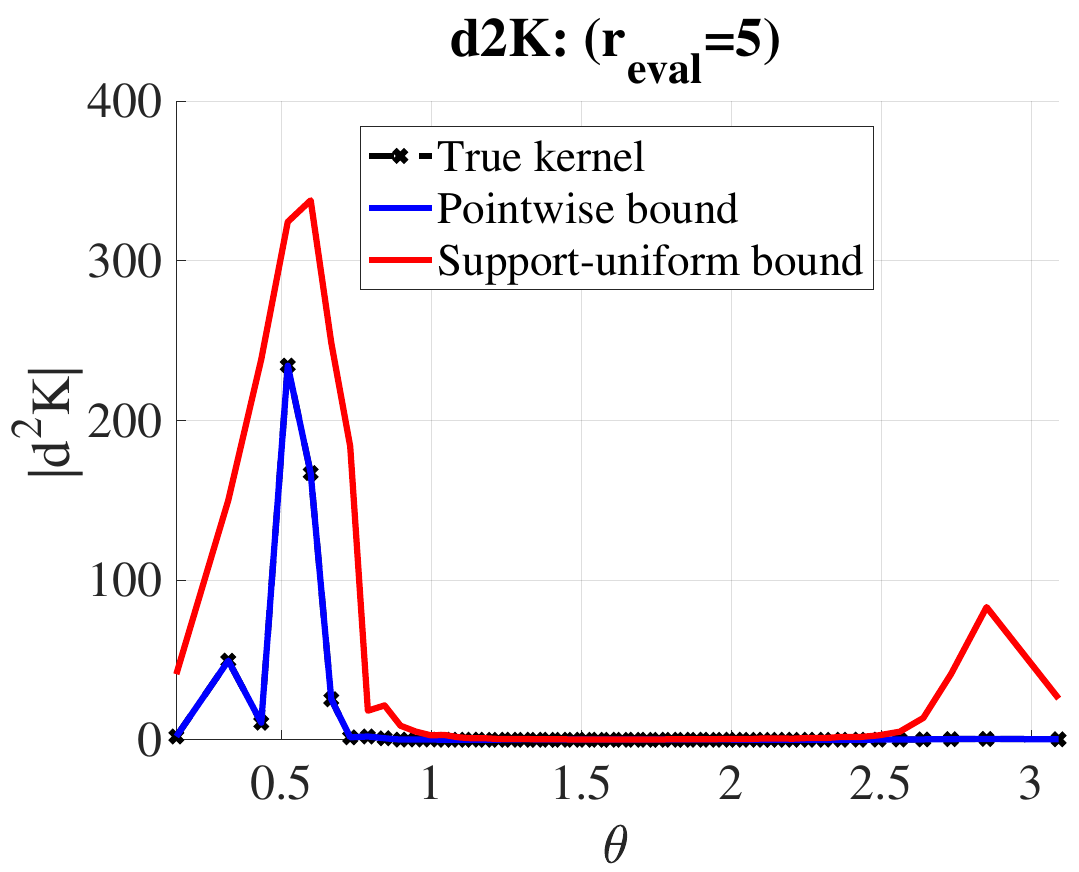}
        \caption{\(\partial_\theta^2K\)}
        \label{fig:cpam-angle-d2K}
    \end{subfigure}
    \hfill
    \begin{subfigure}[t]{0.31\textwidth}
        \centering
        \includegraphics[width=\linewidth]{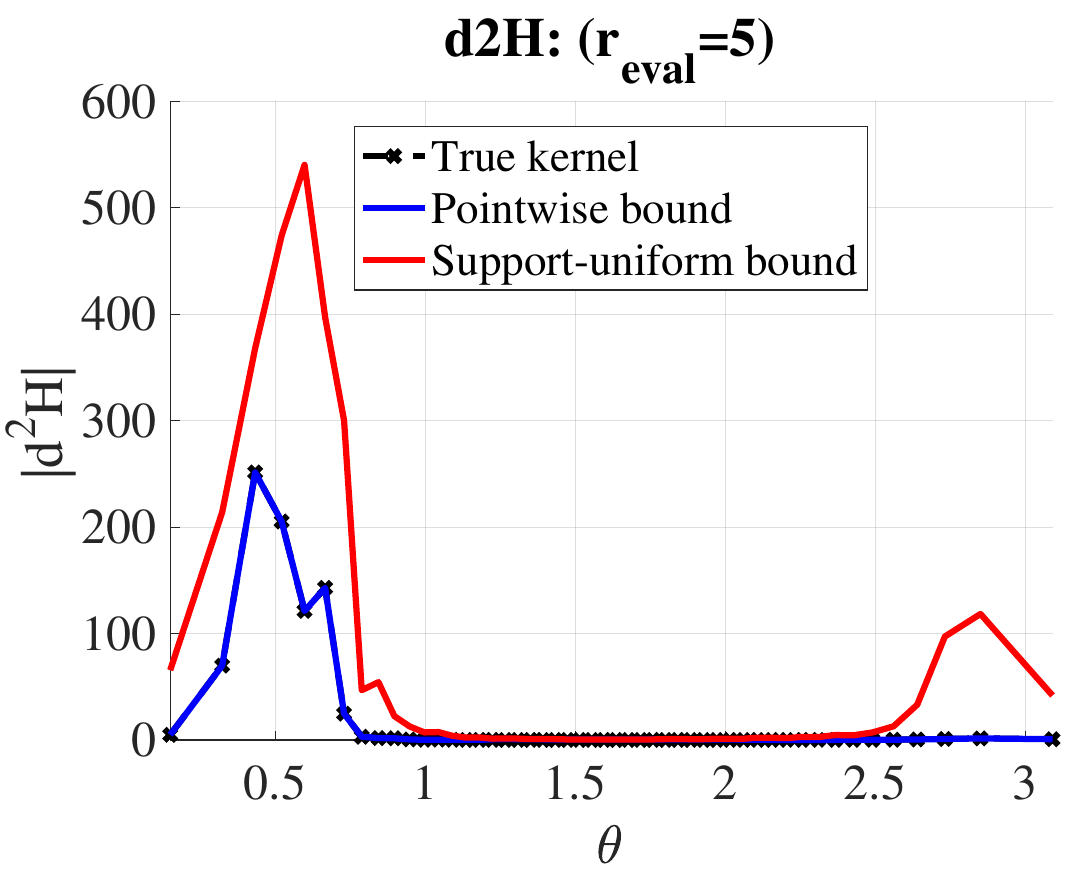}
        \caption{\(\partial_\theta^2H\)}
        \label{fig:cpam-angle-d2H}
    \end{subfigure}

    \vspace{0.8em}

    % ---------- Row 3 ----------
    \makebox[\textwidth][c]{%
    \begin{subfigure}[t]{0.31\textwidth}
        \centering
        \includegraphics[width=\linewidth]{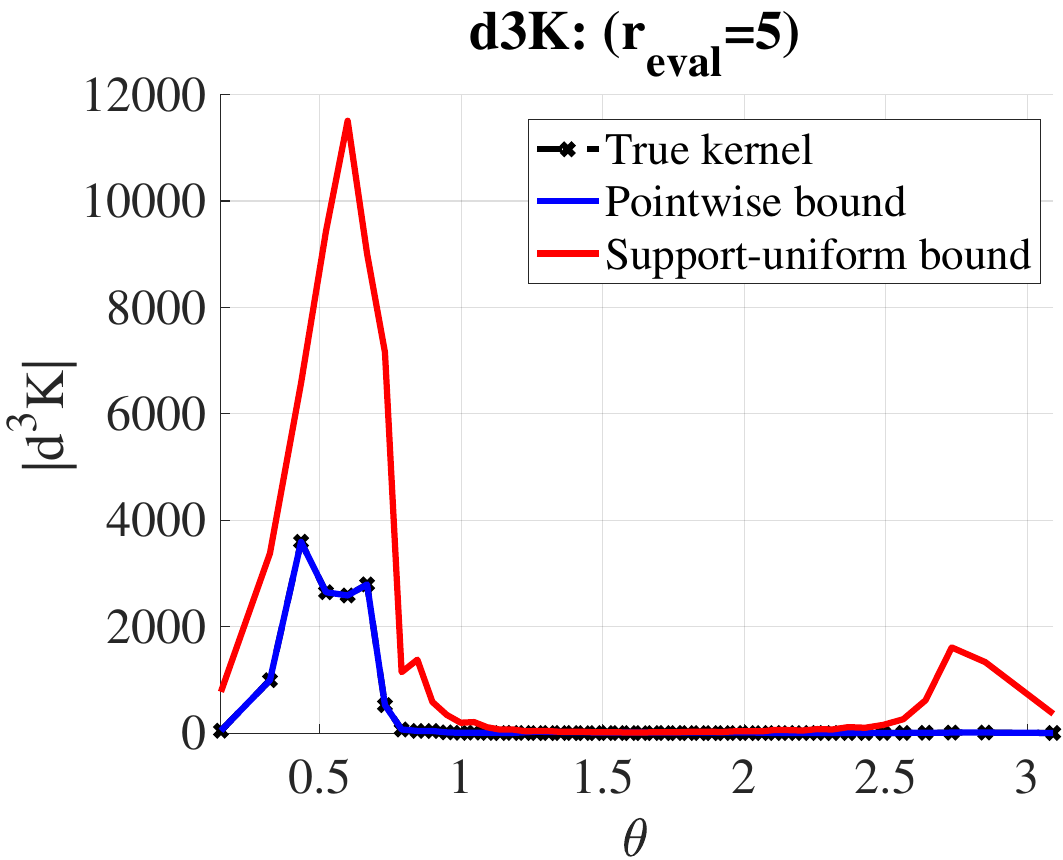}
        \caption{\(\partial_\theta^3K\)}
        \label{fig:cpam-angle-d3K}
    \end{subfigure}
    \hspace{0.05\textwidth}
    \begin{subfigure}[t]{0.31\textwidth}
        \centering
        \includegraphics[width=\linewidth]{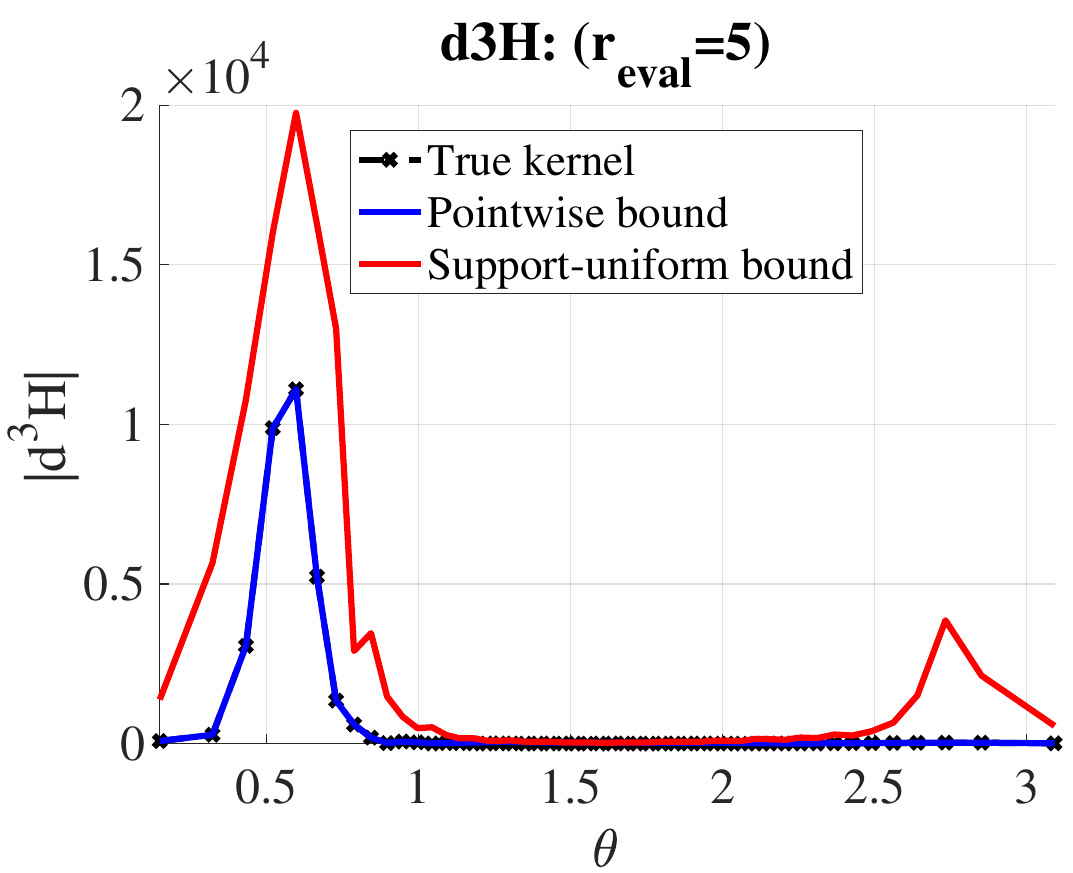}
        \caption{\(\partial_\theta^3H\)}
        \label{fig:cpam-angle-d3H}
    \end{subfigure}%
    }

 \caption{
Angular-slice comparison of exact gauged channel magnitudes, pointwise
QPAC bounds, and cellwise support-uniform QPAC envelopes.  The normalized Hermite
channels \(K,H,dK,dH\) enter interpolation and first-order tangent control,
while the higher derivative channels enter the near-support curvature and
Taylor-remainder estimates. These panels are reproduced by
\texttt{run\_bounds\_true\_kernel\_repro.m} in the reproducibility archive
\cite{Daei2026NearFieldRepro}.
}
    \label{fig:cpam-angle-all-channels}
\end{figure}

\begin{figure}[t]
    \centering

    % ---------- Row 1 ----------
    \begin{subfigure}[t]{0.31\textwidth}
        \centering
        \includegraphics[width=\linewidth]{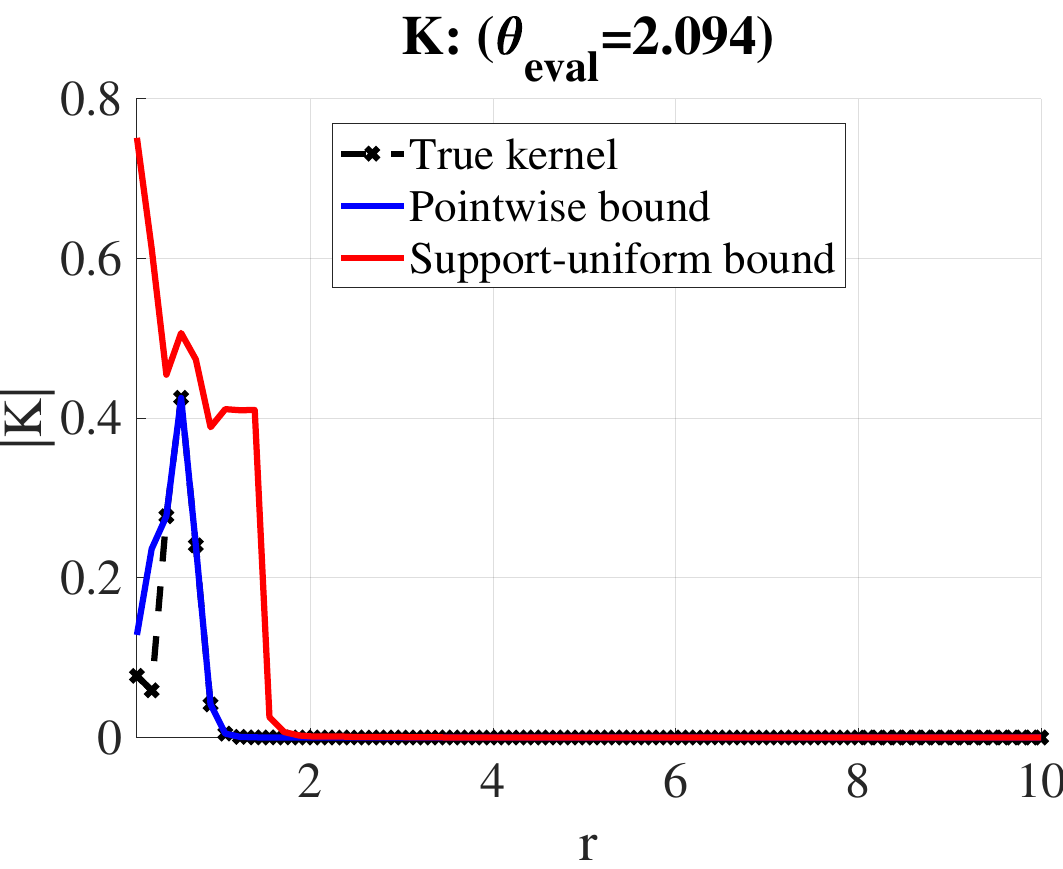}
        \caption{\(K\)}
        \label{fig:cpam-range-K}
    \end{subfigure}
    \hfill
    \begin{subfigure}[t]{0.31\textwidth}
        \centering
        \includegraphics[width=\linewidth,trim={0 0cm 0 0},clip]{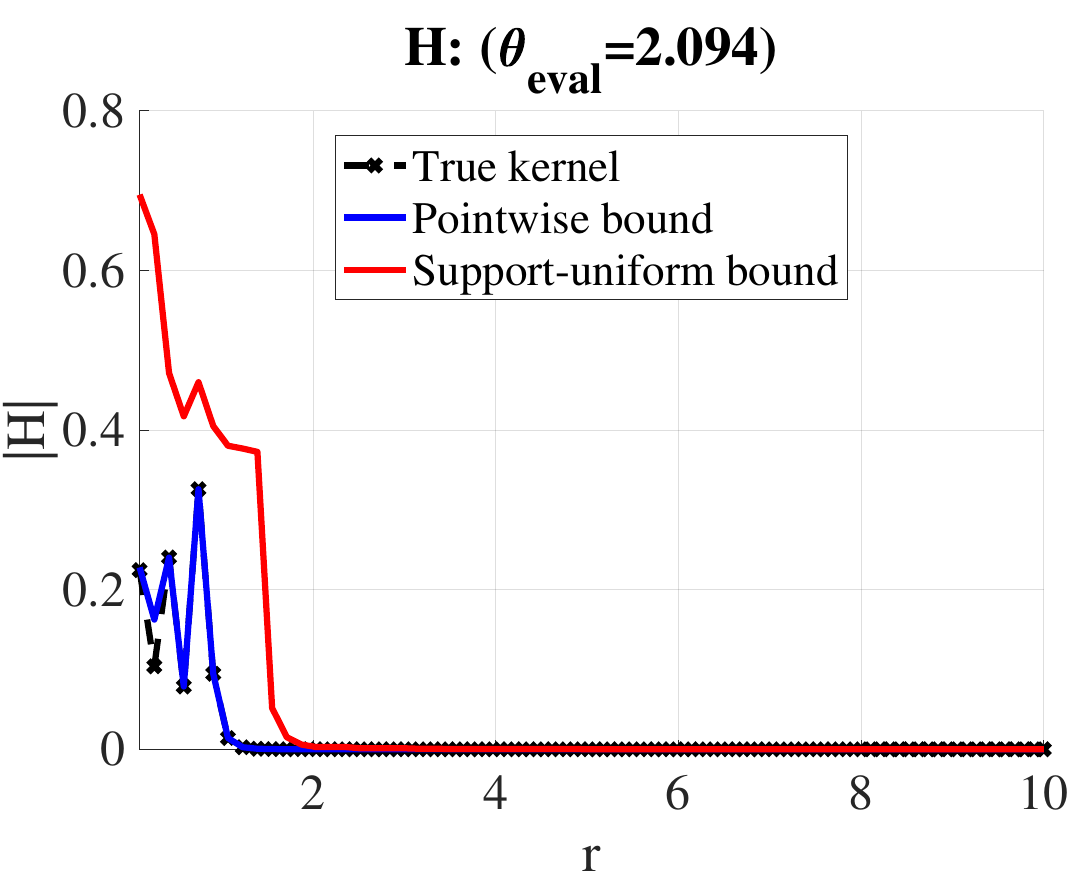}
        \caption{\(H\)}
        \label{fig:cpam-range-H}
    \end{subfigure}
    \hfill
    \begin{subfigure}[t]{0.31\textwidth}
        \centering
        \includegraphics[width=\linewidth]{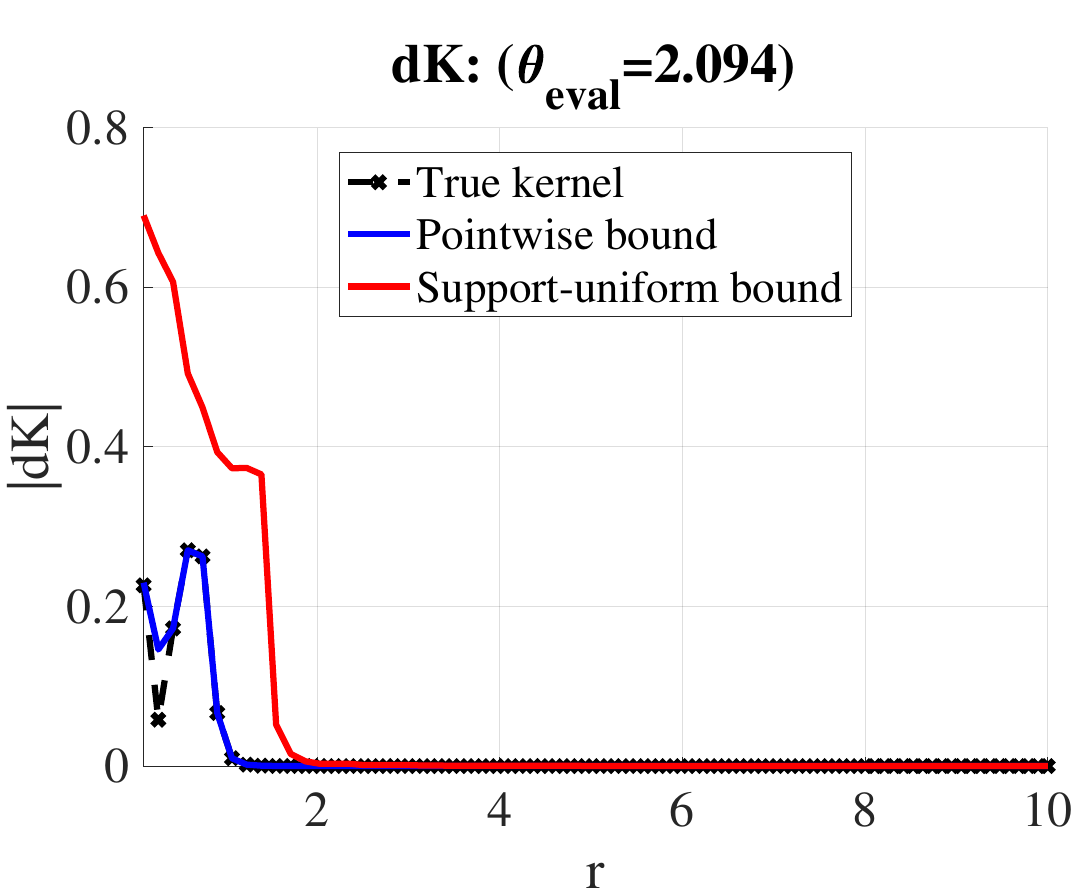}
        \caption{\(dK\)}
        \label{fig:cpam-range-dK}
    \end{subfigure}

    \vspace{0.8em}

    % ---------- Row 2 ----------
    \begin{subfigure}[t]{0.31\textwidth}
        \centering
        \includegraphics[width=\linewidth]{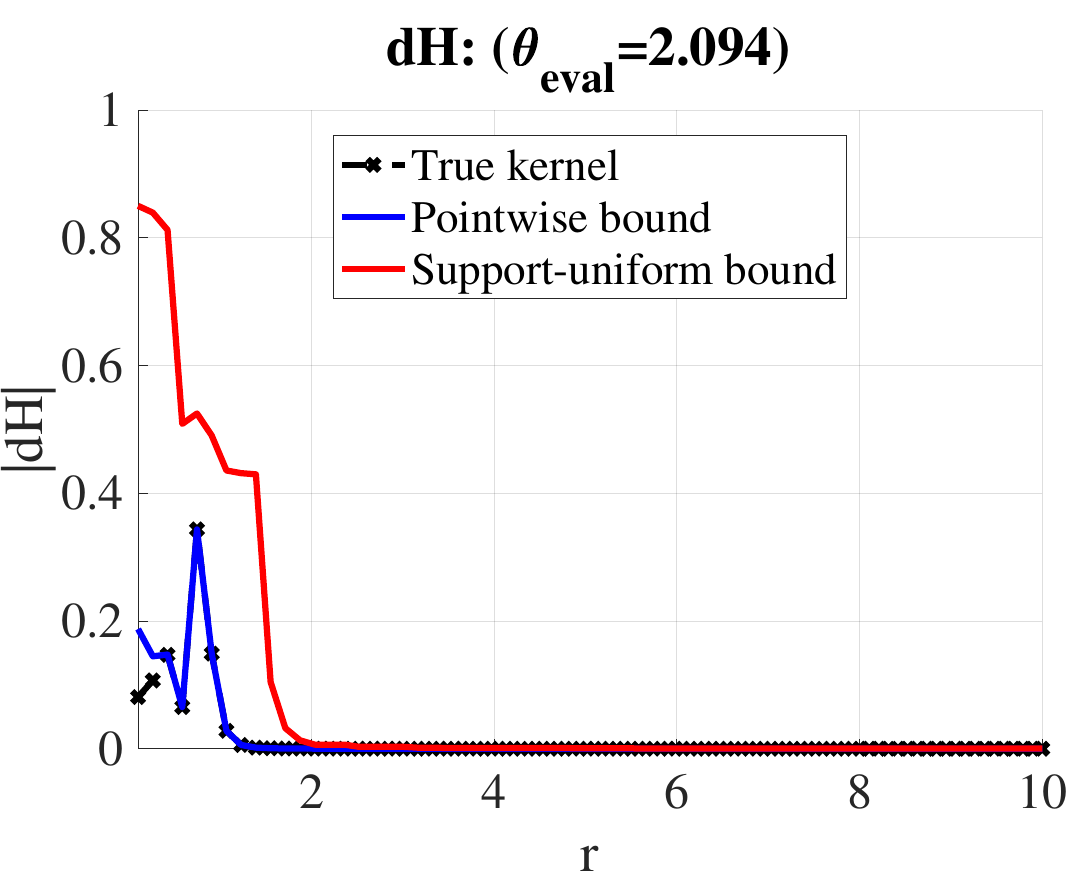}
        \caption{\(dH\)}
        \label{fig:cpam-range-dH}
    \end{subfigure}
    \hfill
    \begin{subfigure}[t]{0.31\textwidth}
        \centering
        \includegraphics[width=\linewidth]{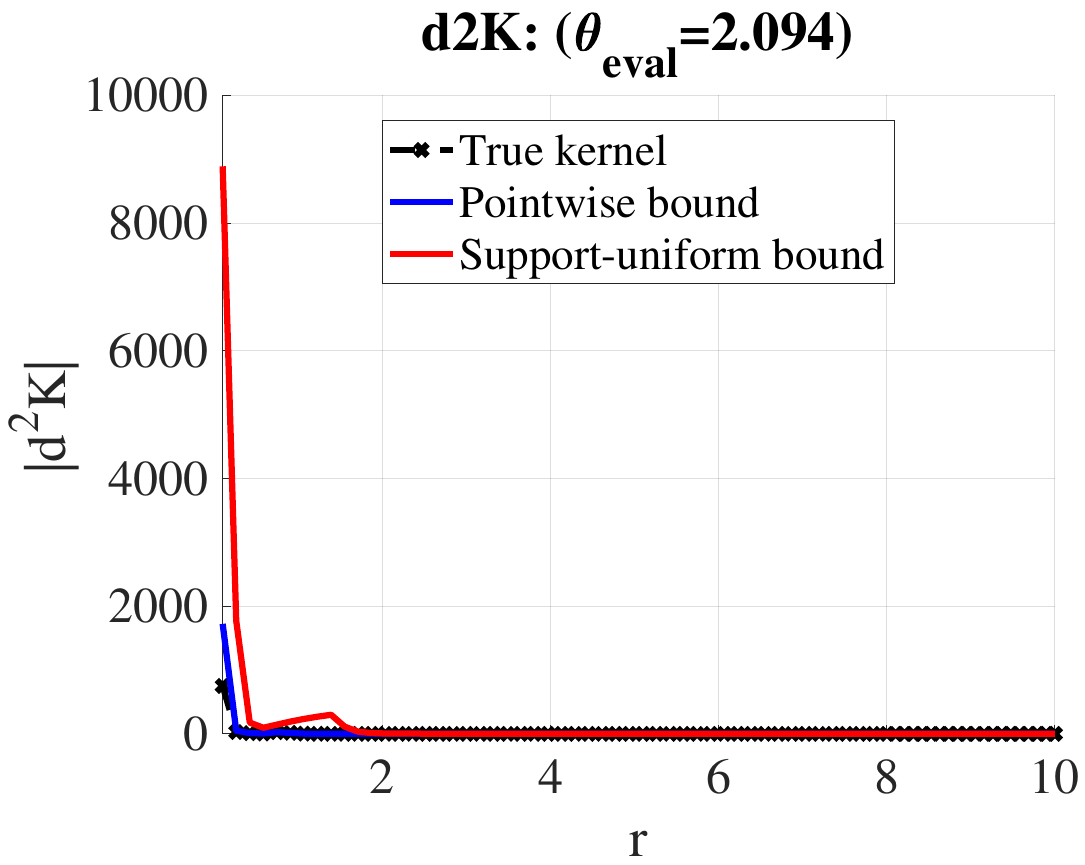}
        \caption{\(\partial_\theta^2K\)}
        \label{fig:cpam-range-d2K}
    \end{subfigure}
    \hfill
    \begin{subfigure}[t]{0.31\textwidth}
        \centering
        \includegraphics[width=\linewidth]{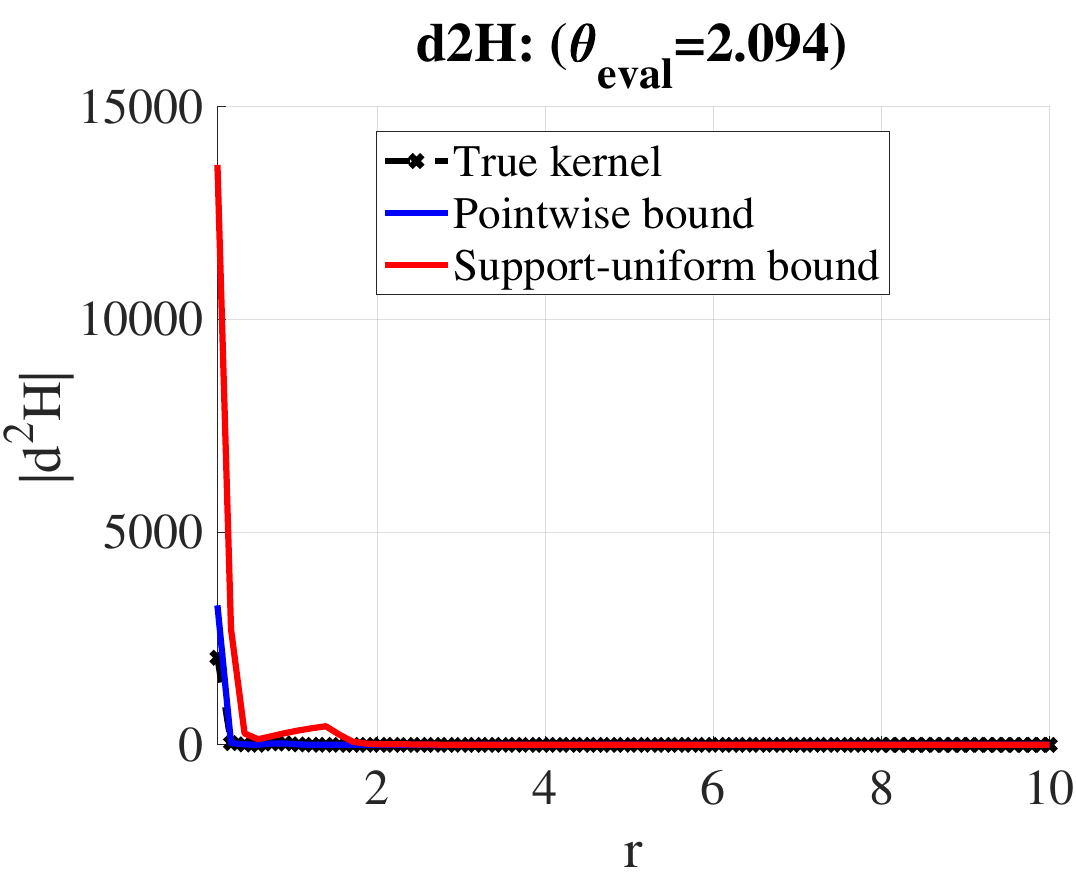}
        \caption{\(\partial_\theta^2H\)}
        \label{fig:cpam-range-d2H}
    \end{subfigure}

    \vspace{0.8em}

    % ---------- Row 3 ----------
    \makebox[\textwidth][c]{%
    \begin{subfigure}[t]{0.31\textwidth}
        \centering
        \includegraphics[width=\linewidth]{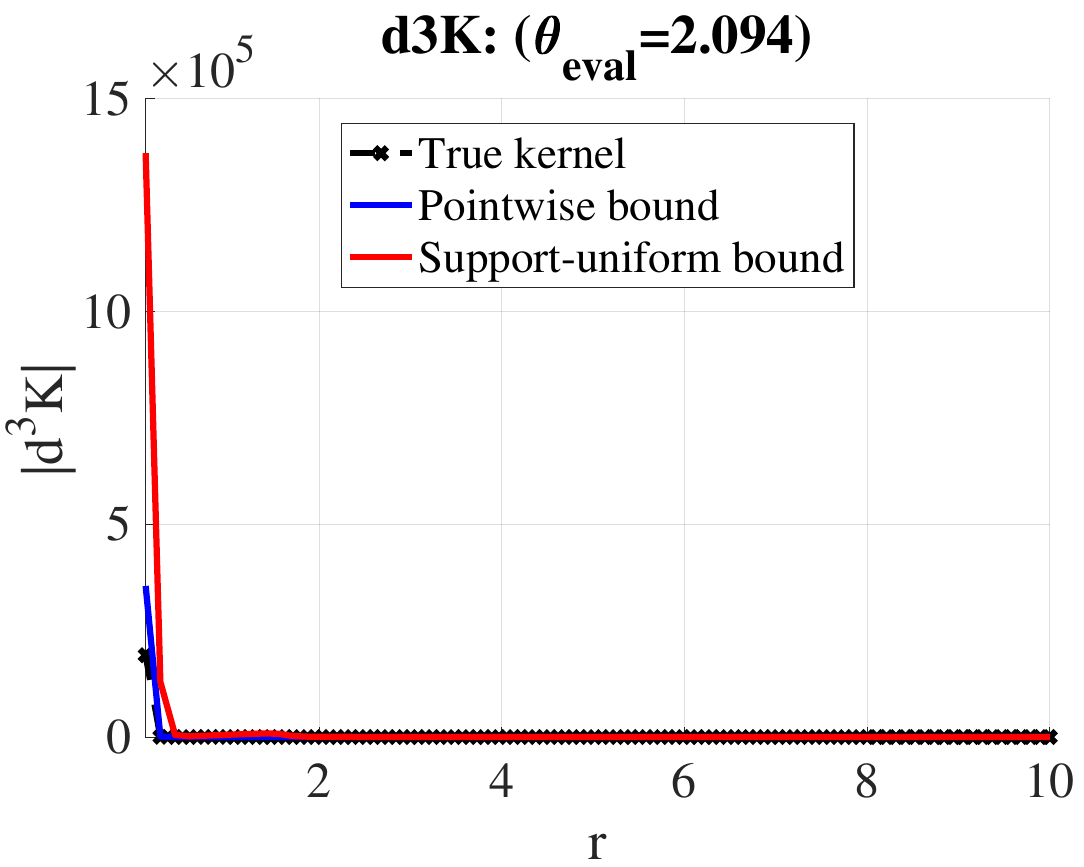}
        \caption{\(\partial_\theta^3K\)}
        \label{fig:cpam-range-d3K}
    \end{subfigure}
    \hspace{0.05\textwidth}
    \begin{subfigure}[t]{0.31\textwidth}
        \centering
        \includegraphics[width=\linewidth]{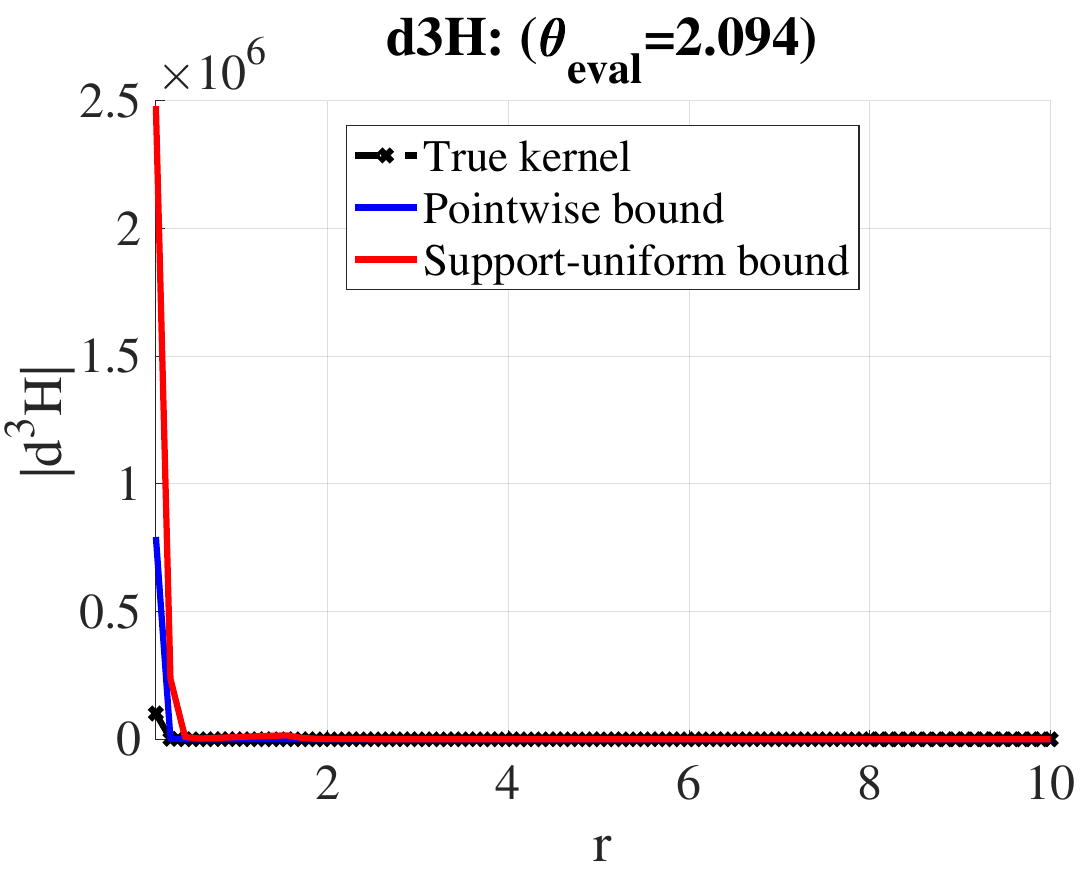}
        \caption{\(\partial_\theta^3H\)}
        \label{fig:cpam-range-d3H}
    \end{subfigure}%
    }

  \caption{
Range-slice comparison of exact gauged channel magnitudes, pointwise QPAC
bounds, and support-uniform QPAC envelopes.  The normalized channels
\(K,H,dK,dH\) quantify interpolation and tangent interactions, while the
second- and third-derivative channels quantify the local curvature and
Taylor-remainder controls used in the strict-decay proof. These panels are reproduced by
\texttt{run\_bounds\_true\_kernel\_repro.m} in the reproducibility archive
\cite{Daei2026NearFieldRepro}.
}
    \label{fig:cpam-range-all-channels}
\end{figure}

% ============================================================
\section{Discussion}
\label{sec:discussion}
% ============================================================

The theorem shows that, for a one-dimensional finite aperture, near-field
super-resolution has its own resolution geometry.  In the narrowband far field,
the aperture response is governed by a linear phase across sensor positions.
The corresponding certificate algebra is Fourier-like, translation invariant
in the angular spatial frequency, and blind to range along a fixed bearing.
Exact recovery is therefore naturally organized by angular separation.  In the
Fresnel near field, wavefront curvature adds a second phase coordinate, and
pairwise interactions are governed by the finite quadratic trajectory
\[
        n\mapsto \omega_1 n+\omega_2 n^2,
        ~ n=0,\ldots,N_r-1 .
\]
Thus the relevant object is not a scalar Fourier distance, but a finite
quadratic phase path seen through the aperture.

This change is structural.  Angular separation alone does not determine
coherence, and curvature separation alone does not determine it either.  The
two coordinates interact over the finite aperture.  A pair of sources that is
close in angle may be separated by curvature, while a pair that appears
well-separated in angle may remain coherent when the quadratic phase is close
to resonance.  Near-field distinguishability is therefore governed by whether
the finite quadratic phase adds coherently or cancels across the aperture.

QPAC is a deterministic sufficient certificate for this cancellation
geometry. It is not intended as a necessary characterization of
recoverability and is not formulated through a universal scalar
minimum-separation threshold. Its branches give
complementary ways of proving finite-aperture cancellation in different regions
of quadratic phase space: nonstationary phases, residue-lag nonalignment, and
near-rational curvature.  The resulting bounds control precisely the kernel,
normalized angular tangent channels, and angular derivative channels that enter the
gauged Hermite dual certificate.

QPAC supplies the three certificate budgets needed by the dual-certificate
proof: Hermite interpolation on the support, strict local angular decay, and
far-region leakage control.  The local Taylor estimate is applied on the
punctured feasible near set
\(
\mathcal N_\theta(S;\varpi_{\rm loc})\setminus S
\).

The theorem is intentionally stated for the semi-discrete Fresnel model, with
range restricted to a finite admissible set and angle left continuous.  This
choice isolates the Fresnel curvature mechanism while leaving only one
continuous local direction, so that a one-dimensional Hermite interpolation
argument suffices.  A fully continuous range-angle theorem would require a
genuinely two-parameter certificate, with range, angle, and mixed derivative
control.  That problem is natural, but mathematically different from the one
solved here.

The finite-harmonic Bessel-Vandermonde lift has a different role.  It is a
computable surrogate for the semi-discrete Fresnel TV problem.  For any fixed
QPAC dual vector, the finite-harmonic polynomial generated by that same vector
is a uniformly controlled perturbation of the corresponding ideal Fresnel
certificate, with an explicit harmonic-truncation error.  This perturbation
statement does not assert that the QPAC dual vector coincides with an optimizer
of the finite-dimensional lifted SDP.  The computed lifted dual optimizer is
interpreted separately through the dual-saturation property of the lifted
problem.

For physical spherical-wave data, harmonic truncation is only one
approximation: the paraxial mismatch between the spherical and Fresnel phases
is separate.  Exact physical recovery would require either an asymptotic
regime in which that mismatch vanishes or a certificate theory built directly
for the spherical-wave kernel.

Several mathematical directions follow from this viewpoint.  The present QPAC
constants are explicit and support-uniform, hence conservative; sharper
certificates should be possible through finer covers or stronger envelopes,
especially near the boundary between resonant and nonresonant quadratic
phases.  A continuous range theory would require a two-parameter Hermite
certificate and a local analysis of the full Fresnel phase map.

Another natural direction is to replace the one-dimensional aperture by
planar, volumetric, sparse, or conformal aperture geometries.  For a general
sensor set \(\Lambda\subset\mathbb R^d\), the Fresnel interaction is no longer
a scalar quadratic trajectory, but a finite quadratic form sampled on the
aperture,
\[
        x\mapsto \boldsymbol\omega_1\cdot x+x^T\Omega_2x,
        ~ x\in\Lambda .
\]
For a uniform planar array, this leads to two-dimensional quadratic sums of
the form
\[
        \sum_{(n_1,n_2)\in\Lambda}
        a_{n_1,n_2}
        \exp i\!\left(
        \omega_{10}n_1+\omega_{01}n_2
        +\omega_{20}n_1^2+\omega_{11}n_1n_2+\omega_{02}n_2^2
        \right),
\]
with possible mixed curvature terms.  Circular, sparse, or conformal apertures
lead instead to quadratic phases sampled over nonrectangular or irregular
finite sets.  Thus the aperture shape is not merely an implementation detail:
it changes the finite oscillatory geometry that the dual certificate must
control.  A QPAC analogue for such arrays would require multidimensional
quadratic-sum envelopes together with the corresponding multi-parameter
Hermite derivative channels.  Broadband measurements would bring time delay
into the same framework as wavefront curvature, potentially combining several
independent mechanisms for resolving range.

The broader conclusion is that near-field super-resolution is not far-field
super-resolution with a modified Rayleigh length.  Once curvature is visible,
the aperture measures a finite quadratic phase trajectory; for more general
apertures, it measures a finite quadratic phase geometry determined jointly by
the source parameters and the sensor set.  The arithmetic and oscillatory
behavior of this geometry determines whether sparse sources add coherently or
cancel across the aperture.  QPAC turns the one-dimensional version of this
finite-aperture quadratic-phase geometry into a deterministic
dual-certificate calculus.  In the far-field limit, the quadratic coordinate
disappears, range becomes unobservable from single-carrier spatial aperture
data, and the theory degenerates to Fourier-type angular geometry.

\paragraph{Code availability:}
The MATLAB code used to reproduce the derivative-certified QPAC example,
the LCS-certified QPAC example, the lifted dual-polynomial experiment, and
the QPAC channel-bound plots is available in the versioned
reproducibility archive \cite{Daei2026NearFieldRepro}.  The archive contains
four principal scripts:
\nolinkurl{run\_derivative\_example\_qpac.m} for the support-uniform derivative
verification,
\nolinkurl{run\_lcs\_example\_qpac.m} for the LCS support-class
verification,
\nolinkurl{run\_lifted\_dual\_figure1\_repro.m} for the finite-harmonic
lifted dual-polynomial experiment, and
\nolinkurl{run\_bounds\_true\_kernel\_repro.m} for the exact-channel,
pointwise-bound, and cellwise support-uniform-bound comparisons.  The archive
also contains README files, licensing information, and instructions for
regenerating the reported tables and figures.  No external input data files
are required.

\section*{Funding}
This work was supported by the European Union through the MULTIPLY-6G
project (Grant Agreement No.~101293106) and the QUEST-6G project
(Grant Agreement No.~101292676). G\'abor Fodor was also supported by
the Swedish Foundation for Strategic Research (SSF) through the
SAICOM project (grant FUS21-0004).
% ====================
\bibliographystyle{amsplain}
\bibliography{refs}

\appendix
% ============================================================

% ============================================================
\section{TV duality and gauged Hermite interpolation}
\label{app:tv-duality}
% ============================================================

This appendix proves the dual-certificate criterion used in the main text and
then gives the gauged Hermite interpolation construction.  We keep all
conjugation conventions explicit.

Throughout this appendix, the Euclidean inner product on \(\C^{N_r}\) is
\(
        \langle u,z\rangle_{\C^{N_r}}
        :=
        \sum_{n=0}^{N_r-1}\overline{u[n]}z[n],
\)
so it is conjugate-linear in the first variable and linear in the second.
For weighted aperture vectors we use
\(
        \langle u,z\rangle_\rho
        :=
        \sum_{n=0}^{N_r-1}\rho_n\overline{u[n]}z[n],
\)
and
\(
        \|u\|_\rho^2
        :=
        \sum_{n=0}^{N_r-1}\rho_n|u[n]|^2 .
\)
Thus all weighted inner products and moments in this appendix are computed with
the QPAC taper \(\rho\).

% ============================================================
\subsection{Overview of the proof}
\label{app:tv-duality-sketch}
% ============================================================

The proof has four steps.

First, the TV dual problem is obtained by introducing a dual vector
\(\boldsymbol\zeta\in\C^{N_r}\).  Its adjoint image is the range-indexed trigonometric
family
\(
        Q_i^{\bs \zeta}(\theta)
        =
        \sum_{n=0}^{N_r-1}
        \zeta_n\overline{\aFR(r_i,\theta)[n]}.
\)
The dual feasibility condition is
\(\sup_{(i,\theta)\in\Qset}|Q_i^{\bs \zeta}(\theta)|\le1.\)

Second, if such a dual family interpolates the signs on the true support,
\(Q_{i_\ell}^\zeta(\theta_\ell)=\sgn(c_\ell),\)
and is strictly bounded by one off the support, then TV minimization is exact.
Indeed, every feasible measure has the same pairing with the adjoint dual
family as the true measure, while dual feasibility bounds this pairing by the
TV norm.  Equality forces every minimizer to be supported only where
\(|Q|=1\), hence on the true support.  Linear independence of the active atoms
then gives uniqueness.

Third, to construct \(Q\), we first construct a gauged certificate \(P\).  The
gauge removes a local phase and makes the value atom orthogonal to the angular
tangent atom at each point.  The gauged certificate is a linear combination of
value and tangent kernels,
\(
        P_{S,\widetilde {\mathbf v}}(p)
        =
        \sum_{\ell=1}^L\alpha_\ell K_\ell(p)
        +
        \sum_{\ell=1}^L \beta_\ell H_\ell(p),
\)
chosen to satisfy Hermite interpolation:
\(P_{S,\widetilde {\mathbf v}}(p_\ell)=\widetilde v_\ell, ~ \partial_\theta P_{S,\widetilde {\mathbf v}}(p_\ell)=0.\)

Fourth, the Hermite interpolation matrix has identity diagonal blocks.  If the
off-diagonal support interactions are dominated by a nonnegative matrix
\(\mathbf G_{\rm SS}\) and
\((L-1)\rho_{\rm sp}(\mathbf G_{\rm SS})<1,\)
then a block Neumann argument proves invertibility and gives coefficient
bounds
\(|\alpha_\ell|\le\Gamma_K, ~ |\beta_\ell|\le\Gamma_H.\)
The same invertibility also implies linear independence of the active Fresnel
atoms.

% ============================================================
\subsection{Vector measures and the TV program}
\label{app:vector-measures-tv-program}
% ============================================================

Recall that
\(\Qset=\{1,\dots,N_d\}\times\Thetaset.\)
A semi-discrete measure on \(\Qset\) is written as
\(
        \nu
        =
        \sum_{i=1}^{N_d}\delta_{r_i}\otimes\mu_i,
        ~
        \mu_i\in\Mset(\Thetaset),
\)
and its total variation is
\(\|\nu\|_{\TV,\Qset}:= \sum_{i=1}^{N_d}\|\mu_i\|_{\TV}.\)
The Fresnel measurement operator is
\(
        (\FFR\nu)[n]
        =
        \sum_{i=1}^{N_d}
        \int_{\Thetaset}
        \aFR(r_i,\theta)[n]\,d\mu_i(\theta),
        ~ n=0,\dots,N_r-1.
\)
Given noiseless data
\(
        \mathbf y=\FFR\nu_\star,
\)
the TV program is
\begin{equation}\label{eq:app_tv_program}
        \min_\nu \|\nu\|_{\TV,\Qset}
        ~
        \text{subject to}
        ~
        \FFR\nu=\mx y .
\end{equation}

For a sparse measure
\(
        \nu_\star
        =
        \sum_{\ell=1}^L
        c_\ell\,\delta_{r_{i_\ell}}\otimes\delta_{\theta_\ell},
        ~ c_\ell\ne0,
\)
we denote
\(p_\ell=(i_\ell,\theta_\ell), ~ S=\{p_1,\dots,p_L\},\)
and define the complex signs by
\(
        v_\ell:=\sgn(c_\ell):=\tfrac{c_\ell}{|c_\ell|}.
\)
For a continuous family
\(
        Q=(Q_i)_{i=1}^{N_d},
        ~
        Q_i:\Thetaset\to\C,
\)
and a vector measure \(\nu=\sum_i\delta_{r_i}\otimes\mu_i\), define the
duality pairing
\(
        \langle Q,\nu\rangle
        :=
        \sum_{i=1}^{N_d}
        \int_{\Thetaset}\overline{Q_i(\theta)}\,d\mu_i(\theta).
\)
With this convention,
\(
        \overline{v_\ell}c_\ell=|c_\ell|.
\)

% ============================================================
\subsection{The adjoint family and the TV dual problem}
\label{app:adjoint-family-tv-dual}
% ============================================================

For \(\boldsymbol\zeta\in\C^{N_r}\), define the adjoint Fresnel family \(Q_i^{\bs \zeta}(\theta)
        :=
        \sum_{n=0}^{N_r-1}
        \zeta_n\overline{\aFR(r_i,\theta)[n]}.\)
Then
\(
        \overline{Q_i^{\bs \zeta}(\theta)}
        =
        \sum_{n=0}^{N_r-1}
        \overline{\zeta_n}\aFR(r_i,\theta)[n].
\)
Therefore, for every vector measure \(\nu\),
\[
\begin{aligned}
        &\langle Q^{\bs \zeta},\nu\rangle
        =
        \sum_{i=1}^{N_d}
        \int_{\Thetaset}
        \sum_{n=0}^{N_r-1}
        \overline{\zeta_n}\aFR(r_i,\theta)[n]
        \,d\mu_i(\theta)                                      
        =
        \sum_{n=0}^{N_r-1}
        \overline{\zeta_n}
        \sum_{i=1}^{N_d}
        \int_{\Thetaset}
        \aFR(r_i,\theta)[n]\,d\mu_i(\theta)   \\                 
       & =
        \sum_{n=0}^{N_r-1}
        \overline{\zeta_n}(\FFR\nu)[n]                         
        =
        \langle \zeta,\FFR\nu\rangle_{\C^{N_r}} .
\end{aligned}
\]
Thus \(Q^{\bs \zeta}\) is exactly the adjoint image of \({\bs \zeta}\), written on the
range-angle parameter domain.

The formal dual of \eqref{eq:app_tv_program} is
\begin{equation}\label{eq:app_tv_dual}
        \max_{\boldsymbol\zeta\in\C^{N_r}}
        \operatorname{Re} \langle \zeta,y\rangle_{\C^{N_r}}
        ~
        \text{subject to}
        ~
        \sup_{(i,\theta)\in\Qset}|Q_i^\zeta(\theta)|\le1 .
\end{equation}

Indeed, the Lagrangian for \eqref{eq:app_tv_program} is
\[
        \mathcal L(\nu,{\bs \zeta})
        =
        \|\nu\|_{\TV,\Qset}
        +
        \operatorname{Re}\langle {\bs \zeta},\mx y-\FFR\nu\rangle_{\C^{N_r}}.
\]
Using the adjoint identity just proved,
\(\operatorname{Re}\langle {\bs \zeta},\FFR\nu\rangle_{\C^{N_r}} = \operatorname{Re}\langle Q^{\bs \zeta},\nu\rangle.\)
Hence
\[
        \mathcal L(\nu,\zeta)
        =
        \operatorname{Re}\langle {\bs \zeta},\mx y\rangle_{\C^{N_r}}
        +
        \left(
        \|\nu\|_{\TV,\Qset}
        -
        \operatorname{Re}\langle Q^{\bs \zeta},\nu\rangle
        \right).
\]
The infimum over \(\nu\) is finite if and only if
\(
        \sup_{(i,\theta)\in\Qset}|Q_i^{\bs \zeta}(\theta)|\le1.
\)
To see this, first suppose the supremum is at most one.  Let
\(d\mu_i(\theta)=s_i(\theta)\,d|\mu_i|(\theta)\)
be the polar decomposition, with \(|s_i(\theta)|=1\) for
\(|\mu_i|\)-almost every \(\theta\).  Then
\[
\begin{aligned}
        \operatorname{Re}\langle Q^{\bs \zeta},\nu\rangle
        &=
        \sum_i
        \int_{\Thetaset}
        \operatorname{Re}\bigl(\overline{Q_i^{\bs \zeta}(\theta)}s_i(\theta)\bigr)
        \,d|\mu_i|(\theta)                                  \le
        \sum_i
        \int_{\Thetaset}
        |Q_i^{\bs \zeta}(\theta)|\,d|\mu_i|(\theta)                  \\
        &\le
        \sum_i\|\mu_i\|_{\TV}
        =
        \|\nu\|_{\TV,\Qset}.
\end{aligned}
\]
Thus the bracketed term is nonnegative and is minimized at \(\nu=0\).

Conversely, suppose that there exists \((i_0,\theta_0)\in\mathcal Q\) such
that
\(
        |Q_{i_0}^\zeta(\theta_0)|>1 .
\)
Choose a phase \(\sigma\in\mathbb C\), \(|\sigma|=1\), such that
\(
        \operatorname{Re}\!\left(
        \overline{Q_{i_0}^{\bs \zeta}(\theta_0)}\,\sigma
        \right)
        =
        |Q_{i_0}^{\bs \zeta}(\theta_0)|.
\)
For \(t>0\), set
\(
        \nu_t
        :=
        t\sigma\,\delta_{r_{i_0}}\otimes\delta_{\theta_0}.
\)
Then
\(
        \|\nu_t\|_{\TV,\Qset}
        =
        t
\)
and
\(
        \operatorname{Re}\langle Q^\zeta,\nu_t\rangle
        =
        t|Q_{i_0}^\zeta(\theta_0)|.
\)
Hence
\(
        \|\nu_t\|_{\TV,\Qset}
        -
        \operatorname{Re}\langle Q^{\bs \zeta},\nu_t\rangle
        =
        t\left(1-|Q_{i_0}^{\bs \zeta}(\theta_0)|\right)
        \to -\infty
\)
as \(t\to\infty\).  Therefore finiteness of the infimum of the Lagrangian over \(\nu\) forces
\(
        \sup_{(i,\theta)\in\mathcal Q}|Q_i^{\bs \zeta}(\theta)|\le1 .
\)
The derivation of the formal dual is useful, but the uniqueness proof below
only needs weak duality and the existence of one strict certificate.

% ============================================================
\subsection{TV dual certificate criterion}
\label{app:tv-dual-certificate-criterion}
% ============================================================

The uniqueness proof requires only three properties: an adjoint polynomial
that interpolates the source phases, strict modulus below one away from the
support, and linear independence of the active measurement atoms.  The
following proposition records this implication using the complex pairing
and conjugation conventions fixed above.

\begin{prop}
\label{prop:app_tv_dual_certificate}
Let
\(
        \nu_\star
        =
        \sum_{\ell=1}^L
        c_\ell\,\delta_{r_{i_\ell}}\otimes\delta_{\theta_\ell},
        ~
        c_\ell\ne0,
\)
and let
\(
        S=\{p_\ell=(i_\ell,\theta_\ell)\}_{\ell=1}^L .
\)
For a raw Fresnel adjoint family \(Q^\zeta\), we use the complex measure pairing
\(
        \langle Q^\zeta,\nu\rangle
        :=
        \int_{\Qset}\overline{Q^\zeta(p)}\,d\nu(p).
\)
Assume there exists \(\boldsymbol\zeta\in\C^{N_r}\) such that
\(
        Q^\zeta
        =
        \FFR^*{\bs \zeta}
\)
in the sense that
\[
        Q_i^{\bs \zeta}(\theta)
        =
        \sum_{n=0}^{N_r-1}
        \zeta_n\overline{\aFR(r_i,\theta)[n]},
        ~
        (i,\theta)\in\Qset .
\]
Assume that \(Q^\zeta\) interpolates the source signs,
\begin{equation}
\label{eq:app_dual_interpolation}
        Q_{i_\ell}^{\bs \zeta}(\theta_\ell)
        =
        \sgn(c_\ell)
        =
        \tfrac{c_\ell}{|c_\ell|},
        ~
        \ell=1,\ldots,L,
\end{equation}
and is strictly bounded off the support,
\begin{equation}
\label{eq:app_dual_strict}
        |Q_i^{\bs \zeta}(\theta)|<1
        ~
        \text{for every }(i,\theta)\in\Qset\setminus S .
\end{equation}
Assume also that the active Fresnel atoms
\(
        \{\aFR(r_{i_\ell},\theta_\ell)\}_{\ell=1}^L
\)
are linearly independent in \(\C^{N_r}\).  Then \(\nu_\star\) is the unique
minimizer of \eqref{eq:app_tv_program} with data \(\mathbf y=\FFR\nu_\star\).
\end{prop}

\begin{proof}
Let \(\nu\) be any feasible point of \eqref{eq:app_tv_program}. Then \(\FFR\nu=\FFR\nu_\star .\) Since \(Q^\zeta\) is the adjoint image of \(\zeta\), and since the inner product
on \(\C^{N_r}\) is conjugate-linear in the first argument, we have \(\langle Q^\zeta,\nu\rangle = \langle \zeta,\FFR\nu\rangle_{\C^{N_r}} .\) Therefore feasibility gives
\[
 \langle Q^\zeta,\nu\rangle
 =
 \langle {\bs \zeta},\FFR\nu\rangle_{\C^{N_r}}
 =
 \langle {\bs \zeta},\FFR\nu_\star\rangle_{\C^{N_r}}
 =
 \langle Q^{\bs \zeta},\nu_\star\rangle .
\]
Taking real parts,
\begin{equation}
\label{eq:app_dual_feasible_pairing_identity}
 \operatorname{Re}\langle Q^\zeta,\nu\rangle
 =
 \operatorname{Re}\langle Q^\zeta,\nu_\star\rangle .
\end{equation}

Using the sign interpolation condition \eqref{eq:app_dual_interpolation}, we
compute
\[
\begin{aligned}
 &\operatorname{Re}\langle Q^\zeta,\nu_\star\rangle
 =
 \operatorname{Re}
 \sum_{\ell=1}^L
 \overline{Q_{i_\ell}^{\bs \zeta}(\theta_\ell)}c_\ell 
 =
 \operatorname{Re}
 \sum_{\ell=1}^L
 \overline{\sgn(c_\ell)}\,c_\ell =
 \sum_{\ell=1}^L |c_\ell|
 =
 \|\nu_\star\|_{\TV,\Qset}.
\end{aligned}
\]
On the other hand, write the polar decomposition of each component measure as
\(
        d\mu_i(\theta)
        =
        s_i(\theta)\,d|\mu_i|(\theta),
        ~
        |s_i(\theta)|=1
        \quad |\mu_i|\text{-a.e.}
\)
Then
\[
\begin{aligned}
 &\operatorname{Re}\langle Q^{\bs \zeta},\nu\rangle
 =
 \sum_i
 \int_{\Theta}
 \operatorname{Re}\left(
 \overline{Q_i^{\bs \zeta}(\theta)}s_i(\theta)
 \right)
 d|\mu_i|(\theta) \le
 \sum_i
 \int_{\Theta}
 |Q_i^{\bs \zeta}(\theta)|\,d|\mu_i|(\theta) \le
 \sum_i
 \|\mu_i\|_{\TV}\\
 &=
 \|\nu\|_{\TV,\Qset}.
\end{aligned}
\]
Combining this estimate with
\eqref{eq:app_dual_feasible_pairing_identity} gives
\(        \|\nu_\star\|_{\TV,\Qset}
        \le
        \|\nu\|_{\TV,\Qset}.
\)
Hence \(\nu_\star\) is a minimizer.
It remains to prove uniqueness. Suppose that \(\nu\) is another minimizer.
Then equality holds in the chain above. In particular, \(\|\nu\|_{\TV,\Qset} = \operatorname{Re}\langle Q^\zeta,\nu\rangle .\) Equivalently,
\[
 \sum_i
 \int_\Theta
 \left[
 1-
 \operatorname{Re}\left(
 \overline{Q_i^{\bs \zeta}(\theta)}s_i(\theta)
 \right)
 \right]
 d|\mu_i|(\theta)
 =
 0 .
\]
The integrand is nonnegative, because \(\operatorname{Re}\left( \overline{Q_i^\zeta(\theta)}s_i(\theta) \right) \le |Q_i^\zeta(\theta)| \le 1 .\) Therefore the integrand must vanish \(|\mu_i|\)-almost everywhere. Hence,
for \(|\mu_i|\)-almost every \(\theta\), \(|Q_i^\zeta(\theta)|=1\) and the phase of \(s_i(\theta)\) matches the phase of \(Q_i^\zeta(\theta)\).
By the strict off-support condition \eqref{eq:app_dual_strict}, the set where
\(|Q_i^\zeta(\theta)|=1\) is contained in the active support \(S\). Therefore
every minimizer is supported on \(S\). Hence there exist coefficients
\(d_1,\ldots,d_L\in\C\) such that \(\nu = \sum_{\ell=1}^L d_\ell\,\delta_{r_{i_\ell}}\otimes\delta_{\theta_\ell}.\) Since \(\nu\) and \(\nu_\star\) are feasible for the same data, \(\FFR\nu=\FFR\nu_\star.\) Thus \(\sum_{\ell=1}^L (d_\ell-c_\ell)\aFR(r_{i_\ell},\theta_\ell) = 0 ~ \text{in }\C^{N_r}.\) The active Fresnel atoms are linearly independent, so \(d_\ell=c_\ell, ~ \ell=1,\ldots,L.\) Therefore \(\nu=\nu_\star\). Hence the minimizer is unique.
\end{proof}

% ============================================================
\subsection{Gauged Fresnel atoms}
\label{app:gauged-fresnel-atoms}
% ============================================================

We now construct the gauged atoms used in the Hermite certificate.
Let
\(
        W_0:=\sum_{n=0}^{N_r-1}\rho_n>0 .
\)
Define the weighted moments
\(
        \bar n
        :=
        \tfrac1{W_0}\sum_{n=0}^{N_r-1}\rho_n n,
        ~
        \overline{n^2}
        :=
        \tfrac1{W_0}\sum_{n=0}^{N_r-1}\rho_n n^2 .
\)
For each range bin \(i\), define the gauge phase \( \chi_i(\theta)
        :=
        k_\lambda d\,\bar n\cos\theta
        +
        \tfrac{k_\lambda d^2}{4r_i}\overline{n^2}\cos(2\theta).\) The gauged normalized Fresnel atom is \(\psi_i(\theta)[n]
        :=
        \tfrac1{\sqrt{W_0}}
        e^{-i\chi_i(\theta)}
        \aFR(r_i,\theta)[n].\) Since \(|e^{-i\chi_i(\theta)}|=1\) and
\(|\aFR(r_i,\theta)[n]|=1\), we have
\[
        \|\psi_i(\theta)\|_\rho^2
        =
        \sum_{n=0}^{N_r-1}\rho_n|\psi_i(\theta)[n]|^2
        =
        \tfrac1{W_0}\sum_{n=0}^{N_r-1}\rho_n
        =
        1.
\]
The angular derivative has the form
\begin{equation}\label{eq:app_gauged_derivative}
        \partial_\theta\psi_i(\theta)[n]
        =
        i\,u_i(n,\theta)\psi_i(\theta)[n],
\end{equation}
where
\begin{equation}\label{eq:app_centered_derivative_phase}
        u_i(n,\theta)
        :=
        -k_\lambda d(n-\bar n)\sin\theta
        -
        \tfrac{k_\lambda d^2}{2r_i}
        (n^2-\overline{n^2})\sin(2\theta).
\end{equation}
Indeed, if
\(\aFR(r_i,\theta)[n] = e^{i\phi_i(n,\theta)},\)
then
\(
        \phi_i(n,\theta)
        =
        k_\lambda dn\cos\theta
        -
        \tfrac{k_\lambda d^2n^2}{2r_i}\sin^2\theta.
\)
Thus
\(
        \partial_\theta\phi_i(n,\theta)
        =
        -k_\lambda dn\sin\theta
        -
        \tfrac{k_\lambda d^2n^2}{2r_i}\sin(2\theta),
\)
while
\(
        \partial_\theta\chi_i(\theta)
        =
        -k_\lambda d\bar n\sin\theta
        -
        \tfrac{k_\lambda d^2}{2r_i}\overline{n^2}\sin(2\theta).
\)
Therefore
\(
        \partial_\theta\phi_i(n,\theta)
        -
        \partial_\theta\chi_i(\theta)
        =
        u_i(n,\theta),
\)
which proves \eqref{eq:app_gauged_derivative}.
By the definitions of \(\bar n\) and \(\overline{n^2}\),
\(
        \sum_{n=0}^{N_r-1}\rho_n u_i(n,\theta)=0 .
\)
Consequently,
\[
\begin{aligned}
        &\langle \partial_\theta\psi_i(\theta),
        \psi_i(\theta)\rangle_\rho
        =
        \sum_{n=0}^{N_r-1}
        \rho_n\overline{i\,u_i(n,\theta)\psi_i(\theta)[n]}
        \psi_i(\theta)[n]                                     =\\
        &-i
        \sum_{n=0}^{N_r-1}
        \rho_n u_i(n,\theta)|\psi_i(\theta)[n]|^2              =
        -\tfrac{i}{W_0}
        \sum_{n=0}^{N_r-1}\rho_n\,u_i(n,\theta)
        =
        0 .
\end{aligned}
\]
Thus value and tangent directions are orthogonal at each point.
For a support point
\(p_\ell=(i_\ell,\theta_\ell),\)
define
\(
        \psi_\ell:=\psi_{i_\ell}(\theta_\ell),
        ~
        \dot\psi_\ell
        :=
        \partial_\theta\psi_{i_\ell}(\theta_\ell),
        ~
        \sigma_\ell:=\|\dot\psi_\ell\|_\rho.\)
Assume
\(\sigma_\ell>0, ~ \ell=1,\dots,L.\)
In the main theorem this is guaranteed by the analytical lower bound for the
gauged tangent energy.  Define the unit tangent atom
\(\mathbf h_\ell:=\tfrac{\dot{\boldsymbol\psi}_\ell}{\sigma_\ell}.\)
Then
\(
        \|\psi_\ell\|_\rho=1,
        ~
        \|h_\ell\|_\rho=1,
        ~
        \langle h_\ell,\psi_\ell\rangle_\rho=0.
\)

% ============================================================
\subsection{Gauged value and tangent kernels}
\label{app:gauged-value-tangent-kernels}
% ============================================================

For \(p=(i,\theta)\in\Qset\), define the gauged value and tangent kernels
\begin{equation}\label{eq:app_KH_def}
        K_\ell(p)
        :=
        \langle \psi_i(\theta),\psi_\ell\rangle_\rho,
        ~
        H_\ell(p)
        :=
        \langle \psi_i(\theta),h_\ell\rangle_\rho.
\end{equation}
Because \(\|\psi_\ell\|_\rho=1\) and \(\langle \psi_\ell,h_\ell\rangle_\rho=0\),
\(K_\ell(p_\ell)=1, ~ H_\ell(p_\ell)=0.\)

The angular derivatives with respect to the evaluation point are
\(\partial_\theta K_\ell(i,\theta) = \langle \partial_\theta\psi_i(\theta),\psi_\ell\rangle_\rho,\)
and
\(\partial_\theta H_\ell(i,\theta) = \langle \partial_\theta\psi_i(\theta),h_\ell\rangle_\rho.\)
At the diagonal point \(p_\ell=(i_\ell,\theta_\ell)\),
\(
        \partial_\theta K_\ell(p_\ell)
        =
        \langle \dot\psi_\ell,\psi_\ell\rangle_\rho
        =
        0,
\)
and
\(
        \partial_\theta H_\ell(p_\ell)
        =
        \langle \dot\psi_\ell,h_\ell\rangle_\rho
        =
        \left\langle \dot\psi_\ell,\tfrac{\dot\psi_\ell}{\sigma_\ell}\right\rangle_\rho
        =
        \sigma_\ell.
\)

% ============================================================
\subsection{The gauged Hermite pre-certificate}
\label{app:gauged-hermite-precertificate}
% ============================================================

Let
\(\mathbf v
=
(v_1,\ldots,v_L)^T, ~ |v_\ell|=1.\)
For recovery of
\(        \nu_\star=\sum_{\ell=1}^L c_\ell
        \delta_{r_{i_\ell}}\otimes\delta_{\theta_\ell},
\)
one takes \(v_\ell=\sgn(c_\ell)\).

Define the gauged signs \( \widetilde v_\ell
        :=
        e^{i\chi_{i_\ell}(\theta_\ell)}v_\ell.\) We seek a gauged certificate of the form
\begin{equation}\label{eq:app_P_expansion}
        P_{S,\widetilde {\mathbf v}}(p)
        =
        \sum_{\ell=1}^L\alpha_\ell K_\ell(p)
        +
        \sum_{\ell=1}^L \beta_\ell H_\ell(p).
\end{equation}
Equivalently, if
\(
    \mathbf z_{S,\widetilde{\mathbf v}}
:=
\sum_{\ell=1}^L
\alpha_\ell\boldsymbol\psi_\ell
+
\sum_{\ell=1}^L
\beta_\ell\mathbf h_\ell,
\)
then
\(P_{S,\widetilde {\mathbf v}}(i,\theta) = \langle \bs \psi_i(\theta),\mathbf z_{S,\widetilde{\mathbf v}}\rangle_\rho.\)

The Hermite interpolation conditions are
\begin{equation}\label{eq:app_hermite_conditions}
        P_{S,\widetilde {\mathbf v}}(p_j)=\widetilde v_j,
        ~
        \partial_\theta P_{S,\widetilde {\mathbf v}}(p_j)=0,
        ~ j=1,\dots,L.
\end{equation}
For normalization, we write the second set of equations as
\(\sigma_j^{-1}\partial_\theta P_{S,\widetilde {\mathbf v}}(p_j)=0.\)

Define the \(2L\times2L\) Hermite interpolation matrix \(\mathcal H_{\rm g}(S)\) as
follows.  Its \(j\)-th value row is
\[
        \bigl(
        K_1(p_j),\dots,K_L(p_j),
        H_1(p_j),\dots,H_L(p_j)
        \bigr).
\]
Its \(j\)-th tangent row is
\[
        \sigma_j^{-1}
        \bigl(
        \partial_\theta K_1(p_j),\dots,\partial_\theta K_L(p_j),
        \partial_\theta H_1(p_j),\dots,\partial_\theta H_L(p_j)
        \bigr).
\]
Let
\(
        \boldsymbol\gamma_\ell
        :=
        \begin{bmatrix}
        \alpha_\ell\\
        \beta_\ell
        \end{bmatrix},
        ~
        \mathbf r_j
:=
\begin{bmatrix}
\widetilde v_j\\
0
\end{bmatrix}.
\)
Then \eqref{eq:app_hermite_conditions} is the block system
\begin{equation}\label{eq:app_hermite_block_system}
        \boldsymbol\gamma_j
        +
        \sum_{\ell\ne j}\mathbf E_{j\ell}\boldsymbol\gamma_\ell
        =
        \mathbf r_j,
        ~ j=1,\dots,L .
\end{equation}
where \(\mathbf E_{j\ell}\in\C^{2\times2}\) is the off-diagonal block
\[
      \mathbf E_{j\ell}
        =
        \begin{bmatrix}
        K_\ell(p_j) & H_\ell(p_j)\\[0.5ex]
        \sigma_j^{-1}\partial_\theta K_\ell(p_j)
        &
        \sigma_j^{-1}\partial_\theta H_\ell(p_j)
        \end{bmatrix}.
\]
The diagonal block is the \(2\times2\) identity.  Indeed, using the diagonal
identities above,
\[
        \begin{bmatrix}
        K_j(p_j) & H_j(p_j)\\[0.5ex]
        \sigma_j^{-1}\partial_\theta K_j(p_j)
        &
        \sigma_j^{-1}\partial_\theta H_j(p_j)
        \end{bmatrix}
        =
        \begin{bmatrix}
        1&0\\0&1
        \end{bmatrix}.
\]

% ============================================================
\subsection{Block Neumann domination}
\label{app:block-neumann-domination}
% ============================================================

We write the Hermite coefficient block at support point \(p_j\) as
\(
       \boldsymbol\gamma_j
:=
\begin{bmatrix}
\alpha_j\\
\beta_j
\end{bmatrix},
~
\mathbf r_j
:=
\begin{bmatrix}
\widetilde v_j\\
0
\end{bmatrix}.
\)
The gauged Hermite normalization gives the diagonal block identity
\[
       \mx E_{jj}
        =
        \begin{bmatrix}
        K(p_j,p_j) & H(p_j,p_j)\\
        dK(p_j,p_j) & dH(p_j,p_j)
        \end{bmatrix}
        =
        \mathbf I_2 .
\]
Thus the support-support Hermite matrix has identity diagonal blocks and
off-diagonal blocks
\[
      \mx E_{j\ell}
        =
        \begin{bmatrix}
        K(p_j,p_\ell) & H(p_j,p_\ell)\\
        dK(p_j,p_\ell) & dH(p_j,p_\ell)
        \end{bmatrix},
        ~ j\ne \ell .
\]

Assume that the off-diagonal blocks are dominated entrywise by the nonnegative
matrix
\[
        \mathbf G_{\rm SS}
        :=
        \begin{bmatrix}
        u_K^{\rm SS} & u_H^{\rm SS}\\
        u_{dK}^{\rm SS} & u_{dH}^{\rm SS}
        \end{bmatrix}
        \in\R_+^{2\times2}.
\]
That is, for every \(j\ne\ell\),
\begin{equation}
\label{eq:app_GSS_domination}
        |\mathbf E_{j\ell}|
        \le
        \mathbf G_{\rm SS}
\end{equation}
entrywise.  Equivalently,
\[
        |K(p_j,p_\ell)|\le u_K^{\rm SS},
        ~
        |H(p_j,p_\ell)|\le u_H^{\rm SS},
\]
and
\[
        |dK(p_j,p_\ell)|\le u_{dK}^{\rm SS},
        ~
        |dH(p_j,p_\ell)|\le u_{dH}^{\rm SS}.
\]
Assume the support-support Neumann condition
\begin{equation}
\label{eq:app_neumann_condition}
        (L-1)\rho_{\rm sp}(\mathbf G_{\rm SS})<1 .
\end{equation}
Define
\begin{equation}
\label{eq:app_Gamma_def}
       \boldsymbol \Gamma
        =
        \begin{bmatrix}
        \Gamma_K\\
        \Gamma_H
        \end{bmatrix}
        :=
        \left(\mathbf I_2-(L-1)\mathbf G_{\rm SS}\right)^{-1}
        \begin{bmatrix}
        1\\
        0
        \end{bmatrix}.
\end{equation}
Define also
\[
        \boldsymbol\Xi
        :=
        \boldsymbol\Gamma-
        \begin{bmatrix}
        1\\
        0
        \end{bmatrix}
        =
        \begin{bmatrix}
        \Xi_K\\
        \Xi_H
        \end{bmatrix}.
\]
Because \(\mathbf G_{\rm SS}\ge0\) and
\((L-1)\rho_{\rm sp}(\mathbf G_{\rm SS})<1\), the Neumann series
\[
        \left(\mathbf I_2-(L-1)\mathbf G_{\rm SS}\right)^{-1}
        =
        \sum_{r=0}^{\infty}
        \bigl((L-1)\mathbf G_{\rm SS}\bigr)^r
\]
converges entrywise to a nonnegative matrix.

The block domination condition reduces the off-diagonal Hermite
interactions to a nonnegative \(2\times2\) comparison system.  Since the
spectral radius of this comparison operator is strictly below one, the
associated Neumann series controls both invertibility of the Hermite matrix
and the size of the interpolation coefficients.

\begin{lem}
\label{lem:app_hermite_neumann}
Under \eqref{eq:app_GSS_domination} and
\eqref{eq:app_neumann_condition}, the gauged Hermite matrix
\(\mathcal H_{\rm g}(S)\) is invertible.  Moreover, for every sign vector
\(|v_\ell|=1\), the solution of the Hermite system satisfies
\[
        |\alpha_\ell|\le \Gamma_K,
        ~
        |\beta_\ell|\le \Gamma_H,
        ~
        \ell=1,\ldots,L .
\]
\end{lem}

\begin{proof}
First we prove invertibility.  Suppose
\[
        \mathcal H_{\rm g}(S)\boldsymbol\gamma=0,
        ~
        \boldsymbol\gamma
        :=
        (\boldsymbol\gamma_1,\ldots,\boldsymbol\gamma_L).
\]
Using the diagonal identity \(\mathbf E_{jj}=\mathbf I_2\), the \(j\)-th block equation is
\[
        \boldsymbol\gamma_j
        +
        \sum_{\ell\ne j}\mathbf E_{j\ell}\boldsymbol\gamma_\ell
        =
        0,
        ~ j=1,\ldots,L .
\]
Taking componentwise absolute values and using
\eqref{eq:app_GSS_domination}, we obtain
\(
        |\boldsymbol\gamma_j|
        \le
        \sum_{\ell\ne j}\mathbf G_{\rm SS}|\boldsymbol\gamma_\ell|.
\)
Let
\(
        \mathbf U:=\max_{1\le j\le L}|\boldsymbol\gamma_j|
\)
componentwise.  Then
\(
       \mathbf U\le (L-1)\mathbf G_{\rm SS}\mathbf U .
\)
Set \(\mathbf M:=(L-1)\mathbf G_{\rm SS}\).  Since \(\mathbf M\ge0\), iteration gives
\(
        \mathbf U\le \mathbf M^r \mathbf U,
        ~ r=1,2,\ldots .
\)
The spectral condition \(\rho_{\rm sp}(\mx M)<1\) implies \(\mathbf M^r\to0\).  Hence
\(\mathbf U=0\).  Therefore \(\boldsymbol\gamma_j=0\) for every \(j\), so
\(\mathcal H_{\rm g}(S)\) is injective.  Since it is a square finite-dimensional
matrix, it is invertible.

We now prove the coefficient bounds.  The inhomogeneous Hermite system has
block equations
\[
        \boldsymbol\gamma_j
        +
        \sum_{\ell\ne j}\mathbf E_{j\ell}\boldsymbol\gamma_\ell
        =
        \mathbf r_j,
        ~
        \mathbf r_j=
        \begin{bmatrix}
        \widetilde v_j\\
        0
        \end{bmatrix}.
\]
Since \(|\widetilde v_j|=1\), we have
\(
   |\mathbf r_j|
=
\mathbf e,
\qquad
\mathbf e:=
\begin{bmatrix}
1\\
0
\end{bmatrix}.
\)
Taking componentwise absolute values gives
\(
        |\boldsymbol\gamma_j|
        \le
        e+
        \sum_{\ell\ne j}\mathbf G_{\rm SS}|\boldsymbol\gamma_\ell|.
\)
Taking the componentwise maximum over \(j\) gives
\(
        \mathbf U\le \mathbf e+\mathbf M\mathbf U,
        ~
        \mx M=(L-1)\mathbf G_{\rm SS}.
\)
Iterating this inequality yields, for every \(r\ge0\),
\(
       \mx  U
        \le
        \sum_{m=0}^{r}\mx M^m e
        +
        \mx M^{r+1}\mx U .
\)
Since \(\mx M\ge0\) and \(\rho_{\rm sp}(\mathbf M)<1\), we have \(\mx M^{r+1}\mx U\to0\), and
the Neumann series converges.  Therefore
\(
       \mx U
        \le
        \sum_{m=0}^{\infty}\mx M^m \mx e
        =
        (\mathbf I_2-\mx M)^{-1}\mx e
        =
       \boldsymbol \Gamma .
\)
Thus, for every \(\ell=1,\ldots,L\),
\(
        |\boldsymbol\gamma_\ell|\le \bs\Gamma .
\)
Equivalently,
\(
        |\alpha_\ell|\le \Gamma_K,
        ~
        |\beta_\ell|\le \Gamma_H .
\)
\end{proof}

For the local curvature analysis, the global coefficient estimates must be
refined at each diagonal block.  Subtracting the ideal self-interpolation
vector from the value and tangent equations gives the correction quantities
\(\Xi_K\) and \(\Xi_H\).

\begin{lem}
\label{lem:app_diagonal_correction_bounds}
Under the assumptions of \Cref{lem:app_hermite_neumann}, the Hermite
coefficients satisfy, for every \(j=1,\ldots,L\),
\[
        |\alpha_j-\widetilde v_j|\le \Xi_K,
        ~
        |\beta_j|\le \Xi_H .
\]
\end{lem}

\begin{proof}
The value interpolation equation at \(p_j\) is
\[
        \alpha_j
        +
        \sum_{\ell\ne j}\alpha_\ell K(p_j,p_\ell)
        +
        \sum_{\ell\ne j}\beta_\ell H(p_j,p_\ell)
        =
        \widetilde v_j .
\]
Hence
\[
        \alpha_j-\widetilde v_j
        =
        -
        \sum_{\ell\ne j}\alpha_\ell K(p_j,p_\ell)
        -
        \sum_{\ell\ne j}\beta_\ell H(p_j,p_\ell).
\]
Using the support-support bounds and
\Cref{lem:app_hermite_neumann}, we obtain
\[
        |\alpha_j-\widetilde v_j|
        \le
        (L-1)
        \left(
        u_K^{\rm SS}\Gamma_K
        +
        u_H^{\rm SS}\Gamma_H
        \right).
\]
Similarly, the tangent interpolation equation gives
\[
        \beta_j
        =
        -
        \sum_{\ell\ne j}\alpha_\ell dK(p_j,p_\ell)
        -
        \sum_{\ell\ne j}\beta_\ell dH(p_j,p_\ell),
\]
and therefore
\[
        |\beta_j|
        \le
        (L-1)
        \left(
        u_{dK}^{\rm SS}\Gamma_K
        +
        u_{dH}^{\rm SS}\Gamma_H
        \right).
\]
Set
\(
       \mathbf M:=(L-1)\mathbf G_{\rm SS},
        ~
        \mathbf e:=
        \begin{bmatrix}
        1\\0
        \end{bmatrix}.
\)
Since
\(
        \boldsymbol\Gamma=(\mathbf I_2-\mathbf M)^{-1}\mx e,
\)
we have
\(
     \boldsymbol\Gamma-\mathbf e
=
\mathbf M\boldsymbol\Gamma.
\)
Thus the first component of
\(\mathbf M\boldsymbol\Gamma\)
is \(\Xi_K\), and the second component is \(\Xi_H\).  The two preceding estimates therefore give
\(
        |\alpha_j-\widetilde v_j|\le \Xi_K,
        ~
        |\beta_j|\le \Xi_H .
\)
\end{proof}
% ============================================================
\subsection{From the gauged certificate to an adjoint TV certificate}
\label{app:gauged-to-ungauged}
% ============================================================

Let \(P_{S,\widetilde {\mathbf v}}\) be the solution of the gauged Hermite system.  We
write points as \(p=(i,\theta)\), and use the shorthand
\(
        \chi_p:=\chi_i(\theta).
\) For \(p=(i,\theta)\in\Qset\), let
\(
        \varphi_p[n]
        :=
        a_{\rm Fr}(r_i,\theta)[n],
        ~
        n=0,\ldots,N_r-1,
\)
denote the raw Fresnel measurement atom.
The gauged atom is related to the raw Fresnel atom by
\[
        \psi_p[n]
        =
        W_0^{-1/2}e^{-i\chi_p}\varphi_p[n],
        ~ n=0,\ldots,N_r-1 .
\]
The gauged signs are defined by
\(
        \widetilde v_\ell
        :=
        e^{i\chi_{p_\ell}}v_\ell,
        ~
        p_\ell=(i_\ell,\theta_\ell).
\)

Define the ungauged family \(   Q_{S, {\mathbf v}}(p)
        :=
        e^{-i\chi_p}P_{S,\widetilde {\mathbf v}}(p),
        ~ p=(i,\theta).\)
Since \(|e^{-i\chi_p}|=1\), we immediately have
\[
        |Q_{S, {\mathbf v}}(p)|
        =
        |P_{S,\widetilde {\mathbf v}}(p)|.
\]
Moreover, at each support point \(p_\ell\),
\[
        Q_{S, {\mathbf v}}(p_\ell)
        =
        e^{-i\chi_{p_\ell}}P_{S,\widetilde {\mathbf v}}(p_\ell)
        =
        e^{-i\chi_{p_\ell}}\widetilde v_\ell
        =
        v_\ell .
\]
Thus the ungauged family interpolates the original TV signs.
It remains to verify that \(Q_{S, {\mathbf v}}\) is a genuine raw Fresnel adjoint
polynomial.  Since
\(
        P_{S,\widetilde {\mathbf v}}(p)
        =
        \langle \boldsymbol\psi_p,\mathbf z_{S,\widetilde{\mathbf{v}}}\rangle_\rho ,
\)
and the weighted inner product is
\(
        \langle u,v\rangle_\rho
        =
        \sum_{n=0}^{N_r-1}\rho_n\overline{u[n]}v[n],
\)
we have
\[
\begin{aligned}
        &P_{S,\widetilde {\mathbf v}}(p)
        =
        \sum_{n=0}^{N_r-1}
        \rho_n\overline{\psi_p[n]}z_{S,\widetilde v}[n] =
        e^{i\chi_p}
        \sum_{n=0}^{N_r-1}
        \tfrac{\rho_n}{\sqrt{W_0}}
        z_{S,\widetilde v}[n]\overline{\varphi_p[n]} .
\end{aligned}
\]
Therefore
\[
        Q_{S, {\mathbf v}}(p)
        =
        e^{-i\chi_p}P_{S,\widetilde {\mathbf v}}(p)
        =
        \sum_{n=0}^{N_r-1}
        \zeta_{S,\mathbf v}[n]\overline{\varphi_p[n]},
\]
where
\(
        \zeta_{S,v}[n]
        :=
        \tfrac{\rho_n}{\sqrt{W_0}}z_{S,\widetilde v}[n].
\)
Hence \(Q_{S,\mathbf v}=Q^{\boldsymbol\zeta_{S,\mathbf v}}\) belongs to the range of the raw Fresnel
adjoint operator.

\begin{rem}[Stationarity is a gauged condition]
The gauged certificate satisfies
\(
        \partial_\theta P_{S,\widetilde {\mathbf v}}(p_\ell)=0.
\)
The ungauged certificate \(Q_{S, {\mathbf v}}\) need not be stationary at \(p_\ell\),
because the gauge factor \(e^{-i\chi_p}\) generally has a nonzero angular
derivative.  This is harmless.  The TV dual certificate criterion requires
sign interpolation and strict off-support boundedness for \(Q_{S, {\mathbf v}}\), not
stationarity of \(Q_{S, {\mathbf v}}\).  The stationarity of \(P_{S,\widetilde {\mathbf v}}\) is
only a device used to prove local modulus decay.
\end{rem}

% ============================================================
\subsection{Atom linear independence}
\label{app:atom-linear-independence}
% ============================================================

The TV dual-certificate criterion also requires independence of the active
Fresnel atoms.  This is not an additional hypothesis in the main theorem.
Indeed, any nontrivial linear dependence among the active atoms would
produce a nonzero vector in the null space of the gauged Hermite matrix.
Hermite invertibility therefore gives the required independence.

\begin{lem}
\label{lem:app_hermite_implies_atom_independence}
If \(\mathcal H_{\rm g}(S)\) is invertible, then the active Fresnel atoms
\(
        \{\aFR(r_{i_\ell},\theta_\ell)\}_{\ell=1}^L
\)
are linearly independent in the full aperture space \(\C^{N_r}\).
\end{lem}

\begin{proof}
Suppose, toward a contradiction, that the raw Fresnel atoms are linearly
dependent in the full aperture space. Then there exist coefficients
\(q_1,\ldots,q_L\), not all zero, such that \(\sum_{\ell=1}^L q_\ell \aFR(r_{i_\ell},\theta_\ell) = 0 ~ \text{in }\C^{N_r}.\) Each gauged atom is obtained from the corresponding raw Fresnel atom by
multiplication by the nonzero scalar \(W_0^{-1/2}e^{-i\chi_{i_\ell}(\theta_\ell)}.\) Therefore there exist coefficients \(\widetilde q_\ell\), not all zero, such
that \(\sum_{\ell=1}^L \widetilde q_\ell\psi_{p_\ell} = 0 ~ \text{in }\C^{N_r}.\) Moreover, \(\widetilde q_\ell=0 \quad\Longleftrightarrow\quad q_\ell=0 .\) In particular, the same identity holds after restriction to the active aperture
set \(\mathcal A_\rho:=\{n:\rho_n>0\}\).  Any full-aperture linear dependence
therefore restricts to a dependence on \(\mathcal A_\rho\).  Thus invertibility
of the weighted Hermite system on the active aperture rules out a full-aperture
dependence.

Define the auxiliary Hermite vector
\(
        \boldsymbol\eta
        :=
        (\boldsymbol\eta_1,\ldots,\boldsymbol\eta_L),
        ~
        \boldsymbol\eta_\ell
        :=
        \begin{bmatrix}
        \widetilde q_\ell\\
        0
        \end{bmatrix}.
\)
For any evaluation point \(p=(i,\theta)\), the corresponding gauged function is
\[
\begin{aligned}
 &P(p)
 =
 \sum_{\ell=1}^L
 \widetilde q_\ell K(p,p_\ell)=
 \left\langle
 \psi_p,
 \sum_{\ell=1}^L\widetilde q_\ell\psi_{p_\ell}
 \right\rangle_\rho
 =
 0 .
\end{aligned}
\]
Hence all value rows of \(\mathcal H_{\rm g}(S)\boldsymbol\eta\) vanish.
Likewise, for every support point \(p_j\),
\[
\begin{aligned}
 &\partial_\theta P(p_j)
 =
 \left\langle
 \partial_\theta\psi_{p_j},
 \sum_{\ell=1}^L\widetilde q_\ell\psi_{p_\ell}
 \right\rangle_\rho
 =
 0 .
\end{aligned}
\]
Hence all tangent rows of \(\mathcal H_{\rm g}(S)\boldsymbol\eta\) vanish. Therefore \(\mathcal H_{\rm g}(S)\boldsymbol\eta=0.\) Since \(\mathcal H_{\rm g}(S)\) is invertible, \(\boldsymbol\eta=0\). Thus \(\widetilde q_\ell=0, ~ \ell=1,\ldots,L.\) This contradicts the fact that the coefficients \(\widetilde q_\ell\) are not
all zero. Therefore no nontrivial full-space dependence among the raw Fresnel
atoms exists. Hence \(\{\aFR(r_{i_\ell},\theta_\ell)\}_{\ell=1}^L\) is linearly independent in \(\C^{N_r}\).
\end{proof}
% ============================================================
% \subsection{Summary of the appendix}
% \label{app:tv-duality-summary}
% % ============================================================

% The results above establish the logical chain used in the main theorem:
% \(\text{block Neumann domination} \Longrightarrow \text{Hermite invertibility and coefficient bounds},\)
% \(\text{Hermite invertibility} \Longrightarrow \text{active atom linear independence},\)
% and
% \(\text{gauged Hermite certificate} \Longrightarrow \text{ungauged adjoint TV certificate}.\)
% Consequently, once the main text proves
% \(|P_{S,\widetilde {\mathbf v}}(p)|<1 ~ \text{for every }p\notin S,\)
% the ungauged family \(Q_{S, {\mathbf v}}\) satisfies the TV dual certificate criterion and
% the exact recovery conclusion follows.

% ============================================================
% ============================================================
% ============================================================
% ============================================================
\section{Proofs of weighted QPAC estimates and gauged kernel bounds}
\label{app:qsum}
% ============================================================

This appendix proves the weighted quadratic-phase aperture certification
(QPAC) estimates used in \Cref{sec:qsum-bounds}.  The proof has three layers.

First, we prove scalar estimates for finite weighted quadratic exponential
sums
\begin{equation}
\label{eq:app_basic_qsum}
        T_N(a;\omega_1,\omega_2)
        :=
        \sum_{n=0}^{N-1}
        a_n e^{i(\omega_1 n+\omega_2 n^2)} ,
\end{equation}
where \(a=(a_0,\ldots,a_{N-1})\in\mathbb C^N\). The derivative branch is obtained by fourth-order Abel summation.  The Lag
Correlation branch is obtained by residue splitting, retaining the complete
residue lag-pair phase, and controlling the real lag contributions through
cosine majorants.  The residue-linear branch is obtained by decomposing the
aperture into residue classes when \(\omega_2\) is close to a rational multiple
of \(\pi\).

Second, we show that each gauged Fresnel channel is exactly a scalar sum of
the form \eqref{eq:app_basic_qsum}, up to a unit-modulus phase factor.  The
four normalized gauged Hermite channels are
\(K,~ H,~ dK,~ dH,\)
corresponding respectively to
\(
        (\psi,\psi),~
        (\psi,h),~
        (h,\psi),~
        (h,h).
\)
The higher angular derivative channels are
\(d^2K,~ d^3K,~ d^2H,~ d^3H,\)
where the angular derivative is taken with respect to the evaluation point
while the support point is fixed.

Third, we convert the pairwise scalar estimates into support-uniform estimates
on compact physical support/evaluation cells.  The support-uniform estimates
are cellwise bounds: they hold simultaneously for every pair of physical
parameters in a prescribed cell. When such cells form a certified cover of the admissible pair class, taking
the supremum of the cellwise bounds yields a support-uniform bound over the
entire class.

The normalized channels \(K,H,dK,dH\) have the weighted
Cauchy-Schwarz cap \(1\). The higher angular derivative channels are not
normalized inner products and therefore receive derivative-size
Cauchy-Schwarz caps instead.

Throughout the appendix, the superscript \({\rm SS}\) means
\emph{support-support}.  It is not a channel label.

% ============================================================
\subsection{Scalar notation and finite differences}
\label{app:scalar_notation}
% ============================================================

For a finite sequence \(c=(c_0,\ldots,c_M)\), define the forward difference
\((\Delta c)_m:=c_{m+1}-c_m, ~ m=0,\ldots,M-1.\)
Higher forward differences are defined recursively:
\(\Delta^{j+1}c:=\Delta(\Delta^jc).\)
The norm
\(\|\Delta^j c\|_1\)
is always taken over the natural index range of \(\Delta^j c\).  Empty sums
are interpreted as zero.
For \(x\in\mathbb R\), define
\(\dist_{2\pi}(x):=\min_{k\in\mathbb Z}|x-2\pi k|\)
and
\(\dist_{\pi}(x):=\min_{k\in\mathbb Z}|x-\pi k|.\)
The signed representative modulo \(2\pi\) is
\(
        \dist_{2\pi}^{\rm sgn}(x)
        :=
        \operatorname{mod}(x+\pi,2\pi)-\pi
        \in[-\pi,\pi).
\)
Thus
\(x-\dist_{2\pi}^{\rm sgn}(x)\in2\pi\mathbb Z.\)
For a phase pair \((\omega_1,\omega_2)\), define the extended derivative
separation
\begin{equation}
\label{eq:app_dN_plus}
        d_N^+(\omega_1,\omega_2)
        :=
        \min_{0\le n\le N-1}
        \dist_{2\pi}\bigl(\omega_1+(2n+1)\omega_2\bigr).
\end{equation}
The index \(n=N-1\) is included because the fourth-order Abel summation
argument introduces one auxiliary endpoint.

% ============================================================
\subsection{The scalar QPAC best-branch estimate}
\label{app:scalar_best_branch}
% ============================================================

Let \(Q_{\max}\in\{2,\ldots,N\}\).  The scalar QPAC estimate is defined as the
minimum of one universal cap and three valid cancellation bounds:

\begin{equation}
\label{eq:app_Bbest_def}
\begin{aligned}
        B_{\rm best}^{(4,Q_{\max})}
        (a;\omega_1,\omega_2;N)
        :=
        \min\Bigl\{
        &\|a\|_1,\,
        B_{\rm der}^{(4)}(a;\omega_1,\omega_2;N),\\
        &B_{\RS}^{(Q_{\max})}(a;\omega_1,\omega_2;N),\,
        B_{\rm res,lin}^{(Q_{\max})}(a;\omega_1,\omega_2;N)
        \Bigr\}.
\end{aligned}
\end{equation}

The separate scalar branches developed in this appendix are combined through
\(B_{\rm best}^{(4,Q_{\max})}\).  The resulting estimate is the pointwise
input used later in the cellwise support-uniform construction.

\begin{thm}
\label{thm:app_scalar_qpac}
Let \(N\ge9\), let \(Q_{\max}\in\{2,\ldots,N\}\), and let
\(a\in\mathbb C^N\) satisfy the fourth-order flat-end condition
\begin{equation}
\label{eq:app_r4_flat_end}
        a_0=a_1=a_2=a_3=0,
        ~
        a_{N-4}=a_{N-3}=a_{N-2}=a_{N-1}=0.
\end{equation}
Then, for every \((\omega_1,\omega_2)\in\mathbb R^2\),
\begin{equation}
\label{eq:app_scalar_qpac_bound}
        |T_N(a;\omega_1,\omega_2)|
        \le
        B_{\rm best}^{(4,Q_{\max})}
        (a;\omega_1,\omega_2;N).
\end{equation}
\end{thm}

The flat-end condition is needed only for the fourth-order derivative branch.
The trivial cap, the lag-correlation branch, and the residue-linear branch are
valid for arbitrary finite coefficient sequences.  In the gauged kernel
application, all coefficient families inherit the flat-end condition from the
taper, so the complete best-branch estimate applies.

% ============================================================
\subsection{Derivative branch}
\label{app:derivative_branch}
% ============================================================

If
\(d_N^+(\omega_1,\omega_2)>0,\)
define
\(s_N:=\sin\left(\tfrac{d_N^+(\omega_1,\omega_2)}{2}\right).\)
If \(d_N^+=0\), the derivative branch is declared to be \(+\infty\).

For \(d_N^+>0\), define
\begin{equation}
\label{eq:app_derivative_branch}
        B_{\rm der}^{(4)}(a;\omega_1,\omega_2;N)
        :=
        \sum_{j=0}^{4}
        \beta_j(\omega_2,s_N)\|\Delta^j a\|_1,
\end{equation}
where
\[
        \beta_0(\omega_2,s_N)
        =
        \tfrac{105|\sin\omega_2|^4}{16s_N^8},
        ~
        \beta_1(\omega_2,s_N)
        =
        \tfrac{105|\sin\omega_2|^3}{16s_N^7},
\]
\[
        \beta_2(\omega_2,s_N)
        =
        \tfrac{45|\sin\omega_2|^2}{16s_N^6},
        ~
        \beta_3(\omega_2,s_N)
        =
        \tfrac{5|\sin\omega_2|}{8s_N^5},
\]
and
\(\beta_4(\omega_2,s_N) = \tfrac{1}{16s_N^4}.\)
If \(d_{N}^+=0\), set
\(B_{\rm der}^{(4)}(a;\omega_1,\omega_2;N):=+\infty.\)

Four successive summation-by-parts steps generate a recursively
differentiated product of the coefficient sequence and the reciprocal phase
increment.  To control the resulting sequence, we first establish a
fourth-order discrete Leibniz estimate in terms of the finite differences of
the original coefficients and of the reciprocal-increment sequence.

\begin{lem}
\label{lem:app_fourth_order_product_bound}
Let \(\phi=(\phi_n)\) be a finite sequence.  Define
\(a^{(0)}:=a,\)
and recursively
\[
        a^{(k+1)}
        :=
        \Delta\bigl(a^{(k)}\phi_{\cdot+k}\bigr),
        ~ k=0,1,2,3.
\]
For \(j=0,\ldots,4\), set
\(M_j:=\max_n |(\Delta^j\phi)_n|,\)
where the maximum is taken over the natural index range.  Then
\begin{align}
\|a^{(4)}\|_1
&\le
M_0^4\|\Delta^4a\|_1
+
10M_0^3M_1\|\Delta^3a\|_1
\nonumber\\
&\quad+
\left(
10M_0^3M_2+25M_0^2M_1^2
\right)\|\Delta^2a\|_1
\nonumber\\
&\quad+
\left(
5M_0^3M_3
+
30M_0^2M_1M_2
+
15M_0M_1^3
\right)\|\Delta a\|_1
\nonumber\\
&\quad+
\left(
M_0^3M_4
+
7M_0^2M_1M_3
+
4M_0^2M_2^2
+
11M_0M_1^2M_2
+
M_1^4
\right)\|a\|_1 .
\label{eq:app_fourfold_product_bound}
\end{align}
\end{lem}

\begin{proof}
We use the forward-difference convention \((\Delta b)_n:=b_{n+1}-b_n .\) All products, differences, and norms are taken over their natural common index
ranges.
The finite-difference product rule is \(\Delta(fg)_n = (\Delta f)_n g_{n+1} + f_n(\Delta g)_n .\) Iterating this identity gives the finite-difference Leibniz inequality \(\|\Delta^r(fg)\|_1 \le \sum_{j=0}^{r} \binom{r}{j} \|\Delta^j f\|_1 \|\Delta^{r-j}g\|_\infty .\) The shifts that appear in the exact Leibniz formula are harmless here, because
they do not change the \(\ell^\infty\)-norm of the \(\phi\)-factor and only
restrict the natural summation range of the \(\ell^1\)-norm.

For \(k=0,1,2,3\), write \(\phi^{[k]}_n:=\phi_{n+k}.\) Then \(a^{(k+1)} = \Delta\bigl(a^{(k)}\phi^{[k]}\bigr).\) Since \(\Delta^m a^{(k+1)} = \Delta^{m+1}\bigl(a^{(k)}\phi^{[k]}\bigr),\) the Leibniz inequality with \(r=m+1\) gives \(\|\Delta^m a^{(k+1)}\|_1 \le \sum_{j=0}^{m+1} \binom{m+1}{j} M_{m+1-j} \|\Delta^j a^{(k)}\|_1 .\) We now keep track of the dependence on the original sequence \(a\). Define
coefficients \(C_{k,m,\ell}\) by \(\|\Delta^m a^{(k)}\|_1 \le \sum_{\ell=0}^{m+k} C_{k,m,\ell}\|\Delta^\ell a\|_1 .\) For \(k=0\), we have
\(
 C_{0,m,\ell}
 =
 \begin{cases}
 1,&\ell=m,\\
 0,&\ell\neq m.
 \end{cases}
\)
Combining this definition with the previous inequality gives the recursion \(C_{k+1,m,\ell} = \sum_{j=0}^{m+1} \binom{m+1}{j} M_{m+1-j}C_{k,j,\ell}.\) We need the coefficients \(C_{4,0,\ell}\), \(\ell=0,\ldots,4\). The final
step of the recursion is particularly simple: \(C_{4,0,\ell} = M_1C_{3,0,\ell} + M_0C_{3,1,\ell}.\) Thus it remains only to record the needed third-level coefficients.

First, \(C_{3,0,0} = M_0^2M_3+4M_0M_1M_2+M_1^3,\) \(C_{3,0,1} = 4M_0^2M_2+7M_0M_1^2,\) \(C_{3,0,2} = 6M_0^2M_1, ~ C_{3,0,3} = M_0^3, ~ C_{3,0,4}=0.\) Second, \(C_{3,1,0} = M_0^2M_4 + 6M_0M_1M_3 + 4M_0M_2^2 + 7M_1^2M_2,\) \(C_{3,1,1} = 5M_0^2M_3 + 26M_0M_1M_2 + 8M_1^3,\) \(C_{3,1,2} = 10M_0^2M_2+19M_0M_1^2,\) \(C_{3,1,3} = 9M_0^2M_1, ~ C_{3,1,4} = M_0^3.\) Substituting these coefficients into \(C_{4,0,\ell} = M_1C_{3,0,\ell} + M_0C_{3,1,\ell}\) gives, for \(\ell=4\), \(C_{4,0,4} = M_1C_{3,0,4} + M_0C_{3,1,4} = M_0^4.\) For \(\ell=3\), \(C_{4,0,3} = M_1M_0^3 + M_0(9M_0^2M_1) = 10M_0^3M_1.\) For \(\ell=2\), \(C_{4,0,2} = M_1(6M_0^2M_1) + M_0(10M_0^2M_2+19M_0M_1^2),\) and hence \(C_{4,0,2} = 10M_0^3M_2+25M_0^2M_1^2.\) For \(\ell=1\), \(C_{4,0,1} = M_1(4M_0^2M_2+7M_0M_1^2) + M_0(5M_0^2M_3+26M_0M_1M_2+8M_1^3),\) so \(C_{4,0,1} = 5M_0^3M_3 + 30M_0^2M_1M_2 + 15M_0M_1^3.\) Finally, for \(\ell=0\),
\[
\begin{aligned}
 C_{4,0,0}
 &=
 M_1
 \left(
 M_0^2M_3+4M_0M_1M_2+M_1^3
 \right) \\
 &\quad+
 M_0
 \left(
 M_0^2M_4
 +
 6M_0M_1M_3
 +
 4M_0M_2^2
 +
 7M_1^2M_2
 \right).
\end{aligned}
\]
Collecting terms gives \(C_{4,0,0} = M_0^3M_4 + 7M_0^2M_1M_3 + 4M_0^2M_2^2 + 11M_0M_1^2M_2 + M_1^4 .\) Therefore \(\|a^{(4)}\|_1 \le \sum_{\ell=0}^{4} C_{4,0,\ell}\|\Delta^\ell a\|_1 .\) Substituting the computed values of \(C_{4,0,\ell}\) gives exactly
\eqref{eq:app_fourfold_product_bound}.
\end{proof}

We now apply the preceding product estimate to
\(
        \phi_n
        =
        \frac{1}{1-e^{i(\omega_1+(2n+1)\omega_2)}}.
\)
Uniform separation of the discrete phase increments from
\(2\pi\mathbb Z\) controls the finite differences of \(\phi\), while the
flat-end condition eliminates all boundary terms in four successive
summation-by-parts steps.  This gives the derivative-branch estimate.

\begin{lem}
\label{lem:app_derivative_branch}
Under the assumptions of \Cref{thm:app_scalar_qpac},
\(|T_N(a;\omega_1,\omega_2)| \le B_{\rm der}^{(4)}(a;\omega_1,\omega_2;N).\)
\end{lem}

\begin{proof}
We use the forward-difference convention \((\Delta b)_n:=b_{n+1}-b_n,\) and all finite differences and norms below are taken over their natural ranges.

If \(d_N^+=0\), then the derivative branch is \(+\infty\), so the claim is
immediate. We therefore assume \(d_N^+>0\).

Define \(u_n:=e^{i(\omega_1 n+\omega_2 n^2)},~ n=0,\ldots,N.\) Then \(u_{n+1} = u_n e^{i\delta_n}, ~ \delta_n:=\omega_1+(2n+1)\omega_2, ~ n=0,\ldots,N-1.\) Since \(\dist_{2\pi}(\delta_n)\ge d_{N}^+,\) and since \(x\mapsto \sin(x/2)\) is increasing on \([0,\pi]\), we have
\[
 |1-e^{i\delta_n}|
 =
 2\left|\sin\left(\tfrac{\delta_n}{2}\right)\right|
 \ge
 2\sin\left(\tfrac{d_{N}^+}{2}\right)
 =
 2s_N .
\]
In particular, \(1-e^{i\delta_n}\neq 0\). Define \(\phi_n:=\tfrac{1}{1-e^{i\delta_n}}, ~ n=0,\ldots,N-1.\) Then \(u_n-u_{n+1} = u_n(1-e^{i\delta_n}),\) and therefore \(u_n = \phi_n(u_n-u_{n+1}).\) We use the summation-by-parts identity
\[
 \sum_{n=0}^{M}x_n(y_n-y_{n+1})
 =
 x_0y_0-x_My_{M+1}
 +
 \sum_{n=1}^{M}(x_n-x_{n-1})y_n .
 \tag{SBP}
 \label{eq:app_sbp_identity}
\]
Set \(a^{(0)}_n:=a_n,~ n=0,\ldots,N-1,\) and, for \(k=0,1,2,3\), define \(a^{(k+1)} := \Delta\bigl(a^{(k)}\phi_{\cdot+k}\bigr),\) that is, \(a^{(k+1)}_n = a^{(k)}_{n+1}\phi_{n+1+k} - a^{(k)}_{n}\phi_{n+k}.\) We first record the endpoint behavior. By the flat-end condition
\eqref{eq:app_r4_flat_end}, \(a^{(0)}\) vanishes at least at its first and last four
entries. Since all relevant values of \(\phi\) are finite, multiplication by
\(\phi_{\cdot+k}\) does not destroy endpoint zeros. Also, if a sequence \(b\)
vanishes at its first and last \(r\) entries, then \(\Delta b\) vanishes at its
first and last \(r-1\) entries. Indeed, for the left endpoint, \((\Delta b)_n=b_{n+1}-b_n=0, ~ n=0,\ldots,r-2,\) and the right endpoint is identical. Hence, by induction, \(a^{(k)} \text{ vanishes at its first and last }4-k\text{ entries}, ~ k=0,1,2,3,4.\) We now prove by induction that \(T_N(a;\omega_1,\omega_2) = \sum_{n=0}^{N-k-1}a^{(k)}_n u_{n+k}, ~ k=0,1,2,3,4.\) For \(k=0\), this is exactly the definition \(T_N(a;\omega_1,\omega_2) = \sum_{n=0}^{N-1}a_nu_n.\) Assume the identity holds for some \(k\in\{0,1,2,3\}\). Since \(u_{n+k} = \phi_{n+k}(u_{n+k}-u_{n+k+1}),\) we get \(T_N(a;\omega_1,\omega_2) = \sum_{n=0}^{N-k-1} a^{(k)}_n\phi_{n+k} (u_{n+k}-u_{n+k+1}).\) Apply \eqref{eq:app_sbp_identity} with \(x_n:=a^{(k)}_n\phi_{n+k}, ~ y_n:=u_{n+k}, ~ M:=N-k-1.\) Because \(a^{(k)}\) vanishes at its first and last \(4-k\) entries, and
\(k\le 3\), we have \(x_0=0, ~ x_M=0.\) Therefore the boundary terms vanish, and
\[
 T_N(a;\omega_1,\omega_2)
 =
 \sum_{n=1}^{N-k-1}
 \left(
 a^{(k)}_n\phi_{n+k}
 -
 a^{(k)}_{n-1}\phi_{n+k-1}
 \right)u_{n+k}.
\]
Reindexing with \(m=n-1\) gives
\[
 T_N(a;\omega_1,\omega_2)
 =
 \sum_{m=0}^{N-k-2}
 \left(
 a^{(k)}_{m+1}\phi_{m+1+k}
 -
 a^{(k)}_{m}\phi_{m+k}
 \right)u_{m+k+1}.
\]
By the definition of \(a^{(k+1)}\), this is \(T_N(a;\omega_1,\omega_2) = \sum_{m=0}^{N-k-2} a^{(k+1)}_m u_{m+k+1}.\) This proves the induction. Taking \(k=4\), we obtain \(T_N(a;\omega_1,\omega_2) = \sum_{n=0}^{N-5}a^{(4)}_n u_{n+4}.\) Since \(|u_{n+4}|=1\), it follows that \(|T_N(a;\omega_1,\omega_2)| \le \sum_{n=0}^{N-5}|a^{(4)}_n| = \|a^{(4)}\|_1 .\) It remains to estimate \(\|a^{(4)}\|_1\). Define \(z_n:=e^{i\delta_n}, ~ w:=e^{i2\omega_2}.\) Then \(z_{n+j}=z_nw^j, ~ \phi_n=\tfrac{1}{1-z_n}.\) For \(r=0,1,2,3,4\), define \(M_r := \|\Delta^r\phi\|_\infty .\) More explicitly, \(M_r\) is the supremum of \(|\Delta^r\phi_n|\) over the
natural range \(0\le n\le N-1-r\).

Every denominator factor appearing in \(\Delta^r\phi_n\) has the form \(1-z_nw^j = 1-e^{i(\delta_n+2j\omega_2)} = 1-e^{i\delta_{n+j}}.\) Whenever it appears, the index \(n+j\) lies in \(\{0,\ldots,N-1\}\). Hence \(|1-z_nw^j| = |1-e^{i\delta_{n+j}}| \ge 2s_N .\) In particular, \(M_0 = \|\phi\|_\infty \le \tfrac{1}{2s_N}.\) We now compute the first four differences of \(\phi\). Since \(\phi_{n+j} = \tfrac{1}{1-z_nw^j},\) we have \((\Delta\phi)_n = \tfrac{1}{1-z_nw} - \tfrac{1}{1-z_n} = \tfrac{z_n(w-1)} {(1-z_n)(1-z_nw)}.\) Therefore \(M_1 \le \tfrac{|w-1|}{(2s_N)^2}.\) Since \(|w-1| = |e^{i2\omega_2}-1| = 2|\sin\omega_2|,\) we get \(M_1 \le \tfrac{|\sin\omega_2|}{2s_N^2}.\) For the second difference, \((\Delta^2\phi)_n = \tfrac{1}{1-z_nw^2} - \tfrac{2}{1-z_nw} + \tfrac{1}{1-z_n}.\) Putting the terms over the common denominator gives \((\Delta^2\phi)_n = \tfrac{ z_n(w-1)^2(1+z_nw) } { (1-z_n)(1-z_nw)(1-z_nw^2) }.\) Since \(|z_n|=|w|=1\), we have \(|1+z_nw|\le 2.\) Hence \(M_2 \le \tfrac{2|w-1|^2}{(2s_N)^3} = \tfrac{|\sin\omega_2|^2}{s_N^3}.\) For the third difference, \((\Delta^3\phi)_n = \tfrac{1}{1-z_nw^3} - \tfrac{3}{1-z_nw^2} + \tfrac{3}{1-z_nw} - \tfrac{1}{1-z_n}.\) Using the common denominator \(\prod_{j=0}^{3}(1-z_nw^j),\) one obtains
\[
 (\Delta^3\phi)_n
 =
 \tfrac{
 z_n(w-1)^3
 \left(
 1+2z_nw+2z_nw^2+z_n^2w^3
 \right)
 }
 {
 (1-z_n)(1-z_nw)(1-z_nw^2)(1-z_nw^3)
 }.
\]
The numerator polynomial satisfies \(\left| 1+2z_nw+2z_nw^2+z_n^2w^3 \right| \le 1+2+2+1 = 6.\) Therefore \(M_3 \le \tfrac{6|w-1|^3}{(2s_N)^4} = \tfrac{3|\sin\omega_2|^3}{s_N^4}.\) For the fourth difference,
\[
 (\Delta^4\phi)_n
 =
 \tfrac{1}{1-z_nw^4}
 -
 \tfrac{4}{1-z_nw^3}
 +
 \tfrac{6}{1-z_nw^2}
 -
 \tfrac{4}{1-z_nw}
 +
 \tfrac{1}{1-z_n}.
\]
Again using the common denominator \(\prod_{j=0}^{4}(1-z_nw^j),\) we get
\[
 (\Delta^4\phi)_n
 =
 \tfrac{
 z_n(w-1)^4
 (1+z_nw^2)
 \left(
 1+3z_nw+4z_nw^2+3z_nw^3+z_n^2w^4
 \right)
 }
 {
 \prod_{j=0}^{4}(1-z_nw^j)
 }.
\]
Since \(|1+z_nw^2|\le 2\) and \(\left| 1+3z_nw+4z_nw^2+3z_nw^3+z_n^2w^4 \right| \le 1+3+4+3+1 = 12,\) we obtain \(M_4 \le \tfrac{24|w-1|^4}{(2s_N)^5} = \tfrac{12|\sin\omega_2|^4}{s_N^5}.\) Next we estimate \(\|a^{(4)}\|_1\) from the recursive definition of
\(a^{(k)}\). We use the finite-difference Leibniz formula \(\Delta^r(fg)_n = \sum_{j=0}^{r} \binom{r}{j} (\Delta^j f)_n (\Delta^{r-j}g)_{n+j}.\) Applying this with \(r=m+1\), \(f=a^{(k)}\), and \(g=\phi_{\cdot+k}\), gives \(\Delta^m a^{(k+1)} = \Delta^{m+1}\bigl(a^{(k)}\phi_{\cdot+k}\bigr).\) Hence \(\|\Delta^m a^{(k+1)}\|_1 \le \sum_{j=0}^{m+1} \binom{m+1}{j} M_{m+1-j} \|\Delta^j a^{(k)}\|_1 .\) Introduce constants \(C_{k,m,\ell}\) by \(\|\Delta^m a^{(k)}\|_1 \le \sum_{\ell\ge 0} C_{k,m,\ell}\|\Delta^\ell a\|_1 .\) For \(k=0\),
\(
 C_{0,m,\ell}
 =
 \begin{cases}
 1,&\ell=m,\\
 0,&\ell\neq m.
 \end{cases}
\)
The previous inequality gives the recursion \(C_{k+1,m,\ell} = \sum_{j=0}^{m+1} \binom{m+1}{j} M_{m+1-j}C_{k,j,\ell}.\) Iterating this recursion four times and then taking \(m=0\) gives
\[
 \|a^{(4)}\|_1
 \le
 C_{4,0,0}\|a\|_1
 +
 C_{4,0,1}\|\Delta a\|_1
 +
 C_{4,0,2}\|\Delta^2a\|_1
 +
 C_{4,0,3}\|\Delta^3a\|_1
 +
 C_{4,0,4}\|\Delta^4a\|_1,
\]
where \(C_{4,0,4}=M_0^4,\) \(C_{4,0,3}=10M_0^3M_1,\) \(C_{4,0,2}=10M_0^3M_2+25M_0^2M_1^2,\) \(C_{4,0,1} = 5M_0^3M_3 + 30M_0^2M_1M_2 + 15M_0M_1^3,\) and \(C_{4,0,0} = M_0^3M_4 + 7M_0^2M_1M_3 + 4M_0^2M_2^2 + 11M_0M_1^2M_2 + M_1^4 .\) Substituting \(M_0\le \tfrac{1}{2s_N}, ~ M_1\le \tfrac{|\sin\omega_2|}{2s_N^2},\) \[
 M_2\le \tfrac{|\sin\omega_2|^2}{s_N^3},
 ~
 M_3\le \tfrac{3|\sin\omega_2|^3}{s_N^4},
 ~
 M_4\le \tfrac{12|\sin\omega_2|^4}{s_N^5},
\]
we get \(C_{4,0,4} \le \tfrac{1}{16s_N^4},\) \(C_{4,0,3} \le \tfrac{5|\sin\omega_2|}{8s_N^5},\) \(C_{4,0,2} \le \tfrac{45|\sin\omega_2|^2}{16s_N^6},\) \(C_{4,0,1} \le \tfrac{105|\sin\omega_2|^3}{16s_N^7},\) and \(C_{4,0,0} \le \tfrac{105|\sin\omega_2|^4}{16s_N^8}.\) Therefore
\[
 \|a^{(4)}\|_1
 \le
 \tfrac{105|\sin\omega_2|^4}{16s_N^8}\|a\|_1
 +
 \tfrac{105|\sin\omega_2|^3}{16s_N^7}\|\Delta a\|_1
\]
\[
 \quad+
 \tfrac{45|\sin\omega_2|^2}{16s_N^6}\|\Delta^2a\|_1
 +
 \tfrac{5|\sin\omega_2|}{8s_N^5}\|\Delta^3a\|_1
 +
 \tfrac{1}{16s_N^4}\|\Delta^4a\|_1 .
\]
Since \(|T_N(a;\omega_1,\omega_2)| \le \|a^{(4)}\|_1,\) this proves exactly \eqref{eq:app_derivative_branch}.
\end{proof}

% ============================================================
\subsection{Pointwise residue splitting}
\label{app:lag_correlation_branch}
% ============================================================

This auxiliary pointwise bound is obtained by splitting the aperture into
residue classes modulo \(Q\) and applying the triangle inequality across the
residue sums.  It is a precursor to, but should not be confused with, the
support-uniform lag-correlation support (LCS) envelope developed below.  The
LCS envelope additionally expands the squared residue sums, retains the full
lag-pair phases, and controls their real contributions by interval cosine
majorants.

Let \(Q\in\{2,\ldots,Q_{\max}\}\).  For each residue class
\(s=0,\ldots,Q-1\), define
\[
        M_s
        :=
        \#\{m\in\mathbb Z_{\ge0}:s+Qm\le N-1\}
        =
        1+\left\lfloor\tfrac{N-1-s}{Q}\right\rfloor .
\]
For \(m=0,\ldots,M_s-1\), define the residue subsequence
\(
        a_m^{(Q,s)}
        :=
        a_{s+Qm}.
\)
Next set
\(
        \alpha_s
        :=
        \omega_1s+\omega_2s^2,
        ~
        \beta_s
        :=
        Q\omega_1+2Qs\omega_2,
        ~
        \gamma
        :=
        Q^2\omega_2 .
\)
For each residue class, define the corresponding residue sum
\(
        T_s^{(Q)}
        :=
        \sum_{m=0}^{M_s-1}
        a_m^{(Q,s)}
        e^{i(\beta_s m+\gamma m^2)} .
\)

Indeed, for \(n=s+Qm\), we have
\[
\begin{aligned}
        \omega_1n+\omega_2n^2
        &=
        \omega_1(s+Qm)+\omega_2(s+Qm)^2        \\
        &=
        \omega_1s+\omega_2s^2
        +
        \left(Q\omega_1+2Qs\omega_2\right)m
        +
        Q^2\omega_2m^2                    =
        \alpha_s+\beta_s m+\gamma m^2 .
\end{aligned}
\]
Therefore the full quadratic sum decomposes as
\(
        T_N(a;\omega_1,\omega_2)
        =
        \sum_{s=0}^{Q-1}
        e^{i\alpha_s}
        T_s^{(Q)} .
\)
The pointwise residue-split bound is defined by
\[
        B_{\RS}^{(Q_{\max})}(a;\omega_1,\omega_2;N)
        :=
        \min_{2\le Q\le Q_{\max}}
        \sum_{s=0}^{Q-1}
        |T_s^{(Q)}|.
        \tag{RS}
        \label{eq:app_BRS_def}
\]

The residue-class decomposition is exact before absolute values are taken.
For a fixed modulus \(Q\), the triangle inequality bounds the complete sum
by the sum of the residue-class magnitudes.  Minimizing this estimate over
the admissible values of \(Q\) gives the following pointwise bound.

\begin{lem}
\label{lem:app_lc_branch}
Let \(N\ge2\), let \(Q_{\max}\in\{2,\ldots,N\}\), and let
\(a\in\mathbb C^N\). Then, for every
\((\omega_1,\omega_2)\in\mathbb R^2\),
\[
        |T_N(a;\omega_1,\omega_2)|
        \le
        B_{\RS}^{(Q_{\max})}(a;\omega_1,\omega_2;N).
\]
\end{lem}

\begin{proof}
Fix \(Q\in\{2,\ldots,Q_{\max}\}\). Split the aperture indices
\(n=0,\ldots,N-1\) into residue classes modulo \(Q\). Thus every admissible
index can be written uniquely as \(n=s+Qm, ~ s\in\{0,\ldots,Q-1\}, ~ m\in\{0,\ldots,M_s-1\}.\) Hence
\[
\begin{aligned}
 T_N(a;\omega_1,\omega_2)=
 \sum_{n=0}^{N-1}
 a_n e^{i(\omega_1n+\omega_2n^2)}=
 \sum_{s=0}^{Q-1}
 \sum_{m=0}^{M_s-1}
 a_{s+Qm}
 e^{i(\omega_1(s+Qm)+\omega_2(s+Qm)^2)} .
\end{aligned}
\]
For fixed \(s\), the phase identity above gives \(\omega_1(s+Qm)+\omega_2(s+Qm)^2 = \alpha_s+\beta_s m+\gamma m^2 .\) Therefore
\[
\begin{aligned}
 T_N(a;\omega_1,\omega_2)=
 \sum_{s=0}^{Q-1}
 e^{i\alpha_s}
 \sum_{m=0}^{M_s-1}
 a_{s+Qm}e^{i(\beta_s m+\gamma m^2)}=
 \sum_{s=0}^{Q-1}
 e^{i\alpha_s}T_s^{(Q)} .
\end{aligned}
\]
Since \(|e^{i\alpha_s}|=1\), the triangle inequality gives \(|T_N(a;\omega_1,\omega_2)| \le \sum_{s=0}^{Q-1} |T_s^{(Q)}|.\) This estimate holds for every \(Q\in\{2,\ldots,Q_{\max}\}\). Taking the
minimum over \(2\le Q\le Q_{\max}\) gives
\[
 |T_N(a;\omega_1,\omega_2)|
 \le
 \min_{2\le Q\le Q_{\max}}
 \sum_{s=0}^{Q-1}|T_s^{(Q)}|
 =
 B_{\RS}^{(Q_{\max})}(a;\omega_1,\omega_2;N).
\]
The claim follows.
\end{proof}

\begin{rem}[Why \(Q=1\) is excluded]
If \(Q=1\) were allowed, there would be only one residue class and the
residue decomposition would reduce to
\(
        T_N(a;\omega_1,\omega_2)=T_0^{(1)}.
\)
Hence the corresponding triangle-inequality estimate would be exactly
\(
        |T_N(a;\omega_1,\omega_2)|,
\)
and would provide no independent cancellation estimate.  We therefore
restrict the residue decomposition to \(Q\ge2\).
\end{rem}
% ============================================================
\subsection{Residue-linear branch}
\label{app:residue_linear_branch}
% ============================================================

For \(q\in\{1,\ldots,Q_{\max}\}\) and
\(A\in\{0,\ldots,2q-1\},\)
define
\(
        \varepsilon_{q,A}
        :=
        \dist_{2\pi}^{\rm sgn}
        \left(
        \omega_2-\tfrac{\pi A}{q}
        \right),
\)
and
\(
        \nu_{q,A}
        :=
        \dist_{2\pi}^{\rm sgn}
        \left(q^2\varepsilon_{q,A}\right).
\)
For \(s=0,\ldots,q-1\), define
\(\Omega_{q,A,s}:= q\omega_1+\pi Aq+2q\varepsilon_{q,A}s.\)

For a finite sequence \(c=(c_0,\ldots,c_{M-1})\) of length \(M\ge1\), define
\[
        b_m:=c_m e^{i\nu m^2},
        ~ m=0,\ldots,M-1,
\]
and
\[
        V_\nu(c)
        :=
        |b_0|+|b_{M-1}|
        +
        \sum_{m=0}^{M-2}|b_{m+1}-b_m|.
\]
For \(M=1\), the empty sum is interpreted as zero.
Define
\begin{equation}
\label{eq:app_Breslin_def}
\begin{aligned}
        B_{\rm res,lin}^{(Q_{\max})}
        (a;\omega_1,\omega_2;N)
        :=
        \min_{1\le q\le Q_{\max}}
        \min_{0\le A\le 2q-1}
        \sum_{s=0}^{q-1}
        \min\left\{
        \|a^{(q,s)}\|_1,\,
        \tfrac{V_{\nu_{q,A}}(a^{(q,s)})}
        {2|\sin(\Omega_{q,A,s}/2)|}
        \right\}.
\end{aligned}
\end{equation}
The quotient is interpreted as \(+\infty\) if
\(\sin(\Omega_{q,A,s}/2)=0.\)

When the quadratic coefficient is close to a rational phase
\(\pi A/q\), splitting the aperture modulo \(q\) removes the rational
quadratic component from each residue class.  The remaining phase consists
of a linear oscillation and a residual quadratic chirp.  Applying the
first-order Abel estimate to this representation gives the residue-linear
bound.

\begin{lem}
\label{lem:app_reslin}
For every \(a\in\mathbb C^N\),
\[
        |T_N(a;\omega_1,\omega_2)|
        \le
        B_{\rm res,lin}^{(Q_{\max})}
        (a;\omega_1,\omega_2;N).
\]
\end{lem}

\begin{proof}
Fix \(q\in\{1,\ldots,Q_{\max}\}\) and
\(A\in\{0,\ldots,2q-1\}\). Choose
\(\varepsilon_{q,A}\in\mathbb R\) such that \(\omega_2 \equiv \tfrac{\pi A}{q} + \varepsilon_{q,A} \pmod{2\pi}.\) We split the aperture into residue classes modulo \(q\). Thus every
\(n=0,\ldots,N-1\) can be written uniquely as \(n=qm+s, ~ s\in\{0,\ldots,q-1\}, ~ m=0,\ldots,M_s-1,\) where \(M_s := \#\{m\in\mathbb Z_{\ge0}:s+qm\le N-1\}.\) For each residue class, define \(a_m^{(q,s)} := a_{s+qm}, ~ m=0,\ldots,M_s-1.\) We first expand the phase. For \(n=qm+s\),
\[
\begin{aligned}
 &\omega_1(qm+s)+\omega_2(qm+s)^2
 =
 \omega_1s+\omega_2s^2 +
 m(q\omega_1+2qs\omega_2)
 +
 m^2(q^2\omega_2).
\end{aligned}
\]
Substituting \(\omega_2 \equiv \tfrac{\pi A}{q} + \varepsilon_{q,A} \pmod{2\pi}\) into the \(m\)-dependent part gives \(m(q\omega_1+2\pi As+2q\varepsilon_{q,A}s) + m^2(\pi Aq+q^2\varepsilon_{q,A}).\) The term \(2\pi Asm\) is an integer multiple of \(2\pi\), hence it does not
change the exponential. Moreover, \(e^{i\pi Aqm^2} = e^{i\pi Aqm},\) because \(m^2-m\) is even and therefore \(\pi Aq(m^2-m)\in 2\pi\mathbb Z.\) Thus the rational quadratic part \(\pi Aqm^2\) may be moved into the linear
phase.

Define
\(
    \Omega_{q,A,s}
    :=
    q\omega_1+\pi Aq+2q\varepsilon_{q,A}s,
\)
and recall that
\(
    \nu_{q,A}
    :=
    \dist_{2\pi}^{\rm sgn}
    \left(q^2\varepsilon_{q,A}\right).
\)
Since
\(
    q^2\varepsilon_{q,A}-\nu_{q,A}\in2\pi\mathbb Z,
\)
we have, for every integer \(m\),
\(
    e^{iq^2\varepsilon_{q,A}m^2}
    =
    e^{i\nu_{q,A}m^2}.
\)
Hence, up to a phase factor independent of \(m\), the \(s\)-th residue-class
phase is
\(
    \Omega_{q,A,s}m+\nu_{q,A}m^2.
\) Consequently, since \(m\)-independent phase factors have modulus one,
\[
 |T_N(a;\omega_1,\omega_2)|
 \le
 \sum_{s=0}^{q-1}
 \left|
 \sum_{m=0}^{M_s-1}
 a_m^{(q,s)}
 e^{i(\Omega_{q,A,s}m+\nu_{q,A}m^2)}
 \right|.
\]

Fix a residue class \(s\). Set \(b_m := a_m^{(q,s)} e^{i\nu_{q,A}m^2}, ~ m=0,\ldots,M_s-1.\) Then \(|b_m|=|a_m^{(q,s)}|\) and the residue sum becomes \(\sum_{m=0}^{M_s-1} b_m e^{i\Omega_{q,A,s}m}.\) The first-order Abel estimate gives
\[
 \left|
 \sum_{m=0}^{M_s-1} b_m e^{i\Omega_{q,A,s}m}
 \right|
 \le
 \tfrac{
 |b_0|+|b_{M_s-1}|+
 \sum_{m=0}^{M_s-2}|b_{m+1}-b_m|
 }
 {
 2\left|\sin\left(\Omega_{q,A,s}/2\right)\right|
 },
\]
whenever \(\sin\left(\Omega_{q,A,s}/2\right)\neq 0.\) If this denominator vanishes, the Abel term is interpreted as \(+\infty\), so
only the trivial estimate below is used for that residue class.

The trivial estimate is always valid:
\[
 \left|
 \sum_{m=0}^{M_s-1} b_m e^{i\Omega_{q,A,s}m}
 \right|
 \le
 \sum_{m=0}^{M_s-1}|b_m|
 =
 \|a^{(q,s)}\|_1 .
\]
Therefore, for each residue class,
\[
\begin{aligned}
 &\left|
 \sum_{m=0}^{M_s-1}
 a_m^{(q,s)}
 e^{i(\Omega_{q,A,s}m+\nu_{q,A}m^2)}
 \right|
 \le
 \min\Bigg\{
 \|a^{(q,s)}\|_1,~
 \tfrac{
 |b_0|+|b_{M_s-1}|+
 \sum_{m=0}^{M_s-2}|b_{m+1}-b_m|
 }
 {
 2\left|\sin\left(\Omega_{q,A,s}/2\right)\right|
 }
 \Bigg\}.
\end{aligned}
\]
Here the second term is understood as \(+\infty\) if its denominator vanishes.
Summing this estimate over \(s=0,\ldots,q-1\) gives the residue-linear bound
for the fixed pair \((q,A)\). Since the argument holds for every
\(q\in\{1,\ldots,Q_{\max}\}\) and every \(A\in\{0,\ldots,2q-1\}\), taking the
minimum over all such pairs gives \(|T_N(a;\omega_1,\omega_2)| \le B_{\rm res,lin}^{(Q_{\max})} (a;\omega_1,\omega_2;N),\) which is exactly the claimed estimate.
\end{proof}

\begin{proof}[Proof of \Cref{thm:app_scalar_qpac}]
The trivial estimate gives
\(|T_N(a;\omega_1,\omega_2)|\le\|a\|_1.\)
The derivative branch follows from \Cref{lem:app_derivative_branch}.  The
pointwise residue-split bound follows from \Cref{lem:app_lc_branch}, while its
support-uniform LCS refinement is developed in the cellwise envelope
construction below.  The residue-linear branch follows from
\Cref{lem:app_reslin}. Since every term in
the minimum defining \(B_{\rm best}^{(4,Q_{\max})}\) is an upper
bound for \(|T_N|\), their minimum is also an upper bound.
\end{proof}

% ============================================================
\subsection{Gauged Fresnel atoms and phase increments}
\label{app:gauged_fresnel_atoms}
% ============================================================

Let
\(k_\lambda:=\tfrac{2\pi}{\lambda}, ~ \kappa_0:=k_\lambda d.\)
For a physical point \(a=(r_a,\theta_a)\), define
\(
        \alpha_a:=\tfrac{d}{r_a},~
        \tau_a:=\tfrac{d\cos\theta_a}{r_a}=\alpha_a\cos\theta_a,
\)
\(
        \mu_a:=k_\lambda d\cos\theta_a,~
        \eta_a:=\tfrac{k_\lambda d^2}{2r_a}\sin^2\theta_a .
\)
The Fresnel atom is written as
\(
        s_a[n]
        =
        e^{i\mu_a n-i\eta_a n^2},
        ~ n=0,\ldots,N_r-1.
\)
Throughout this appendix, \(u\) denotes the evaluation point and \(v\) denotes
the support point.  The ordered phase increments are \(  \omega_{1,uv}:=\mu_v-\mu_u,~
        \omega_{2,uv}:=\eta_u-\eta_v .\)
% \begin{equation}
% \label{eq:app_ordered_phase_increments}
% \end{equation}
Then
\(\overline{s_u[n]}s_v[n] = e^{i(\omega_{1,uv}n+\omega_{2,uv}n^2)}.\)
Let \(\rho_n\ge0\) be the aperture taper, and set
\(
        W_0:=\sum_{n=0}^{N_r-1}\rho_n,~
        b_n:=\tfrac{\rho_n}{W_0}.
\)
Let
\(\bar n:=\sum_n b_n n, ~ \overline{n^2}:=\sum_n b_n n^2,\)
and define centered aperture coordinates
\(x_n:=n-\bar n, ~ y_n:=n^2-\overline{n^2}.\)
For \(\tau\) in the admissible tangent interval, define
\(q(\tau):=\sum_n b_n(x_n+\tau y_n)^2.\)
We assume
\(q(\tau)>0\)
on the tangent interval under consideration.

The normalized tangent profile is
\(h_\tau[n]:= \tfrac{x_n+\tau y_n}{\sqrt{q(\tau)}}.\)
It satisfies
\(\sum_n b_n h_\tau[n]^2=1.\)
The gauged atom has the form
\(\psi_a[n] = W_0^{-1/2}e^{-i\chi_a}s_a[n],\)
where \(\chi_a\) is independent of \(n\).  The corresponding normalized
angular tangent has the form
\[
        h_a[n]
        =
        \text{\rm unit phase depending only on }a
        \cdot
        W_0^{-1/2} h_{\tau_a}[n]s_a[n].
\]
The exact constant phase is irrelevant for all absolute-value estimates.

% ============================================================
\subsection{Pairwise reductions for \texorpdfstring{\(K,H,dK,dH\)}{K, H, dK, dH}}
\label{app:pairwise_basic_channel_reductions}
% ============================================================

We first record the scalar quadratic-sum reductions for the four normalized
gauged Hermite channels.  Write the ungauged Fresnel phase at a physical point
\(a=(r_a,\theta_a)\) as
\[
        s_a[n]
        =
        e^{i\mu_a n-i\eta_a n^2},
        ~
        \mu_a:=k_\lambda d\cos\theta_a,
        ~
        \eta_a:=\tfrac{k_\lambda d^2}{2r_a}\sin^2\theta_a .
\]
For an ordered pair \((u,v)\), define the ordered phase increments by
\[
        \overline{s_u[n]}s_v[n]
        =
        e^{i(\omega_{1,uv}n+\omega_{2,uv}n^2)}.
\]
Equivalently, \( \omega_{1,uv}
        =
        \mu_v-\mu_u,
        ~
        \omega_{2,uv}
        =
        \eta_u-\eta_v .\) Thus the first point \(u\) is the evaluation point and the second point \(v\)
is the support/source point.  This is the same ordered convention used in the
main theorem.

Let
\(
        b_n:=\tfrac{\rho_n}{W_0},
        ~
        W_0:=\sum_{n=0}^{N_r-1}\rho_n .
\)
The four normalized gauged pairwise channels are
\(
        K_{uv}:=\langle\psi_u,\psi_v\rangle_\rho,~
        H_{uv}:=\langle\psi_u,h_v\rangle_\rho,
\)
and
\(
        dK_{uv}:=\langle h_u,\psi_v\rangle_\rho,~
        dH_{uv}:=\langle h_u,h_v\rangle_\rho,
\)
where
\(
        \langle f,g\rangle_\rho
        :=
        \sum_{n=0}^{N_r-1}\rho_n\overline{f[n]}g[n].
\)
Define the coefficient families
\(
        a^K[n]
        :=
        b_n,
\)
\(
        a^H_\tau[n]
        :=
        b_nh_\tau[n],
        ~
        a^{dK}_\tau[n]
        :=
        b_nh_\tau[n],
\)
and
\(
        a^{dH}_{\tau,\tau'}[n]
        :=
        b_nh_\tau[n]h_{\tau'}[n].
\)
Here \(h_\tau\) is the real normalized tangent profile, and \(\tau_u,\tau_v\)
denote the tangent parameters associated with \(u\) and \(v\), respectively.

Throughout this subsection, ``unit phase'' denotes a complex scalar of modulus
one that may depend on the ordered pair \((u,v)\), but is independent of the
aperture index \(n\).

The coefficient families above were chosen so that all four normalized
Hermite interactions share the same ordered quadratic phase and differ only
through their amplitude sequences.  Consequently, each channel is exactly a
weighted scalar quadratic sum, up to a unit-modulus factor independent of
the aperture index.

\begin{prop}
\label{prop:app_basic_pairwise_reductions}
For every ordered pair \((u,v)\),
\[
        K_{uv}
        =
        \text{\rm unit phase}\cdot
        T_{N_r}(a^K;\omega_{1,uv},\omega_{2,uv}),
\]
\[
        H_{uv}
        =
        \text{\rm unit phase}\cdot
        T_{N_r}(a^H_{\tau_v};\omega_{1,uv},\omega_{2,uv}),
\]
\[
        dK_{uv}
        =
        \text{\rm unit phase}\cdot
        T_{N_r}(a^{dK}_{\tau_u};\omega_{1,uv},\omega_{2,uv}),
\]
and
\[
        dH_{uv}
        =
        \text{\rm unit phase}\cdot
        T_{N_r}
        \left(
        a^{dH}_{\tau_u,\tau_v};
        \omega_{1,uv},\omega_{2,uv}
        \right).
\]
\end{prop}

\begin{proof}
The gauged atom differs from \(W_0^{-1/2}s_a\) only by an
\(n\)-independent unit-modulus factor. Therefore
\[
\begin{aligned}
 K_{uv}
 &=
 \sum_{n=0}^{N_r-1}
 \rho_n\overline{\psi_u[n]}\psi_v[n]=
 \text{unit phase}\cdot
 \sum_{n=0}^{N_r-1}
 b_n\overline{s_u[n]}s_v[n] =\\
 &\text{unit phase}\cdot
 \sum_{n=0}^{N_r-1}
 b_n
 e^{i(\omega_{1,uv}n+\omega_{2,uv}n^2)} .
\end{aligned}
\]
This is the claimed scalar reduction for \(K\).
For \(H_{uv}=\langle\psi_u,h_v\rangle_\rho\), the tangent factor comes from the
support point \(v\). Up to an \(n\)-independent unit phase, the support-side
tangent contributes the real profile \(h_{\tau_v}[n]\). Hence
\[
\begin{aligned}
 H_{uv}=
 \sum_{n=0}^{N_r-1}
 \rho_n\overline{\psi_u[n]}h_v[n]=
 \text{unit phase}\cdot
 \sum_{n=0}^{N_r-1}
 b_n h_{\tau_v}[n]
 e^{i(\omega_{1,uv}n+\omega_{2,uv}n^2)} .
\end{aligned}
\]
This gives the coefficient family \(a^H_{\tau_v}\).

For \(dK_{uv}=\langle h_u,\psi_v\rangle_\rho\), the tangent factor comes from
the evaluation point \(u\). Thus
\[
\begin{aligned}
 dK_{uv}
 &=
 \text{unit phase}\cdot
 \sum_{n=0}^{N_r-1}
 b_n h_{\tau_u}[n]
 e^{i(\omega_{1,uv}n+\omega_{2,uv}n^2)} ,
\end{aligned}
\]
which gives the coefficient family \(a^{dK}_{\tau_u}\).
Finally, \(dH_{uv}=\langle h_u,h_v\rangle_\rho\) contains both tangent factors.
Therefore
\[
\begin{aligned}
 dH_{uv}
 &=
 \text{unit phase}\cdot
 \sum_{n=0}^{N_r-1}
 b_n h_{\tau_u}[n]h_{\tau_v}[n]
 e^{i(\omega_{1,uv}n+\omega_{2,uv}n^2)} .
\end{aligned}
\]
This is the claimed reduction for \(dH\).
\end{proof}

Since the taper has fourth-order flat ends, \(b_n\) vanishes at the first and
last four aperture indices.  Each coefficient family above is \(b_n\) multiplied
by a polynomial or rationally normalized tangent factor that is finite on the
admissible physical domain.  Hence all four coefficient sequences inherit the
fourth-order flat-end condition.

% ============================================================
\subsection{First angular derivatives}
\label{app:first_angular_derivatives}
% ============================================================

Let the angular derivative act on the evaluation point \(u\), while the support
point \(v\) is fixed.  Write
\(
        K_v(u):=K_{uv},
        ~
        H_v(u):=H_{uv}.
\)
Let
\(
        \sigma_u
        :=
        \|\partial_\theta\psi_u\|_\rho .
\)
By construction,
\(
        h_u
        =
        \tfrac{\partial_\theta\psi_u}{\sigma_u}.
\)

The normalized tangent \(h_u\) is obtained by dividing the physical angular
derivative of the gauged atom by its weighted norm \(\sigma_u\).  Restoring
this normalization identifies the physical first derivatives of \(K\) and
\(H\) with the normalized evaluation-side tangent channels \(dK\) and
\(dH\).

\begin{lem}\label{lem:app_first_derivative_channels}
For every ordered pair \((u,v)\),
\[
        \partial_\theta K_v(u)
        =
        \sigma_u dK_{uv},
        ~
        \partial_\theta H_v(u)
        =
        \sigma_u dH_{uv}.
\]
Consequently, if \(\sigma_u\le\overline\sigma_\theta\), then
\[
        |\partial_\theta K_v(u)|
        \le
        \overline\sigma_\theta |dK_{uv}|,
        ~
        |\partial_\theta H_v(u)|
        \le
        \overline\sigma_\theta |dH_{uv}|.
\]
\end{lem}

\begin{proof}
The weighted inner product is conjugate-linear in its first argument. Since the
derivative acts on the evaluation point \(u\),
\(
 \partial_\theta K_v(u)
 =
 \partial_\theta\langle\psi_u,\psi_v\rangle_\rho
 =
 \langle\partial_\theta\psi_u,\psi_v\rangle_\rho .
\)
Using \(\partial_\theta\psi_u=\sigma_u h_u\), we obtain \(\partial_\theta K_v(u) = \sigma_u\langle h_u,\psi_v\rangle_\rho = \sigma_u dK_{uv}.\) The same argument applies to \(H_v(u)=\langle\psi_u,h_v\rangle_\rho\), because
\(h_v\) is fixed when differentiating with respect to the evaluation angle:
\[
 \partial_\theta H_v(u)
 =
 \langle\partial_\theta\psi_u,h_v\rangle_\rho
 =
 \sigma_u\langle h_u,h_v\rangle_\rho
 =
 \sigma_u dH_{uv}.
\]
The inequalities follow immediately from
\(\sigma_u\le\overline\sigma_\theta\).
\end{proof}

% ============================================================
\subsection{Higher angular derivative reductions}
\label{app:higher_angular_derivatives}
% ============================================================

Let
\(
        \alpha:=\tfrac{d}{r},
        ~
        c:=\cos\theta,
        ~
        s:=\sin\theta,
        ~
        \kappa_0:=k_\lambda d .
\)
Define the real aperture profile
\(
        u_{\alpha,\theta}[n]
        :=
        -\kappa_0 s\,(x_n+\alpha c\,y_n).
\)
Its first and second angular derivatives are
\(
        u_{\theta,\alpha,\theta}[n]
        :=
        -\kappa_0\left(c\,x_n+\alpha\cos(2\theta)y_n\right),
\)
and
\(
        u_{\theta\theta,\alpha,\theta}[n]
        :=
        \kappa_0 s(x_n+4\alpha c\,y_n).
\)
The gauged atom satisfies
\(
        \partial_\theta\psi
        =
        i\,u_{\alpha,\theta}\psi .
\)
Differentiating once more gives
\(
        \partial_\theta^2\psi
        =
        \left(
        i\,u_{\theta,\alpha,\theta}
        -
        u_{\alpha,\theta}^2
        \right)\psi .
\)
Differentiating a third time gives
\(
        \partial_\theta^3\psi
        =
        \left(
        i\,u_{\theta\theta,\alpha,\theta}
        -
        3u_{\alpha,\theta}u_{\theta,\alpha,\theta}
        -
        i\,u_{\alpha,\theta}^3
        \right)\psi .
\)
Because the derivative acts on the first argument of the inner product and the
inner product is conjugate-linear in that argument, the scalar-sum coefficient
is the complex conjugate of the derivative multiplier.  Since
\(u_{\alpha,\theta}\), \(u_{\theta,\alpha,\theta}\), and
\(u_{\theta\theta,\alpha,\theta}\) are real-valued profiles, define \( \widetilde \Pi_2(\alpha,\theta;n)
        :=
        -i\,u_{\theta,\alpha,\theta}[n]
        -
        u_{\alpha,\theta}[n]^2,\)
and \(  \widetilde\Pi_3(\alpha,\theta;n)
        :=
        -i\,u_{\theta\theta,\alpha,\theta}[n]
        -
        3u_{\alpha,\theta}[n]u_{\theta,\alpha,\theta}[n]
        +
        i\,u_{\alpha,\theta}[n]^3 .\)
For \(m=2,3\), define
\(
        a^{K,m}_{\alpha,\theta}[n]
        :=
        b_n\widetilde\Pi_m(\alpha,\theta;n),
\)
and
\(
        a^{H,m}_{\alpha,\theta;\tau_v}[n]
        :=
        b_nh_{\tau_v}[n]\widetilde\Pi_m(\alpha,\theta;n).
\)

The second and third angular derivatives introduce the explicit aperture
multipliers \(\widetilde\Pi_2\) and \(\widetilde\Pi_3\), but they do not
change the underlying ordered quadratic phase.  The higher derivative
channels therefore admit scalar quadratic-sum representations of the same
form as the normalized Hermite channels.

\begin{prop}
\label{prop:app_higher_derivative_reductions}
For \(m=2,3\),
\[
        \partial_\theta^m K_v(u)
        =
        \text{\rm unit phase}\cdot
        T_{N_r}
        \left(
        a^{K,m}_{\alpha_u,\theta_u};
        \omega_{1,uv},\omega_{2,uv}
        \right),
\]
and
\[
        \partial_\theta^m H_v(u)
        =
        \text{\rm unit phase}\cdot
        T_{N_r}
        \left(
        a^{H,m}_{\alpha_u,\theta_u;\tau_v};
        \omega_{1,uv},\omega_{2,uv}
        \right).
\]
Here, as above, the unit phase is independent of \(n\).
\end{prop}

\begin{proof}
Let \(\Pi_2(\alpha,\theta;n) := i\,u_{\theta,\alpha,\theta}[n] - u_{\alpha,\theta}[n]^2\) and
\[
 \Pi_3(\alpha,\theta;n)
 :=
 i\,u_{\theta\theta,\alpha,\theta}[n]
 -
 3u_{\alpha,\theta}[n]u_{\theta,\alpha,\theta}[n]
 -
 i\,u_{\alpha,\theta}[n]^3 .
\]
Then \(\partial_\theta^m\psi_u[n] = \Pi_m(\alpha_u,\theta_u;n)\psi_u[n], ~ m=2,3.\) By construction, \(\widetilde\Pi_m(\alpha_u,\theta_u;n) = \overline{\Pi_m(\alpha_u,\theta_u;n)}, ~ m=2,3.\) For \(K_v(u)=\langle\psi_u,\psi_v\rangle_\rho\), the derivative acts only on
the first argument. Therefore \(\partial_\theta^mK_v(u) = \langle\partial_\theta^m\psi_u,\psi_v\rangle_\rho .\) Using the representation above,
\[
\begin{aligned}
 \partial_\theta^mK_v(u)
 &=
 \sum_{n=0}^{N_r-1}
 \rho_n
 \overline{\Pi_m(\alpha_u,\theta_u;n)\psi_u[n]}
 \psi_v[n] \\
 &=
 \text{unit phase}\cdot
 \sum_{n=0}^{N_r-1}
 b_n\widetilde\Pi_m(\alpha_u,\theta_u;n)
 e^{i(\omega_{1,uv}n+\omega_{2,uv}n^2)} .
\end{aligned}
\]
This is the claimed reduction for \(K\).
For \(H_v(u)=\langle\psi_u,h_v\rangle_\rho\), the support-side tangent is fixed
with respect to the evaluation angle and contributes the real factor
\(h_{\tau_v}[n]\), up to an \(n\)-independent unit phase. Hence
\[
\begin{aligned}
 \partial_\theta^mH_v(u)
 &=
 \langle\partial_\theta^m\psi_u,h_v\rangle_\rho \\
 &=
 \text{unit phase}\cdot
 \sum_{n=0}^{N_r-1}
 b_nh_{\tau_v}[n]\widetilde \Pi_m(\alpha_u,\theta_u;n)
 e^{i(\omega_{1,uv}n+\omega_{2,uv}n^2)} .
\end{aligned}
\]
This is the claimed reduction for \(H\).
\end{proof}
% ============================================================
\subsection{Dictionary form of
\texorpdfstring{\(\widetilde{\Pi}_2\) and \(\widetilde{\Pi}_3\)}
{Pi2 and Pi3}}\label{app:Q_dictionary}
% ============================================================

The higher derivative profiles can be written as finite dictionaries of fixed
aperture monomials with physical scalar coefficients.  This form is useful for
certified support-uniform envelopes.

For \(\widetilde \Pi_2\), using \(c=\cos\theta\) and \(s^2=1-c^2\),
\[
\begin{aligned}
        &\widetilde\Pi_2(\alpha,\theta;n)
        =
        i\kappa_0 c\,x_n
        +
        i\kappa_0\alpha(2c^2-1)y_n~
        -
        \kappa_0^2(1-c^2)x_n^2
        -
        2\kappa_0^2\alpha c(1-c^2)x_ny_n~
        \\&-
        \kappa_0^2\alpha^2c^2(1-c^2)y_n^2 .
\end{aligned}
\]
Thus
\(
        \widetilde \Pi_2(\alpha,\theta;n)
        =
        \sum_{\ell=1}^{5}
        \Gamma_\ell^{(2)}(\alpha,c)A_\ell^{(2)}[n],
\)
where
\[
        A_1^{(2)}=x,\quad
        A_2^{(2)}=y,\quad
        A_3^{(2)}=x^2,\quad
        A_4^{(2)}=xy,\quad
        A_5^{(2)}=y^2,
\]
and
\(\Gamma_1^{(2)}=i\kappa_0 c, ~ \Gamma_2^{(2)}=i\kappa_0\alpha(2c^2-1),\)
\(\Gamma_3^{(2)}=-\kappa_0^2(1-c^2), ~ \Gamma_4^{(2)}=-2\kappa_0^2\alpha c(1-c^2),\)
\(\Gamma_5^{(2)}=-\kappa_0^2\alpha^2c^2(1-c^2).\)
For \(\widetilde\Pi_3\), write \(s=\sin\theta\).  Then
\(
        \widetilde\Pi_3(\alpha,\theta;n)
        =
        \sum_{\ell=1}^{9}
        \Gamma_\ell^{(3)}(\alpha,c,s)A_\ell^{(3)}[n],
\)
where
\[
        A_1^{(3)}=x,\quad
        A_2^{(3)}=y,\quad
        A_3^{(3)}=x^2,\quad
        A_4^{(3)}=xy,\quad
        A_5^{(3)}=y^2,
\]
\[
        A_6^{(3)}=x^3,\quad
        A_7^{(3)}=x^2y,\quad
        A_8^{(3)}=xy^2,\quad
        A_9^{(3)}=y^3,
\]
and
\(
        \Gamma_1^{(3)}
        =
        -i\kappa_0 s,
        ~
        \Gamma_2^{(3)}
        =
        -i4\kappa_0\alpha c s,
\)
\(
        \Gamma_3^{(3)}
        =
        -3\kappa_0^2 c s,
        ~
        \Gamma_4^{(3)}
        =
        -3\kappa_0^2\alpha(3c^2-1)s,
\)
\(\Gamma_5^{(3)} = -3\kappa_0^2\alpha^2c(2c^2-1)s,\)
\(
        \Gamma_6^{(3)}
        =
        -i\kappa_0^3(1-c^2)s,
        ~
        \Gamma_7^{(3)}
        =
        -i3\kappa_0^3\alpha c(1-c^2)s,
\)
\(
        \Gamma_8^{(3)}
        =
        -i3\kappa_0^3\alpha^2c^2(1-c^2)s,
        ~
        \Gamma_9^{(3)}
        =
        -i\kappa_0^3\alpha^3c^3(1-c^2)s.
\)

On a compact physical cell, suppose certified scalar envelopes satisfy
\(|\Gamma_\ell^{(m)}|\le \gamma_\ell^{(m),{\rm cell}}, ~ m=2,3.\)
Then, for every valid scalar branch envelope \(\mathcal B^{\rm cell}\), the
triangle inequality gives
\begin{equation}
\label{eq:app_dictionary_K_bound}
        |T_{N_r}(a^{K,m}_{\alpha,\theta};\omega_1,\omega_2)|
        \le
        \sum_\ell
        \gamma_\ell^{(m),{\rm cell}}\,
        \mathcal B^{\rm cell}\bigl(b\odot A_\ell^{(m)}\bigr).
\end{equation}
Indeed,
\(
        a^{K,m}_{\alpha,\theta}
        =
        \sum_\ell
        \Gamma_\ell^{(m)}
        \bigl(b\odot A_\ell^{(m)}\bigr),
\)
and hence
\[
\begin{aligned}
        &|T(a^{K,m}_{\alpha,\theta})|
        \le
        \sum_\ell
        |\Gamma_\ell^{(m)}|
        |T(b\odot A_\ell^{(m)})|\le
        \sum_\ell
        \gamma_\ell^{(m),{\rm cell}}
        \mathcal B^{\rm cell}(b\odot A_\ell^{(m)}).
\end{aligned}
\]
The \(H\)-type higher derivative channels are treated similarly, with the
additional normalized support tangent factor:
\begin{equation}
\label{eq:app_dictionary_H_bound}
        |T_{N_r}(a^{H,m}_{\alpha,\theta;\tau_s};
        \omega_1,\omega_2)|
        \le
        \sum_\ell
        \gamma_\ell^{(m),{\rm cell}}\,
        \sup_{\tau_s\in I_{\tau_s}}
        \mathcal B^{\rm cell}
        \bigl(b\odot h_{\tau_s}\odot A_\ell^{(m)}\bigr).
\end{equation}
The right-hand side is a certified support-uniform bound whenever the
supremum over \(\tau_s\) is enclosed rigorously.  This dictionary construction
does not claim to be the exact supremum over the nonlinear higher derivative
family.  It is a certified implementable upper bound obtained from finite
atoms and scalar QPAC branch envelopes.

% ============================================================
\subsection{Physical cells and induced phase boxes}
\label{app:physical_cells}
% ============================================================

Let
\(\mathcal C_{\rm phys} = R_s\times\Theta_s\times R_e\times\Theta_e\)
be a compact support/evaluation cell.  Thus
\(
        r_s\in R_s,~ \theta_s\in\Theta_s,
        ~
        r_e\in R_e,~ \theta_e\in\Theta_e.
\)
The support point is \(v=(r_s,\theta_s)\), and the evaluation point is
\(u=(r_e,\theta_e)\).
Define interval enclosures
\(\mu_s\in I_{\mu_s},~ \mu_e\in I_{\mu_e},\)
\(\eta_s\in I_{\eta_s},~ \eta_e\in I_{\eta_e},\)
where
\(
        \mu=k_\lambda d\cos\theta,
        ~
        \eta=\tfrac{k_\lambda d^2}{2r}\sin^2\theta .
\)
Then the induced phase boxes are
\(
        \omega_1=\mu_s-\mu_e
        \in
        I_{\omega_1}:=I_{\mu_s}-I_{\mu_e},
\)
and
\(
        \omega_2=\eta_e-\eta_s
        \in
        I_{\omega_2}:=I_{\eta_e}-I_{\eta_s}.
\)
The tangent intervals are
\(
        \tau_s=\tfrac{d\cos\theta_s}{r_s}\in I_{\tau_s},
        ~
        \tau_e=\tfrac{d\cos\theta_e}{r_e}\in I_{\tau_e}.
\)
Because \(\theta\in(0,\pi)\), \(\cos\theta\) may be negative.  Therefore
\(I_{\tau_s}\) and \(I_{\tau_e}\) are computed as endpoint hulls of
\(d\cos\theta/r\) over their corresponding cells, not by assuming
\(\tau\ge0\).

The evaluation inverse range parameter satisfies
\(
        \alpha_e=\tfrac{d}{r_e}
        \in
        I_{\alpha_e}
        :=
        \left[
        \tfrac{d}{\max R_e},
        \tfrac{d}{\min R_e}
        \right].
\)

% ============================================================
\subsection{Cellwise branch admissibility}
\label{app:cellwise_branch_admissibility}
% ============================================================

A scalar branch is admissible on the physical cell \(\mathcal C_{\rm phys}\)
if all denominator or phase-interval conditions needed by that branch hold
uniformly for every phase pair
\(
        (\omega_1,\omega_2)\in I_{\omega_1}\times I_{\omega_2}.
\)
We write
\(
        \Omega_{\rm cell}
        :=
        I_{\omega_1}\times I_{\omega_2}.
\)

\paragraph{Derivative branch:}
Fix \(d_0>0\) and \(s_2^{\rm max}\in[0,1]\).  The derivative branch is
admissible on the cell if
\[
        \underline d_{N_r}^+\ge d_0
        ~\text{and}~
        s_2^{\rm cell}\le s_2^{\rm max},
\]
where \(\underline d_{N_r}^+\) is a certified lower bound satisfying
\[
        \underline d_{N_r}^+
        \le
        \inf_{(\omega_1,\omega_2)\in\Omega_{\rm cell}}
        d_{N_r}^+(\omega_1,\omega_2),
\]
and \(s_2^{\rm cell}\) is a certified upper bound satisfying
\[
        s_2^{\rm cell}
        \ge
        \sup_{\omega_2\in I_{\omega_2}}|\sin\omega_2|.
\]

The lower bound \(\underline d_{N_r}^+\) is computed as follows.  For each
\(n=0,\ldots,N_r-1\), enclose the affine phase increment
\[
        \omega_1+(2n+1)\omega_2
\]
over the rectangle \(I_{\omega_1}\times I_{\omega_2}\).  The distance of this
interval to \(2\pi\mathbb Z\) gives a certified lower bound for the
\(n\)-th derivative separation.  Taking the minimum over
\(n=0,\ldots,N_r-1\) gives \(\underline d_{N_r}^+\).

When the branch is admissible, set
\(
        s_0:=\sin\left(\tfrac{d_0}{2}\right).
\)
Then the derivative-branch coefficients are bounded uniformly on the cell by
\[
        \beta_0(d_0,s_2^{\rm max})
        =
        \tfrac{105(s_2^{\rm max})^4}{16s_0^8},
        ~
        \beta_1(d_0,s_2^{\rm max})
        =
        \tfrac{105(s_2^{\rm max})^3}{16s_0^7},
\]
\[
        \beta_2(d_0,s_2^{\rm max})
        =
        \tfrac{45(s_2^{\rm max})^2}{16s_0^6},
        ~
        \beta_3(d_0,s_2^{\rm max})
        =
        \tfrac{5s_2^{\rm max}}{8s_0^5},
\]
and
\[
        \beta_4(d_0,s_2^{\rm max})
        =
        \tfrac{1}{16s_0^4}.
\]

\paragraph{Lag-correlation support branch:}
Fix \(Q\in\{2,\ldots,Q_{\max}\}\).  The LCS branch has no denominator
condition.  Its admissibility on a cell is determined by certified interval
bounds for the real lag contributions.
For a channel \(X\), a point \(z\in\mathcal C_{\rm phys}\), and a residue class
\(s=0,\ldots,Q-1\), let \(a_X^{(Q,s)}(z)\) denote the corresponding residue
subsequence.  For each lag
\(
        h=1,\ldots,M_s-1
\)
and lag position
\(
        m=0,\ldots,M_s-1-h,
\)
define
\[
        d_{s,h,m}(z)
        :=
        a_X^{(Q,s)}(z)_{m+h}
        \overline{a_X^{(Q,s)}(z)_m}.
\]
Terms with \(d_{s,h,m}(z)=0\) contribute zero to the lag estimate and are
omitted.  For \(d_{s,h,m}(z)\ne0\), define the full residue lag-phase interval
\[
        \Jlc_{s,h,m}(z,\Omega_{\rm cell})
        :=
        \arg(d_{s,h,m}(z))
        +
        hQ\,I_{\omega_1}
        +
        \bigl(2Qsh+Q^2(h^2+2hm)\bigr)I_{\omega_2}.
\]
The argument may be chosen as any real representative modulo \(2\pi\);
the resulting cosine majorant is unchanged because \(\Cmax\) is
\(2\pi\)-periodic in the sense used below.
Here the interval operations are real interval operations before any
modulo-\(2\pi\) reduction.
For a compact real interval \(J=[L,U]\), define
\(
        \Cmax(J):=\sup_{\xi\in J}\cos\xi.
\)
Equivalently,
\[
        \Cmax([L,U])
        =
        \begin{cases}
        1,
        &\text{if there exists } k\in\mathbb Z
        \text{ such that } 2\pi k\in[L,U],\\[1mm]
        \max\{\cos L,\cos U\},
        &\text{otherwise}.
        \end{cases}
\]
This definition handles intervals that cross a multiple of \(2\pi\).  In particular,
\(
        -1\le \Cmax(J)\le 1.
\)
For the residue lag phase, define
\[
        \Phi_{s,h,m}(\omega_1,\omega_2)
        :=
        hQ\omega_1+
        \bigl(2Qsh+Q^2(h^2+2hm)\bigr)\omega_2 .
\]
The interval \(\Jlc_{s,h,m}(z,\Omega_{\rm cell})\) contains the full phase
\[
        \arg(d_{s,h,m}(z))
        +
        \Phi_{s,h,m}(\omega_1,\omega_2)
\]
for every \((\omega_1,\omega_2)\in\Omega_{\rm cell}\).  Therefore, for
\(d_{s,h,m}(z)\ne0\),
\[
        \operatorname{Re}
        \left(
        d_{s,h,m}(z)
        e^{i\Phi_{s,h,m}(\omega_1(z),\omega_2(z))}
        \right)
        \le
        |d_{s,h,m}(z)|
        \Cmax(\Jlc_{s,h,m}(z,\Omega_{\rm cell})).
\]
Terms with \(d_{s,h,m}(z)=0\) contribute zero.

\paragraph{Residue-linear branch:}
Fix \(q\in\{1,\ldots,Q_{\max}\}\) and
\(A\in\{0,\ldots,2q-1\}\).  Let
\(
        \varepsilon
        :=
        \omega_2-\tfrac{\pi A}{q}.
\)
Choose a certified upper bound
\(
        \bar\nu_{q,A}
        \ge
        \sup_{\omega_2\in I_{\omega_2}}
        \dist_{2\pi}(q^2\varepsilon).
\)
For each residue class \(s=0,\ldots,q-1\), choose a certified lower bound
\[
        \gamma_s^{(q,A)}
        \le
        \inf_{(\omega_1,\omega_2)\in\Omega_{\rm cell}}
        \left|
        \sin\left(
        \tfrac{q\omega_1+\pi Aq+2q\varepsilon s}{2}
        \right)
        \right|.
\]
The residue-linear branch is admissible if
\(
        \gamma_s^{(q,A)}>0,
        ~ s=0,\ldots,q-1.
\)
In that case, the first-order Abel denominators in all residue classes are
uniformly bounded below on the cell.

\subsection{Cellwise coefficient envelopes}
\label{app:cellwise_coefficient_envelopes}
% ============================================================

Let \(X\) denote one of the eight channels
\[
        K,H,dK,dH,d^2K,d^3K,d^2H,d^3H.
\]
On a physical cell, the channel \(X\) has a coefficient family
\[
        a_X(z),
        ~
        z=(r_s,\theta_s,r_e,\theta_e)\in\mathcal C_{\rm phys}.
\]
A cellwise branch envelope for \(X\) is any number
\(\mathcal E_X^{(b,{\rm cell})}\) such that
\[
        B_b(a_X(z);\omega_1(z),\omega_2(z);N_r)
        \le
        \mathcal E_X^{(b,{\rm cell})}
\]
for every \(z\in\mathcal C_{\rm phys}\), where \(b\) is an admissible scalar
branch.

\paragraph{Derivative branch envelope:}
For the derivative branch, suppose that for \(j=0,\ldots,4\) we have certified
bounds
\(
        \mathcal D_{X,j}^{\rm cell}
        \ge
        \sup_{z\in\mathcal C_{\rm phys}}
        \|\Delta^j a_X(z)\|_1 .
\)
Then the derivative envelope
\(
        \mathcal E_X^{({\rm der},{\rm cell})}
        :=
        \sum_{j=0}^{4}
        \beta_j(d_0,s_2^{\rm max})
        \mathcal D_{X,j}^{\rm cell}
\)
is valid on the entire cell.

\paragraph{Lag-correlation support envelope:}
Fix \(Q\in\{2,\ldots,Q_{\max}\}\).  For each residue class \(s\), let
\(a_X^{(Q,s)}(z)\) be the residue subsequence of \(a_X(z)\).  Define the
residue-class energy
\(
        E_{Q,s}(z)
        :=
        \sum_{m=0}^{M_s-1}
        |a_X^{(Q,s)}(z)_m|^2 .
\)
Suppose we have a certified energy bound
\(
        \mathcal E_{Q,s}^{\rm cell}
        \ge
        \sup_{z\in\mathcal C_{\rm phys}} E_{Q,s}(z).
\)
For each lag
\(
        h=1,\ldots,M_s-1
\)
and lag position
\(
        m=0,\ldots,M_s-1-h,
\)
define
\(
        d_{s,h,m}(z)
        :=
        a_X^{(Q,s)}(z)_{m+h}
        \overline{a_X^{(Q,s)}(z)_m}.
\)
For \(d_{s,h,m}(z)\ne0\), define
\[
        \Jlc_{s,h,m}(z,\Omega_{\rm cell})
        :=
        \arg(d_{s,h,m}(z))
        +
        hQ\,I_{\omega_1}
        +
        \bigl(2Qsh+Q^2(h^2+2hm)\bigr)I_{\omega_2}.
\]
Terms with \(d_{s,h,m}(z)=0\) are omitted because they contribute zero.

For each residue class \(s\) and lag \(h\), suppose we have a certified real
lag bound
\[
        \mathcal U_{Q,s,h}^{\LCS,{\rm cell}}
        \ge
        \sup_{z\in\mathcal C_{\rm phys}}
        \sum_{\substack{0\le m\le M_s-1-h\\ d_{s,h,m}(z)\ne0}}
        |d_{s,h,m}(z)|
        \Cmax(\Jlc_{s,h,m}(z,\Omega_{\rm cell})) .
\]

For every \(z\in\mathcal C_{\rm phys}\), the residue-class expansion gives
\[
\begin{aligned}
        |T_s^{(Q)}(z)|^2
        &=
        \sum_{m=0}^{M_s-1}
        |a_X^{(Q,s)}(z)_m|^2                                      \\
        &\quad+
        2\operatorname{Re}
        \sum_{h=1}^{M_s-1}
        \sum_{m=0}^{M_s-1-h}
        d_{s,h,m}(z)
        e^{i\Phi_{s,h,m}(\omega_1(z),\omega_2(z))}.
\end{aligned}
\]
Here
\[
        \Phi_{s,h,m}(\omega_1,\omega_2)
        =
        hQ\omega_1+
        \bigl(2Qsh+Q^2(h^2+2hm)\bigr)\omega_2 .
\]
By the definition of \(\Jlc_{s,h,m}(z,\Omega_{\rm cell})\), for every
nonzero \(d_{s,h,m}(z)\),
\[
        \operatorname{Re}
        \left(
        d_{s,h,m}(z)
        e^{i\Phi_{s,h,m}(\omega_1(z),\omega_2(z))}
        \right)
        \le
        |d_{s,h,m}(z)|
        \Cmax(\Jlc_{s,h,m}(z,\Omega_{\rm cell})).
\]
Zero coefficients contribute zero.  Therefore
\[
        |T_s^{(Q)}(z)|^2
        \le
        \mathcal E_{Q,s}^{\rm cell}
        +
        2
        \sum_{h=1}^{M_s-1}
        \mathcal U_{Q,s,h}^{\LCS,{\rm cell}} .
\]
Consequently,
\[
        |T_s^{(Q)}(z)|
        \le
        \left(
        \mathcal E_{Q,s}^{\rm cell}
        +
        2
        \sum_{h=1}^{M_s-1}
        \mathcal U_{Q,s,h}^{\LCS,{\rm cell}}
        \right)_+^{1/2}.
\]
The quantity under the square root is nonnegative whenever the certified
bounds are valid, because it upper-bounds \(|T_s^{(Q)}(z)|^2\).
Since
\(
        T_N
        =
        \sum_{s=0}^{Q-1}
        e^{i\alpha_s}T_s^{(Q)},
\)
the triangle inequality gives
\(
        |T_N|
        \le
        \sum_{s=0}^{Q-1}|T_s^{(Q)}|.
\)
Thus the cellwise LCS envelope for this fixed \(Q\) is
\[
        \mathcal E_X^{(\LCS,Q,{\rm cell})}
        :=
        \sum_{s=0}^{Q-1}
        \left(
        \mathcal E_{Q,s}^{\rm cell}
        +
        2
        \sum_{h=1}^{M_s-1}
        \mathcal U_{Q,s,h}^{\LCS,{\rm cell}}
        \right)_+^{1/2}.
\]
Taking the minimum over \(Q\) gives
\[
        \mathcal E_X^{(\LCS,{\rm cell})}
        :=
        \min_{2\le Q\le Q_{\max}}
        \mathcal E_X^{(\LCS,Q,{\rm cell})}.
\]

\paragraph{Residue-linear branch envelope:}
Fix an admissible pair \((q,A)\).  Suppose that, for each residue class
\(s=0,\ldots,q-1\), we have certified bounds
\(
        \mathcal L_{q,s}^{\rm cell}
        \ge
        \sup_{z\in\mathcal C_{\rm phys}}
        \|a_X^{(q,s)}(z)\|_1
\)
and
\(        \mathcal V_{q,A,s}^{\rm cell}
        \ge
        \sup_{z\in\mathcal C_{\rm phys}}
        V_{\nu}(a_X^{(q,s)}(z))
\)
for every admissible residual \(\nu\) satisfying
\(
        \dist_{2\pi}(\nu)\le \bar\nu_{q,A}.
\)
Then
\[
        \mathcal E_X^{({\rm res,lin},q,A,{\rm cell})}
        :=
        \sum_{s=0}^{q-1}
        \min\left\{
        \mathcal L_{q,s}^{\rm cell},\,
        \tfrac{
        \mathcal V_{q,A,s}^{\rm cell}
        }{
        2\gamma_s^{(q,A)}
        }
        \right\}
\]
is valid on the cell.  Taking the minimum over admissible \((q,A)\) gives
\[
        \mathcal E_X^{({\rm res,lin},{\rm cell})}
        :=
        \min_{(q,A)\ {\rm admissible}}
        \mathcal E_X^{({\rm res,lin},q,A,{\rm cell})}.
\]
The minimum over an empty admissible family is defined to be \(+\infty\).
\paragraph{Trivial envelope:}
The trivial branch is given by any certified bound
\(
        \mathcal E_X^{(\ell_1,{\rm cell})}
        \ge
        \sup_{z\in\mathcal C_{\rm phys}}\|a_X(z)\|_1.
\)
% ============================================================
\subsection{Cauchy-Schwarz bounds}
\label{app:cauchy_caps}
% ============================================================

The scalar QPAC estimates apply to arbitrary coefficient sequences and do
not encode the normalization of the Hermite channels.  Independently of any
oscillatory cancellation, the gauged value and tangent atoms have unit
weighted norm. The Cauchy-Schwarz inequality therefore provides a universal unit
bound for the four normalized channels.

\begin{lem}
\label{lem:app_unit_caps}
For every pair \((u,v)\),
\(|K_{uv}|,\quad |H_{uv}|,\quad |dK_{uv}|,\quad |dH_{uv}| \le 1.\)
\end{lem}

\begin{proof}
By construction, \(\|\psi_u\|_\rho=\|\psi_v\|_\rho=1, ~ \|h_u\|_\rho=\|h_v\|_\rho=1.\) The result follows immediately from the Cauchy-Schwarz inequality applied to
\[
 \langle\psi_u,\psi_v\rangle_\rho,\quad
 \langle\psi_u,h_v\rangle_\rho,\quad
 \langle h_u,\psi_v\rangle_\rho,\quad
 \langle h_u,h_v\rangle_\rho .
\]
\end{proof}

The second and third angular derivative channels are not normalized inner
products, so the preceding unit bound does not apply to them.  Their
universal Cauchy bounds are instead determined by the weighted
\(\ell^2\)-size of the derivative multipliers
\(\widetilde\Pi_m\).

\begin{lem}
\label{lem:app_higher_cauchy_caps}
For \(m=2,3\),
\(
        |\partial_\theta^m K_v(u)|
        \le
        \left(
        \sum_n b_n |\widetilde\Pi_m(\alpha_u,\theta_u;n)|^2
        \right)^{1/2},
\)
and
\(
        |\partial_\theta^m H_v(u)|
        \le
        \left(
        \sum_n b_n |\widetilde\Pi_m(\alpha_u,\theta_u;n)|^2
        \right)^{1/2}.
\)
\end{lem}

\begin{proof}
For \(K\), by \Cref{prop:app_higher_derivative_reductions},
\[
 \partial_\theta^mK_v(u)
 =
 \text{unit phase}\cdot
 \sum_n b_n \widetilde\Pi_m(\alpha_u,\theta_u;n)
 e^{i(\omega_{1,uv}n+\omega_{2,uv}n^2)} .
\]
Cauchy-Schwarz gives
\[
 |\partial_\theta^mK_v(u)|
 \le
 \left(\sum_n b_n|\widetilde\Pi_m(\alpha_u,\theta_u;n)|^2\right)^{1/2}
 \left(\sum_n b_n\right)^{1/2}.
\]
Since \(\sum_n b_n=1\), the first bound follows.

For \(H\), the coefficient contains \(h_{\tau_v}[n]\):
\[
 \partial_\theta^mH_v(u)
 =
 \text{unit phase}\cdot
 \sum_n b_n h_{\tau_v}[n]\widetilde\Pi_m(\alpha_u,\theta_u;n)
 e^{i(\omega_{1,uv}n+\omega_{2,uv}n^2)} .
\]
Apply Cauchy-Schwarz in the form
\[
\begin{aligned}
 |\partial_\theta^mH_v(u)|
 &\le
 \left(\sum_n b_n|\widetilde\Pi_m(\alpha_u,\theta_u;n)|^2\right)^{1/2}
 \left(\sum_n b_n h_{\tau_v}[n]^2\right)^{1/2}.
\end{aligned}
\]
The second factor is \(1\) by the normalization of \(h_{\tau_v}\).
\end{proof}

Thus, on a physical cell, valid derivative-size Cauchy-Schwarz caps are
\(
        \mathcal C_{K,m}^{\rm cell}
        :=
\sup_{\substack{\alpha\in I_{\alpha_e}\\ \theta\in \Theta_e}}
\left(
        \sum_n b_n |\widetilde\Pi_m(\alpha,\theta;n)|^2
        \right)^{1/2},
\)
and
\(\mathcal C_{H,m}^{\rm cell}:= \mathcal C_{K,m}^{\rm cell}.\)
These caps apply to \(d^mK\) and \(d^mH\), respectively, for \(m=2,3\).
We use the shorthand
\[
        B_X^{\ell_1,{\rm cell}}
        :=
        \mathcal E_X^{(\ell_1,{\rm cell})},
        ~
        B_X^{{\rm der},{\rm cell}}
        :=
        \mathcal E_X^{({\rm der},{\rm cell})},
\]
\[    B_X^{\LCS,{\rm cell}}
        :=
        \mathcal E_X^{(\LCS,{\rm cell})},
        ~
        B_X^{{\rm res,lin},{\rm cell}}
        :=
        \mathcal E_X^{({\rm res,lin},{\rm cell})}.\]
% ============================================================
\subsection{Cellwise support-uniform QPAC theorem}
\label{app:cellwise_support_uniform_theorem}
% ============================================================
Define the Cauchy-Schwarz cell cap by
\[
B_X^{{\rm Cauchy},{\rm cell}}
:=
\begin{cases}
1,
& X\in\{K,H,dK,dH\},\\[0.5ex]
\mathcal C_{K,2}^{\rm cell},
& X=d^2K,\\[0.5ex]
\mathcal C_{K,3}^{\rm cell},
& X=d^3K,\\[0.5ex]
\mathcal C_{H,2}^{\rm cell},
& X=d^2H,\\[0.5ex]
\mathcal C_{H,3}^{\rm cell},
& X=d^3H.
\end{cases}
\]
Here the constants \(\mathcal C_{K,m}^{\rm cell}\) and
\(\mathcal C_{H,m}^{\rm cell}\) are the derivative-size Cauchy-Schwarz caps from
\Cref{lem:app_higher_cauchy_caps}.
For a channel \(X\), define the cellwise support-uniform bound
\begin{equation}
\label{eq:app_cellwise_support_bound}
        B_X^{\rm cell}
        :=
        \min
        \left\{
        B_X^{\ell_1,{\rm cell}},
        B_X^{{\rm Cauchy},{\rm cell}},
        B_X^{{\rm der},{\rm cell}},
        B_X^{\LCS,{\rm cell}},
        B_X^{{\rm res,lin},{\rm cell}}
        \right\},
\end{equation}
where a branch term is omitted if its admissibility conditions fail on the
cell.  For the normalized channels \(K,H,dK,dH\), the Cauchy-Schwarz  term includes the
unit cap \(1\).  For \(d^2K,d^3K,d^2H,d^3H\), the Cauchy-Schwarz term is the
corresponding derivative-size cap from \Cref{lem:app_higher_cauchy_caps}.

Every finite candidate entering \(B_X^{\rm cell}\) has been constructed as
an upper bound that is valid simultaneously over the entire physical cell.
Taking the minimum over these valid candidates preserves the upper-bound
property.  The resulting theorem is the step that converts the scalar QPAC
branches into support-uniform channel envelopes.

\begin{thm}
\label{thm:app_cellwise_support_uniform_qpac}
Let
\(
        \mathcal C_{\rm phys}
        =
        R_s\times\Theta_s\times R_e\times\Theta_e
\)
be a compact physical support/evaluation cell.  For each channel
\[
        X\in
        \{K,H,dK,dH,d^2K,d^3K,d^2H,d^3H\},
\]
let \(B_X^{\rm cell}\) be defined by
\eqref{eq:app_cellwise_support_bound}.  The minimum in
\eqref{eq:app_cellwise_support_bound} is taken only over candidates that are
certified on the whole cell: the trivial \(\ell^1\) cap, the appropriate
Cauchy-Schwarz cap, and those scalar QPAC branches whose admissibility conditions hold
on the entire induced phase box of the cell.  Invalid branches are omitted,
equivalently assigned the value \(+\infty\).
Then
\[
        |X((r_e,\theta_e),(r_s,\theta_s))|
        \le
        B_X^{\rm cell}
\]
for every
\(
        (r_s,\theta_s,r_e,\theta_e)
        \in
        \mathcal C_{\rm phys}.
\)
Here the first argument of \(X\) is the evaluation point and the second
argument is the support point, consistently with the main convention
\(X(q,p)\).
\end{thm}

\begin{proof}
Fix
\(
        z=(r_s,\theta_s,r_e,\theta_e)\in\mathcal C_{\rm phys},
\)
and set
\(
        p_s=(r_s,\theta_s),
        ~
        p_e=(r_e,\theta_e).
\)
The ordered channel value is \(X(p_e,p_s)\), and \(z\) induces the phase pair
\((\omega_1(z),\omega_2(z))\) with evaluation point \(p_e\) and support point
\(p_s\).  For brevity in this proof, write
\(
        X(z):=X(p_e,p_s).
\)
By construction of the cellwise phase box, this pointwise phase pair belongs
to the phase box associated with \(\mathcal C_{\rm phys}\).
For the derivative and residue-linear candidates, the scalar
estimates apply at the induced phase pair, and the coefficient
enclosures bound their values uniformly on the cell.
For the LCS candidate, the squared-residue expansion in
\Cref{app:cellwise_coefficient_envelopes} directly gives
\[
|X(z)|\le\mathcal E_X^{(\LCS,{\rm cell})}.
\]
Thus each included cancellation candidate bounds
\(|X(z)|\) throughout the cell.
The trivial \(\ell^1\) candidate is valid by the triangle inequality: \(|X(z)| \le \|a_X(z)\|_1 \le \mathcal E_X^{(\ell^1,{\rm cell})}.\) The Cauchy-Schwarz candidate is valid by \Cref{lem:app_unit_caps} for the normalized
channels \(K,H,dK,dH,\) and by \Cref{lem:app_higher_cauchy_caps} for the higher derivative channels \(d^2K,d^3K,d^2H,d^3H.\) In particular, the higher derivative channels use derivative-size Cauchy-Schwarz caps,
not the unit cap.
Therefore every finite candidate appearing in
\eqref{eq:app_cellwise_support_bound} is an upper bound for \(|X(z)|\).
Taking the minimum over valid upper bounds preserves the upper-bound property.
Hence \(|X(z)|\le B_X^{\rm cell}.\) Since \(z\in\mathcal C_{\rm phys}\) was arbitrary, the bound holds uniformly on
the whole cell.
\end{proof}

% ============================================================
\subsection{Raw tangent comparison}
\label{app:raw_tangent_comparison}
% ============================================================

The estimates above use the gauged tangent \(h_a\).  If one uses instead the
raw ungauged angular tangent, the same scalar QPAC theorem still applies after
replacing the centered gauged profile by the raw derivative profile.  The
gauged tangent is used in the Hermite certificate because it enforces
\(
        \langle \psi_a,h_a\rangle_\rho=0.
\)
Thus the diagonal value-tangent block is exactly orthogonal.  This
orthogonality is the reason the gauged tangent is the natural object for the
QPAC certificate.

% ============================================================
\section{Near-curvature estimates}
\label{app:near-curvature}
% ============================================================

This appendix proves the local strict-decay estimate used in the QPAC recovery
theorem.  The argument is deterministic and support-uniform.  It uses three
ingredients:

\begin{itemize}
        \item coefficient bounds for the gauged Hermite system;
        \item a uniform lower bound on the tangent energy at the support;
        \item local support-uniform derivative bounds for \(K\) and \(H\) up to
        order three.
\end{itemize}

The conclusion is that the gauged Hermite certificate
\(P_{S,\widetilde {\mathbf v}}\) satisfies
\(
        |P_{S,\widetilde {\mathbf v}}(p)|<1
\)
throughout every punctured feasible near interval, provided the near-curvature
margin is positive and the QPAC near number is strictly smaller than one.

% ============================================================
\subsection{Local notation}
\label{app:near-local-notation}
% ============================================================

Let
\(
        S=\{p_\ell=(i_\ell,\theta_\ell)\}_{\ell=1}^{L}
        \subset
        \Qset
        =
        \{1,\ldots,N_d\}\times\Theta
\)
be a support set.  Let
\(
        \widetilde{\mathbf v}=(\widetilde v_1,\ldots,\widetilde v_L)
        \in\mathbb C^L,
        ~
        |\widetilde v_\ell|=1,
\)
be an arbitrary gauged interpolation phase vector.  In the exact-recovery
argument, this vector is chosen as
\(
        \widetilde v_\ell
        =
        e^{i\chi_{i_\ell}(\theta_\ell)}
        \tfrac{c_\ell}{|c_\ell|}.
\)
The gauged Hermite certificate has the form
\begin{equation}
\label{eq:near_certificate_def}
        P_{S,\widetilde {\mathbf v}}(p)
        =
        \sum_{\ell=1}^{L}
        \alpha_\ell K_\ell(p)
        +
        \sum_{\ell=1}^{L}
        \beta_\ell H_\ell(p),
\end{equation}
where
\(
        K_\ell(p)
        :=
        K_{p,p_\ell}
        =
        \langle \bs \psi_p,\bs \psi_{p_\ell}\rangle_\rho,
        ~
        H_\ell(p)
        :=
        H_{p,p_\ell}
        =
        \langle \bs \psi_p,\mx h_{p_\ell}\rangle_\rho .
\)
The first argument is always the evaluation point, and the second argument is
the support point.
The gauged Hermite interpolation conditions are
\begin{equation}
\label{eq:near_Hermite_conditions}
        P_{S,\widetilde {\mathbf v}}(p_j)=\widetilde v_j,
        ~
        \partial_\theta P_{S,\widetilde {\mathbf v}}(p_j)=0,
        ~
        j=1,\ldots,L.
\end{equation}
For each support point \(p_j=(i_j,\theta_j)\), define the local angular path
\[
        p_j(t):=(i_j,\theta_j+t),
\]
and the feasible local interval
\[
        \mathcal I_j(\varpi_{\rm loc})
        :=
        \left\{
        t\in\mathbb R:
        |t|\le \varpi_{\rm loc},
        \quad
        \theta_j+t\in\Theta
        \right\}.
\]
Define the punctured feasible local interval by
\[
        \mathcal I_j^\circ(\varpi_{\rm loc})
        :=
        \mathcal I_j(\varpi_{\rm loc})\setminus\{0\}.
\]

We study the scalar functions \( f_j(t)
        :=
        |P_{S,\widetilde {\mathbf v}}(p_j(t))|^2\)
and \( F_j(t)
        :=
        1-f_j(t)
        =
        1-|P_{S,\widetilde {\mathbf v}}(p_j(t))|^2.\)
By the Hermite interpolation conditions,
\(
        f_j(0)=1,
        ~
        f_j'(0)=0.
\)
Equivalently,
\(
        F_j(0)=0,
        ~
        F_j'(0)=0.
\)
The goal is to prove
\(
        f_j(t)<1,
        ~
        t\in\mathcal I_{j}^{\circ}(\varpi_{\rm loc}),
\)
or, equivalently,
\(
        F_j(t)>0,
        ~
        t\in\mathcal I_{j}^{\circ}(\varpi_{\rm loc}).
\)

% ============================================================
\subsection{Tangent-energy lower bound}
\label{app:tangent-energy-lower-bound}
% ============================================================

Recall that
\(
        \tau=\tfrac{d\cos\theta}{r}.
\)
Let
\(
        W_0:=\sum_{n=0}^{N_r-1}\rho_n>0,
\)
and define the weighted aperture moments
\(
        \bar n
        :=
        \tfrac1{W_0}\sum_{n=0}^{N_r-1}\rho_n n,
        ~
        \overline{n^2}
        :=
        \tfrac1{W_0}\sum_{n=0}^{N_r-1}\rho_n n^2.
\)
Set
\(
        x_n:=n-\bar n,
        ~
        y_n:=n^2-\overline{n^2}.
\)
For \(\tau\in I_\tau\), define \( D_g(\tau)
        :=
        \sum_{n=0}^{N_r-1}
        \rho_n(x_n+\tau y_n)^2 .\)
Since \(b_n=\rho_n/W_0\), this is equivalently
\(
        D_g(\tau)=W_0q(\tau).
\)
Thus the constants \(D_{\min}/W_0\) and \(D_{\max}/W_0\) used below coincide
with \(q_{\min}\) and \(q_{\max}\) in the main theorem.  Assume
\(
        D_g(\tau)>0,
        ~
        \tau\in I_\tau .
\)
The gauge is chosen so that
\(
        \partial_\theta\psi_{r,\theta}[n]
        =
        i\widetilde u_{r,\theta}(n)\psi_{r,\theta}[n],
\)
where
\(
        \widetilde u_{r,\theta}(n)
        =
        -k_\lambda d\sin\theta\,(x_n+\tau y_n).
\)
Since
\(
        |\psi_{r,\theta}[n]|^2=\tfrac1{W_0},
\)
we obtain
\begin{align*}
       & \|\partial_\theta\psi_{r,\theta}\|_\rho^2
        =
        \sum_{n=0}^{N_r-1}
        \rho_n
        |\widetilde u_{r,\theta}(n)|^2
        |\psi_{r,\theta}[n]|^2
        =
        (k_\lambda d\sin\theta)^2
        \tfrac1{W_0}
        \sum_{n=0}^{N_r-1}
        \rho_n(x_n+\tau y_n)^2
        \nonumber\\&=
        (k_\lambda d\sin\theta)^2
        \tfrac{D_g(\tau)}{W_0}.
\end{align*}
Let
\(
        D_{\min}:=\min_{\tau\in I_\tau}D_g(\tau),
        ~
        D_{\max}:=\max_{\tau\in I_\tau}D_g(\tau).
\)
Define \( \underline\sigma^2
        :=
        (k_\lambda d)^2s_{\min}^2
        \tfrac{D_{\min}}{W_0},\) \(
        ~
        \overline\sigma_\theta
        :=
        k_\lambda d
        \sqrt{\tfrac{D_{\max}}{W_0}} .\) Then, uniformly over the admissible physical domain,
\begin{equation}
\label{eq:near_sigma_uniform_bounds}
        \|\partial_\theta\psi_{r,\theta}\|_\rho^2
        \ge
        \underline\sigma^2,
        ~
        \|\partial_\theta\psi_{r,\theta}\|_\rho
        \le
        \overline\sigma_\theta .
\end{equation}

% ============================================================
\subsection{Coefficient bounds from the Hermite system}
\label{app:near-coefficient-bounds}
% ============================================================

The support-support QPAC estimates give the nonnegative matrix
\[
        \mathbf G_{\rm SS}
        =
        \begin{bmatrix}
        u_K^{\rm SS} & u_H^{\rm SS}\\
        u_{dK}^{\rm SS} & u_{dH}^{\rm SS}
        \end{bmatrix}.
\]
Assume
\[
        \eta_{\rm SS}
        :=
        (L-1)\rho_{\rm sp}(\mathbf G_{\rm SS})<1.
\]
Define
\begin{equation}
\label{eq:near_Gamma_def}
       \boldsymbol \Gamma
        =
        \begin{bmatrix}
        \Gamma_K\\
        \Gamma_H
        \end{bmatrix}
        :=
        \left(\mathbf I_2-(L-1)\mathbf G_{\rm SS}\right)^{-1}
        \begin{bmatrix}
        1\\
        0
        \end{bmatrix}.
\end{equation}
Also define
\begin{equation}
\label{eq:near_Xi_def}
        \boldsymbol\Xi
        :=
        \boldsymbol\Gamma-
        \begin{bmatrix}
        1\\
        0
        \end{bmatrix}
        =
        \begin{bmatrix}
        \Xi_K\\
        \Xi_H
        \end{bmatrix}.
\end{equation}

The Hermite-system estimates proved earlier imply, uniformly over the support
class and over all gauged phase vectors \(\widetilde{\mx v}\),
\begin{equation}
\label{eq:near_alpha_beta_bounds}
        |\alpha_\ell|\le \Gamma_K,
        ~
        |\beta_\ell|\le \Gamma_H,
        ~
        \ell=1,\ldots,L.
\end{equation}
They also imply the diagonal correction bounds
\begin{equation}
\label{eq:near_diagonal_correction_bounds}
        |\alpha_j-\widetilde v_j|\le \Xi_K,
        ~
        |\beta_j|\le \Xi_H,
        ~
        j=1,\ldots,L.
\end{equation}
These are the only coefficient estimates used in the near-curvature argument.

% ============================================================
\subsection{Local derivative bounds}
\label{app:near-derivative-bounds}
% ============================================================

All derivatives below are taken with respect to the evaluation angle.
First, assume that the self second-derivative bounds
\(
        U_{K,2}^{\rm self},
        ~
        U_{H,2}^{\rm self}
\)
satisfy, for every admissible support point \(p_j\),
\begin{equation}
\label{eq:near_self_second_derivative_bounds}
        |\partial_\theta^2 K_j(p_j)|
        \le
        U_{K,2}^{\rm self},
        ~
        |\partial_\theta^2 H_j(p_j)|
        \le
        U_{H,2}^{\rm self}.
\end{equation}

Second, assume that the off-diagonal support-support second-derivative bounds
\(
        U_{K,2}^{\rm SS},
        ~
        U_{H,2}^{\rm SS}
\)
satisfy, for every support \(S\) in the support class and every \(j\ne\ell\),
\begin{equation}
\label{eq:near_SS_second_derivative_bounds}
        |\partial_\theta^2 K_\ell(p_j)|
        \le
        U_{K,2}^{\rm SS},
        ~
        |\partial_\theta^2 H_\ell(p_j)|
        \le
        U_{H,2}^{\rm SS}.
\end{equation}

Third, use the local channel envelopes and aggregate derivative bounds
\(B_{X,a}^{\rm loc}\) and \(D_a(\varpi_{\rm loc})\), \(a=2,3\), from
\eqref{eq:local_channel_envelopes_sec2}-\eqref{eq:local_Da_sec2}.
They are certified on the complete near sets
\(\mathcal N_\theta(S;\varpi_{\rm loc})\).

For
\[
 P(t):=P_{S,\widetilde{\mathbf v}}(p_j(t)),
\]
one has
\begin{equation}
\label{eq:near_P_derivative_D_bound}
 |P^{(a)}(t)|\le D_a(\varpi_{\rm loc}),
 ~ a=2,3,~
 t\in\mathcal I_j(\varpi_{\rm loc}).
\end{equation}
Indeed, differentiation of \eqref{eq:near_certificate_def} gives
\[
 P^{(a)}(t)
 =
 \sum_{\ell=1}^{L}\alpha_\ell
       \partial_\theta^aK_\ell(p_j(t))
 +
 \sum_{\ell=1}^{L}\beta_\ell
       \partial_\theta^aH_\ell(p_j(t)).
\]
Using the coefficient bounds and then the local channel envelopes,
\[
\begin{aligned}
 |P^{(a)}(t)|
 &\le
 \sum_{\ell=1}^{L}
 \left[
 \Gamma_K
 \left|\partial_\theta^aK(p_j(t),p_\ell)\right|
 +
 \Gamma_H
 \left|\partial_\theta^aH(p_j(t),p_\ell)\right|
 \right] \\
 &\le
 \mathcal D_a(\varpi_{\rm loc})
 \le
 D_a(\varpi_{\rm loc}),
 \qquad a=2,3.
\end{aligned}
\]
The last inequality follows because
\(p_j(t)\in\mathcal N_\theta(S;\varpi_{\rm loc})\).

% ============================================================

% ============================================================
\subsection{Curvature at the support}
\label{app:near-curvature-at-support}
% ============================================================

The local argument begins with the diagonal kernel.  Because the gauged atom
is normalized along the angular path, the real second derivative of its
self-correlation is exactly the negative weighted tangent energy.

\begin{lem}
\label{lem:self_curvature_gauged_atom}
Let
\(
        K_j(t)
        :=
        K_j(p_j(t))
        =
        \langle\psi_{p_j(t)},\psi_{p_j}\rangle_\rho .
\)
Then
\(
        \operatorname{Re} K_j''(0)
        =
        -\|\partial_\theta\psi_{p_j}\|_\rho^2 .
\)
\end{lem}

\begin{proof}
The atom is normalized: \(\|\psi_{p_j(t)}\|_\rho^2=1 ~ \text{for all admissible }t.\) The inner product is conjugate-linear in the first argument, so \(\|\psi_{p_j(t)}\|_\rho^2 = \langle\psi_{p_j(t)},\psi_{p_j(t)}\rangle_\rho .\) Differentiating once and evaluating at \(t=0\) gives
\[
 0
 =
 \langle\partial_\theta\psi_{p_j},\psi_{p_j}\rangle_\rho
 +
 \langle\psi_{p_j},\partial_\theta\psi_{p_j}\rangle_\rho
 =
 2\operatorname{Re}
 \langle\partial_\theta\psi_{p_j},\psi_{p_j}\rangle_\rho .
\]
Differentiating twice and evaluating at \(t=0\) gives
\[
\begin{aligned}
 0
 &=
 \langle\partial_\theta^2\psi_{p_j},\psi_{p_j}\rangle_\rho
 +
 2\langle\partial_\theta\psi_{p_j},
 \partial_\theta\psi_{p_j}\rangle_\rho
 +
 \langle\psi_{p_j},\partial_\theta^2\psi_{p_j}\rangle_\rho .
\end{aligned}
\]
The first and third terms are conjugates. Hence
\[
 2\operatorname{Re}
 \langle\partial_\theta^2\psi_{p_j},\psi_{p_j}\rangle_\rho
 +
 2\|\partial_\theta\psi_{p_j}\|_\rho^2
 =
 0.
\]
Since \(K_j''(0) = \langle\partial_\theta^2\psi_{p_j},\psi_{p_j}\rangle_\rho,\) the claim follows.
\end{proof}
Define the curvature-loss quantity
\begin{equation}
\label{eq:near_Ecurv_def}
        E_{\rm curv}
        :=
        2\Xi_K U_{K,2}^{\rm self}
        +
        2\Xi_H U_{H,2}^{\rm self}
        +
        2(L-1)
        \left(
        \Gamma_K U_{K,2}^{\rm SS}
        +
        \Gamma_H U_{H,2}^{\rm SS}
        \right).
\end{equation}
The near-curvature margin is
\(
        m_{\rm near}
        :=
        2\underline\sigma^2-E_{\rm curv}.
\)
The ideal diagonal curvature is reduced by coefficient corrections and by
off-diagonal interactions with the remaining support atoms.  The quantity
\(E_{\rm curv}\) collects these losses, so \(m_{\rm near}\) is the
curvature retained uniformly over the full support class.

\begin{lem}
\label{lem:support_uniform_curvature_bound}
For every support \(S\) in the class, every gauged phase vector
\(\widetilde{\mx v}\), and every \(j=1,\ldots,L\),
\(
        f_j''(0)\le -m_{\rm near}.
\)
Equivalently,
\(
        F_j''(0)\ge m_{\rm near}.
\)
\end{lem}

\begin{proof}
Fix \(j\), and write \(P(t):=P_{S,\widetilde {\mathbf v}}(p_j(t)).\) Then \(f_j(t)=|P(t)|^2.\) By Hermite interpolation, \(P(0)=\widetilde v_j, ~ P'(0)=0.\) Therefore \(f_j''(0) = 2\operatorname{Re} \left( \overline{\widetilde v_j}P''(0) \right).\) Next,
\[
 P''(0)
 =
 \sum_{\ell=1}^{L}
 \alpha_\ell\partial_\theta^2K_\ell(p_j)
 +
 \sum_{\ell=1}^{L}
 \beta_\ell\partial_\theta^2H_\ell(p_j).
\]
Separating the diagonal term \(\ell=j\), and writing \(\alpha_j=\widetilde v_j+(\alpha_j-\widetilde v_j),\) we obtain
\[
\begin{aligned}
 &f_j''(0)
 =
 2\operatorname{Re}\bigl(\partial_\theta^2K_j(p_j)\bigr)+
 2\operatorname{Re}
 \left(
 \overline{\widetilde v_j}(\alpha_j-\widetilde v_j)
 \partial_\theta^2K_j(p_j)
 \right)+
 2\operatorname{Re}
 \left(
 \overline{\widetilde v_j}\beta_j
 \partial_\theta^2H_j(p_j)
 \right)\\
 &+
 2\operatorname{Re}
 \left(
 \overline{\widetilde v_j}
 \sum_{\ell\ne j}
 \alpha_\ell\partial_\theta^2K_\ell(p_j)
 \right)+
 2\operatorname{Re}
 \left(
 \overline{\widetilde v_j}
 \sum_{\ell\ne j}
 \beta_\ell\partial_\theta^2H_\ell(p_j)
 \right).
\end{aligned}
\]

By \Cref{lem:self_curvature_gauged_atom}, \(2\operatorname{Re}\bigl(\partial_\theta^2K_j(p_j)\bigr) = -2\|\partial_\theta\psi_{p_j}\|_\rho^2 .\) Using \(|\widetilde v_j|=1,\) the diagonal correction bounds \(|\alpha_j-\widetilde v_j|\le \Xi_K, ~ |\beta_j|\le \Xi_H,\) the self derivative bounds
\[
 |\partial_\theta^2K_j(p_j)|\le U_{K,2}^{\rm self},
 ~
 |\partial_\theta^2H_j(p_j)|\le U_{H,2}^{\rm self},
\]
and the off-diagonal bounds \(|\alpha_\ell|\le \Gamma_K, ~ |\beta_\ell|\le \Gamma_H,\) \[
 |\partial_\theta^2K_\ell(p_j)|\le U_{K,2}^{\rm SS},
 ~
 |\partial_\theta^2H_\ell(p_j)|\le U_{H,2}^{\rm SS},
 ~
 \ell\ne j,
\]
we get
\[
\begin{aligned}
 &f_j''(0)
 \le
 -2\|\partial_\theta\psi_{p_j}\|_\rho^2
 +
 2\Xi_K U_{K,2}^{\rm self}
 +
 2\Xi_H U_{H,2}^{\rm self} +
 2(L-1)
 \left(
 \Gamma_K U_{K,2}^{\rm SS}
 +
 \Gamma_H U_{H,2}^{\rm SS}
 \right).
\end{aligned}
\]
By the tangent-energy lower bound, \(\|\partial_\theta\psi_{p_j}\|_\rho^2 \ge \underline\sigma^2.\) Therefore \(f_j''(0) \le -2\underline\sigma^2+E_{\rm curv} = -m_{\rm near}.\) Since \(F_j=1-f_j\), this is equivalent to \(F_j''(0)\ge m_{\rm near}.\) \end{proof}

% ============================================================
\subsection{Taylor proof of local strict decay}
\label{app:near-taylor-local-decay}

The local estimate uses the exact Hermite conditions before taking
absolute values.  Recall from \eqref{eq:eta_near_sec2} that, when the
curvature margin is positive and the derivative bounds are finite,
\begin{equation}
\label{eq:near_eta_near_def}
 \eta_{\rm near}(\varpi_{\rm loc})
 =
 \frac{2\varpi_{\rm loc}D_3(\varpi_{\rm loc})}{3m_{\rm near}}
 +
 \frac{\varpi_{\rm loc}^{2}D_2(\varpi_{\rm loc})^{2}}
      {2m_{\rm near}}.
\end{equation}
Otherwise \(\eta_{\rm near}(\varpi_{\rm loc})=+\infty\).

\begin{thm}
\label{thm:near_support_strict_decay_app}
Suppose
\[
 \eta_{\rm SS}<1,
 \qquad
 \eta_{\rm near}(\varpi_{\rm loc})<1.
\]
For every support \(S\) in the prescribed class, every gauged phase
vector \(\widetilde{\mathbf v}\), every \(j=1,\ldots,L\), and every
\(t\in\mathcal I_j(\varpi_{\rm loc})\),
\begin{equation}
\label{eq:near_quadratic_gap}
 |P_{S,\widetilde{\mathbf v}}(p_j(t))|^2
 \le
 1-\frac{m_{\rm near}}{2}
 \bigl(1-\eta_{\rm near}(\varpi_{\rm loc})\bigr)t^2.
\end{equation}
In particular,
\[
 |P_{S,\widetilde{\mathbf v}}(p_j(t))|<1,
 \qquad
 t\in\mathcal I_{j}^{\circ}(\varpi_{\rm loc}).
\]
\end{thm}

\begin{proof}
Fix \(S,\widetilde{\mathbf v},j\), and write
\[
 P(t)=P_{S,\widetilde{\mathbf v}}(p_j(t)),
 \qquad
 v=P(0)=\widetilde v_j,
 \qquad
 D_a=D_a(\varpi_{\rm loc}).
\]
Hermite interpolation gives
\[
 |v|=1,\qquad P'(0)=0.
\]
By \Cref{lem:support_uniform_curvature_bound},
\[
 2\operatorname{Re}\bigl(\overline vP''(0)\bigr)
 =
 \left.\frac{d^2}{dt^2}|P(t)|^2\right|_{t=0}
 \le -m_{\rm near}.
\]

Let \(t\in\mathcal I_j(\varpi_{\rm loc})\).
Because \(\Theta\) is an interval, one has
\(\theta_j+s\in\Theta\) for every \(s\) between \(0\) and \(t\).
Taylor's formula with integral remainder gives
\[
 P(t)-v=\frac12P''(0)t^2+R_3(t),
 \qquad
 R_3(t)=\frac12\int_0^t(t-s)^2P'''(s)\,ds.
\]
By \eqref{eq:near_P_derivative_D_bound},
\[
 |R_3(t)|\le\frac{D_3}{6}|t|^3.
\]
Independently, integrating \(P''\) and using \(P'(0)=0\),
\[
 P(t)-v=\int_0^t(t-s)P''(s)\,ds,
\]
so that
\[
 |P(t)-v|\le\frac{D_2}{2}t^2.
\]

Consequently,
\[
\begin{aligned}
 |P(t)|^2
 &=1+2\operatorname{Re}
       \bigl(\overline v(P(t)-v)\bigr)
       +|P(t)-v|^2\\
 &\le
 1-\frac{m_{\rm near}}2t^2
   +\frac{D_3}{3}|t|^3
   +\frac{D_2^2}{4}t^4\\
 &\le
 1-\frac{m_{\rm near}}2
 \left(
 1-\frac{2\varpi_{\rm loc}D_3}{3m_{\rm near}}
  -\frac{\varpi_{\rm loc}^2D_2^2}{2m_{\rm near}}
 \right)t^2.
\end{aligned}
\]
The last inequality uses \(|t|\le\varpi_{\rm loc}\).
By \eqref{eq:near_eta_near_def}, this is
\eqref{eq:near_quadratic_gap}.

Since \(m_{\rm near}>0\) and
\(\eta_{\rm near}(\varpi_{\rm loc})<1\), the right-hand side is
strictly smaller than one whenever \(t\ne0\).
The strict modulus bound follows on the punctured feasible interval.
\end{proof}

The strict-decay theorem also constrains the support geometry. If a second
support point at the same range lay within the localization radius, Hermite
interpolation would force the certificate modulus to equal one there, while
the preceding theorem would force it to be strictly smaller than one.  The
two conditions are incompatible.

\begin{cor}
\label{cor:qpac_near_implies_noncollision}
Under the hypotheses of \Cref{thm:near_support_strict_decay_app}, no two distinct same-range support points
\(
        p_j=(i,\theta_j),
        ~
        p_\ell=(i,\theta_\ell)
\)
satisfy
\(
        |\theta_j-\theta_\ell|
        \le
        \varpi_{\rm loc}.
\)
\end{cor}

\begin{proof}
Suppose, to the contrary, that
\(
        0<
        |\theta_\ell-\theta_j|
        \le
        \varpi_{\rm loc}.
\)
Set
\(
        t:=\theta_\ell-\theta_j.
\)
Since \(\theta_j,\theta_\ell\in\Theta\), one has
\(
        t\in
        \mathcal I_{j}^{\circ}(\varpi_{\rm loc}).
\)
Hence \Cref{thm:near_support_strict_decay_app} gives
\(
        |P_{S,\widetilde{\mathbf v}}(p_\ell)|
        =
        |P_{S,\widetilde{\mathbf v}}(p_j(t))|
        <1.
\)
On the other hand, Hermite interpolation gives
\(
        P_{S,\widetilde{\mathbf v}}(p_\ell)
        =
        \widetilde v_\ell,
        ~
        |\widetilde v_\ell|=1,
\)
which is a contradiction.
\end{proof}

\subsection{Support-uniform near set statement}
\label{app:near-support-uniform-statement}
% ============================================================

Recall that
\(
        \mathcal N^\circ(S;\varpi_{\rm loc})
        =
        \mathcal N_\theta(S;\varpi_{\rm loc})\setminus S.
\)
We show directly that
\Cref{thm:near_support_strict_decay_app} gives strict boundedness on this set.
Fix
\(
        p=(i,\theta)
        \in
        \mathcal N^\circ(S;\varpi_{\rm loc}).
\)
By the definition of
\(\mathcal N_\theta(S;\varpi_{\rm loc})\), there exists a support point
\(
        p_j=(i_j,\theta_j)\in S
\)
such that
\(
        i=i_j,
        ~
        |\theta-\theta_j|\le\varpi_{\rm loc}.
\)
Set
\(
        t:=\theta-\theta_j.
\)
Since \(p\notin S\), one has \(t\ne0\).  Moreover,
\(\theta_j,\theta\in\Theta\subset(0,\pi)\), and therefore
\(
        [\theta_j,\theta]
        \subset
        [\theta_{\min},\theta_{\max}]\cap(0,\pi).
\)
Consequently,
\(
        t\in
        \mathcal I_{j}^{\circ}(\varpi_{\rm loc}).
\)
By \Cref{thm:near_support_strict_decay_app},
\(
        |P_{S,\widetilde{\mathbf v}}(p)|
        =
        |P_{S,\widetilde{\mathbf v}}(p_j(t))|
        <1.
\)
Since
\(p\in\mathcal N^\circ(S;\varpi_{\rm loc})\)
was arbitrary, we obtain
\begin{equation}
\label{eq:near_final_strict_decay}
        |P_{S,\widetilde{\mathbf v}}(p)|<1,
        ~
        p\in
        \mathcal N^\circ(S;\varpi_{\rm loc}).
\end{equation}

The derivative estimates used in the Taylor proof are taken on
\(\mathcal I_j(\varpi_{\rm loc})\), equivalently on
\(\mathcal N_\theta(S;\varpi_{\rm loc})\).
Strict dual boundedness is required on the feasible off-support set
\(\mathcal Q\setminus S\).
% ============================================================
\section{Proof of the support-uniform Fresnel theorem}
\label{app:exact-proof}
% ============================================================

This appendix proves the main exact-recovery theorem from the support-uniform
QPAC conditions.  The proof assembles four ingredients proved elsewhere:
Hermite-system invertibility and coefficient bounds, support-uniform QPAC
kernel estimates, near-support strict decay, and TV duality.

The theorem is support-uniform: the same deterministic QPAC inequalities hold
for every support in the prescribed support class and for every choice of
nonzero amplitudes.

% ============================================================
\subsection{Statement proved in this appendix}
\label{app:exact-proof-statement}
% ============================================================

Let
\(
        \Qset
        =
        \{1,\ldots,N_d\}\times\Theta
\)
be the semi-discrete Fresnel parameter domain, and let
\(
        \varnothing\ne\mathfrak S_L
        \subseteq
        \mathfrak S_L^{\rm SS}
        (\mathbf u_{\rm SS};\mathscr E)
\)
be the prescribed class of \(L\)-point supports.  Thus every
\(
        S
        =
        \{p_\ell=(i_\ell,\theta_\ell)\}_{\ell=1}^{L}
        \in\mathfrak S_L
\)
satisfies, for every ordered pair of distinct support points,
\(|K(p_j,p_\ell)|\le u_K^{\rm SS}, ~ |H(p_j,p_\ell)|\le u_H^{\rm SS},\)
and
\(|dK(p_j,p_\ell)|\le u_{dK}^{\rm SS}, ~ |dH(p_j,p_\ell)|\le u_{dH}^{\rm SS}.\)
Assume that, for some \(\varpi_{\rm loc}>0\),
\[
        \mathfrak C_{\rm QPAC}^{\mathscr E}
        (\mathfrak S_L,\varpi_{\rm loc})
        =
        \max\left\{
        \eta_{\rm SS},
        \eta_{\rm near}(\varpi_{\rm loc}),
        \eta_{\rm far}(\varpi_{\rm loc})
        \right\}
        <1 .
\]
In particular,
\(
        \eta_{\rm SS}<1,
        ~
        \eta_{\rm near}(\varpi_{\rm loc})<1,
        ~
        \eta_{\rm far}(\varpi_{\rm loc})<1.
\)
The near-support conclusion
\(
        |P_{S,\widetilde{\mathbf v}}(p)|<1,
        ~
        p\in
        \mathcal N_\theta(S;\varpi_{\rm loc})\setminus S,
\)
follows from
\Cref{app:near-support-uniform-statement}.
We prove that every measure
\(
        \nu_\star
        =
        \sum_{\ell=1}^L
        c_\ell\,\delta_{p_\ell},
        ~
        c_\ell\neq0,
        ~
        S=\{p_\ell\}_{\ell=1}^L\in\mathfrak S_L,
\)
is the unique solution of
\begin{equation}\label{eq:exact_TV_program}
        \min_{\nu}\|\nu\|_{\TV,\Qset}
        ~
        \text{\rm subject to}
        ~
        \FFR\nu=\FFR\nu_\star .
\end{equation}

% ============================================================
\subsection{Raw atoms, adjoint polynomials, and TV duality}
\label{app:exact-atoms-duality}

For \(p=(i,\theta)\in\Qset\), recall the raw Fresnel measurement atom
\(
        \varphi_p[n]
        :=
        a_{\rm Fr}(r_i,\theta)[n],
        ~ n=0,\ldots,N_r-1 .
\)
Thus
\(\FFR\nu = \int_{\Qset}\varphi_p\,d\nu(p).\)
The adjoint Fresnel family associated with
\(\zeta\in\mathbb C^{N_r}\) is
\(
        Q^\zeta(p)
        :=
        \langle \varphi_p,\zeta\rangle_{\mathbb C^{N_r}}
        =
        \sum_{n=0}^{N_r-1}
        \zeta_n\overline{\varphi_p[n]} .
\)
With the inner-product convention of the paper,
\(
        \langle \zeta,\FFR\nu\rangle_{\mathbb C^{N_r}}
        =
        \int_{\Qset}\overline{Q^\zeta(p)}\,d\nu(p).
\)
Therefore the TV dual certificate criterion proved in
\Cref{prop:app_tv_dual_certificate} applies directly to the raw Fresnel atoms
\(\varphi_p\).  In the rest of this appendix we only construct such an adjoint
family by first building a gauged Hermite certificate and then removing the
gauge.
% ============================================================

\subsection{Hermite interpolation and coefficient bounds}
\label{app:exact-Hermite-system}

Fix a support
\(S=\{p_\ell\}_{\ell=1}^L\in\mathfrak S_L\)
and a raw phase vector \(\mx v=(v_1,\ldots,v_L)\), \(|v_\ell|=1\).  Define the
gauged phases
\(\widetilde v_\ell:= e^{i\chi_{i_\ell}(\theta_\ell)}v_\ell.\)
The gauged Hermite certificate is
\[
        P_{S,\widetilde {\mathbf v}}(p)
        =
        \sum_{\ell=1}^{L}\alpha_\ell K_\ell(p)
        +
        \sum_{\ell=1}^{L}\beta_\ell H_\ell(p),
\]
with interpolation conditions
\(P_{S,\widetilde {\mathbf v}}(p_j)=\widetilde v_j, ~ \partial_\theta P_{S,\widetilde {\mathbf v}}(p_j)=0.\)
The support-support QPAC bounds give the block domination matrix
\(\mathbf G_{\rm SS}\).  Since
\(\eta_{\rm SS} = (L-1)\rho_{\rm sp}(\mathbf G_{\rm SS})<1,\)
\Cref{lem:app_hermite_neumann} gives existence, uniqueness, and the coefficient
bounds
\[
        |\alpha_\ell|\le \Gamma_K,
        ~
        |\beta_\ell|\le \Gamma_H,
        ~ \ell=1,\ldots,L.
\]
By \Cref{lem:app_diagonal_correction_bounds}, the diagonal correction bounds are
\[
        |\alpha_j-\widetilde v_j|\le \Xi_K,
        ~
        |\beta_j|\le \Xi_H,
        ~ j=1,\ldots,L.
\]
Finally, \Cref{lem:app_hermite_implies_atom_independence} gives linear
independence of the active raw Fresnel atoms.
% ============================================================
\subsection{Strict dual feasibility on the far set}
\label{app:exact-far-set}

Recall that
\[
 \mathcal F(S;\varpi_{\rm loc})
 =
 \Qset\setminus\mathcal N_\theta(S;\varpi_{\rm loc}).
\]
Fix \(p\in\mathcal F(S;\varpi_{\rm loc})\).
The certificate representation, the coefficient bounds, and the
certified far-pair envelopes give
\[
\begin{aligned}
|P_{S,\widetilde{\mathbf v}}(p)|
&\le
\sum_{\ell=1}^{L}
\left(
|\alpha_\ell|\,|K(p,p_\ell)|
+
|\beta_\ell|\,|H(p,p_\ell)|
\right)\\
&\le
\sum_{\ell=1}^{L}
\left(
\Gamma_K|K(p,p_\ell)|
+
\Gamma_H|H(p,p_\ell)|
\right)\le \eta_{\rm far}(\varpi_{\rm loc}).
\end{aligned}
\]
The support \(S\) and evaluation point \(p\) remain fixed throughout
each sum.

Since \(\eta_{\rm far}(\varpi_{\rm loc})<1\), we obtain
\begin{equation}
\label{eq:exact_far_strict}
 |P_{S,\widetilde{\mathbf v}}(p)|<1,
 \qquad p\in\mathcal F(S;\varpi_{\rm loc}).
\end{equation}
If the far set is empty, this assertion is vacuous and the corresponding
budget is zero.

% ============================================================

% ============================================================
\subsection{Global dual certificate}
\label{app:exact-global-certificate}
% ============================================================

By definition,
\[
        \mathcal N^\circ(S;\varpi_{\rm loc})
        =
        \mathcal N_\theta(S;\varpi_{\rm loc})\setminus S
        \subseteq
        \mathcal Q\setminus S.
\]
By \eqref{eq:near_final_strict_decay},
\(
        |P_{S,\widetilde{\mathbf v}}(p)|<1,
        ~
        p\in
        \mathcal N^\circ(S;\varpi_{\rm loc}).
\)

The domain decomposes as
\(
        \Qset
        =
        S
        \cup
        \mathcal N^\circ(S;\varpi_{\rm loc})
        \cup
        \mathcal F(S;\varpi_{\rm loc}).
\)
On the support \(S\), the gauged Hermite interpolation conditions give
\[
        P_{S,\widetilde {\mathbf v}}(p_\ell)
        =
        \widetilde v_\ell,
        ~
        |\widetilde v_\ell|=1,
        ~
        \ell=1,\ldots,L .
\]
On the punctured near set, \eqref{eq:near_final_strict_decay} gives
\[
        |P_{S,\widetilde {\mathbf v}}(p)|<1,
        ~
        p\in\mathcal N^\circ(S;\varpi_{\rm loc}).
\]
On the far set, \eqref{eq:exact_far_strict} gives
\[
        |P_{S,\widetilde {\mathbf v}}(p)|<1,
        ~
        p\in\mathcal F(S;\varpi_{\rm loc}).
\]
Therefore
\(
        |P_{S,\widetilde {\mathbf v}}(p)|\le1,
        ~
        p\in\Qset,
\)
and
\(
        |P_{S,\widetilde {\mathbf v}}(p)|<1,
        ~
        p\in\Qset\setminus S .
\)
We now remove the gauge.  By \Cref{app:gauged-to-ungauged}, the ungauged
polynomial
\(
        Q_{S, {\mathbf v}}(p)
        :=
        e^{-i\chi_p}P_{S,\widetilde {\mathbf v}}(p)
\)
belongs to the range of the raw Fresnel adjoint operator and has the raw
adjoint representation
\(
        Q_{S, {\mathbf v}}(p)
        =
        \sum_{n=0}^{N_r-1}
        \zeta_{S,v}[n]\overline{\varphi_p[n]},
\)
where
\(
      \zeta_{S,v}[n]
:=
\tfrac{\rho_n}{\sqrt{W_0}}z_{S,\widetilde v}[n].
\)
Since the gauge factor is unimodular,
\[
        |Q_{S, {\mathbf v}}(p)|
        =
        |P_{S,\widetilde {\mathbf v}}(p)|,
        ~
        p\in\Qset .
\]
Hence
\[
        |Q_{S, {\mathbf v}}(p)|\le1,
        ~
        p\in\Qset,
\]
and
\[
        |Q_{S, {\mathbf v}}(p)|<1,
        ~
        p\in\Qset\setminus S .
\]

It remains only to check sign interpolation.  For each support point
\(p_\ell=(i_\ell,\theta_\ell)\), the definition of the gauged signs gives
\(
        \widetilde v_\ell
        =
        e^{i\chi_{p_\ell}}v_\ell.
\)
Therefore
\[
        Q_{S, {\mathbf v}}(p_\ell)
        =
        e^{-i\chi_{p_\ell}}P_{S,\widetilde {\mathbf v}}(p_\ell)
        =
        e^{-i\chi_{p_\ell}}\widetilde v_\ell
        =
        v_\ell .
\]
Thus \(Q_{S, {\mathbf v}}\) satisfies
\[
        Q_{S, {\mathbf v}}(p_\ell)=v_\ell,
        ~
        \ell=1,\ldots,L,
\]
\[
        |Q_{S, {\mathbf v}}(p)|\le1,
        ~
        p\in\Qset,
\]
and
\[
        |Q_{S, {\mathbf v}}(p)|<1,
        ~
        p\in\Qset\setminus S .
\]
This is precisely the raw TV dual certificate required by
\Cref{prop:app_tv_dual_certificate}.
% ============================================================
\subsection{Completion of the exact-recovery proof}
\label{app:exact-completion}
% ============================================================

Let
\(
        \nu_\star
        =
        \sum_{\ell=1}^{L}
        c_\ell\delta_{p_\ell},
        ~ c_\ell\neq0,
\)
with
\(S=\{p_\ell\}_{\ell=1}^{L}\in\mathfrak S_L.\)
Set
\(v_\ell:=\tfrac{c_\ell}{|c_\ell|}.\)
The construction above gives a raw Fresnel adjoint polynomial \(Q_{S, {\mathbf v}}\)
satisfying
\(Q_{S, {\mathbf v}}(p_\ell)=v_\ell, ~ \ell=1,\ldots,L,\)
\(|Q_{S, {\mathbf v}}(p)|\le1, ~ p\in\Qset,\)
and
\(|Q_{S, {\mathbf v}}(p)|<1, ~ p\in\Qset\setminus S.\)
The Hermite-system argument also proves that the raw active Fresnel atoms
\(\{\varphi_{p_\ell}\}_{\ell=1}^{L}\)
are linearly independent.
Therefore all hypotheses of \Cref{prop:app_tv_dual_certificate} hold.  Hence
\(\nu_\star\) is the unique minimizer of
\[
        \min_{\nu}\|\nu\|_{\TV,\Qset}
        ~
        \text{\rm subject to}
        ~
        \FFR\nu=\FFR\nu_\star .
\]
This proves the support-uniform Fresnel exact-recovery theorem.
% ============================================================
\section{Finite-harmonic Bessel-Vandermonde expansions}
\label{app:bv-expansion}
% ============================================================

This appendix proves the finite-harmonic expansion used by the lifted
computational model.  The object expanded here is the semi-discrete Fresnel
atom.  It admits an infinite Bessel-Vandermonde expansion in the angular
variable, and truncating the two Jacobi-Anger series gives a finite-harmonic
approximation with explicit uniform error
\(\Delta_{P_{\rm JA},Q_{\rm JA}}^{\rm Fr}\to0 ~ \text{as }P_{\rm JA},Q_{\rm JA}\to\infty.\)
For physical spherical-wave data, this Fresnel lift incurs an additional
deterministic paraxial mismatch.  Thus
\(
        \Delta_{P_{\rm JA},Q_{\rm JA}}^{\rm NF}
        \le
        \Delta_{P_{\rm JA},Q_{\rm JA}}^{\rm Fr}
        +
        \Delta_{\rm par}.
\)
The first term is controlled by the harmonic truncation order.  The second is
the deterministic Fresnel approximation error and does not vanish by increasing
\(P_{\rm JA},Q_{\rm JA}\) alone.  This distinction is used only to interpret computations with
physical data; the exact recovery theorem is proved for the semi-discrete
Fresnel TV model.
% ============================================================
\subsection{Aperture and atom conventions}
\label{app:bv-conventions}
% ============================================================

Let the aperture index be
\(n=0,\ldots,N_r-1,\)
and set
\(n_{\max}:=N_r-1.\)
We keep the notation \(k_\lambda=2\pi/\lambda\) for the wavenumber.  The admissible ranges satisfy
\(r_i\in[r_{\min},r_{\max}], ~ i=1,\ldots,N_d,\)
and angles satisfy
\(\theta\in[\theta_{\min},\theta_{\max}] \subset(0,\pi).\)
We assume the aperture is in the non-occluding near-field regime \(\varepsilon_\star
        :=
        \tfrac{n_{\max}d}{r_{\min}}
        <1 .\)
This ensures that all square-root expressions below are single-valued and
smooth on the physical domain.
The semi-discrete Fresnel atom at range \(r_i\) is
\begin{equation}\label{eq:fresnel_atom_bv_app}
        a_{\rm Fr}(r_i,\theta)[n]
        :=
        \exp\left(
        ik_{\lambda} dn\cos\theta
        -
        i\tfrac{k_{\lambda} d^2n^2}{2r_i}\sin^2\theta
        \right).
\end{equation}
The physical spherical-wave atom is
\begin{equation}\label{eq:physical_atom_bv_app}
        a_{\rm NF}(r_i,\theta)[n]
        :=
        \exp\left(
        -ik_{\lambda}
        \left(
        R_n(r_i,\theta)-r_i
        \right)
        \right),
\end{equation}
where
\(R_n(r,\theta):= \sqrt{r^2-2rnd\cos\theta+n^2d^2}.\)
With this sign convention, the Fresnel approximation of
\(R_n(r,\theta)-r\)
is
\(-nd\cos\theta + \tfrac{n^2d^2}{2r}\sin^2\theta,\)
which gives exactly \eqref{eq:fresnel_atom_bv_app}.
For a semi-discrete scene
\(
        \nu
        =
        \sum_{i=1}^{N_d}
        \delta_{r_i}\otimes\mu_i,
        ~
        \mu_i\in\mathcal M(\Theta),
\)
define the physical near-field measurement operator
\(\FNF\) by
\begin{equation}\label{eq:physical_near_field_operator}
        (\FNF\nu)[n]
        :=
        \sum_{i=1}^{N_d}
        \int_{\Theta}
        a_{\rm NF}(r_i,\theta)[n]\,d\mu_i(\theta),
        ~
        n=0,\ldots,N_r-1 .
\end{equation}
% ============================================================
\subsection{Jacobi-Anger identities and tail bounds}
\label{app:jacobi-anger-tail-bounds}
% ============================================================

We use the Jacobi-Anger identity\cite[Sec.~10.12]{OlverEtAl2010}:
\begin{equation}\label{eq:JA_identity}
        e^{iz\cos\theta}
        =
        \sum_{p\in\mathbb Z}
        i^pJ_p(z)e^{ip\theta},
        ~ z\in\mathbb R,
\end{equation}
where \(J_p\) is the Bessel function of the first kind.  Since
\(J_{-p}(z)=(-1)^pJ_p(z),\)
the two-sided tail is controlled by the positive orders.

For \(X\ge0\) and \(M\in\mathbb Z_{\ge0}\), define the Bessel tail envelope
\begin{equation}\label{eq:Bessel_tail_def}
        \mathcal E_M(X)
        :=
        2\sum_{p=M+1}^{\infty}
        \sup_{0\le z\le X}|J_p(z)|.
\end{equation}
Then, uniformly for
\(z\in[0,X]\) and \(\theta\in\mathbb R\), we have
\begin{equation}\label{eq:JA_tail_bound}
        \left|
        e^{iz\cos\theta}
        -
        \sum_{|p|\le M}
        i^pJ_p(z)e^{ip\theta}
        \right|
        \le
        \mathcal E_M(X).
\end{equation}

The quantity \(\mathcal E_M(X)\) is deterministic and can be computed by
certified interval bounds for Bessel functions.  For a completely explicit
upper bound, one may use
\[
        |J_p(z)|\le I_p(z)
        \le
        e^{X^2/4}\tfrac{(X/2)^p}{p!},
        ~ 0\le z\le X,
\]
where \(I_p\) is the modified Bessel function.  Hence
\begin{equation}\label{eq:Bessel_tail_explicit_bound}
        \mathcal E_M(X)
        \le
        2e^{X^2/4}
        \sum_{p=M+1}^{\infty}
        \tfrac{(X/2)^p}{p!}.
\end{equation}
In particular,
\(\mathcal E_M(X)\to0 ~ \text{as }M\to\infty.\)

% ============================================================
\subsection{Fresnel Bessel-Vandermonde expansion}
\label{app:fresnel-bv-expansion}
% ============================================================

Define
\(
        X_1:=k_{\lambda}d\,n_{\max},
        ~
        X_2:=\tfrac{k_{\lambda}d^2n_{\max}^2}{4r_{\min}} .
\)
For fixed \(i\) and \(n\), set
\(
        \xi_{n,i}:=\tfrac{k_{\lambda}d^2n^2}{4r_i}.
\)
Using
\(
        \sin^2\theta=\tfrac{1-\cos(2\theta)}{2},
\)
we rewrite the Fresnel atom as
\[
\begin{aligned}
        a_{\rm Fr}(r_i,\theta)[n]
        &=
        \exp(ik_{\lambda}dn\cos\theta)
        \exp\left(
        -i\tfrac{k_{\lambda}d^2n^2}{2r_i}\sin^2\theta
        \right) \\
        &=
        \exp(ik_{\lambda}dn\cos\theta)
        \exp(-i\xi_{n,i})
        \exp(i\xi_{n,i}\cos(2\theta)).
\end{aligned}
\]
Thus the Fresnel atom is the product of one angular Jacobi-Anger factor, one
constant phase, and one second-harmonic Jacobi-Anger factor.

Applying \eqref{eq:JA_identity} to both oscillatory factors gives the exact
infinite expansion
\begin{equation}
\label{eq:fresnel_exact_infinite_BV}
        a_{\rm Fr}(r_i,\theta)[n]
        =
        \sum_{p\in\mathbb Z}
        \sum_{q\in\mathbb Z}
        c_{n,i,p,q}^{\rm Fr}
        e^{i(p+2q)\theta},
\end{equation}
where \(  c_{n,i,p,q}^{\rm Fr}
        :=
        e^{-i\xi_{n,i}}
        i^{p+q}
        J_p(k_{\lambda}dn)
        J_q(\xi_{n,i}).\)
For Jacobi-Anger truncation orders
\(P_{\rm JA},Q_{\rm JA}\ge0\), define
\(
        \mathcal J_{P_{\rm JA},Q_{\rm JA}}
        :=
        \left\{
        (p,q)\in\mathbb Z^2:
        |p|\le P_{\rm JA},\
        |q|\le Q_{\rm JA}
        \right\}.
\)
The finite Bessel-Vandermonde approximation is \[ a_{\rm Fr}^{P_{\rm JA},Q_{\rm JA}}(r_i,\theta)[n]
        :=
        \sum_{(p,q)\in\mathcal J_{P_{\rm JA},Q_{\rm JA}}}
        c_{n,i,p,q}^{\rm Fr}
        e^{i(p+2q)\theta}.\]
The retained harmonics have the form
\(
        m=p+2q.
\)
Since \(|p|\le P_{\rm JA}\) and \(|q|\le Q_{\rm JA}\), all retained harmonics satisfy
\(
        |m|\le P_{\rm JA}+2Q_{\rm JA}.
\)
We therefore set
\(
        I:=P_{\rm JA}+2Q_{\rm JA}
\)
and index the lifted harmonic variables by
\(
        m=-I,\ldots,I.
\)
For \(-I\le m\le I\), define the aggregated Fresnel coefficient \( C_{n,i,m}^{\rm Fr,P_{\rm JA},Q_{\rm JA}}
        :=
        \sum_{\substack{|p|\le P_{\rm JA},\ |q|\le Q_{\rm JA}\\ p+2q=m}}
        c_{n,i,p,q}^{\rm Fr},\) with the convention that an empty sum is zero.  Then
\begin{equation}
\label{eq:fresnel_finite_BV_aggregated}
        a_{\rm Fr}^{P_{\rm JA},Q_{\rm JA}}(r_i,\theta)[n]
        =
        \sum_{m=-I}^{I}
        C_{n,i,m}^{\rm Fr,P_{\rm JA},Q_{\rm JA}}e^{im\theta}.
\end{equation}
Let
\(
        \mathbf v_I(\theta)
        :=
        \bigl(e^{im\theta}\bigr)_{m=-I}^{I}
\)
and
\(
        \mathbf A_i(\theta)
        :=
        \mathbf e_i\mathbf v_I(\theta)^H .
\)
The lifted sensing matrices are chosen as
\begin{equation}
\label{eq:fresnel_lifted_matrix_coefficients}
        \bigl(\boldsymbol\Phi_{n,{\rm Fr}}^{P_{\rm JA},Q_{\rm JA}}\bigr)_{i,m}
        :=
        \overline{C_{n,i,-m}^{\rm Fr,P_{\rm JA},Q_{\rm JA}}},
        ~
        m=-I,\ldots,I.
\end{equation}
With the Frobenius inner product
\(
        \langle X,Y\rangle_F
        =
        \operatorname{trace}(X^HY)
        =
        \sum_{i,m}\overline{X_{i,m}}Y_{i,m},
\)
we obtain
\begin{align*}
   & \left\langle
        \boldsymbol\Phi_{n,{\rm Fr}}^{P_{\rm JA},Q_{\rm JA}},
        \mathbf A_i(\theta)
        \right\rangle_F=
        \sum_{m=-I}^{I}
        \overline{
        \bigl(\boldsymbol\Phi_{n,{\rm Fr}}^{P_{\rm JA},Q_{\rm JA}}\bigr)_{i,m}}
        e^{-im\theta} =
       \nonumber\\
       &\sum_{m=-I}^{I}
        C_{n,i,-m}^{\rm Fr,P_{\rm JA},Q_{\rm JA}}e^{-im\theta} =
        \sum_{m=-I}^{I}
        C_{n,i,m}^{\rm Fr,P_{\rm JA},Q_{\rm JA}}e^{im\theta}=
        a_{\rm Fr}^{P_{\rm JA},Q_{\rm JA}}(r_i,\theta)[n].
\end{align*} 

The finite-harmonic representation differs from the exact Fresnel atom only
through the omitted Jacobi-Anger tails.  Bounding the linear-phase and
quadratic-phase tails uniformly and controlling their product gives an
atomwise approximation error that is independent of the sensor index, the
range bin, and the admissible angle.

\begin{prop}
\label{prop:fresnel_truncation_bound}
For every \(i=1,\ldots,N_d\), every \(n=0,\ldots,N_r-1\), and every
\(\theta\in[\theta_{\min},\theta_{\max}]\),
\begin{equation}
\label{eq:fresnel_trunc_error_bound}
        \left|
        a_{\rm Fr}(r_i,\theta)[n]
        -
        a_{\rm Fr}^{P_{\rm JA},Q_{\rm JA}}(r_i,\theta)[n]
        \right|
        \le
        \Delta_{P_{\rm JA},Q_{\rm JA}}^{\rm Fr},
\end{equation}
where \(\Delta_{P_{\rm JA},Q_{\rm JA}}^{\rm Fr}
:=
\mathcal E_{P_{\rm JA}}(X_1)
+
\mathcal E_{Q_{\rm JA}}(X_2)
+
\mathcal E_{P_{\rm JA}}(X_1)
\mathcal E_{Q_{\rm JA}}(X_2).\)
Consequently,
\[
        \Delta_{P_{\rm JA},Q_{\rm JA}}^{\rm Fr}\to0
        ~
        \text{as }P_{\rm JA},Q_{\rm JA}\to\infty.
\]
\end{prop}

\begin{proof}
Let \(A(\theta):=\exp(ik_{\lambda}dn\cos\theta), ~ B(\theta):=\exp(i\xi_{n,i}\cos(2\theta)).\) Let \(A_{P_{\rm JA}}\) and \(B_{Q_{\rm JA}}\) be their truncated Jacobi-Anger expansions: \(A_{P_{\rm JA}}(\theta) := \sum_{|p|\le P_{\rm JA}} i^pJ_p(k_{\lambda}dn)e^{ip\theta},\) and \(B_{Q_{\rm JA}}(\theta) := \sum_{|q|\le Q_{\rm JA}} i^qJ_q(\xi_{n,i})e^{i2q\theta}.\) By the definition of the Jacobi-Anger tails, \(|A(\theta)-A_{P_{\rm JA}}(\theta)| \le \mathcal E_{P_{\rm JA}}(X_1),\) because \(0\le k_{\lambda}dn\le X_1.\) Similarly, \(|B(\theta)-B_{Q_{\rm JA}}(\theta)| \le \mathcal E_{Q_{\rm JA}}(X_2),\) because \(0\le \xi_{n,i}\le X_2.\) The factor \(e^{-i\xi_{n,i}}\) has modulus one, so it does not affect the
truncation error.

Since \(|A(\theta)|=|B(\theta)|=1,\) we have
\[
\begin{aligned}
 |A(\theta)B(\theta)-A_{P_{\rm JA}}(\theta)B_{Q_{\rm JA}}(\theta)|
 &\le
 |A(\theta)-A_{P_{\rm JA}}(\theta)|\,|B(\theta)|+
 |A_{P_{\rm JA}}(\theta)|\,|B(\theta)-B_{Q_{\rm JA}}(\theta)|.
\end{aligned}
\]
Also, \(|A_{P_{\rm JA}}(\theta)| \le |A(\theta)|+|A(\theta)-A_{P_{\rm JA}}(\theta)| \le 1+\mathcal E_{P_{\rm JA}}(X_1).\) Therefore
\[
\begin{aligned}
 |A(\theta)B(\theta)-A_{P_{\rm JA}}(\theta)B_{Q_{\rm JA}}(\theta)|
 &\le
 \mathcal E_{P_{\rm JA}}(X_1)
 +
 \bigl(1+\mathcal E_{P_{\rm JA}}(X_1)\bigr)\mathcal E_{Q_{\rm JA}}(X_2) \\
 &=
 \mathcal E_{P_{\rm JA}}(X_1)
 +
 \mathcal E_{Q_{\rm JA}}(X_2)
 +
 \mathcal E_{P_{\rm JA}}(X_1)\mathcal E_{Q_{\rm JA}}(X_2).
\end{aligned}
\]
Multiplication by the unit-modulus factor \(e^{-i\xi_{n,i}}\) preserves this
bound. Hence
\[
 \left|
 a_{\rm Fr}(r_i,\theta)[n]
 -
 a_{\rm Fr}^{P_{\rm JA},Q_{\rm JA}}(r_i,\theta)[n]
 \right|
 \le
 \Delta_{P_{\rm JA},Q_{\rm JA}}^{\rm Fr}.
\]
The convergence \(\Delta_{P_{\rm JA},Q_{\rm JA}}^{\rm Fr}\to0\) follows from the convergence of the Jacobi-Anger tails \(\mathcal E_{P_{\rm JA}}(X_1)\to0, ~ \mathcal E_{Q_{\rm JA}}(X_2)\to0\) as \(P_{\rm JA},Q_{\rm JA}\to\infty\). This proves the proposition.
\end{proof}
% ============================================================
\subsection{Lifted Bessel-Vandermonde factorization}
\label{app:lifted-BV-factorization}
% ============================================================

Collect the Bessel-product terms with the same angular harmonic
\(m=p+2q.\)
Then
\(
        a_{\rm Fr}^{P_{\rm JA},Q_{\rm JA}}(r_i,\theta)[n]
        =
        \sum_{m=-P_{\rm JA}-2Q_{\rm JA}}^{P_{\rm JA}+2Q_{\rm JA}}
        C_{n,i,m}^{\rm Fr,P_{\rm JA},Q_{\rm JA}}e^{im\theta},
\)
where
\(
        C_{n,i,m}^{\rm Fr,P_{\rm JA},Q_{\rm JA}}
        :=
        \sum_{\substack{|p|\le P_{\rm JA},\ |q|\le Q_{\rm JA}\\ p+2q=m}}
        c_{n,i,p,q}^{\rm Fr}.
\)
Let
\(
        I:=P_{\rm JA}+2Q_{\rm JA},
        ~
        \mathbf v_I(\theta)
        :=
        \bigl(e^{-iI\theta},e^{-i(I-1)\theta},
        \ldots,e^{iI\theta}\bigr)^T,
\)
and define
\(\mathbf A_i(\theta):= \mathbf e_i\mathbf v_I(\theta)^H.\)
Define the lifted coefficient matrices
\(\boldsymbol\Phi_{n,{\rm Fr}}^{P_{\rm JA},Q_{\rm JA}} \in\mathbb C^{N_d\times(2I+1)}\)
by indexing columns with \(m=-I,\ldots,I\) and setting
\(\bigl(\boldsymbol\Phi_{n,{\rm Fr}}^{P_{\rm JA},Q_{\rm JA}}\bigr)_{i,m}:= \overline{C_{n,i,-m}^{\rm Fr,P_{\rm JA},Q_{\rm JA}}}.\)
Then, with the Frobenius convention
\(\langle X,Y\rangle_F = \operatorname{trace}(X^HY),\)
one has
\(
        a_{\rm Fr}^{P_{\rm JA},Q_{\rm JA}}(r_i,\theta)[n]
        =
        \left\langle
        \boldsymbol\Phi_{n,{\rm Fr}}^{P_{\rm JA},Q_{\rm JA}},
        \mathbf A_i(\theta)
        \right\rangle_F.
\)
Combining this identity with \Cref{prop:fresnel_truncation_bound}, we obtain
\begin{equation}
\label{eq:fresnel_lifted_approximation_app}
        a_{\rm Fr}(r_i,\theta)[n]
        =
        \left\langle
        \boldsymbol\Phi_{n,{\rm Fr}}^{P_{\rm JA},Q_{\rm JA}},
        \mathbf A_i(\theta)
        \right\rangle_F
        +
        e_{n,i,{\rm Fr}}^{P_{\rm JA},Q_{\rm JA}}(\theta).
\end{equation}
Moreover,
\begin{equation}
\label{eq:fresnel_lifted_uniform_error_app}
        \max_{0\le n\le N_r-1}
        \max_{1\le i\le N_d}
        \sup_{\theta\in\Theta}
        |e_{n,i,{\rm Fr}}^{P_{\rm JA},Q_{\rm JA}}(\theta)|
        \le
        \Delta_{P_{\rm JA},Q_{\rm JA}}^{\rm Fr}.
\end{equation}
% ============================================================
\subsection{Physical spherical-wave atom and Fresnel remainder}
\label{app:physical-Fresnel-remainder}
% ============================================================

We now compare the physical atom to the Fresnel atom.  Define
\(\varepsilon:=\tfrac{nd}{r}, ~ c:=\cos\theta,\)
and
\(F(\varepsilon,c):= \sqrt{1-2\varepsilon c+\varepsilon^2}.\)
Then
\(R_n(r,\theta)-r = r\left(F(\varepsilon,\cos\theta)-1\right).\)
For fixed \(c\in[-1,1]\), the Taylor expansion at \(\varepsilon=0\) gives
\[
        F(\varepsilon,c)
        =
        1-\varepsilon c+\tfrac{\varepsilon^2}{2}(1-c^2)
        +
        \mathcal R_3(\varepsilon,c).
\]
Since
\(1-c^2=\sin^2\theta,\)
this is exactly the Fresnel path expansion
\[
        R_n(r,\theta)-r
        =
        -nd\cos\theta
        +
        \tfrac{n^2d^2}{2r}\sin^2\theta
        +
        r\,\mathcal R_3(\varepsilon,\cos\theta).
\]

We now bound the remainder uniformly.  Differentiating
\(F(\varepsilon,c) = (1-2\varepsilon c+\varepsilon^2)^{1/2}\)
with respect to \(\varepsilon\), one obtains
\[
        F'''(\varepsilon,c)
        =
        -3(1-c^2)(\varepsilon-c)
        (1-2\varepsilon c+\varepsilon^2)^{-5/2}.
\]
For
\(0\le \varepsilon\le\varepsilon_\star<1, ~ -1\le c\le1,\)
we have
\[
        1-2\varepsilon c+\varepsilon^2
        \ge
        (1-\varepsilon)^2
        \ge
        (1-\varepsilon_\star)^2,
\]
and
\(|\varepsilon-c|\le1+\varepsilon_\star, ~ 1-c^2\le1.\)
Therefore
\[
        |F'''(\varepsilon,c)|
        \le
        \tfrac{3(1+\varepsilon_\star)}
        {(1-\varepsilon_\star)^5}.
\]
Taylor's theorem gives
\begin{equation}\label{eq:Fresnel_remainder_scalar_bound}
        |\mathcal R_3(\varepsilon,c)|
        \le
        C_{\rm par}(\varepsilon_\star)\varepsilon^3,
        ~
        C_{\rm par}(\varepsilon_\star)
        :=
        \tfrac{1+\varepsilon_\star}
        {2(1-\varepsilon_\star)^5}.
\end{equation}
Consequently,
\begin{equation}\label{eq:path_remainder_bound}
\left|
        R_n(r,\theta)-r
        +
        nd\cos\theta
        -
        \tfrac{n^2d^2}{2r}\sin^2\theta
\right|
\le
        C_{\rm par}(\varepsilon_\star)
        \tfrac{(nd)^3}{r^2}.
\end{equation}
Uniformly over the admissible domain,
\begin{equation}\label{eq:path_remainder_uniform}
        \Delta_{\rm path}
        :=
        C_{\rm par}(\varepsilon_\star)
        \tfrac{(n_{\max}d)^3}{r_{\min}^2}
\end{equation}
satisfies
\[
\left|
        R_n(r_i,\theta)-r_i
        +
        nd\cos\theta
        -
        \tfrac{n^2d^2}{2r_i}\sin^2\theta
\right|
\le
        \Delta_{\rm path}.
\]

% ============================================================
\subsection{Physical atom approximated by the Fresnel Bessel-Vandermonde lift}
\label{app:NF-by-Fresnel-lift}
% ============================================================

Define the physical-Fresnel phase remainder
\[
        \delta_n(r,\theta)
        :=
        R_n(r,\theta)-r
        +
        nd\cos\theta
        -
        \tfrac{n^2d^2}{2r}\sin^2\theta.
\]
Then
\(
        a_{\rm NF}(r,\theta)[n]
        =
        a_{\rm Fr}(r,\theta)[n]
        e^{-ik_{\lambda}\delta_n(r,\theta)}.
\)
Hence
\[
        |a_{\rm NF}(r,\theta)[n]-a_{\rm Fr}(r,\theta)[n]|
        =
        |e^{-ik_{\lambda}\delta_n(r,\theta)}-1|
        \le
        \min\{2,k_{\lambda}|\delta_n(r,\theta)|\}.
\]
Using \eqref{eq:path_remainder_uniform}, define \(\Delta_{\rm par}
        :=
        \min\{2,k_{\lambda}\Delta_{\rm path}\}.\)
Then
\begin{equation}\label{eq:physical_minus_fresnel_bound}
        \max_{0\le n\le N_r-1}
        \max_{1\le i\le N_d}
        \sup_{\theta\in\Theta}
        |a_{\rm NF}(r_i,\theta)[n]
        -
        a_{\rm Fr}(r_i,\theta)[n]|
        \le
        \Delta_{\rm par}.
\end{equation}

Combining \eqref{eq:fresnel_lifted_approximation_app} with
\eqref{eq:physical_minus_fresnel_bound} gives
\begin{equation}\label{eq:NF_lifted_approx}
        a_{\rm NF}(r_i,\theta)[n]
        =
        \left\langle
        \boldsymbol\Phi_{n,{\rm Fr}}^{P_{\rm JA},Q_{\rm JA}},
        \mathbf A_i(\theta)
        \right\rangle_F
        +
        e_{n,i,{\rm NF}}^{P_{\rm JA},Q_{\rm JA}}(\theta),
\end{equation}
where
\begin{equation}\label{eq:Delta_NF_def}
        \max_{0\le n\le N_r-1}
        \max_{1\le i\le N_d}
        \sup_{\theta\in\Theta}
        |e_{n,i,{\rm NF}}^{P_{\rm JA},Q_{\rm JA}}(\theta)|
        \le
        \Delta_{P_{\rm JA},Q_{\rm JA}}^{\rm NF},
        ~
        \Delta_{P_{\rm JA},Q_{\rm JA}}^{\rm NF}
        :=
        \Delta_{P_{\rm JA},Q_{\rm JA}}^{\rm Fr}+\Delta_{\rm par}.
\end{equation}

\begin{rem}[What \(\Delta_{P_{\rm JA},Q_{\rm JA}}^{\rm NF}\) does and does not mean]
The quantity \(\Delta_{P_{\rm JA},Q_{\rm JA}}^{\rm NF}\) is the error of approximating the
physical spherical-wave atom by the \emph{Fresnel} Bessel-Vandermonde lift.
Thus
\(\limsup_{P_{\rm JA},Q_{\rm JA}\to\infty}\Delta_{P_{\rm JA},Q_{\rm JA}}^{\rm NF} \le \Delta_{\rm par}.\)
It does not generally converge to zero as \(P_{\rm JA},Q_{\rm JA}\to\infty\) unless the
paraxial approximation error \(\Delta_{\rm par}\) is also sent to zero.
\end{rem}

% ============================================================
\subsection{Operator error bounds for semi-discrete measures}
\label{app:operator-error-from-atom-error}
% ============================================================

Let a semi-discrete measure be written as
\(
        \nu
        =
        \sum_{i=1}^{N_d}
        \delta_{r_i}\otimes \mu_i,
\)
where each \(\mu_i\) is a finite complex measure on \(\Theta\).  Its total
variation is
\(\|\nu\|_{\TV,\Qset} = \sum_{i=1}^{N_d}\|\mu_i\|_{\TV}.\)
Define the lifted object
\begin{equation}\label{eq:Xnu_def_bv}
    \mathbf X_\nu^{P_{\rm JA},Q_{\rm JA}}
        :=
        \sum_{i=1}^{N_d}
        \int_{\Theta}
        \mathbf A_i(\theta)\,d\mu_i(\theta)
        \in \mathbb C^{N_d\times(2I+1)}.
\end{equation}
Define the finite-harmonic measurement operator
\(
        \mathcal B_{\rm Fr}^{P_{\rm JA},Q_{\rm JA}}:
        \mathbb C^{N_d\times(2I+1)}
        \to
        \mathbb C^{N_r}
\)
by
\begin{equation}\label{eq:Bop_def_bv}
        (\mathcal B_{\rm Fr}^{P_{\rm JA},Q_{\rm JA}}X)[n]
        :=
        \left\langle
        \boldsymbol\Phi_{n,{\rm Fr}}^{P_{\rm JA},Q_{\rm JA}},
        X
        \right\rangle_F,
        ~ n=0,\ldots,N_r-1.
\end{equation}

Recall that \(\FFR\) is the Fresnel measurement operator defined in
\eqref{eq:fresnel_forward_model_sec2}, while \(\FNF\) is the physical
near-field measurement operator defined in
\eqref{eq:physical_near_field_operator}.  From the atomwise bound
\eqref{eq:fresnel_trunc_error_bound},
\[
\begin{aligned}
        |(\FFR\nu)[n]-(\mathcal B_{\rm Fr}^{P_{\rm JA},Q_{\rm JA}}\mathbf X_\nu^{P_{\rm JA},Q_{\rm JA}})[n]|
        &\le
        \sum_{i=1}^{N_d}
        \int_{\Theta}
        |e_{n,i,{\rm Fr}}^{P_{\rm JA},Q_{\rm JA}}(\theta)|\,d|\mu_i|(\theta)\le
        \Delta_{P_{\rm JA},Q_{\rm JA}}^{\rm Fr}
        \|\nu\|_{\TV,\Qset}.
        \end{aligned}
\]
Therefore
\begin{equation}\label{eq:Fresnel_operator_error_bv}
        \|\FFR\nu-\mathcal B_{\rm Fr}^{P_{\rm JA},Q_{\rm JA}}\mathbf X_\nu^{P_{\rm JA},Q_{\rm JA}}\|_2
        \le
        \sqrt{N_r}\,
        \Delta_{P_{\rm JA},Q_{\rm JA}}^{\rm Fr}
        \|\nu\|_{\TV,\Qset}.
\end{equation}
Similarly, using \eqref{eq:Delta_NF_def},
\begin{equation}\label{eq:NF_operator_error_bv}
        \|\FNF\nu-\mathcal B_{\rm Fr}^{P_{\rm JA},Q_{\rm JA}}\mathbf X_\nu^{P_{\rm JA},Q_{\rm JA}}\|_2
        \le
        \sqrt{N_r}\,
        \Delta_{P_{\rm JA},Q_{\rm JA}}^{\rm NF}
        \|\nu\|_{\TV,\Qset}.
\end{equation}

These are the deterministic approximation bounds used in the lifted physical
model.  They separate the harmonic truncation error from the paraxial
spherical-wave error.

% This shows the relation between the exact Fresnel QPAC
% theory and the finite lifted physical reconstruction method.
% ============================================================
% ============================================================
\section{Lifted atomic norm duality and SDP representation}
\label{app:lifted-duality}
% ============================================================

This appendix gives the convex-analytic and semidefinite details for the
finite-harmonic lifted model used in \Cref{sec:lifted_model}.  The model is a
continuous-angle, finite-range atomic norm problem.  Its atoms are indexed by a
range bin \(i\in\{1,\dots,N_d\}\) and a continuous angle
\(\theta\in\Theta\).  We derive the primal and dual problems, the
trigonometric dual-polynomial constraint, the Toeplitz SDP representation, the
dual-contact condition used for support extraction, and the perturbation bound
connecting the finite-harmonic lift to the Fresnel QPAC certificate.

% ============================================================
\subsection{From the tensor Bessel lift to the harmonic lift}
\label{app:lifted-harmonic-aggregation}
% ============================================================

Let \(P_{\rm JA},Q_{\rm JA}\ge0\), and set
\(
        I:=P_{\rm JA}+2Q_{\rm JA},
        ~
        N_h:=2I+1.
\)
The finite Fresnel Bessel-Vandermonde expansion has the form
\[
        a_{\rm Fr}^{P_{\rm JA},Q_{\rm JA}}(r_i,\theta)[n]
        =
        \sum_{\substack{
        |p|\le P_{\rm JA}\\
        |q|\le Q_{\rm JA}}}
        c_{n,i,p,q}^{\rm Fr}e^{i(p+2q)\theta}.
\]
For each harmonic
\(m\in\{-I,\ldots,I\},\)
define the aggregated coefficient
\begin{equation}
\label{eq:agg_coeff_lifted_duality}
        C_{n,i,m}^{\rm Fr,P_{\rm JA},Q_{\rm JA}}
        :=
        \sum_{\substack{|p|\le P_{\rm JA},\ |q|\le Q_{\rm JA}\\ p+2q=m}}
        c_{n,i,p,q}^{\rm Fr},
\end{equation}
with the convention that an empty sum is zero.  Then
\begin{equation}
\label{eq:harmonic_lifted_atom_expansion}
        a_{\rm Fr}^{P_{\rm JA},Q_{\rm JA}}(r_i,\theta)[n]
        =
        \sum_{m=-I}^{I}
        C_{n,i,m}^{\rm Fr,P_{\rm JA},Q_{\rm JA}}e^{im\theta}.
\end{equation}

The centered harmonic vector is
\(
        \mathbf v_I(\theta)
        :=
        \bigl(e^{-iI\theta},e^{-i(I-1)\theta},
        \ldots,e^{iI\theta}\bigr)^T
        \in\mathbb C^{N_h}.
\)
For Toeplitz SDP formulas it is often more convenient to use the one-sided
vector
\(\mathbf b_I(\theta):= \bigl(1,e^{i\theta},\ldots,e^{i(2I)\theta}\bigr)^T.\)
The two conventions are related by
\(\mathbf v_I(\theta)=e^{-iI\theta}\mathbf b_I(\theta).\)
Multiplication by \(e^{-iI\theta}\) has unit modulus.  Hence passing between
the centered and one-sided conventions only changes the representing measure
by a unit-modulus density and does not change total variation.  The main text
uses the centered convention, while the Toeplitz SDP below uses the one-sided
convention.

% ============================================================
\subsection{Lifted atomic norm}
\label{app:lifted-atomic-norm}
% ============================================================

For each range bin \(i\), define the lifted atom
\(
        \mathbf A_i(\theta)
        :=
        \mathbf e_i\mathbf v_I(\theta)^H
        \in\mathbb C^{N_d\times N_h}.
\)
The atomic set is
\(
        \mathcal A
        :=
        \{\mathbf A_i(\theta):i=1,\ldots,N_d,\ \theta\in\Theta\}.
\)
The lifted atomic norm is the gauge
\[
        \|\mathbf X\|_{\mathcal A}
        :=
        \inf
        \left\{
        \sum_s |c_s|:
        \mathbf X=\sum_s c_s\mathbf A_{i_s}(\theta_s)
        \right\}.
\]
Equivalently, if
\(\nu=\sum_{i=1}^{N_d}\delta_{r_i}\otimes\mu_i,\)
then
\(
        \mathbf X_\nu
        =
        \sum_{i=1}^{N_d}
        \int_\Theta
        \mathbf A_i(\theta)\,d\mu_i(\theta),
\)
and
\(
        \|\mathbf X\|_{\mathcal A}
        =
        \inf_{\mu_1,\ldots,\mu_{N_d}}
        \sum_{i=1}^{N_d}\|\mu_i\|_{\rm TV},
\)
where the infimum is taken over all representing measures generating
\(\mathbf X\).
If the angular constraint is imposed on the full circle \(\mathbb T\), the
corresponding one-dimensional trigonometric moment norm has an exact Toeplitz
SDP representation.  For the physical angular set
\(\Theta\subset(0,\pi)\), exact sector-restricted constraints require
additional localizing Toeplitz constraints.  If these localizing constraints
are omitted, one obtains the conservative full-circle formulation, which is
the one stated explicitly below.

% ============================================================
\subsection{Primal and dual lifted problems}
\label{app:lifted-primal-dual}
% ============================================================

The lifted Fresnel operator is
\[
        \bigl(\mathcal B_{\rm Fr}^{P_{\rm JA},Q_{\rm JA}}\mathbf X\bigr)[n]
        :=
        \left\langle
        \boldsymbol\Phi_{n,{\rm Fr}}^{P_{\rm JA},Q_{\rm JA}},
        \mathbf X
        \right\rangle_F,
        ~ n=0,\ldots,N_r-1 .
\]
The noiseless lifted atomic norm program is
\[
        \min_{\mathbf X}
        \|\mathbf X\|_{\mathcal A}
        ~
        \text{subject to}
        ~
        \mathcal B_{\rm Fr}^{P_{\rm JA},Q_{\rm JA}}\mathbf X=\mathbf y.
\]
The robust lifted program is
\[
        \min_{\mathbf X}
        \|\mathbf X\|_{\mathcal A}
        ~
        \text{subject to}
        ~
        \|\mathcal B_{\rm Fr}^{P_{\rm JA},Q_{\rm JA}}\mathbf X-\mathbf y\|_2\le\varepsilon.
\]

For
\(\boldsymbol\lambda\in\mathbb C^{N_r},\)
the Hilbert adjoint of the lifted operator is
\[
        \mathbf Z_{\boldsymbol\lambda}
        :=
        (\mathcal B_{\rm Fr}^{P_{\rm JA},Q_{\rm JA}})^*\boldsymbol\lambda
        =
        \sum_{n=0}^{N_r-1}
        \lambda_n\boldsymbol\Phi_{n,{\rm Fr}}^{P_{\rm JA},Q_{\rm JA}}.
\]
The dual polynomial is
\[
        p_i^{\boldsymbol\lambda}(\theta)
        =
        \left\langle
        \mathbf A_i(\theta),
        \mathbf Z_{\boldsymbol\lambda}
        \right\rangle_F
        =
        \mathbf e_i^H
        \mathbf Z_{\boldsymbol\lambda}
        \mathbf v_I(\theta).
\]
The dual norm is therefore
\[
        \|\mathbf Z_{\boldsymbol\lambda}\|_{\mathcal A}^*
        =
        \sup_{1\le i\le N_d}\sup_{\theta\in\Theta}
        |p_i^{\boldsymbol\lambda}(\theta)|.
\]

The Lagrangian of the noiseless problem is
\[
        \mathcal L(\mathbf X,\boldsymbol\lambda)
        =
        \|\mathbf X\|_{\mathcal A}
        +
        \operatorname{Re}\left\langle
        \boldsymbol\lambda,
        \mathbf y-\mathcal B_{\rm Fr}^{P_{\rm JA},Q_{\rm JA}}\mathbf X
        \right\rangle .
\]
By the adjoint relation and the conjugate-linear-in-the-first convention,
\(
        \left\langle
        \mathcal B_{\rm Fr}^{P_{\rm JA},Q_{\rm JA}}\mathbf X,
        \boldsymbol\lambda
        \right\rangle
        =
        \left\langle
        \mathbf X,
        \mathbf Z_{\boldsymbol\lambda}
        \right\rangle_F .
\)
Equivalently, after conjugating both sides,
\(
        \left\langle
        \boldsymbol\lambda,
        \mathcal B_{\rm Fr}^{P_{\rm JA},Q_{\rm JA}}\mathbf X
        \right\rangle
        =
        \left\langle
        \mathbf Z_{\boldsymbol\lambda},
        \mathbf X
        \right\rangle_F .
\)
Thus the infimum over \(\mathbf X\) is finite if and only if
\(\|\mathbf Z_{\boldsymbol\lambda}\|_{\mathcal A}^*\le1.\)
Thus the noiseless dual is
\begin{equation}
\label{eq:lifted_dual_noiseless_app}
        \max_{\boldsymbol\lambda\in\mathbb C^{N_r}}
        \operatorname{Re}\langle \boldsymbol\lambda,\mathbf y\rangle
        ~
        \text{subject to}
        ~
        \sup_{1\le i\le N_d}\sup_{\theta\in\Theta}|p_i^{\boldsymbol\lambda}(\theta)|\le1.
\end{equation}

For the robust problem, Fenchel duality gives
\[
        \max_{\boldsymbol\lambda\in\mathbb C^{N_r}}
        \operatorname{Re}\langle \boldsymbol\lambda,\mathbf y\rangle
        -
        \varepsilon\|\boldsymbol\lambda\|_2
        ~
        \text{subject to}
        ~
        \sup_{1\le i\le N_d}\sup_{\theta\in\Theta}|p_i^{\boldsymbol\lambda}(\theta)|\le1.
\]
Indeed, the support function of the Euclidean ball
\(\varepsilon\mathbb B_2\) is \(\varepsilon\|\boldsymbol\lambda\|_2\).

% ============================================================
\subsection{Fej\'er-Riesz dual SDP}
\label{app:fejer-riesz-dual-sdp}
% ============================================================

Let \(\mx z\in\mathbb C^{N_h}\), and use the one-sided harmonic vector
\[
        \mathbf b_I(\theta)
        :=
        \bigl(1,e^{i\theta},\ldots,e^{i2I\theta}\bigr)^T
        \in\mathbb C^{N_h}.
\]
For a row polynomial originally written as
\[
        p_i(\theta)
        =
        \mathbf z_i^H\mathbf v_I(\theta),
\]
we multiply by the unit-modulus factor \(e^{iI\theta}\) and apply the
Toeplitz representation to the one-sided polynomial
\[
        e^{iI\theta}p_i(\theta)
        =
        \widetilde{\mathbf z}_i^H\mathbf b_I(\theta),
\]
for the corresponding shifted coefficient vector \(\widetilde{\mathbf z}_i\).
This transformation does not change the modulus constraint.
The full-circle dual constraint
\[
        \sup_{\theta\in\mathbb T}
        |\mx z^H\mathbf b_I(\theta)|\le1
\]
is equivalent to
\[
        1-|\mx z^H\mathbf b_I(\theta)|^2\ge0,
        ~ \theta\in\mathbb T .
\]
By the Fej\'er-Riesz theorem \cite{dumitrescu2017positive}, this nonnegative trigonometric polynomial admits
a positive semidefinite Gram representation.  Equivalently, the following LMI
is necessary and sufficient: there exists a Hermitian matrix
\(\mathbf R\in\mathbb C^{N_h\times N_h}\) such that
\begin{equation}
\label{eq:dual_poly_lmi_block}
        \begin{bmatrix}
       \mx  R&\mx z\\
        \mx z^H&1
        \end{bmatrix}
        \succeq0,
\end{equation}
and
\begin{equation}
\label{eq:dual_poly_lmi_diag_sums}
       \sum_{j=0}^{N_h-1-\ell}R_{j,j+\ell}
        =
        \begin{cases}
        1,& \ell=0,\\
        0,& \ell=1,\ldots,N_h-1.
        \end{cases}
\end{equation}

Indeed, if \eqref{eq:dual_poly_lmi_block} holds, the Schur complement gives
\(\mx R-\mx z\mx z^H\succeq0.\)
For every \(\theta\),
\(\mathbf b_I(\theta)^H(\mx R-\mx z\mx z^H)\mathbf b_I(\theta)\ge0.\)
The diagonal-sum constraints imply
\(\mathbf b_I(\theta)^HR\mathbf b_I(\theta)=1,\)
and therefore
\(|\mx z^H\mathbf b_I(\theta)|\le1.\)
Conversely, if \(|\mx z^H\mathbf b_I(\theta)|\le1\), then
\(1-|\mx z^H\mathbf b_I(\theta)|^2\) is a nonnegative trigonometric polynomial.
Fej\'er-Riesz gives a positive semidefinite matrix \(\mx S\) such that
\[
        1-|\mx z^H\mathbf b_I(\theta)|^2
        =
        \mathbf b_I(\theta)^H S \mathbf b_I(\theta).
\]
Setting
\(\mx R:=\mx S+\mx z\mx z^H\)
gives \eqref{eq:dual_poly_lmi_block} and the diagonal-sum constraints.

Applying this representation independently to each range row gives an exact
SDP representation of the dual constraint when the angular domain is the full
circle \(\mathbb T\).  When the physical angular domain satisfies
\(\Theta\subsetneq\mathbb T\), imposing the same full-circle constraints gives
a conservative SDP relaxation of
\eqref{eq:lifted_dual_noiseless_app}.  An exact SDP representation of the
restricted-angle constraint requires the corresponding localizing Toeplitz
constraints.  The robust full-circle SDP formulation is obtained by replacing
the objective by
\(
    \operatorname{Re}\langle\boldsymbol\lambda,\mathbf y\rangle
    -
    \varepsilon\|\boldsymbol\lambda\|_2.
\)

% ============================================================
\subsection{Proof of dual saturation}
\label{app:proof-dual-saturation}
% ============================================================

\begin{proof}[Proof of \Cref{prop:lifted_dual_saturation}]
Write
\[
        \mathcal B
        :=
        \mathcal B_{\rm Fr}^{P_{\rm JA},Q_{\rm JA}},
        \qquad
        \widehat{\mathbf r}
        :=
        \mathbf y-\mathcal B\widehat{\mathbf X}.
\]
Primal feasibility gives
\(
        \|\widehat{\mathbf r}\|_2\le\varepsilon.
\)
By strong duality,
\[
\begin{aligned}
\|\widehat{\mathbf X}\|_{\mathcal A}
&=
\operatorname{Re}
\langle\widehat{\boldsymbol\lambda},\mathbf y\rangle
-
\varepsilon\|\widehat{\boldsymbol\lambda}\|_2\\
&\le
\operatorname{Re}
\langle\widehat{\boldsymbol\lambda},
\mathbf y-\widehat{\mathbf r}\rangle\\
&=
\operatorname{Re}
\langle\widehat{\boldsymbol\lambda},
\mathcal B\widehat{\mathbf X}\rangle\\
&=
\operatorname{Re}
\left\langle
\mathcal B^*\widehat{\boldsymbol\lambda},
\widehat{\mathbf X}
\right\rangle_F\\
&\le
\|\widehat{\mathbf X}\|_{\mathcal A}.
\end{aligned}
\]
The first inequality follows from
\[
\operatorname{Re}
\langle\widehat{\boldsymbol\lambda},
\widehat{\mathbf r}\rangle
\le
\|\widehat{\boldsymbol\lambda}\|_2
\|\widehat{\mathbf r}\|_2
\le
\varepsilon\|\widehat{\boldsymbol\lambda}\|_2,
\]
and the last follows from dual feasibility. Hence equality holds throughout.

Using the norm-attaining decomposition and the definition of the dual
polynomial,
\[
\sum_{\ell=1}^{\widehat L}|\widehat c_\ell|
=
\sum_{\ell=1}^{\widehat L}
\operatorname{Re}\left(
\widehat c_\ell
\overline{
p_{\widehat i_\ell}^{\widehat{\boldsymbol\lambda}}
(\widehat\theta_\ell)}
\right)
\le
\sum_{\ell=1}^{\widehat L}|\widehat c_\ell|.
\]
Every summand must therefore attain equality. Consequently,
\[
p_{\widehat i_\ell}^{\widehat{\boldsymbol\lambda}}
(\widehat\theta_\ell)
=
\frac{\widehat c_\ell}{|\widehat c_\ell|},
\qquad
\ell=1,\ldots,\widehat L.
\]
The proof is unchanged when the angular atom domain is \(\mathbb T\).
\end{proof}
% ============================================================
\subsection{Extracting \texorpdfstring{\((r,\theta)\)}{(r, theta)} from a lifted solution}\label{app:coefficient-extraction}
% ============================================================

The lifted variable is block-structured by range.  Therefore the range
estimate is discrete:
\(\widehat r=r_i\)
whenever the \(i\)-th range block contains an active angular atom.

There are two standard extraction routes.

\paragraph{Dual peak extraction:}
Given a dual optimizer \(\widehat{\boldsymbol\lambda}\), form
\[
        \widehat{\mathbf Z}
        =
        (\mathcal B_{\rm Fr}^{P_{\rm JA},Q_{\rm JA}})^*
        \widehat{\boldsymbol\lambda},
        ~
                p_i^{\widehat{\boldsymbol\lambda}}(\theta)
        =
        \left\langle
        \mathbf A_i(\theta),
        \widehat{\mathbf Z}
        \right\rangle_F .
\]
Under the assumptions of
\Cref{prop:lifted_dual_saturation}, define the dual contact set
\[
\mathcal C_\star
:=
\left\{
(i,\theta)\in\{1,\ldots,N_d\}\times\Theta:
|p_i^{\widehat{\boldsymbol\lambda}}(\theta)|=1
\right\}.
\]
The support of every norm-attaining primal atomic decomposition is contained
in \(\mathcal C_\star\).
Numerically, one uses the near-contact set
\[
        \widehat{\mathcal C}_{\tau_{\rm sat}}
        =
        \{(i,\theta)\in\{1,\ldots,N_d\}\times\Theta:
        |p_i^{\widehat{\boldsymbol\lambda}}(\theta)|\ge1-\tau_{\rm sat}\}.
\]
The angles may be obtained either by high-resolution peak search with local
refinement or by rooting the trigonometric polynomial
\(1-|p_i^{\widehat{\boldsymbol\lambda}}(\theta)|^2.\)

\paragraph{Toeplitz decomposition route:}
If the standard primal full-circle moment SDP is used, let
\(
        \mathbf T_i\succeq0
\)
denote the Hermitian Toeplitz moment matrix associated with range bin \(i\).
If
\(
        \operatorname{rank}(\mathbf T_i)<N_h,
\)
its finite Vandermonde decomposition has the form
\[
        \mathbf T_i
        =
        \sum_s
        \varrho_{i,s}
        \mathbf b_I(\theta_{i,s})
        \mathbf b_I(\theta_{i,s})^H,
        ~
        \varrho_{i,s}>0.
\]
The angles \(\theta_{i,s}\) may be extracted using annihilating filters, matrix
pencil methods, or Prony-type methods.  Once the support
\(
        \widehat S
        =
        \{(\widehat i_\ell,\widehat\theta_\ell)\}_{\ell=1}^{\widehat L}
\)
has been estimated, define the finite lifted sensing matrix
\[
        \mathbf D_{\widehat S}
        :=
        \left[
        \mathcal B_{\rm Fr}^{P_{\rm JA},Q_{\rm JA}}
        \mathbf A_{\widehat i_1}(\widehat\theta_1)
        \ \cdots\
        \mathcal B_{\rm Fr}^{P_{\rm JA},Q_{\rm JA}}
        \mathbf A_{\widehat i_{\widehat L}}
        (\widehat\theta_{\widehat L})
        \right].
\]
In the noiseless case, if
\(\mathbf D_{\widehat S}\)
has full column rank, then
\(
        \widehat{\mathbf c}
        =
        \mathbf D_{\widehat S}^{\dagger}\mathbf y
        =
        \left(
        \mathbf D_{\widehat S}^{H}
        \mathbf D_{\widehat S}
        \right)^{-1}
        \mathbf D_{\widehat S}^{H}\mathbf y.
\)
Equivalently,
\(
        \mathbf D_{\widehat S}\widehat{\mathbf c}
        =
        \mathbf y.
\)
In the noisy or finite-truncation case,
\(
        \widehat{\mathbf c}
        =
        \mathbf D_{\widehat S}^{\dagger}\mathbf y
\)
is the minimum-norm least-squares amplitude estimate.  If
\(\mathbf D_{\widehat S}\)
has full column rank, then
\(
        \widehat{\mathbf c}
        =
        \left(
        \mathbf D_{\widehat S}^{H}
        \mathbf D_{\widehat S}
        \right)^{-1}
        \mathbf D_{\widehat S}^{H}\mathbf y.
\)

% ============================================================
\subsection{Relation to the Fresnel TV certificate}
\label{app:lifted-relation-fresnel-TV}
% ============================================================

The QPAC theorem proves exact recovery for the semi-discrete Fresnel TV
problem.  The finite lifted atomic norm problem is a finite-dimensional
harmonic approximation of that Fresnel model.  These two objects are connected
but not identical.
The approximation relation is
\(
         a_{\rm Fr}(r_i,\theta)[n]
        =
        \left(
        \mathcal B_{\rm Fr}^{P_{\rm JA},Q_{\rm JA}}
        \mathbf A_i(\theta)
        \right)[n]
        +
        e_{n,i,{\rm Fr}}^{P_{\rm JA},Q_{\rm JA}}(\theta),
\)
with
\[
        \max_{0\le n\le N_r-1}
        \max_{1\le i\le N_d}
        \sup_{\theta\in\Theta}
        |e_{n,i,{\rm Fr}}^{P_{\rm JA},Q_{\rm JA}}(\theta)|
        \le
        \Delta_{P_{\rm JA},Q_{\rm JA}}^{\rm Fr}.
\]
Thus, as \(P_{\rm JA},Q_{\rm JA}\to\infty\), the finite lifted Fresnel model converges uniformly
to the semi-discrete Fresnel model.  At finite \(P_{\rm JA},Q_{\rm JA}\), the lifted SDP should
be viewed as a computable finite-harmonic surrogate unless the discarded
harmonics are exactly zero.
The perturbation estimate in
\Cref{prop:lifted_fresnel_dual_perturbation} gives the precise bridge between
the ideal QPAC certificate and the finite lifted polynomial generated by the
same QPAC vector.  It does not assert that the finite SDP optimizer is equal to
the QPAC vector.  The practical extraction of estimates from the computed SDP
dual polynomial is instead justified by the lifted dual-saturation property in
\Cref{prop:lifted_dual_saturation}.

\end{document}